\documentclass{msc}

\usepackage{amsmath}

\usepackage{amssymb}
\usepackage{booktabs}   
\usepackage{subcaption} 
\usepackage{tikz-cd}
\usepackage{tikz}
\usepackage{graphicx}
\usepackage{marvosym}
\usepackage{okexplain}
\usepackage{tikz}
\usepackage{wrapfig}
\usepackage{array,longtable}
\usepackage{pdfcomment}

\usepackage[shortlabels]{enumitem}

\newcommand\UnicodeCharacter{}

\allowdisplaybreaks

\usepackage{xcolor}
\usepackage{amsmath}
\usepackage{mathtools}
\usepackage{suffix}
\usepackage{amsfonts}
\usepackage{mathpartir}
\usepackage{enumitem}
\usepackage{stmaryrd}
\usepackage[all]{xy}
\usepackage{twoopt}

\usepackage{amsthm}

\numberwithin{equation}{section}

\newcommand\type{\mathrm{type}}
\newcommand\ltype{\mathrm{ltype}}

\newcommand\computationtype[1]{\underline{#1}}
\makeatletter
\newcommand\@KiAlph[1]{%
\ifcase #1\or \kappa\or \kappa'\or \kappa''\else \@ctrerr \fi%
}
\newcommand\ki[1][1]{{\@KiAlph{#1}}}

\newcommand\@TyAlph[1]{%
\ifcase #1\or \tau\or \sigma\or \rho\else \@ctrerr \fi%
}
\newcommand\ty[1][1]{{\@TyAlph{#1}}}

\newcommand\@CTyAlph[1]{%
\computationtype{\ifcase #1\or \tau\or \sigma\or \rho\else \@ctrerr \fi}%
}
\newcommand\cty[1][1]{{\@CTyAlph{#1}}}

\newcommand\tvar[1][1]{{\@TyVarAlph{#1}}}
\newcommand\@TyVarAlph[1]{%
\ifcase #1\or \alpha\or \beta\or \gamma\else \@ctrerr \fi%
}
\newcommand\ltvar[1][1]{{\@LTyVarAlph{#1}}}
\newcommand\@LTyVarAlph[1]{%
\ifcase #1\or \underline{\alpha}\or \underline{\beta}\or \underline{\gamma}\else \@ctrerr \fi%
}
\newcommand\var[1][1]{{\@VarAlph{#1}}}
\newcommand\@VarAlph[1]{%
\ifcase #1\or x\or y\or z\or u\or v\or w\else \@ctrerr \fi%
}

\newcommand\trm[1][1]{{\@TermAlph{#1}}}
\newcommand\@TermAlph[1]{%
\ifcase #1\or t\or s\or r\else \@ctrerr \fi%
}

\newcommand\op{\mathrm{op}}

\makeatother
\newcommand\lop{\mathsf{lop}}
\newcommand\Op{\mathsf{Op}}
\newcommand\Ty{\mathsf{Ty}}
\newcommand\LOp{\mathsf{LOp}}

\newcommand\cnst{\underline{c}}
\newcommand\zero{\underline{0}}
\newcommand\sigmoid{\varsigma}
\newcommand\tSum{\mathrm{sum}}

\newcommand\bp[1]{\boldsymbol{(}#1\boldsymbol{)}}
\newcommand\tUnit{\tTuple{}}

\newcommand\tPair[2]{\langle #1, #2\rangle}

\newcommand\tTriple[3]{\langle #1, #2, #3\rangle}
\newcommand\tTuple[1]{\langle #1\rangle}

\newcommand\Cns{\ell}
\newcommand\Inj[2][\,]{\mathsf{#2}#1}

\newcommand\fun[1]{\lambda #1.}

\newcommand\lfun[1]{\underline{\lambda} #1.}
\newcommand\lapp[2]{#1\bullet #2}
\newcommand\letin[3]{\mathbf{let}\,#1=\,#2\,\mathbf{in}\,#3}
\newcommand\pletin[4]{\letin{\tPair{#1}{#2}}{#3}{#4}}

\newcommand\lfp[2]{\mu#1.#2}
\newcommand\gfp[2]{\nu#1.#2}
\newcommand\llfp[2]{\underline{\mu}#1.#2}
\newcommand\lgfp[2]{\underline{\nu}#1.#2}

\newcommand\tRoll[2][]{\mathbf{roll}_{#1}\,#2}
\newcommand\invtRoll[2][]{{#1}.\mathbf{roll}^{-1}\,#2}
\newcommand\tUnroll[1]{\mathbf{unroll}\,#1}

\newcommand\KI{ \mathfrak{L} }
\newcommand\sind[1]{ \mathfrak{s}_{\left( #1 \right) } }

\newcommand\okCO{ \mathfrak{A} }
\newcommand\okS{ \mathfrak{S} }
\newcommand\TreeIndex{\mathbf{T}}

\newcommand\NNN[1]{ \mathbb{I}_{#1} }
\newcommand\MUNORMAG{\mathfrak{N}}
\newcommand\MUCANG{\mathfrak{n}}
\newcommand\mind[1]{ \mathtt{m}_{\left(#1\right)} }
\newcommand\NORMAL[1]{ \prescript{}{\mathfrak{e}}{\mathcal{N}}\left( {#1}\right)  }

\newcommand\CANsh[1]{ \prescript{\mathfrak{e}}{}{\mathfrak{n}}_{#1} }
\newcommand\sh[1]{ \mathfrak{can}_{#1} }
\newcommand\obdeck[2]{deck_{#1} \left( #2\right)  }
\newcommand\deckk[1]{\mathfrak{deck}_{#1}}
\newcommand{\EuclidU}{\mathsf{U}_e}
\newcommand{\Euclidean}{\mathsf{Diff}}
\newcommand{\Euc}{\mathfrak{E}}

\newcommand\productclosure[1]{\underline{\overline{\times #1}}}
\newcommand\morterminal[1]{!_{#1}}
\newcommand{\Fscone}{\overleftrightarrow{\mathbf{Scone}} }
\newcommand{\Rscone}{\overleftrightarrow{\mathbf{Scone}} }

\newcommand\catLi{\mathbf{Li}}
\newcommand\catFLi{\mathbf{FLi}}
\newcommand\catFSet{\mathbf{FSet}}
\newcommand\catFV{\mathbf{FVect}}

\newcommand\lan{\mathsf{lan}}
\newcommand\ran{\mathsf{ran}}
\newcommand\famX{\mathcal{X}}
\newcommand\famY{\mathcal{Y}}
\newcommand\famZ{\mathcal{Z}}
\newcommand\pvar{{p}}
\newcommand\lvar{\mathsf{v}}
\newcommand\objects{\mathrm{obj} }

\newcommand\ff[1]{ {#1}' }
\newcommand\sobjects[1]{\mathrm{obj}\left( #1 \right)}
\newcommand\repP[1]{ \mathfrak{re}_{#1} }

\newcommand\ceq{\coloneqq }

\newcommand\prodstrict{\,\underline{\times }\, }
\newcommand\equivalenceextensive{\mathcal{S}}
\newcommand\initialobject{\mathbb{0} }
\newcommand\ic{\iota }  

\newcommand\Ds[1]{\mathfrak{D}#1 }
\newcommand\Dsr[1]{\mathfrak{D}^t#1 }

\newcommand\sDsr[1]{\Dsr{\left( #1 \right)} }
\newcommand\sDs[1]{\Ds{\left( #1 \right)} }

\newcommand\GrothSet{\Sigma_\Set}
\newcommand\GrothE{\Sigma_\catE}
\newcommand\GrothD{\Sigma_\catD}
\newcommand\GrothC{\Sigma_\catC}

\newcommand\GrotED{\Sigma_{\catE\times\catD} }

\newcommand\suitL{\ast }
\newcommand\suitHL{\mu\nu\catL }

\newcommand\Produ{\overline{\prodstrict } }
\newcommand\bijpoly{\displaystyle\overline{{\Sigma }} }
\newcommand\inversebijpoly{\displaystyle\overline{{\partial }} }

\newcommand\ind{\mathfrak{in}}
\newcommand\coind{\mathfrak{out}}
\newcommand\unfold{\mathrm{unfold}}
\newcommand\fold{\mathrm{fold}}

\newcommand\MPoly{\mu\mathsf{Poly}}
\newcommand\mnPoly{\mu\nu\mathsf{Poly}}

\newcommand\GrothSpFib{\mathrm{S}\mathfrak{pFib}}
\newcommand\GrothInd{\mathfrak{Ind}}

\newcommand\ie{ \underline{e} }
\newcommand\ienn{ \overline{e} }
\newcommand\iem{e_{\mu\overline{E}}}
\newcommand\ienu{e_{\nu\overline{E}}}

\newcommand\ih{\underline{h} }
\newcommand\ihnu{ \overline{h}  }
\newcommand\ihnXx{h_{\left( X, \nu\overline{H} ^X\right)}}
\newcommand\ihmXx{h_{\left(X, \mu\overline{H} ^X\right)}}

\newcommand\liste{\mathsf{E}}
\newcommand\listee{\mathcal{E}}
\newcommand\evp{\mathsf{ev}_{poly}}

\newcommand\iE{\overline{E}}
\newcommand\iH{\overline{H}}
\newcommand\iJ{\overline{J}}

\newcommand\forgetfulS{{\mathsf{L}}}

\newcommand\AAlg{\textrm{-}\mathrm{Alg}}
\newcommand\CCoAlg{\textrm{-}\mathrm{CoAlg}}

\newcommand{\Fam}[1]{\mathbf{Fam}(#1)}
\newcommand\Vect{\mathbf{Vect}}
\newcommand\FVect{\mathbf{FVect}}
\newcommand\CMon{\mathbf{CMon}}

\newcommand{\forsem}[1]{\overleftrightarrow{\llbracket #1\rrbracket} }

\newcommand{\sforG}[1]{\overleftrightarrow{G}\left( #1 \right)}
\newcommand{\forpi}{\overleftrightarrow{\pi}}
\newcommand{\forG}{\overleftrightarrow{G}}
\newcommand{\forF}{\overleftrightarrow{F}}

\newcommand{\revG}{\overleftrightarrow{G}}

\newcommand{\incLRR}{\mathrm{inc}}
\newcommand{\revR}{\overleftrightarrow{\RR ^n }}

\newcommand{\nonList}[1]{\left[ #1\right] _\ast}
\newcommand\summ{\mathrm{sum}}
\newcommand\productt{\mathrm{product}}

\newcommand\projection{\pi}

\newcommand\tensMatch[5][\,]{\mathbf{case}\,#2\,\mathbf{of}#1{!#3}\otimes{#4}\To#5}

\newcommand\vMatch[3][\,]{\mathbf{case}\,#2\,\mathbf{of}#1\{#3\}}
\newcommand\leftrightvMatch[3][\,]{\mathbf{case}\,#2\,\mathbf{of}#1\left\{#3\right\}}

\newcommand\tFold[3]{\mathbf{fold}\,#1\,\mathbf{with}\,#2\To#3}
\newcommand\tGen[3]{\mathbf{gen\,from}\,#1\,\mathbf{with}\,#2\To#3}

\newcommand\tInl{\mathbf{inl}\,}
\newcommand\tInr{\mathbf{inr}\,}
\newcommand\tFst{\mathbf{fst}\,}
\newcommand\tSnd{\mathbf{snd}\,}
\newcommand\tProj[1]{\mathbf{proj}_{#1}\,}
\newcommand\tCoProj[1]{\mathbf{coproj}_{#1}\,}
\newcommand\idx[2]{\mathbf{idx}(#1; #2)\,}
\newcommand\vGamma{\overline{\Gamma}}

\newcommand\Cat{\mathbf{Cat}}

\newcommand\ctx{\Gamma}
\newcommand\kctx{\Delta}

\newcommand\tinf{\vdash}
\newcommand\DGinf[3][]{\kctx\mid\ctx #1\tinf #2 : #3}
\newcommand\Ginf[3][]{\ctx #1\tinf #2 : #3}
\newcommand\Dinf[3][]{\kctx #1\tinf #2 : #3}

\newcommand\subst[2]{#1{}[#2]}
\newcommand\sfor[2]{^{#2}\!/\!_{#1}}

\newcommand\creals{\underline{\mathbf{real}}}

\newcommand\reals{\mathbf{real}}

\newcommand\Unit{\mathbf{1}}

\WithSuffix\newcommand\t*{\boldsymbol{\mathop{*}}}
\WithSuffix\newcommand\t+{\boldsymbol{\mathop{+}}}

\newcommand\To{\to}

\newcommand\bProd[2]{\bp{#1 \t* #2}}
\newcommand\tProd[3]{\bp{#1 \t* #2 \t* #3}}

\newcommand\Dsynsymbol[1][]{\scalebox{0.8}{$\overrightarrow{\mathcal{D}}$}_{#1}}
\newcommand\Dsyn[2][]{\Dsynsymbol[#1](#2)}

\newcommand\Dsynrevsymbol[1][]{\scalebox{0.8}{$\overleftarrow{\mathcal{D}}$}_{#1}}
\newcommand\Dsynrev[2][]{\Dsynrevsymbol[#1](#2)}

\newcommand\CSyn{{\mathbf{CSyn}}}
\newcommand\LSyn{{\mathbf{LSyn}}}
\newcommand\Syn{\mathbf{Syn}}

\newcommand\tFromMaybe[1]{\mathrm{fromMaybe}}

\newcommand\freeeq[1]{\stackrel{\# #1}{=}}
\newcommand\beeq{\stackrel{\beta\eta}{=}}
\newcommand\bepeq{\!\stackrel{\beta\eta+}{=}\!}

\makeatletter
\newcommand{\pushright}[1]{\ifmeasuring@#1\else\omit$\displaystyle#1$\ignorespaces\fi}
\makeatother

\newcommand\explainr[1]{&\pushright{\color{gray}\scriptsize\{\;\textnormal{#1}\;\}}}

\newcommand\citepappx[1]{\ifx\fossacsversion\undefined Appx.~#1\else\citep[Appx.~#1]{vakar2020reverse}\fi}

\definecolor{shade}{RGB}{223,223,223}
\definecolor{unshade}{RGB}{255,255,255}

\usepackage{tcolorbox}
\newtcbox{\shadebox}{on line,arc=1pt, outer arc=2pt,%
  colback=shade,colframe=shade,boxsep=0pt,%
  left=1pt,right=1pt,top=2pt,bottom=2pt,%
  boxrule=0pt,bottomrule=1pt,toprule=1pt}
\newcommand{\shade}[1]{%
        \shadebox{\ensuremath{#1}}%
}
\newtcbox{\unshadebox}{on line,arc=1pt, outer arc=2pt,%
  colback=unshade,colframe=shade,boxsep=0pt,%
  left=1pt,right=1pt,top=2pt,bottom=2pt,%
  boxrule=0pt,bottomrule=1pt,toprule=1pt}

\newcommand\syncat[1]{\mspace{-25mu}\synname{#1}}
\newcommand\synname[1]{\qquad\text{#1}}

\newenvironment{syntax}[1][]{%
\(
  \begin{array}[t]{#1l@{\quad\!\!}*3{l@{}}@{\,}l}
}{
\end{array}
\)%
}

\newcommand\gdefinedby{::=}
\newcommand\gor{\mathrel{\lvert}}

\newcommand\vor{\mathrel{\big\lvert}}

\newcommand{\multI}{\mathsf{multi}}
\newcommand{\plus}{\mathsf{plus}}
\newcommand{\semt}[1]{\prescript{ }{\Sigma}{\sem{#1}} }
\newcommand{\semtt}[1]{\prescript{t}{\Sigma }{\sem{#1}} }

\newcommand{\Diff}{\mathbf{Diff}}

\makeatletter
\def\MTrightharpoonupfill{%
  \arrowfill@\relbar\relbar\rightharpoonup}
\def\MTleftharpoondownfill{%
  \arrowfill@\leftharpoondown\relbar\relbar}
\def\MTleftharpoonupfill{%
  \arrowfill@\leftharpoonup\relbar\relbar}
\def\MTrightharpoondownfill{%
  \arrowfill@\relbar\relbar\rightharpoondown}
\newcommand*\xhookrightleftharpoons[2][]{\mathrel{%
  \raise.22ex\hbox{%
    $\lhook\joinrel\ext@arrow 0359\MTrightharpoonupfill{\phantom{#1}}{#2}$}%
  \setbox0=\hbox{%
    $\ext@arrow 3095\MTleftharpoondownfill{#1}{\phantom{\lhook\joinrel#2}}$}%
  \kern-\wd0 \lower.22ex\box0}}
\newcommand*\xleftrighthookharpoons[2][]{\mathrel{%
  \raise.22ex\hbox{%
    $\ext@arrow 3095\MTleftharpoonupfill{\phantom{#1\mspace{15mu}}}{#2}$}%
  \setbox0=\hbox{%
    $\mathrel{\raise-.4837ex\hbox{$\lhook$}}\joinrel\ext@arrow 0359\MTrightharpoondownfill{#1}{\phantom{#2}}$}%
  \kern-\wd0 \lower.22ex\box0}}
\makeatother

\newcommand\pair[2]{\parent{#1, #2}}
\newcommand\parent[1]{\left(#1\right)}
\WithSuffix\newcommand\pair-[2]{(#1, #2)}

\usepackage[bbgreekl]{mathbbol}

\newcommand{\lUnit}{\underline{\Unit}}

\newcommand{\Man}{\mathbf{Man}}

\newcommand{\Set}{\mathbf{Set}}

\newcommand\inv[1]{#1^{-1}}
\newcommand\sqinv[1]{#1^{\!\!-1}}
\WithSuffix\newcommand\inv+[1]{\parent{#1}^{-1}}

\newcommand\initial{\mathbb{0}}
\newcommand\terminal{\mathbb{1}}

\newcommand\isomorphic\cong

\newcommand{\sem}[1]{\llbracket #1\rrbracket}

\newcommand{\RR}{\mathbb{R}}
\newcommand{\NN}{\mathbb{N}}
\newcommand\cat[1]{\mathcal{#1}}
\newcommand\catA{\cat{A}}
\newcommand\catB{\cat{B}}
\newcommand\catC{\cat{C}}
\newcommand\catD{\cat{D}}
\newcommand\catE{\cat{E}}
\newcommand\catF{\cat{F}}

\newcommand\catL{\cat{L}}
\newcommand\catS{\cat{S}}

\newcommand\depproj[2]{\mathbf{p}_{#1,#2}}
\newcommand\depv[2]{\mathbf{v}_{#1,#2}}
\newcommand\depq[2]{\mathbf{q}_{#1,#2}}

\newcommand\Hom{\mathrm{Hom}}

\newcommand\Dsemsymbol[1][]{\mathcal{T}_{#1}}
\newcommand\Dsem[2][]{\Dsemsymbol[#1]#2}
\newcommand\Dsemrevsymbol[1][]{\mathcal{T}^*_{#1}}
\newcommand\Dsemrev[2][]{\Dsemrevsymbol[#1]#2}
\newcommand\ev[1][]{\mathrm{ev}^{#1}}
\newcommand\evf[1][]{\mathrm{ev1}^{#1}}
\newcommand\evRsymbol[1][]{\mathrm{evR}^{#1}}
\newcommandtwoopt\evR[3][][]{\evRsymbol[#2]_{#1}(#3)}
\newcommand\lamRsymbol[1][]{\mathrm{lamR}^{#1}}
\newcommandtwoopt\lamR[3][][]{\lamRsymbol[#2]_{#1}(#3)}

\newcommand{\sPair}[2]{( #1, #2 )}

\newcommand\transpose[1]{{#1}^{t}}

\newcommand\copower[3][]{!#2\otimes_{#1}#3}

\newcommand\id[1][]{{\mathrm{id}_{#1}}}

\newcommand\xto\xrightarrow

\makeatletter
\newcommand*\bigcdot{\mathpalette\bigcdot@{.6}}
\newcommand*\bigcdot@[2]{\mathbin{\vcenter{\hbox{\scalebox{#2}{$\m@th#1\bullet$}}}}}
\makeatother

\newcommand\seq[2][]{\left(#2\right)_{#1}}
\newcommand\coseq[2][]{\left[#2\right]_{#1}}
\newcommand\set[1]{\left\{#1\right\}}

\newcommand\proj[1]{\pi_{#1}}
\newcommand\coproj[1]{\iota_{#1}}
\newcommand\inj[1]{\iota_{#1}}

\newcommand\colim{\mathrm{colim}\,}
\renewcommand\lim{\mathrm{lim}\,}

\newcommand\ob[1]{\objects\left( #1\right)}

\newcommand{\defeq}{\stackrel {\mathrm{def}}=}

\usepackage{todonotes}

\iftrue
\newcommand\mvremark[1]{\marginpar{\pdfcomment[color=yellow]{MV: #1}}}
\newcommand\flnremark[1]{\marginpar{\pdfcomment[color=yellow]{FLN: #1}}}

\else
\newcommand\mvremark[1]{}
\newcommand\flnremark[1]{}

\fi

\makeatletter
\RequirePackage{zref}

\zref@newprop{oktheoremfreetext}{\oktheorem@parameter}

\newcommand\OKTheoremAddReferences[2]{
  \expandafter\newcommand\csname#1ref\endcsname[1]{#2~\ref{#1:##1}}
  \expandafter\newcommand\csname#1label\endcsname[1]{\label{#1:##1}}
  \WithSuffix\expandafter\newcommand\csname#1ref\endcsname*[1]{\ref{#1:##1}}
  \WithSuffix\expandafter\newcommand\csname#1label\endcsname+[1]{\hypertarget{#1+:##1}{}\zref@labelbyprops{#1:##1}{oktheoremfreetext}}
  \WithSuffix\expandafter\newcommand\csname#1ref\endcsname+[1]{\hyperlink{#1+:##1}{{{\let\ref\@refstar#2~\zref@extract{#1:##1}{oktheoremfreetext}}}}}
  \WithSuffix\expandafter\newcommand\csname#1ref\endcsname-[1]{\hyperlink{#1+:##1}{{\let\ref\@refstar{\zref@extract{#1:##1}{oktheoremfreetext}}}}}
}
\makeatother
\theoremstyle{definition}
\newtheorem{insight}{Insight}

\newtheorem{thmx}{Theorem}

\newtheorem*{exampleEnv*}{Example}
\newtheorem*{exampleEnv+}{Example \oktheorem@parameter}
\newenvironment{example*}{\begin{exampleEnv*}}{\qed\end{exampleEnv*}}
\newenvironment{example+}[1]{\def\oktheorem@parameter{#1}\begin{exampleEnv+}}{\qed\end{exampleEnv+}}
\OKTheoremAddReferences{example}{Example}

\newcommand\sqsubsection[1]{\vspace{-4pt}\subsection{#1}\vspace{-3pt}}

  \newcommand\LDomain[1]{\mathbf{LDom}(#1)}
  \newcommand\LCDomain[1]{\mathbf{CDom}(#1)}

\newcommand\ct[1]{\underline{#1}}
\newcommand\cRR{\ct{\RR}}

\newcommand{\deptrans}[1]{{#1}^\dagger}

\newcommand{\lintrans}[1]{{#1}^\ddagger}

\newcommand{\ovplus}{\vee}

\makeatletter
\renewcommand\tableofcontents{%
	    \section*{\contentsname
		\@mkboth{%
		   \MakeUppercase\contentsname}{\MakeUppercase\contentsname}}%
	    \@starttoc{toc}%
      }
      \renewcommand*\l@part[2]{%
      \ifnum \c@tocdepth >-2\relax
        \addpenalty\@secpenalty
        \addvspace{2.25em \@plus\p@}%
        \setlength\@tempdima{3em}%
        \begingroup
          \parindent \z@ \rightskip \@pnumwidth
          \parfillskip -\@pnumwidth
          {\leavevmode
           \large \bfseries #1\hfil
           \hb@xt@\@pnumwidth{\hss #2%
                              \kern-\p@\kern\p@}}\par
           \nobreak
           \if@compatibility
             \global\@nobreaktrue
             \everypar{\global\@nobreakfalse\everypar{}}%
          \fi
        \endgroup
      \fi}
    \newcommand*\l@section[2]{%
      \ifnum \c@tocdepth >\z@
        \addpenalty\@secpenalty
        \addvspace{1.0em \@plus\p@}%
        \setlength\@tempdima{1.5em}%
        \begingroup
          \parindent \z@ \rightskip \@pnumwidth
          \parfillskip -\@pnumwidth
          \leavevmode \bfseries
          \advance\leftskip\@tempdima
          \hskip -\leftskip
          #1\nobreak\hfil
          \nobreak\hb@xt@\@pnumwidth{\hss #2%
                                     \kern-\p@\kern\p@}\par
        \endgroup
      \fi}
    \newcommand*\l@subsection{\@dottedtocline{2}{1.5em}{2.3em}}
    \renewcommand*\l@subsubsection{\@dottedtocline{3}{3.8em}{3.2em}}
    \renewcommand*\l@paragraph{\@dottedtocline{4}{7.0em}{4.1em}}
    \newcommand*\l@subparagraph{\@dottedtocline{5}{10em}{5em}}
    
\makeatother

\newcommand\ndots{\cdot\cdot\cdot\cdot\cdot\cdot}

\newcommand\nndots{\ndots\ndots\ndots} 

\begin{document}

\lefttitle{CHAD for Expressive Total Languages}
\righttitle{CHAD for Expressive Total Languages}

\papertitle{Article}

\jnlPage{1}{00}
\jnlDoiYr{2019}
\doival{10.1017/xxxxx}

\title{CHAD for Expressive Total Languages}

\begin{authgrp}
  \author{Fernando Lucatelli Nunes}
  \author{ Matthijs V\'ak\'ar}
  \affiliation{
    {Department of Information and Computing Sciences}\\           
    {Utrecht University}\\           
    {Netherlands}                    
  }
\end{authgrp}

\history{(Received xx xxx xxx; revised xx xxx xxx; accepted xx xxx xxx)}

\begin{abstract}
  We show how to apply forward and reverse mode Combinatory Homomorphic Automatic Differentiation (CHAD) \citep{vakar2021chad,vakar2020reverse}
  to total functional programming languages with expressive type systems featuring the combination of
  \begin{itemize}
\item tuple types;
\item sum types;
\item inductive types;
\item coinductive types;
\item function types.
  \end{itemize}
  We achieve this by analysing the categorical semantics of such types
  in $\Sigma$-types (Grothendieck constructions) of suitable categories.
  Using a novel categorical logical relations technique for such expressive type 
  systems,
  we give a correctness proof of CHAD in this setting by showing that it computes the usual mathematical 
  derivative of the function that the original program implements. 
  The result is a principled, purely functional and provably correct method for performing 
  forward and reverse mode automatic differentiation (AD) on total functional programming languages with expressive type systems.
\end{abstract}


\maketitle

\clearpage

\tableofcontents
\clearpage

\section{Introduction}\label{sec:introduction}
Automatic differentiation (AD) is a popular technique for computing 
derivatives of functions implemented by computer programs,
essentially by applying the chain rule across the program code.
It is typically the method of choice for computing derivatives in
machine learning and scientific computing because of its efficiency and numerical stability.
AD has two main variants: forward mode AD, which calculates the derivative of a function, and reverse mode AD,
which calculates the (matrix) transpose of the derivative.
Roughly speaking, for a function $f:\RR^n\to \RR^m$, reverse mode is the more efficient technique if 
$n\gg m$ and forward mode is if $n\ll m$.
Seeing that we are usually interested in computing derivatives (or gradients) of functions $f:\RR^n\to\RR$ 
with very large $n$, reverse AD tends to be the more important algorithm in practice \citep{baydin2018automatic}.

While the study of AD has a long history in the numerical methods community, which we will not survey (see, for example, \citep{griewank2008evaluating}), there has recently been a proliferation of work by the programming languages community 
examining the technique from a new angle.
New goals pursued by this community include 
\begin{itemize}
\item giving a \emph{concise, clear and easy-to-implement definition} of various AD algorithms;
\item \emph{expanding the languages} and programming techniques that AD can be applied to;
\item relating AD to its mathematical \emph{foundations} in differential geometry and proving that AD implementations \emph{correctly} calculate derivatives;
\item performing AD at \emph{compile time} through \emph{source-code transformation}, to maximally expose optimization opportunities to the compiler and to avoid interpreter overhead that other AD approaches can incur;
\item providing formal \emph{complexity guarantees} for AD implementations.
\end{itemize}
We provide a brief summary of some of this more recent work in section \ref{sec:related-work}.
The present paper adds to this new body of work by advancing the state of the art of the first four goals.
We leave the fifth goal when applied to our technique mostly to future work (with the exception of Corollary \ref{cor:no-code-blow-up}).
Specifically, we extend the scope of the Combinatory Homomorphic Automatic Differentiation (CHAD) method of forward and reverse AD \citep{vakar2021chad,vakar2020reverse}  (from the previous state of the art: a simply typed $\lambda$-calculus) to apply to total functional programming languages with expressive type systems, i.e. the combination of:
\begin{itemize}
\item \emph{tuple types}, to enable programs that return or take as an argument more than one value;
\item \emph{sum types}, to enable programs that define and branch on variant data types;
\item \emph{inductive types}, to include programs that operate on labeled-tree-like data structures;
\item \emph{coinductive types}, to deal with programs that operate on lazy infinite data structures such as streams;
\item \emph{function types}, to encompass programs that use popular higher order programming idioms such as maps and folds.
\end{itemize}
This conceptually simple extension requires a considerable extension of existing techniques in 
denotational semantics.
The pay-offs of this challenging development are surprisingly simple AD algorithms as well as 
reusable abstract semantic techniques.

The main contributions of this paper are:
\begin{itemize}
\item developing an abstract categorical semantics (\S \ref{sec:background-categorical-semantics}) of such expressive type systems in 
suitable $\Sigma$-types of categories (\S \ref{sec:grothendieck-constructions});
\item presenting, as the initial instantiation of this abstract semantics,  an idealised target language for CHAD when applied to such type systems (\S \ref{sec:target-language});
\item deriving the forward and reverse CHAD algorithms  (\S \ref{sec:AD-transformations}) when applied to expressive type systems as the uniquely defined homomorphic functors (\S \ref{sect:structure-preserving-functors}) from the source (\S \ref{sec:source-language}) to the target language (\S \ref{sec:target-language});
\item introducing (categorical) logical relations techniques (aka sconing) for reasoning about expressive functional languages that include 
both inductive and coinductive types (\S \ref{sec:subsconing});
\item using such a logical relations construction over the concrete denotational semantics  (\S \ref{sec:concrete-semantics-specification}) of the source and target languages (\S \ref{sec:concrete-models}) that demonstrates that CHAD correctly calculates the usual mathematical derivative (\S \ref{sec:logical-relations-argument}), even for programs between inductive types (\S \ref{sect:correctness-inductive-types});
\item discussing examples (\S \ref{sect:EXAMPLE}) and applied considerations around implementing this extended CHAD method in practice (\S \ref{sec:practical-considerations}). 
\end{itemize}  

We start by giving a high-level 
overview of the key insights and theorems in this paper 
in \S\ref{sec:key-ideas}.

\section{Key ideas}\label{sec:key-ideas}

\subsection{Origins in semantic derivatives and chain rules}
CHAD starts from the observation that for a differentiable function 
$$f: \RR^n\to \RR^m$$
it is useful to pair the primal function value $f(x)$ with $f$'s derivative $Df(x)$ at $x$ 
if we want to calculate derivatives in a compositional way (where we underline the spaces $\cRR^n$ of tangent vectors to emphasize their algebraic structure and we write a linear 
function type for the derivative to indicate its linearity in its tangent vector argument):
\begin{align*}
\Dsem{f} : &\RR^n \to \RR^m\times (\cRR^n\multimap \cRR^m)\\
&x\mapsto \sPair{f(x)}{Df(x)}.
\end{align*}
Indeed, the chain rule for derivatives teaches us that we compute the derivative 
of a composition $g\circ f$ of functions as follows, where we write $\Dsem[1]{f}\defeq \pi_1\circ \Dsem{f}$ and $\Dsem[2]{f}\defeq \pi_2\circ \Dsem{f}$ for the first and second components of $\Dsem{f}$, respectively:
$$
\Dsem{(g\circ f)}(x) = \sPair{\Dsem[1]{g}(\Dsem[1]{f}(x))}{\Dsem[2]{g}(\Dsem[1]{f}(x))\circ \Dsem[2]{f}(x)}.
$$
We make two observations:
\begin{enumerate}
\item the derivative of $g\circ f$ does not only depend on the derivatives of $g$ and $f$ but also on the primal value of $f$;
\item the primal value of $f$ is used twice: once in the primal value of $g\circ f $ and once in its derivative; we want to share these repeated subcomputations.
\end{enumerate}
\begin{insight}
This shows that it is wise to \emph{pair up} computations of primal function values and derivatives and to \emph{share} computation between both if we  want to calculate derivatives of functions compositionally and efficiently.
\end{insight}

Similar observations can be made for $f$'s transposed (adjoint) derivative $\transpose{Df}$, which propagates not tangent vectors but cotangent vectors and which we can pair up as 
\begin{align*}
    \Dsemrev{f} : &\RR^n \to \RR^m\times (\cRR^m\multimap \cRR^n)\\
    &x\mapsto \sPair{f(x)}{\transpose{Df}(x)}
    \end{align*}
to get the following chain rule
$$
\Dsemrev{(g\circ f)}(x) = \sPair{\Dsemrev[1]{g}(\Dsemrev[1]{f}(x))}{\Dsemrev[2]{f}(x)\circ\Dsemrev[2]{g}(\Dsemrev[1]{f}(x))}.
$$

CHAD directly implements the operations $\Dsemsymbol$ and $\Dsemrevsymbol$ as source code transformations $\Dsynsymbol$ and $\Dsynrevsymbol$ on a functional language to implement forward and reverse mode AD, respectively.
These code transformations are defined compositionally through structural induction on the syntax, by making use of the chain rules above.

\subsection{CHAD on a first-order functional language}
We first discuss what the technique looks like on a standard typed first-order functional language.
Despite our different presentation in terms of a $\lambda$-calculus rather than Elliott's categorical combinators, this is essentially the algorithm of \citep{elliott2018simple}.
Types $\ty,\ty[2],\ty[3]$ are either statically sized arrays of $n$ real numbers $\reals^n$ or tuples $\ty\t*\ty[2]$ of  types $\ty,\ty[2]$.
We consider programs $\trm$ of type $\ty[2]$ in typing context $\Gamma=\var_1:\ty_1,\ldots,\var_n:\ty_n$, where $\var_i$ are identifiers.
We write such a typing judgement for programs in context as $\Gamma\vdash\trm:\ty[2]$.
As long as our language has certain primitive operations (which we represent schematically)
$$
\inferrule{\Ginf{\trm_1}{\reals^{n_1}}\quad\cdots\quad \Ginf{\trm_k}{\reals^{n_k}}}{\Ginf{\op(\trm_1,\ldots,\trm_k)}{\reals^m}}
$$
such as constants (as nullary operations), (elementwise) addition and multiplication of arrays,
inner products and certain non-linear functions like sigmoid functions,  we can write complex programs by sequencing together such operations.
For example, writing $\reals$ for $\reals^1$, we can write a program $\var_1:\reals,\var_2:\reals,\var_3:\reals,\var_4:\reals\vdash \trm[2]:\reals$ by 
\begin{align*}
&\letin{\var[2]}{\var_1 * \var_4 + 2 * x_2 }{}\\
&\letin{\var[3]}{\var[2]* \var_3}{}\\
&\letin{w}{\var[3]+ \var_4}{\sin(w)},
\end{align*}
where we indicate shared subcomputations with $\mathbf{let}$-bindings.

CHAD observes that we can define for each language type $\ty$ associated types of
\begin{itemize}
    \item forward mode primal values $\Dsyn{\ty}_1$;\\
    we define $\Dsyn{\reals^n}=\reals^n$ and $\Dsyn{\ty\t*\ty[2]}_1=\Dsyn{\ty}_1\t*\Dsyn{\ty[2]}_1$; that is, for now $\Dsyn{\ty}_1=\ty$; 
    \item reverse mode primal values $\Dsynrev{\ty}_1$;\\ we define $\Dsynrev{\reals^n}=\reals^n$ and $\Dsynrev{\ty\t*\ty[2]}_1=\Dsynrev{\ty}_1\t*\Dsynrev{\ty[2]}_1$; that is, for now $\Dsynrev{\ty}_1=\ty$; 
    \item forward mode tangent values $\Dsyn{\ty}_2$;\\
    we define $\Dsyn{\reals^n}_2=\creals^n$ and $\Dsyn{\ty\t*\ty[2]}=\Dsyn{\ty}_2\t*\Dsyn{\ty[2]}_2$;
    \item reverse mode cotangent values $\Dsynrev{\ty}_2$;\\
    we define $\Dsynrev{\reals^n}_2=\creals^n$ and $\Dsynrev{\ty\t*\ty[2]}=\Dsynrev{\ty}_2\t*\Dsynrev{\ty[2]}_2$.
\end{itemize}
Indeed, the justification for these definitions is the crucial observation that \emph{a (co)tangent vector to a product of spaces is precisely a pair of tangent (co)vectors to the two spaces}.
Put differently, the space $\Dsem[(x,y)]{(X\times Y)}$ of (co)tangent vectors to $X\times Y$ at a point $(x,y)$ equals the product space $(\Dsem[x] X) \times (\Dsem[y] Y)$  \citep{tu2011manifolds}. 

We write the (co)tangent types associated with $\reals^n$ as $\creals^n$ to emphasize that it is a linear type and to 
distinguish it from the cartesian type $\reals^n$.
In particular, we will see that tangent and cotangent values are elements of linear types that come equipped with a commutative monoid structure $(\zero,+)$.
Indeed, (transposed) derivatives are linear functions: homomorphisms of this monoid structure\footnote{In fact, the (co)tangent vectors form a 
vector space and (transposed) derivatives are vector space homomorphisms. Surprisingly, it is only the monoid structure that is relevant to phrasing and proving correct 
CHAD. Therefore, we choose to emphasize this monoid structure over the full vector space structure.
For example, CHAD-like algorithms also works for more general data types than the real numbers, as long as they form a commutative monoid.
An interesting example is a datatype that implements saturation arithmetic, as is commonly used as a cheap alternative to floating point arithmetic in machine learning.
}. 
We extend these operations $\Dsynsymbol$ and $\Dsynrevsymbol$ to act on typing contexts $\Gamma$:
\begin{align*}
\Dsyn{\var_1:\ty_1,\ldots,\var_n:\ty_n}_1&=\var_1:\Dsyn{\ty_1}_1,\ldots, \var_n:\Dsyn{\ty_n}_1 \\ 
\Dsynrev{\var_1:\ty_1,\ldots,\var_n:\ty_n}_1&=\var_1:\Dsynrev{\ty_1}_1,\ldots, \var_n:\Dsynrev{\ty_n}_1 \\
\Dsyn{\var_1:\ty_1,\ldots,\var_n:\ty_n}_2&=\Dsyn{\ty_1}_2\t*\cdots\t*\Dsyn{\ty_n}_2\\
\Dsynrev{\var_1:\ty_1,\ldots,\var_n:\ty_n}_2&=\Dsynrev{\ty_1}_2\t*\cdots\t*\Dsynrev{\ty_n}_2.
\end{align*}
To each program $\Gamma\vdash\trm:\ty[2]$, CHAD associates programs calculating the forward mode  and reverse mode derivatives $\Dsyn[\vGamma]{\trm}$ and $\Dsynrev[\vGamma]{\trm}$, which are indexed by the list $\vGamma$ of identifiers that occur in $\Gamma$:
\begin{align*}
&\Dsyn{\Gamma}_1\vdash \Dsyn[\vGamma]{\trm} : \Dsyn{\ty[2]}\t* \left( \Dsyn{\Gamma}_2\multimap \Dsyn{\ty[2]}\right)\\
&\Dsynrev{\Gamma}_1\vdash \Dsynrev[\vGamma]{\trm} : \Dsynrev{\ty[2]}\t* \left( \Dsynrev{\ty[2]} \multimap\Dsynrev{\Gamma}_2 \right).
\end{align*}
Observing that each program $\trm$ computes a differentiable function $\sem{\trm}$ between Euclidean spaces, as long as all primitive operations $\op$ are differentiable, the key property that we prove for these code transformations is that they actually calculate derivatives:
\begin{thmx}[Correctness of CHAD, Theorem \ref{theorem:correctness-theorem-for-data-typess}]\label{thm:chad-correctness}
For any well-typed program $$\var_1:\reals^{n_1},\ldots,\var_k:\reals^{n_k}\vdash {\trm}:\reals^m$$
we have that
$\sem{\Dsyn[\var_1,\ldots,\var_k]{\trm}}=\Dsem{\sem{\trm}}\;\text{ and }\;\sem{\Dsynrev[\var_1,\ldots,\var_k]{\trm}}=\Dsemrev{\sem{\trm}}.$
\end{thmx}
Once we fix the semantics for the source and target languages, we can show that this theorem holds if we define $\Dsynsymbol$ and $\Dsynrevsymbol$ on programs using the chain rule.
The proof works by plain induction on the syntax.
For example, we can correctly define reverse mode CHAD on a first-order language as follows:
\begin{flalign*}
    &\Dsynrev[\vGamma]{\op(\trm_1,\ldots,\trm_k)} \defeq && \pletin{\var_1}{\var_1'}{\Dsynrev[\vGamma]{\trm_1}}{\cdots\\
    &&& \pletin{\var_k}{\var_k'}{\Dsynrev[\vGamma]{\trm_k}}{\\
    &&&\tPair{\op(\var_1,\ldots,\var_k)}{\lfun\lvar \letin{\lvar}{\transpose{D\op}(\var_1,\ldots,\var_k;\lvar)}{\\
    &&&\phantom{\tPair{\op(\var_1,\ldots,\var_k)}{\lfun\lvar }}\lapp{\var_1'}{\tProj{1}{\lvar}}+\cdots+\lapp{\var_k'}{\tProj{k}{\lvar}}}}}}
\end{flalign*}
\begin{flalign*}
&\Dsynrev[\vGamma]{\var} \defeq &&  \tPair{\var}{\lfun{\lvar} \tCoProj{\idx{\var}{\vGamma}}(\lvar)}
\end{flalign*}
\begin{flalign*}
    &
\Dsynrev[\vGamma]{\letin{\var}{\trm}{\trm[2]}}  
\defeq &&
\pletin{\var}{\var'}{\Dsynrev[\vGamma]{\trm}}{\\ &&&
    \pletin{\var[2]}{\var[2]'}{\Dsynrev[\vGamma,\var]{\trm[2]}}{\\ &&&
        \tPair{\var[2]}{\lfun\lvar 
        \letin{\lvar}{\lapp{\var[2]'}{\lvar}}{
            \tFst\lvar+\lapp{\var'}{(\tSnd \lvar)}
        }}
    }}
\end{flalign*}
\begin{flalign*}&
\Dsynrev[\vGamma]{\tPair{\trm}{\trm[2]}} \defeq && 
\pletin{\var}{\var'}{\Dsynrev[\vGamma]{\trm}}{ \\ &&&
\pletin{\var[2]}{\var[2]'}{\Dsynrev[\vGamma]{\trm[2]}}{\\ &&&
\tPair{\tPair{\var}{\var[2]}}{\lfun\lvar \lapp{\var'}{(\tFst\lvar)} + {\lapp{\var[2]'}{(\tSnd \lvar)}}}}}
\end{flalign*}
\begin{flalign*}&
\Dsynrev[\vGamma]{\tFst\trm} \defeq && 
\pletin{\var}{\var'}{\Dsynrev[\vGamma]{\trm}}
{\tPair{\tFst\var}{\lfun\lvar \lapp{\var'}{\tPair{\lvar}{\zero}}}}
\end{flalign*}
\begin{flalign*}&
\Dsynrev[\vGamma]{\tSnd\trm} \defeq && 
\pletin{\var}{\var'}{\Dsynrev[\vGamma]{\trm}}
{\tPair{\tSnd\var}{\lfun\lvar \lapp{\var'}{\tPair{\zero}{\lvar}}}}
\end{flalign*}
Here, we write $\lfun\lvar\trm$ for a linear function abstraction (merely a notational convention -- it can simply be thought of as a plain function abstraction) and $\lapp{\trm}{\trm[2]}$ for a linear function application (which again can be thought of as a plain function application).
Furthermore,
given
$\Gamma;\lvar:\cty\vdash \trm:\tProd{\cty[2]_1}{\cdots}{\cty[2]_n}$, we write $\Gamma;\lvar:\cty\vdash \tProj{i}(\trm):\cty[2]_i$ for the 
$i$-th projection of $\trm$.
Similarly,
given $\Gamma;\lvar:\cty\vdash \trm:\cty[2]_i$, we write the $i$-th coprojection $\Gamma;\lvar:\cty\vdash 
\tCoProj{i}(\trm)= \tTuple{\zero,\ldots,\zero,\trm,\zero,\ldots,\zero}
:\tProd{\cty[2]_1}{\cdots}{\cty[2]_n}$ and we write $\idx{\var_i}{\var_1,\ldots,\var_n}=i$ for the index of an identifier in a list of identifiers.  
Finally, $\transpose{D\op}$ here is a linear operation that implements the transposed derivative of the primitive operation $\op$.

Note, in particular, that CHAD pairs up primal and (co)tangent values and shares common subcomputations. 
We see that what CHAD achieves is a compositional efficient reverse mode AD algorithm that computes the (transposed) derivatives of a composite program in terms of the (transposed) derivatives $\transpose{D\op}$ of the basic building blocks $\op$.

\subsection{CHAD on a higher-order language: a categorical perspective saves the day}
So far, this account of CHAD has been smooth sailing:
we can simply follow the usual mathematics of (transposed) derivatives of functions $\RR^n\to \RR^m$ 
and implement it in code. 
A challenge arises when trying to extend the algorithm to more expressive languages with features that do not have an obvious counterpart in multivariate calculus, like higher-order functions.

\citep{vakar2021chad,vakar2020reverse} solve this problem by observing that 
we can understand CHAD through the categorical structure of Grothendieck constructions (aka $\Sigma$-types of categories).
In particular, they observe that the syntactic category of the target language for CHAD, a language with both cartesian and linear types, forms a locally indexed category $\LSyn:\CSyn^{op}\to \Cat$, i.e. functor to the category of categories and functors for which $\ob\LSyn(\ty)=\ob\LSyn(\ty[2])$ for all $\ty,\ty[2]\in\ob\CSyn$ and $\LSyn(\ty\xto{\trm}\ty[2]):\LSyn(\ty[2])\to\LSyn(\ty)$ is identity on objects.
Here, $\CSyn $ is the syntactic category whose objects are cartesian types $\ty,\ty[2],\ty[3]$ and morphisms $\ty\to \ty[2]$ are programs $\var:\ty\vdash \trm:\ty[2]$, up to a standard program equivalence.
Similarly, $\LSyn(\ty)$ is the syntactic category whose objects are linear types $\cty,\cty[2],\cty[3]$ and morphisms $\cty[2]\to\cty[3]$ are programs $\var:\ty;\lvar:\cty[2]\vdash \trm:\cty[3]$ of type $\cty[3]$ that have a free variable $\var$ of cartesian type $\ty$ and a free variable $\lvar$ of linear type $\cty[2]$.
The key observation then is the following.
\begin{thmx}[CHAD from a universal property, Corollary \ref{cor:CHAD-definition}] \label{thm:chad-universal-property}
Forward and reverse mode CHAD are the unique structure preserving functors
\begin{align*}
    &\Dsyn{-}:\Syn\to \Sigma_{\CSyn}\LSyn\\
    &\Dsynrev{-}:\Syn\to \Sigma_{\CSyn}\LSyn^{op}
\end{align*}
from the syntactic category $\Syn$ of the source language to (opposite) Grothendieck construction of the target language $\LSyn:\CSyn^{op}\to \Cat$
that send primitive operations $\op$ to their derivative $D\op$ and transposed 
derivative $\transpose{D\op}$, respectively.
\end{thmx}
In particular, they prove that this is true for the unambiguous definitions of CHAD for a source language that is the first-order functional language we have considered above, which we can see as the freely generated category $\Syn$ with finite products, generated by the objects $\reals^n$ and morphisms $\op$.
That is, for this limited language, ``structure preserving functor'' should be interpreted as ``finite product preserving functor''.

This leads \citep{vakar2021chad,vakar2020reverse} to the idea to try to use 
Theorem \ref{thm:chad-universal-property} as a definition of CHAD on more expressive programming languages.
In particular, they consider a higher-order functional source language $\Syn$, i.e. the freely generated cartesian closed category on the objects $\reals^n$ and morphisms $\op$ and try to define $\Dsyn{-}$ and $\Dsynrev{-}$ as the (unique)
structure preserving (meaning: cartesian closed) functors to $\Sigma_{\CSyn}\LSyn$ and $\Sigma_{\CSyn}\LSyn^{op}$ for a suitable linear target language $\LSyn:\CSyn^{op}\to \Cat$.
The main contribution then is to identify conditions on a locally indexed category
$\catL:\catC^{op}\to \Cat$ that guarantee that $\Sigma_{\catC}\catL$ and $\Sigma_{\catC}\catL^{op}$ are cartesian closed and to take the target language $\LSyn:\CSyn^{op}\to \Cat$ as a freely generated such category.
\begin{insight}
To understand how to perform CHAD on a source 
language with language feature $X$ (e.g., higher-order functions),
we need to understand the categorical semantics of language feature $X$ (e.g., categorical exponentials) in 
categories of the form $\Sigma_\catC\catL$ and $\Sigma_\catC\catL^{op}$.
Giving sufficient conditions on $\catL$ for such a semantics to exist yields a 
suitable target language for CHAD, with the definition of the algorithm 
falling from the universal property of the source language.
\end{insight}

Furthermore, we observe in these papers that Theorem \ref{thm:chad-correctness} again holds 
for this extended definition of CHAD on higher-order languages.
However, to prove this, plain induction no longer suffices and we instead need to use a logical relations construction over the semantics (in the form of categorical sconing) that relates differentiable curves to their associated primal and (co)tangent curves.
This is necessary because the program $\trm$ may use higher-order constructions such as $\lambda$-abstractions and function applications in its definition, even if the input and output types are plain first-order types that implement some Euclidean space.
\begin{insight}
To obtain a correctness proof of CHAD on source languages with language 
feature $X$, it suffices to give a concrete denotational semantics for the 
source and target languages as well as a categorical semantics of language 
feature $X$ in a category of logical relations (a scone) over these concrete semantics.
The main technical challenge is to analyse logical relations techniques for language 
feature $X$.    
\end{insight}

Finally, these papers observe that the resulting target language can be implemented as a shallowly embedded DSL in standard functional languages, using a module system to implement the required linear types as abstract types, with a reference Haskell implementation available at \url{https://github.com/VMatthijs/CHAD}.
In fact, \citep{vytiniotis2019differentiable} had proposed the same CHAD algorithm for higher-order languages, arriving at it from practical considerations rather than abstract categorical observations.
\begin{insight}
The code generated by CHAD naturally comes equipped with very precise (e.g., linear) types. These types
emphasize the connections to its mathematical foundations and provide scaffolding for its correctness proof.
However, they are unnecessary for a practical implementation of the algorithm: CHAD can be made to generate
 standard functional (e.g., Haskell) code; the type safety can even be rescued by implementing the linear types 
as abstract types. 
\end{insight}

\subsection{CHAD for sum types: a challenge -- (co)tangent spaces of varying dimension}
A natural approach, therefore, when extending CHAD to yet more expressive source languages is to try to use Theorem \ref{thm:chad-universal-property} as a definition.
In the case of sum types (aka variant types), therefore, we should consider their categorical equivalent, distributive coproducts, and seek conditions on $\catL:\catC^{op}\to\Cat$ under which $\Sigma_{\catC}\catL$ and $\Sigma_{\catC}\catL^{op}$ have distributive coproducts.
The difficulty is that these categories tend not to have coproducts if $\catL$ is locally indexed.
Instead, the desire to have coproducts in $\Sigma_{\catC}\catL$ and $\Sigma_{\catC}\catL^{op}$ naturally leads us to consider more general strictly indexed categories $\catL:\catC^{op}\to\Cat$.

In fact, this is compatible with what we know from differential geometry \citep{tu2011manifolds}:
coproducts allow us to construct spaces with multiple connected components, each of which may have a distinct dimension.
To make things concrete: the space $\Dsem[x]{(\RR^2 \sqcup \RR^3)}$ of tangent vectors to $\RR^2 \sqcup \RR^3$ is either $\cRR^2$ or $\cRR^3$ depending on whether the base point $x$ is chosen in the left or right component of the coproduct.
More generally, a differentiable function $f:X\to Y$ between spaces of varying dimension (which can be formalized as manifolds with multiple connected components), induces functions on the spaces of tangent and cotangent vectors\footnote{In the case of tangent vectors, this often presented in terms of the (equivalent) induced lift $(\Sigma_{x\in X} \Dsem[x]X)\to (\Sigma_{y\in Y} \Dsem[y] Y)$ of $f: X\to Y$ to the tangent bundles.}:
\begin{align*}
\Dsem{f}&:\Pi_{x\in X}\Sigma_{y\in Y}(\Dsem[x] X\multimap \Dsem[y]Y)\\
\Dsemrev{f}&:\Pi_{x\in X}\Sigma_{y\in Y}(\Dsemrev[y] Y\multimap \Dsemrev[x]X),
\end{align*}
whose first component is $f$ itself and whose second component is the action on (co)tangent vectors that $f$ induces.

If the types $\Dsyn{\ty}_2$ and $\Dsynrev{\ty}_2$ are to represent spaces of tangent and cotangent vectors to the spaces that $\Dsyn{\ty}_1$ and $\Dsynrev{\ty}_1$ represent, we would expect them to be types that vary with the particular base point (primal) we choose.
This leads to a refined view of CHAD: while $\vdash \Dsyn{\ty}_1:\type $ and $\vdash\Dsynrev{\ty}_1:\type$ can remain (closed/non-dependent) cartesian types, $\pvar:\Dsyn{\ty}_1\vdash \Dsyn{\ty}_2:\ltype$ and
$\pvar:\Dsynrev{\ty}_1\vdash \Dsynrev{\ty}_2:\ltype$
are, in general, linear dependent types.

\begin{insight}
To accommodate sum types in CHAD, it is natural to consider a target language with dependent types: this allows the dimension of the spaces of (co)tangent vectors to vary with the chosen primal.
In categorical terms: we need to consider general strictly indexed categories $\catL:\catC^{op}\to \Cat$ instead of merely locally indexed ones.
\end{insight}
\noindent The CHAD transformations of the program now becomes typed in the following more precise way:
\[
\begin{array}{l}
    \Dsyn{\Gamma}_1\vdash \Dsyn[\vGamma]{\trm}:\Sigma{\pvar:\Dsyn{\ty}_1}.{\Dsyn{\Gamma}_2\multimap \Dsyn{\ty}_2}\\
    \Dsynrev{\Gamma}_1\vdash \Dsynrev[\vGamma]{\trm}:\Sigma{\pvar:\Dsynrev{\ty}_1}.{\Dsynrev{\ty}_2\multimap \Dsynrev{\Gamma}_2},
\end{array}\]
where the action of $\Dsyn{-}_2$ and $\Dsynrev{-}_2$ on typing contexts
$\Gamma=\var_1:\ty_1,\ldots,\var_n:\ty_n$ has been refined~to
\[
    \Dsyn{\Gamma}_2\defeq \tProd{\subst{\Dsyn{\ty_1}_2}{\sfor{\pvar}{\var_1}}}{\cdots}{\subst{\Dsyn{\ty_n}_2}{\sfor{\pvar}{\var_n}}}\qquad\quad
    \Dsynrev{\Gamma}_2\defeq \tProd{\subst{\Dsynrev{\ty_1}_2}{\sfor{\pvar}{\var_1}}}{\cdots}{\subst{\Dsynrev{\ty_n}_2}{\sfor{\pvar}{\var_n}}}.
\]
All given definitions remain valid, where we simply reinterpret some tuples as having a $\Sigma$-type rather than the more limited original tuple type.

We prove the following novel results.
\begin{thmx}[Bicartesian closed structure of $\Sigma$-categories, Prop. \ref{prop:grothendieck-products-covariant} and \ref{theo:grothendieck-products-contravariant}, Theorem \ref{theo:grothendieck-ccc-covariant}, \ref{theo:grothendieck-ccc-contravariant} and \ref{theo:distributive-property-total-category} and Corollary \ref{coro:cocartesianstructure-in-the-cocartesian-csategory} and \ref{coro:cocartesian-structure-GrothCL-contravariant}]
For a category $\catC$ and a strictly indexed category $\catL:\catC^{op}\to \Cat$, $\Sigma_\catC \catL$ and $\Sigma_\catC \catL^{op}$ have
\begin{itemize} 
    \item (fibred) finite products, if $\catC$ has finite coproducts and $\catL$ has strictly indexed products and coproducts;
    \item (fibred) finite coproducts, if $\catC$ has finite coproducts and $\catL$ is extensive;
    \item exponentials, if $\catL$ is a biadditive model of the dependently typed enriched effect calculus (we intentially keep this vague here to aid legibility -- the point is that these are relatively standard conditions).
\end{itemize}
Furthermore, the coproducts in $\Sigma_\catC \catL$ and $\Sigma_\catC \catL^{op}$  distribute over the products, as long as those in $\catC$ do, even in absence of exponentials. 
Notably, the exponentials are not generally fibred over $\catC$.
\end{thmx}
\noindent The crucial notion here is our (novel) notion of extensivity of an indexed category, which generalizes well-known notions of extensive categories.
In particular, we call $\catL:\catC^{op}\to \Cat$ \emph{extensive} if the canonical functor $\catL(\sqcup_{i=1}^n C_i)\to \prod_{i=1}^n \catL(C_i)$ is an equivalence.
Furthermore, we note that we need to re-establish the product and exponential structures of $\Sigma_\catC \catL$ and $\Sigma_\catC \catL^{op}$ due to the generalization from locally indexed to arbitrary strictly indexed categories $\catL$.

Using these results, we construct a suitable target language
$\LSyn:\CSyn^{op}\to\Cat$ for CHAD on a source language with sum types (and tuple and function types), derive the forward and reverse CHAD algorithms for such a language and reestablish Theorems. \ref{thm:chad-correctness} and \ref{thm:chad-universal-property}  
 in this more general context.
This target language is a standard dependently typed enriched effect calculus with cartesian sum types 
and extensive families of linear types (i.e., dependent linear types that can be defined through case distinction).
Again, the correctness proof of Theorem \ref{thm:chad-correctness}  uses the universal property of Theorem \ref{thm:chad-universal-property} and a logical relations (categorical sconing) construction over the denotational semantics of the source and target languages.
This logical relations construction is relatively straightforward and relies on well-known sconing methods for bicartesian closed categories.
In particular, we obtain the following formulas for a sum type $\set{\ell_1\ty_1\mid\cdots\mid \ell_n\ty_n}$ with constructors $\ell_1,\ldots,\ell_n$
that take arguments of type $\ty_1,\ldots,\ty_n$:
\begin{align*}
    &\Dsyn{\set{\ell_1\ty_1\mid\cdots\mid \ell_n\ty_n}}_1 \defeq \set{\ell_1\Dsyn{\ty_1}_1\mid \cdots \mid\ell_n\Dsyn{\ty_n}_1}\\
    &\Dsyn{\set{\ell_1\ty_1\mid\cdots\mid \ell_n\ty_n}}_2\defeq 
    \vMatch{\pvar}{\ell_1\pvar\To \Dsyn{\ty_1}_2\mid\cdots\mid 
    \ell_n\pvar\To\Dsyn{\ty_n}_2}\\
    &\Dsynrev{\set{\ell_1\ty_1\mid\cdots\mid \ell_n\ty_n}}_1 \defeq \set{\ell_1\Dsynrev{\ty_1}_1\mid \cdots \mid\ell_n\Dsynrev{\ty_n}_1}\\
    &\Dsynrev{\set{\ell_1\ty_1\mid\cdots\mid \ell_n\ty_n}}_2\defeq 
    \vMatch{\pvar}{\ell_1\pvar\To \Dsynrev{\ty_1}_2\mid\cdots\mid 
    \ell_n\pvar\To\Dsynrev{\ty_n}_2},
\end{align*}
mirroring our intuition that the (co)tangent bundle to a coproduct of spaces decomposes (extensively) into the (co)tangent bundles to the component spaces.

\subsection{CHAD for (co)inductive types: where do we begin?}
If we are to really push forward the dream of differentiable programming,
we need to learn how to perform AD on programs that operate on data types.
To this effect, we analyse CHAD for inductive and coinductive types.
If we want to follow our previous methodology to find suitable definitions and correctness proofs, we first need a good categorical axiomatization of such types.
It is well-known that inductive types correspond to initial algebras of functors, 
while coinductive types are precisely terminal coalgebras.
The question, however, is what class of functors to consider.
That choice makes the vague notion of (co)inductive types precise.

Following \citep{ITA_2002__36_2_195_0}, we work with the class of \emph{$\mu\nu$-polynomials}, a relatively standard choice: i.e. functors that can be defined inductively through the combination of
\begin{itemize}
\item constants for primitive types $\reals^n$;
\item type variables $\tvar$;
\item unit and tuple types $\Unit$ and $\ty\t*\ty[2]$ of $\mu\nu$-polynomials;
\item sum types $\set{\Cns_1\ty_1\mid\cdots\mid \Cns_n\ty_n}$ of $\mu\nu$-polynomials;
\item initial algebras $\lfp{\tvar}{\ty}$ of $\mu\nu$-polynomials;
\item terminal coalgebras $\gfp{\tvar}{\ty}$ of $\mu\nu$-polynomials.
\end{itemize}
Notably, we exclude function types, as the non-fibred nature of exponentials in  $\Sigma_\catC \catL$ and $\Sigma_\catC \catL^{op}$ would significantly complicate the technical development.
While this excludes certain examples like the free state monad 
(which, for type $\ty[2]$ state would be the intial algebra $\lfp{\tvar}{\set{Get (\ty[2]\To \tvar)\mid Put (\ty[2]\t* \tvar)}}$),
it still includes the vast majority of examples of eager and lazy types that one uses in practice: e.g., lists $\lfp{\tvar}{\set{Empty\,\Unit\mid Cons (\ty[2]\t* \tvar)}}$, (finitely branching) labelled trees like $\lfp{\tvar}{\set{Leaf\,\Unit\mid Node (\ty[2]\t* \tvar\t* \tvar)}}$, streams $\gfp{\tvar}{\ty[2]\t* \tvar}$, and many more.

We characterize conditions on a strictly indexed category $\catL:\catC^{op}\to\Cat$ that guarantee that $\Sigma_\catC\catL$ and $\Sigma_\catC \catL^{op}$ have this precise notion of inductive and coinductive types.
The first step is to give a characterization of initial algebras and terminal coalgebras of split fibration endofunctors on $\Sigma_\catC\catL$ and $\Sigma_\catC \catL^{op}$. 
For legibility, we state the results here for simple endofunctors and (co)algebras, but they 
generalize to parameterized endofunctors and (co)algebras.
\begin{thmx}[Characterization of initial algebras and terminal coalgebras in $\Sigma$-categories, Corollary \ref{coro:finalresult-of-initial-algebras-fibrations} and Theorem \ref{theo:parameterized-terminal-coalgebras-for-indexed-functors}]
Let $E$ be a split fibration endofunctor on $\Sigma_\catC\catL$ (resp. $\Sigma_\catC \catL^{op}$) and let $(\iE,e)$ be the corresponding strictly indexed endofunctor on $\catL$.
Then, $E$ has 
a (fibred) initial algebra if
\begin{itemize}
    \item $\iE:\catC\to\catC$ has an initial algebra $\ind _{\iE}:\iE(\mu\iE)\to \mu\iE$;
    \item $\catL(\ind _{\iE} )^{-1} \iem :\catL(\mu\iE)\to \catL(\mu\iE)$ has an initial algebra (resp. terminal coalgebra);
    \item $\catL(f)$ preserves initial algebras (resp. terminal coalgebras) for all morphisms $f\in \catC$;
\end{itemize}
and $E$ has a (fibred) terminal coalgebra if 
\begin{itemize}
    \item $\iE:\catC\to\catC$ has a terminal coalgebra $\coind _{\iE}:\nu\iE\to \iE(\nu\iE)$;
    \item $\catL(\coind _{\iE} ) \iem :\catL(\nu\iE)\to \catL(\nu\iE)$ has a terminal  coalgebra (resp. initial algebra)
    \item $\catL(f)$ preserves terminal coalgebras (resp. initial algebras) for all morphisms $f\in \catC$.
\end{itemize}
\end{thmx}
We use this result to give sufficient conditions for (fibred) $\mu\nu$-polynomials (including their fibred initial algebras and terminal coalgebras) to exist in $\Sigma_\catC\catL$ and $\Sigma_\catC\catL^{op}$.
In particular, we show that it suffices to extend the target language $\LSyn:\CSyn^{op}\to \Cat$ with both cartesian and linear inductive and coinductive types to perform CHAD on a source language $\Syn$ with inductive and coinductive types.
Again, an equivalent of Theorem \ref{thm:chad-universal-property} holds.

We write $\tRoll{\var}$ for the constructor of inductive types (applied to an identifier $\var$),
$\tUnroll{\var}$ for the destructor of coinductive types,
and $\invtRoll[\ty]{\var}\defeq \tFold{\var}{\var[2]}{\subst{\ty}{\sfor{\tvar}{\var[2]\vdash \tRoll\var[2]}}}$, where we write 
$\subst{\ty}{\sfor{\tvar}{\var[2]\vdash \tRoll\var[2]}}$ for the functorial action of the parameterized type $\ty$ with type parameter $\tvar$ on the term $\tRoll\var[2]$ in context $\var[2]$.
This yields the following formula for spaces of primals and (co)tangent vectors to (co)inductive types, where:
\begin{align*}
&\Dsyn{\tvar}_1\defeq \tvar \qquad\qquad & \Dsyn{\tvar}_2 = \ltvar\\
&\Dsyn{\lfp{\tvar}{\ty}}_1\defeq \lfp{\tvar}{\Dsyn{\ty}_1}
\qquad\qquad
&\Dsyn{\lfp{\tvar}{\ty}}_2\defeq \llfp{\ltvar}{\subst{\Dsyn{\ty}_2}{\sfor{\pvar}{\invtRoll[\Dsyn{\ty}_1]{\pvar}}}}\\
&\Dsyn{\gfp{\tvar}{\ty}}_1\defeq \gfp{\tvar}{\Dsyn{\ty}_1}
\qquad\qquad
&\Dsyn{\gfp{\tvar}{\ty}}_2\defeq \lgfp{\ltvar}{\subst{\Dsyn{\ty}_2}{\sfor{\pvar}{\tUnroll\pvar}}}\\
&\Dsynrev{\tvar}_1\defeq \tvar \qquad\qquad & \Dsynrev{\tvar}_2 = \ltvar\\
&\Dsynrev{\lfp{\tvar}{\ty}}_1\defeq \lfp{\tvar}{\Dsynrev{\ty}_1}
\qquad\qquad
&\Dsynrev{\lfp{\tvar}{\ty}}_2\defeq \lgfp{\ltvar}{\subst{\Dsynrev{\ty}_2}{\sfor{\pvar}{\invtRoll[\Dsynrev{\ty}_1]{\pvar}}}}\\
&\Dsynrev{\gfp{\tvar}{\ty}}_1\defeq \gfp{\tvar}{\Dsynrev{\ty}_1}
\qquad\qquad
&\Dsynrev{\gfp{\tvar}{\ty}}_2\defeq \llfp{\ltvar}{\subst{\Dsynrev{\ty}_2}{\sfor{\pvar}{\tUnroll\pvar}}}
\end{align*}
\begin{insight}
Types of primals to (co)inductive types are (co)inductive types of primals, types of tangents to (co)inductive types are linear (co)inductive types of tangents, and types of cotangents to inductive types are linear coinductive types of cotangents and vice versa.
\end{insight}

For example, for a type $\ty=\lfp{\tvar}{\set{Empty\,\Unit\mid Cons (\ty[2]\t* \tvar)}}$ of lists of elements of type $\ty[2]$,
we have a cotangent space 
\[
\Dsynrev{\ty}_2 = \lgfp{\ltvar}{\vMatch{\inv\tRoll{\pvar}}{Empty\,\_\To \lUnit\mid Cons\, \pvar\To \subst{\Dsynrev{\ty[2]}_2}{\sfor{\pvar}{\tFst\pvar}}\t*\ltvar}}\qquad\text{where}\]
$\sqinv\tRoll{\pvar}=\tFold{\pvar}{\var[2]}{\vMatch{\var[2]}{Empty\,\var[2]\To Empty\,\var[2]\mid Cons \, \var[2] \To Cons\tPair{\tFst\var[2]}{\tRoll\!(\tSnd\var[2])}}}\hspace{-40pt}$\\[8pt]
and, for a type $\ty=\gfp{\tvar}{\ty[2]\t* \tvar}$ of streams, we have a cotangent space
$$
\Dsynrev{\ty}_2 = \llfp{\ltvar}{\subst{\Dsynrev{\ty[2]}_2}{\sfor{\pvar}{\tFst(\tUnroll\pvar)}}\t*\ltvar}.
$$

We demonstrate that the strictly indexed category $\FVect:\Set^{op}\to \Cat$ of 
families of vector spaces also satisfies our conditions, so it gives a concrete 
denotational semantics of the target language $\LSyn:\CSyn^{op}\to\Cat$,
by Theorem \ref{thm:chad-universal-property}.
To reestablish the correctness theorem \ref{thm:chad-correctness},
existing logical relations techniques do not suffice, as far as we are aware.
Instead, we achieve it by developing a novel theory of categorical logical relations (sconing)
for languages with expressive type systems like our AD source language.
\begin{insight}
    We can obtain powerful logical relations techniques for reasoning 
    about expressive type systems by analysing when the forgetful functor 
    from a category of logical relations to the underlying category is comonadic and monadic.
\end{insight}
\indent In almost all instances, the forgetful functor from a category 
of logical relations to the underlying category is comonadic and in many instances, including ours,
it is even monadic.
This gives us the following logical relations techniques for expressive type systems:
\begin{thmx}[Logical relations for expressive types, \S\ref{sec:subsconing}] \label{thm:logical-relations-expressive} Let $G:\catC\to\catD$ be a functor. We observe
    \begin{itemize}
    \item If $\catD$ has binary products, then the forgetful functor from the scone (the comma category) $\catD\downarrow G\to \catD\times\catC$ is comonadic
    (Theorem~\ref{theo:basic-comonadicity-sconing}).
    \item If $G$ has a left adjoint and $\catC $ has binary coproducts, then $\catD\downarrow G\to \catD\times\catC$ is monadic
    (Corollary~\ref{coro:basic-comonadicity-monadicity-sconing}).
    \end{itemize}
    This is relevant because:
    \begin{itemize}
    \item comonadic functors create initial algebras (Theorem~\ref{theo:creation-of-initial-algebras});
    \item monadic functors create terminal coalgebras (Theorem~\ref{theo:creation-of-initial-algebras});
    \item monadic-comonadic functors create $\mu\nu$-polynomials (Corollary \ref{coro:creating-munu-polynomials});
    \item if $\catE$ is monadic-comonadic over $\catE'$, then $\catE$ is finitely complete cartesian closed if $\catE'$ is (Proposition \ref{prop:closed-structure-comonadic-monadic}). 
    \end{itemize}
\end{thmx}
\noindent As a consequence, we can lift our concrete denotational semantics of all types, including inductive and coinductive types to our categories of logical relations over the semantics.

These logical relations techniques are suffient to yield the 
correctness theorem \ref{thm:chad-correctness}.
Indeed, as long as derivatives of primitive operations are correctly implemented in the sense that $\sem{D\op}=D\op$ and $\sem{\transpose{D\op}}=\transpose{D\sem{\op}}$, Theorem \ref{thm:logical-relations-expressive} tells us that 
the unique structure preserving functors 
\begin{align*}&\sPair{\sem{-}}{\sem{\Dsyn{-}}}:\Syn\to \Set\times \Sigma_\Set \FVect\\ & \sPair{\sem{-}}{\sem{\Dsynrev{-}}}:\Syn\to \Set\times \Sigma_\Set\FVect^{op}\end{align*}
lift to the scones of  $\Hom((\RR^k,(\RR^k,\cRR^k)),-) :\Set\times \Sigma_\Set \FVect\to\Set$ and $\Hom((\RR^k,(\RR^k,\cRR^k)),-) :\Set\times \Sigma_\Set \FVect^{op}\to \Set$ where we lift the image of $\reals^n$, respectively, to the logical relations
\begin{align*}
&\set{(f,(g,h))\mid f=g\text{ and } h = Df\phantom{{}^t}}\hookrightarrow (\Set\times \Sigma_\Set \FVect\phantom{{}^{op}})\left((\RR^k,(\RR^k,\cRR^k)), (\RR^n,(\RR^n,\cRR^n))\right)\\ 
&\set{(f,(g,h))\mid f=g\text{ and } h = \transpose{Df}}\hookrightarrow  (\Set\times \Sigma_\Set \FVect^{op})\left((\RR^k,(\RR^k,\cRR^k)), (\RR^n,(\RR^n,\cRR^n))\right).
\end{align*}
We see that $\sem{\Dsyn{\trm}}$ and $\sem{\Dsynrev{\trm}}$ propagate derivatives and transposed derivatives of differentiable $k$-surfaces (differentiable functions $\RR^k\to \mathrm{dom}\sem{\trm}$) correctly for all programs $\trm$.
Seeing that $(\id,(\id,x\mapsto \id))$ is one such $k$-surface in the logical relation associated with $\reals^k$, we see that $(\sem{\trm},(\pi_1\circ\sem{\Dsyn{\trm}},\pi_2\circ\sem{\Dsyn{\trm}}))$ and $(\sem{\trm},(\pi_1\circ \sem{\Dsynrev{\trm}},\pi_2\circ \sem{\Dsynrev{\trm}}))$ are $k$-surfaces in the relations as well, for any $x:\reals^k\vdash \trm:\reals^n$.
That is, Theorem \ref{thm:chad-correctness} holds.

Our novel logical relations machinery is in no way restricted to the context of CHAD, however.
In fact, it is widely applicable for reasoning about total functional languages with expressive type systems.

\subsection{Inductive types and derivatives}
So far, we have only phrased the CHAD correctness theorem \ref{thm:chad-correctness} only for programs $\trm$ with domain/codomain isomorphic to some Euclidean space $\RR^n$, even if $\trm$ may make use of any complex types (including variant, inductive, coinductive and function types) in its computation.
The reason for this restriction, is that this limited context of functions $f:\RR^n\to \RR^m$ is an obvious setting where we have a simple, canonical, unambiguous notion of derivative $\Dsem{f}:\RR^n\to \RR^m\times (\cRR^n\multimap \cRR^m)$, allowing us to phrase an obvious correctness criterion.

More generally, for $f:X\to Y$ where $X$ and $Y$ are manifolds, we also have an unambiguous notion of derivative $\Dsem{f}:\Pi_{x\in X} \Sigma_{y\in Y}\Dsem[x]X\multimap \Dsem[y]Y$, which allows us to strengthen our correctness result.
In fact, for our purposes it suffices to consider the relatively simple context of differentiable functions $f:\coprod\limits_{i\in I}\RR^{n_i}\to \coprod\limits_{j\in J}\RR^{m_j}$
between very simple manifolds that arise as disjoint unions of (finite dimensional) Euclidean spaces.
Such functions $f$ decompose uniquely as copairings $f=[\iota_{\phi(i)}\circ g_i]_{i\in I}$ where we write $\iota_k$ for the $k$-th coprojection and where $\phi:I\to J$ is some function and $g_i:\RR^{n_i}\to \RR^{m_{\phi(j)}}$.
That is, $f$ can be understood as the family $(g_i)_{i\in I}$ and its derivative $\Dsem{f}$ decomposes uniquely as the family of plain derivatives $\Dsem{g_i}$ in the usual sense.
We have a similar decomposition for the transposed derivatives $\Dsemrev{f}$.

This notion of derivatives of functions between disjoint unions of Euclidean spaces is relevant to our context, as we have the following result.
\begin{thmx}[Canonical form of $\mu$-polynomial semantics, Corollary \ref{coro:normal-form-inductive-data-types}]
For any types $\ty_i$ built from Euclidean spaces $\reals^n$, tuple types $\ty_i\t*\ty_j$, variant types $\set{\ell_1\ty_1\mid\cdots\mid \ell_n\ty_n}$, type variables $\tvar$ and inductive types $\mu\tvar.\ty_i$ (so-called $\mu$-polynomials), its denotation $\sem{\ty_i}$ is isomorphic to a manifold of the form
$
\coprod\limits_{i\in I}\RR^{n_i}
$
for some countable set $I$ and some $n_i\in \NN$.
\end{thmx}
Consequently, we can strengthen Theorem \ref{thm:chad-correctness} in the following form:
\begin{thmx}[Correctness of CHAD (Generalized), Theorem \ref{theorem:correctness-theorem-for-inductive-data-types}]\label{thm:chad-correctness-strong}
    For any well-typed program $$\var_1:\ty_1,\ldots,\var_k:\ty_n\vdash {\trm}:\ty[2],$$
    where $\ty_i,\ty[2]$ are all (closed) $\mu$-polynomials,
    we have that
    $\sem{\Dsyn[\var_1,\ldots,\var_k]{\trm}}=\Dsem{\sem{\trm}}\;\text{ and }\;\sem{\Dsynrev[\var_1,\ldots,\var_k]{\trm}}=\Dsemrev{\sem{\trm}}.$
    \end{thmx}
Again, $\trm$ can make use of coinductive types and function types in the middle of its computation, but they may not occur in the input or output types.
The reason is that, as far as we are aware, there is no canonical\footnote{In fact, on such infinite dimensional spaces, we have many inequivalent definitions of derivative (that all coicide for finite dimensional spaces) \cite{iglesias2013diffeology,christensen2014tangent}.} notion of semantic derivative for functions between the sort of infinite dimensional spaces that co-datatypes such as coinductive types and function types implement.
This makes it challenging to even phrase what semantic correctness at such types would mean.

\subsection{How does CHAD for expressive types work in practice?}
The CHAD code transformations we describe in this papers are well-behaved in practical implementations in the sense of the following 
compile-time complexity result.
\begin{thmx}[No code blow-up, Corollary \ref{cor:no-code-blow-up}]
    \label{thm:no-code-blow-up}
    The size of the code of the CHAD transformed programs $\Dsyn[\vGamma]{\trm}$ and $\Dsynrev[\vGamma]{\trm}$  grows linearly with the size of the original source program $\trm$.
\end{thmx}
We have ensured to pair up the primal and (co)tangent computations in our CHAD transformation and to exploit any possible sharing of common subcomputations, using $\mathbf{let}$-bindings.
However, we leave a formal study of the run-time complexity of our technique to future work.

As formulated in this paper, CHAD generates code with linear dependent types.
This seems very hard to implement in practice.
However, this is an illusion: we can use the code generated by CHAD and interpret it as less precise types.
We sketch how all type dependency can be erased and how all linear types other than the linear (co)inductive types 
can be implemented as abstract types in a standard functional language like Haskell.
In fact, we describe three practical implementation strategies for our treatment of sum types, none of which 
require linear or dependent types.
All three strategies have been shown to work in the CHAD reference implementation.
We suggest how linear (co)inductive types might be implemented in practice, based on their concrete denotational semantics, but leave the actual implementation to future work.
\section{Background: categorical semantics of expressive total languages}
\label{sec:background-categorical-semantics}
In this section, we fix some notation and recall the well-known abstract categorical semantics 
of total functional languages with expressive type systems \citep{pitts1995categorical,zbMATH00589495,ITA_2002__36_2_195_0}, which builds on the usual semantics of the simply typed $\lambda$-calculus in Cartesian closed categories \citep{lambek1988introduction}.
In this paper, we will be interested in a few particular 
instantiations (or models) of such an abstract categorical semantics $\catC$:
\begin{itemize}
\item the initial model $\Syn$ (\S \ref{sec:source-language}), which represents the programming language under consideration, up to $\beta\eta$-equivalence; this will be the source language of our AD code transformation;
\item the concrete denotational model $\Set$  (\S \ref{sec:concrete-models}) in terms of sets and functions, which represents our default denotational semantics of the source language;
\item models $\Sigma_{\catC}\catL$ and $\Sigma_{\catC}\catL^{op}$  (\S \ref{sec:grothendieck-constructions}) in the the $\Sigma$-types of suitable indexed categories $\catL:\catC^{op}\To\Cat$;
\item in particular, the models $\Sigma_{\CSyn}\LSyn$ and $\Sigma_{\CSyn}\LSyn^{op}$  (\S \ref{sec:target-language}) built out of the target language, which yield forward and reverse mode CHAD code transformations, respectively;
\item sconing (categorical logical relations) constructions $\Fscone$ and $\Rscone$  (\S \ref{sec:subsconing}) over the models 
$\Set\times \Sigma_{\Set}\FVect$ and $\Set\times \Sigma_{\Set}\FVect^{op}$ that yield the correctness arguments for forward and reverse mode CHAD, respectively,
where $\FVect:\Set^{op}\to \Cat$ is the strictly indexed category of families of real vector spaces.
\end{itemize}
We deem it relevant to discuss the abstract categorical semantic framework for our language as we need these various instantiations of the framework.

\subsection{Basics}
We use standard definitions from category theory; see, for instance, \citep{zbMATH03367095, zbMATH06319720}.
A category $\catC$ can be seen as a semantics for a typed functional programming language, whose types correspond to objects of $\catC$ and whose programs that take an input of type $A$ and produce an output of type $B$ are represented by the homset $\catC(A,B)$.
Identity morphisms $\id[A]$ represent programs that simply return their input (of type $A$) unchanged as output and composition $g\circ f$ of morphisms $f$ and $g$ represents running the program $g$ after the program $f$.
Notably, the equations that hold between morphisms represent program equivalences that hold for the particular notion of semantics that $\catC$ represents.
Some of these program equivalences are so fundamental that we demand them as structural equalities that need to hold in any categorical model (such as the associativity law $f\circ (g \circ h)=(f\circ g)\circ h$).
In programming languages terms, these are known as the $\beta$- and $\eta$-equivalences of programs.

\subsection{Tuple types}

Tuple types represent a mechanism for letting programs take more than one input or produce more than one output.
Categorically, a tuple type corresponds to a product $\prod_{i\in I}A_i$ of a \textit{finite} family of types $\set{A_i}_{i\in I}$, which we also write $\terminal$ or $A_1\times A_2$ in the case of nullary and binary products. For basic aspects of products, we refer the reader to \citep[Chapter~III]{zbMATH03367095}.

We write $\seq[i\in I]{f_i}:C\to \prod_{i\in I}A_i$ for the product pairing of $\set{f_i:C:A_i}_{i\in I}$ and 
$\proj{j}:\prod_{i\in I}A_i\to A_j$ for the $j$-th product projection, for $j\in I$.
As such, we say that a categorical semantics $\catC$ models 
(finite) tuples if $\catC$ has (chosen) finite products.

\subsection{Primitive types and operations}
We are interested in programming languages that have support for a certain set $\Ty$ of ground types such as integers and (floating point) real numbers as well as certain sets $\Op(T_1,\ldots, T_n; S)$, for $T_1,\ldots,T_n,S\in\Ty$, of operations on these basic types such as addition, multiplication, sine functions, etc. 
We model such primitive types and operations categorically by demanding that our category has a distinguished object $C_T$ for each $T\in \Ty$ to represent the primitive types and a distinguished morphism $f_{\op}\in \catC(C_{T_1}\times \ldots\times C_{T_n}, C_S)$ for all primitive operations $\op\in \Op(T_1,\ldots, T_n; S)$. For basic aspects of categorical type theory, see, for instance, \citep[Chapters~3\&4]{zbMATH00589495}.

\subsection{Function types}
Function types let us type popular higher order programming idioms such as maps and folds, which capture common control flow abstractions.
Categorically, a type of functions from $A$ to $B$ is modelled as an exponential $A\Rightarrow B$.
We write $\ev:(A\Rightarrow B)\times A\To B$ (evaluation) for the co-unit of the adjunction $(-)\times A\dashv A\Rightarrow(-)$ and $\Lambda$ for the Currying natural isomorphism $\catC(A\times B, C)\To \catC(A,B\Rightarrow C)$.
We say that a categorical semantics $\catC$ with tuple types models function types if $\catC$ has a chosen right adjoint  $(-)\times A\dashv A\Rightarrow(-)$.

\subsection{Sum types  (aka variant types)}
Sum types (aka variant types) let us model data that exists in multiple different variants and branch in our code on these different possibilities.
Categorically, a sum type is modelled as a coproduct $\coprod_{i\in I}A_i$ of a collection of a finite family $\set{A_i}_{i\in I}$ of types, which we also write $\initial$ or $A_1\sqcup A_2$ in the case of nullary and binary coproducts.
We write $\coseq[i\in I]{f_i}:\coprod_{i\in I}C_i\to A$ for the copairing of $\set{f_i:C_i\to A}_{i\in I}$ and $\coproj{j}:A_j\to \coprod_{i\in I}A_i$ for the $j$-th coprojection.
In fact, in presence of tuple types, a more useful programming interface is obtained if one restricts to 
\emph{distributive} coproducts, i.e. coproducts  $\coprod_{i\in I}A_i$ such that the map $\coseq[i\in I]{\pair{\coproj{i}\circ\proj{1}}{\proj{2}}}:\coprod_{i\in I}(A_i\times B)\To (\coprod_{i\in I}A_i)\times B$ is an isomorphism; see, for instance, \citep{MR1201048, MR2864762}.
Note that in presence of function types, coproducts are automatically distributive since the left adjoint functors $(-)\times A$ preserve colimits; see, for instance, \citep[6.3]{zbMATH06319720}.
As such, we say that a categorical semantics $\catC$ models (finite) sum types if $\catC$ has (chosen) finite distributive coproducts.

\subsection{Inductive and coinductive types}\label{SUBSECT:INDUCTIVE-COINDUCTIVE-TYPES}
We employ the usual semantic interpretation of \textit{inductive and coinductive types} as, respectively, \textit{initial algebras} and \textit{terminal coalgebras} of a certain class of functors. We refer the reader, for instance, to \citep[Chapter~9]{MR2178101}, \citep{ITA_2002__36_2_195_0} and \citep{adamekinformal}.

Most of this section is dedicated to describing precisely which class of functors we consider initial algebras and terminal coalgebras, a class we call \emph{$\mu\nu$-polynomials}.
Roughly speaking, we define $\mu\nu$-polynomials to be functors that can be constructed from products, coproducts, projections, diagonals, constants, initial algebras and terminal  coalgebras.

To fix terminology and for future reference of the detailed constructions, we recall below basic aspects of parameterized initial algebras and parameterized terminal  coalgebras.

\begin{definition}[The category of $E$-algebras]
Let  $E : \catD\to \catD$ be an endofunctor.
The category of $E$-algebras, denoted by $E\AAlg$, is defined as follows. The objects are pairs
$(W, \zeta ) $ in which $W\in \catD $ and $ \zeta : E(W)\to W $ is a morphism of $\catD $.
A morphism between $E$-algebras  $(W, \zeta ) $ and $(Y, \xi) $ is a morphism 
$g: W\to Y $ of $\catD $ such that 
\begin{equation}
	\begin{tikzcd}
		E(W) \arrow[rr, "E(  g)"] \arrow[swap,dd, "\zeta"] && E(Y) \arrow[dd, "\xi"] \\
		&&\\
		W \arrow[swap, rr, "g"] && Y                              
	\end{tikzcd}
\end{equation}
commutes. Dually, we define the category $E\CCoAlg$ of $E$-coalgebras by
\begin{equation}
	E\CCoAlg := \left(E^{\op}\AAlg\right) ^{\op}
\end{equation}
in which $E^\op : \catD ^\op\to \catD ^\op $ is the image of $E$ by $\op :\Cat\to \Cat $.
\end{definition}
\begin{definition}[Initial algebra and terminal  coalgebra] 
Let  $E : \catD\to \catD$ be an endofunctor.
Provided that they exist, the initial object $(\mu E, \ind _E ) $ of $E\AAlg$ and the
terminal  object  $(\nu E, \coind _E ) $ of $E\CCoAlg $ are respectively referred to as the initial $E$-algebra and the terminal  $E$-coalgebra.
\end{definition}

\begin{remark} 
By Lambek's Theorem,  
provided that the initial algebra $(\mu E, \ind _E ) $ of an endofunctor $E$ exists, we have that $\ind _E$   is invertible.
Dually, we get the result for terminal  coalgebras.
\end{remark} 

Assuming the existence of the initial $E$-algebra and the terminal  $E$-coalgebra, we denote by
\begin{equation}\label{eq:basic-fold-semantics-appendix}
	\fold_E (Y, \xi): \mu E \to Y, \quad \unfold_E (X, \varrho ): X\to \nu E 
\end{equation}
the unique morphisms in $\catD $ such that
\begin{equation}
	\begin{tikzcd}
		E(\mu E) \arrow[rrr, "{E\left(\fold_E (Y, \xi)\right)}"] \arrow[swap,dd, "\ind_E"] &&& E(Y) \arrow[dd, "\xi"] \\
		&&&\\
		\mu E \arrow[swap, rrr, "{\fold_E (Y, \xi)}"] &&& Y                              
	\end{tikzcd}\qquad
	\begin{tikzcd}
		X \arrow[rrr, "{\unfold_E (X, \varrho)}"] \arrow[swap,dd, "\varrho"] &&& \nu E \arrow[dd, "\coind _E"] \\
		&&&\\
		E(X) \arrow[swap, rrr, "{E\left(\unfold_E (X, \varrho)\right)}"] &&& E(\nu E)                              
	\end{tikzcd}
\end{equation}
commute.  Whenever it is clear from the context, we denote $ \fold_E (Y, \xi) $ by $\fold_E \xi $, and  $\unfold_E (X, \varrho )$
by $\unfold_E  \varrho  $.

Given a functor
$	H : \catD '\times\catD \to \catD $ and an object
$X $ of $\catD ' $, we denote by $H^X $ the endofunctor 
\begin{equation}\label{eq:basic-functor-inductive-types-restricted}
	H(X, -): \catD \to \catD .
\end{equation}   
In this setting, if $\mu H^X$ exists for any object $X\in\catD'$ then the universal properties of the initial algebras induce a functor denoted by $\mu H : \catD' \to \catD$, called the parameterized initial algebra. In the following, we spell out how to construct parameterized initial algebras and terminal  coalgebras.

\begin{proposition}[$\mu$-operator and $\nu$-operator]\label{prop:general_parameterized_initial_algebras}
	Let $ H:\catD '\times \catD \to\catD $ be a functor. Assume that, for each object $X\in\catD '$, the functor $H^X = H(X,-) $ is such that
	$\mu H ^X $ exists. In this setting, we have the induced functor
	\begin{eqnarray*}
		\mu H :  \catD ' & \to     & \catD\\
		X            & \mapsto & \mu H^X\\
		\left( f: X\to Y \right) & \mapsto & \fold _{H^X} \left( \ind_{H^Y}\circ  H(f, \mu H ^Y)\right).
	\end{eqnarray*}	
	Dually, assuming that, for each object $X\in\catD '$, 	$\nu H ^X $ exists, we have the induced functor
	\begin{eqnarray*}
		\nu H :  \catD ' & \to     & \catD\\
		X            & \mapsto & \nu H^X\\
		\left( f: X\to Y \right) & \mapsto & \unfold _{H^Y} \left(  H(f, \nu H ^X)\circ \coind_{H^X}\right).
	\end{eqnarray*}	
\end{proposition}
\begin{proof}
	We assume that the functor $ H:\catD '\times \catD \to\catD $  is such that, for any object $X\in\catD ' $, $\mu H ^X $ exists. 	For each morphism $f: X\to Y $, we define $\mu H (f) = \fold _{H^X} \left(  \ind_{H^Y}\circ  H(f, \mu H ^Y)\right) $ as above. We prove below that this makes $\mu H(f) $ a functor.
	
	Given $X\in\catD '$,
	\begin{align*}
		& \mu H( \id _  X ) \\
		&  = \fold _{H^X} \left(  \ind_{H^X}\circ  H( \id _ X , \mu H ^X)\right) \\
		& = \fold _{H^X} \left(  \ind_{H^X} \right) \\
		& = \id _{\mu H ^X} .
	\end{align*}
	Moreover, given morphisms $f: X\to Y $ and $g: Y\to Z $ in $\catD '$, we have that
	\begin{align*}
		&  \mu H (g) \circ \mu H (f)  \circ \ind _{H^X} \\
		& = \mu H (g) \circ \ind _{H ^Y }\circ  H \left( f, \mu H (f ) \right)   \explainr{$\mu H(f) = \fold _{H^X} \left(  \ind_{H^Y}\circ  H(f, \mu H ^Y)\right) $} \\
		& = \ind _{H ^Z}\circ H \left( g, \mu H ( g ) \right) \circ  H \left( f, \mu H (f ) \right)   \explainr{$\mu H(g) = \fold _{H^Y} \left(  \ind_{H^Z}\circ  H(g, \mu H ^Z)\right) $}\\
		& = \ind _{H ^Z}\circ H \left( g f, \mu H ( g )\circ\mu H (f ) \right) \\
		& =   \ind _{H ^Z}\circ H \left(  g f , \mu H ^Z \right)\circ H \left( X, \mu H ( g )\circ \mu H (f )  \right)
	\end{align*}
	and, hence, the diagram
	\begin{equation*}
		\begin{tikzcd}
			H \left( X, \mu H^X \right) \arrow[rrr, "{H \left( X, \mu H ( g )\circ \mu H (f )  \right)}"] \arrow[swap,ddd, "{ \ind_{H ^X} }"] &&& H \left( X, \mu H^Z \right) \arrow[ddd, "{\ind _{H ^Z}\circ H \left(  g\circ f , \mu H ^Z \right) }"] \\
			&&&\\
			&&&\\
			\mu H ^X \arrow[swap, rrr, "{\mu H ( g )\circ \mu H (f )}"] &&& \mu H^Z                               
		\end{tikzcd}
	\end{equation*}
	commutes. By the universal property of the initial algebra $\left( \mu H^X, \ind _{H ^X}\right) $,
	we conclude that 
	\begin{align*}
		&  \mu H (g) \circ \mu H (f)   \\
		& = \fold _{H^X}\left( \ind _{H ^Z}\circ H \left(  g\circ f , \mu H ^Z \right) \right) \\
		& = \mu H ( g\circ f ) . \explainr{definition}
	\end{align*}
\end{proof}

It is worth noting that in Proposition \ref{prop:general_parameterized_initial_algebras}, $\catD'$ can be any category. However, in the standard setting of initial algebra semantics, there is a special interest in the case where $\catD' = \catD^{n-1}$ and $n>1$, which is described below.

\begin{proposition}[Parameterized initial algebras and terminal  coalgebras]\label{prop:basic_parameterized_initial_algebras}
	Let $ H:\catD ^n\to\catD $ be a functor in which $n>1 $. Assume that, for each object $X\in\catD ^{n-1}$,
	$\mu H ^X $ exists. In this setting, we have the induced functor
	\begin{eqnarray*}
		\mu H :  \catD ^{n-1} & \to     & \catD\\
		X            & \mapsto & \mu H^X\\
		\left( f: X\to Y \right) & \mapsto & \fold _{H^X} \left(  \ind_{H^Y}\circ  H(f, \mu H ^Y)\right).
	\end{eqnarray*}	
	Dually, if $\nu H ^X $ exists for any $X\in\catD ^{n-1}$, we have the induced functor
	\begin{eqnarray*}
		\nu H :  \catD^{n-1} & \to     & \catD\\
		X            & \mapsto & \nu H^X\\
		\left( f: X\to Y \right) & \mapsto & \unfold _{H^Y} \left(  H(f, \nu H ^X)\circ \coind_{H^X}\right).
	\end{eqnarray*}	
\end{proposition}

In order to model inductive and coinductive types coming from parameterized types not involving function types,
we introduce the following notions.

\begin{definition}[$\mu\nu$-polynomials]\label{def:basic-definition-munupolynomials}
	Assuming that $\catD $ has finite coproducts and finite products, 
	the category $\mnPoly _ \catD $  
	is the smallest subcategory of $\Cat $ satisfying the following.
\begin{enumerate}
\item[O)] The objects are defined inductively by:
	\begin{enumerate}[O1)]
		\item the terminal   category $\terminal $ is an object of $\mnPoly _ \catD $;
		\item the category $\catD $ is an object of $\mnPoly _ \catD $; 
		\item for any pair of objects  $\left( \catD ', \catD ''  \right) \in \mnPoly _ \catD\times \mnPoly _ \catD $, the product $\catD '\times \catD '' $  is an object of  $\mnPoly _ \catD $.
	\end{enumerate}
\item[M)] The morphisms satisfy the following properties:
	\begin{enumerate}[M1)]
		\item\label{eq:munu-M1} for any object $\catD '$ of $\mnPoly _ \catD $, the unique functor 
		$\catD '\to \terminal $ is a morphism of $\mnPoly _ \catD $; 
		\item\label{eq:munu-M2} for any object $\catD '$ of $\mnPoly _ \catD $,  all the functors $\terminal \to \catD ' $ are morphisms of $\mnPoly _ \catD $;
		\item\label{eq:munu-M3}  the binary product $\times : \catD \times\catD   \to \catD $ is a morphism of  
		$\mnPoly _ \catD $;
		\item\label{eq:munu-M4}  the binary coproduct $\sqcup  : \catD\times \catD  \to \catD $ is a morphism of
		$\mnPoly _ \catD $;
		\item\label{eq:munu-M5} for any pair of objects  $\left( \catD ', \catD ''  \right) \in \mnPoly _ \catD\times \mnPoly _ \catD $,  the projections 
		\begin{equation*}
		\pi _1 : \catD '\times \catD '' \to \catD ',\qquad   \pi _2 : \catD '\times \catD '' \to \catD ''
		\end{equation*}
		are morphisms of  $\mnPoly _ \catD $; 
		\item\label{eq:munu-M6} given objects $ \catD ', \catD '' ,  \catD '''$  of $\mnPoly _\catD  $, if $E: \catD '   \to \catD ''    $ and $J : \catD '  \to \catD '''    $ are morphisms of $\mnPoly _ \catD $, then so is the induced functor $(E,J) :\catD '   \to \catD '' \times \catD '''   $; 
		\item\label{eq:munu-M7}   if $\catD '$ is an object of $\mnPoly _ \catD $,  $H: \catD '\times \catD   \to\catD  $ is a morphism of $\mnPoly _ \catD $  and $\mu H : \catD ' \to \catD    $ exists,
		then $\mu H $ is a morphism of  $\mnPoly _ \catD $;
		\item\label{eq:munu-M8}  if $\catD '$ is an object of $\mnPoly _ \catD $,  $H: \catD '\times \catD     \to\catD $ is a morphism of $\mnPoly _ \catD $  and $\nu H : \catD '  \to \catD   $ exists,
	then $\nu H $ is a morphism of  $\mnPoly _ \catD $.
	\end{enumerate}
\end{enumerate} 
	We say that $\catD $ has \textit{$\mu\nu$-polynomials} if $\catD $ has finite coproducts and products and, for any endofunctor $ E: \catD \to\catD  $ in $ \mnPoly _ \catD $, $\mu E $ and $\nu E $ exist.
	We say that $\catD $ has \textit{chosen} $\mu\nu$-polynomials if we have additionally made a choice of initial algebras and terminal   coalgebras for all $\mu\nu$-polynomials.
\end{definition}

\begin{remark}[Self-duality]\label{remark:sel-duality-munupolynomials}
	A category $\catD $ has $\mu\nu$-polynomials if and only if $\catD^{\op} $ has $\mu\nu$-polynomials as well.
\end{remark}

Another suitably equivalent way of defining $\mnPoly _ \catD$ is the following. The category
$\mnPoly _ \catD$ is the smallest subcategory of $\Cat $ such that:
\begin{enumerate}[-]
	\item the inclusion $\mnPoly _ \catD\to\Cat $ creates finite products;
	\item 	$\catD$ is an object of the subcategory $\mnPoly _ \catD$;
	\item for any object $\catD '$ of $\mnPoly _ \catD$, all the functors $\terminal \to \catD ' $ are morphisms of $\mnPoly _ \catD$;
	\item and the binary product $\times : \catD \times\catD  \to \catD $ is a morphism of  $\mnPoly _ \catD $;
	\item  the binary coproduct $\sqcup  : \catD\times \catD \to \catD $ is a morphism of  $\mnPoly _ \catD $;
	\item if $\catD '$ is an object of $\mnPoly _ \catD $, $H: \catD '\times \catD  \to\catD   $ is a morphism of $\mnPoly _ \catD $  and $\mu H : \catD ' \to \catD   $ exists,
	then $\mu H $ is a morphism of  $\mnPoly _ \catD $;
	\item  if $\catD '$ is an object of $\mnPoly _ \catD $, $H: \catD '\times \catD  \to\catD    $ is a morphism of $\mnPoly _ \catD $  and $\nu H : \catD '\to \catD    $ exists,
	then $\nu H $ is a morphism of  $\mnPoly _ \catD $.
\end{enumerate}	

\begin{lemma}\label{lem:about-completeness-of-munupolynomials}
Let $\catC$ be a category with $\mu\nu$-polynomials. If $\catD $ is an object of $\mnPoly _ \catC $ and  $$H: \catD \times \catC     \to\catC $$ is a functor in $\mnPoly _ \catC $,
	then $\mu H : \catD\to\catC $ and $\nu H : \catD\to\catC $  exist (and, hence, they are morphisms of $\mnPoly _ \catC $).
\end{lemma}
\begin{proof}
Let $X$ be any object of $ \catD $. Denoting by $X: \terminal \to \catD $ the functor constantly equal to $X $, the functor $H^X $ is the composition below.
\begin{equation*}
\begin{tikzcd}
\catC \arrow[rr, "{ \left( 1, \id _{\catC} \right) }"] 
\arrow[swap,bend right=12, rrrrrrrrrr, "{ H^X }"] && 
\terminal\times\catC 
\arrow[rrrr, "{ \left( X\circ \pi _1 , \id _{\catC}\circ\pi _2 \right) }"] 
&&&&\catD\times\catC \arrow[rrrr, "{ H }"] &&&& \catC
\end{tikzcd}
\end{equation*}
Since all the morphisms above are in  $ \mnPoly _ \catC $, we conclude that $H^X $ is an endomorphism of $ \mnPoly _ \catC $. Therefore, since $\catC $ has $\mu\nu$-polynomials,
$\mu H^X $ and $\nu H ^X $ exist.

By Proposition \ref{prop:general_parameterized_initial_algebras}, since $\mu H^X $ and $\nu H ^X $ exist for any 
$X$ in $\catD $, $\mu H $ and $\nu H  $ exist. 
\end{proof}

We say that a categorical semantics $\catC$ with (finite) sum and tuple types supports inductive and coinductive types if $\catC$ has chosen $\mu\nu$-polynomials.
Note that we do not consider the more general notion of (co)inductive types defined by endofunctors that may contain function types in their construction.

\section{Structure-preserving functors}\label{sect:structure-preserving-functors}
In this paper, the definition of our AD macro, the definitions of the concrete semantics and logical relations are all framed in terms of appropriate structure-preserving functors. This fact highlights the significance of the suitable notions of structure-preserving functors in our work. 

A structure-preserving functor between bicartesian closed categories are, of course, bicartesian closed functors. We usually assume that those are strict, which means that the functors preserve the structure on the nose.
 
It remains to establish the notion of structure-preserving functor between categories with $\mu\nu$-polynomials. We do it below, starting by establishing the notion of preservation/creation/reflection of initial algebras and terminal coalgebras. 
\subsection{Preservation, reflection and creation of initial algebras}\label{sect:appendix-has-inductive-types}

We begin by recalling a fundamental result on lifting functors from the base categories to the categories of algebras in Lemma \ref{lem:comparsion-of-algebras-basic-lemma}. This is actually related to the universal property of the categories of algebras.

\begin{lemma}\label{lem:comparsion-of-algebras-basic-lemma}
	Let $F : \catD\to\catC $ be a functor. 
	Given endofunctors $E : \catC\to\catC $, $E': \catD\to\catD $ and a natural transformation $\gamma : E  \circ F\longrightarrow F\circ E '   $, we have an induced functor defined by
	\begin{eqnarray*}
		\check{F}_{\gamma } : & E'\AAlg       & \to E\AAlg  \\
		&  \left( X, \zeta \right) & \mapsto \left( F(X), F ( \zeta ) \circ \gamma _X \right)\\
		& g      & \mapsto  F(g).
	\end{eqnarray*}               
\end{lemma}	
\begin{proof}
	Indeed, if $g : W\to Z $ is the underlying morphism of an algebra morphism between $(W, \zeta ) $ and $(Z, \xi ) $, we have that
	\begin{align*}
		&  F(g)\circ  F( \zeta ) \circ \gamma _W  \\
		& = F( \xi ) \circ FE '(g) \circ \gamma _ W \explainr{$f:(W, \zeta )\to (Z, \xi)$}\\
		& =  F( \xi ) \circ \gamma _Z \circ EF (g) \explainr{naturality of $\gamma $}
	\end{align*}
	which proves that $F(g) $ in fact gives a morphism between the algebras $ \left( F(W), F(\zeta ) \circ \gamma _W \right) $ and $\left( F (Z), F (\xi ) \circ \gamma _Z \right) $. The functoriality of $\check{F} _\gamma $ follows, then, from that of $F$.
\end{proof}

Dually, we have:

\begin{lemma}\label{lem:comparsion-of-terminal-coalgebras-basic-lemma}
	Let $E : \catC\to\catC $, $G : \catC\to\catD $, and $E': \catD\to\catD $ be functors. 
	Each natural transformation $\beta :G\circ E \longrightarrow  E'\circ G$ induces a 
	functor 
	\begin{eqnarray*}
		\tilde{G}^{\beta } : & E\CCoAlg      & \to E'\CCoAlg  \\
		& (W, \xi ) & \mapsto  \left( G(W), \beta _W\circ G ( \xi ) \right)\\
		& f      & \mapsto  G(f).
	\end{eqnarray*}               
\end{lemma}	

Below, whenever we talk about strict preservation,
we are assuming that we have chosen initial objects (terminal objects) in the respective categories of (co)algebras.

We can, now, establish the definition of preservation, reflection and creation of initial algebras 
using the respective notions for the induced functor. More precisely:

\begin{definition}[Preservation, reflection and creation of initial algebras]\label{def:preservation-of-initial-algebras-of-a-specific-endofunctor}\label{def:preservation-of-initial-algebras}
	We say that a functor $F : \catD\to \catC $ \textit{(strictly) preserves the initial algebra/reflects the initial algebra/creates the initial algebra of the endofunctor $E: \catC\to \catC $} if, whenever $E' : \catD\to\catD $ is such that $\gamma : E\circ F\cong F\circ E '  $ (or, in the strict case, $F\circ E ' = E\circ F$),
	the functor 
	\begin{eqnarray*}
		\check{F}_ \gamma  : & E'\AAlg       & \to E\AAlg  \\
		&  \left( X , \zeta \right) & \mapsto  \left( F (X), F ( \zeta )\circ \gamma _X  \right)\\
		& g      & \mapsto  F (g).
	\end{eqnarray*} 	 
	induced by $\gamma$  strictly) preserves the initial object/reflects the initial object/creates the initial object.

Finally, we say that a functor $F : \catD\to \catC $ \textit{(strictly) preserves initial algebras/reflects initial algebras/creates initial algebras} if $F$ (strictly) preserves initial algebras/reflects initial algebras/creates initial algebras of any endofunctor on $\catD $.
\end{definition}

\begin{remark}	
	In other words, let $F : \catD\to \catC $ be a functor. 
	\begin{enumerate}[(I)]
		\item We say that $F$ (strictly) preserves initial algebras, if:  for any natural isomorphism  $\gamma : E\circ F\cong F\circ E ' $ (or, in the strict case, for each identity $E\circ F = F\circ E ' $) in which $E$ and $E'$ are endofunctors, assuming that  $ \left( \mu E', \ind _{E'} \right) $ is the initial $E '$-algebra, the $E$-algebra $ \left( F \left( \mu E'\right), F \left( \ind _{E'} \right)\circ \gamma _{\mu E' } \right) $ is an initial object of $E\AAlg $ (the chosen initial object of $E\AAlg $, in the strict case).
		\item We say that $F$ reflects initial algebras, if:  for any natural isomorphism  $\gamma : E\circ F\cong F\circ E '  $ 
		in which $E $ and $E'$ are endofunctors, if  $ \left( F(Y), F\left( \xi\right)\circ \gamma _ Y  \right) $ is an initial $E $-algebra and
		$(Y, \xi) $ is an $E'$-algebra, then 	$(Y, \xi) $ is an initial $E'$-algebra.
		\item We say that $F$ creates initial algebras if: (A) $F$ reflects and preserves initial algebras and, moreover, (B) for any $\gamma : E\circ F\cong F\circ E '  $ in which $E$ and $E'$ are endofunctors, $E '\AAlg $ has an initial algebra if $E\AAlg $ does.	
	\end{enumerate}
\end{remark}

\begin{definition}[Preservation, reflection and creation of terminal coalgebras]\label{def:preservation-of-terminal-coalgebras-specific-endofunctor}\label{def:preservation-of-terminal-coalgebras}
	We say that a functor $G : \catC\to \catD $ \textit{(strictly) preserves the initial algebra/reflects the initial algebra/creates the initial algebra} of an endofunctor $E:\catC\to\catC $  if, for any 
	natural isomorphism  $\beta : G\circ E \cong E'\circ G $ (or, in the strict case, $GE = E'G$),
	the functor 
	\begin{eqnarray*}
		\tilde{G}^\beta : & E\CCoAlg       & \to E'\CCoAlg \\
		& \left( W, \xi \right) & \mapsto \left( G(W), \beta_W\circ G ( \xi ) \right)\\
		& f     &  \mapsto G(f).
	\end{eqnarray*} 	 
	induced by $\beta$  (strictly) preserves the terminal object/reflects the terminal object/creates the terminal object.
	
	Finally, we say that \textit{$G : \catC\to \catD$ (strictly) preserves terminal coalgebras/reflects terminal coalgebras/creates terminal coalgebras} if $G$ (strictly) preserves terminal coalgebras/reflects terminal coalgebras/creates terminal coalgebras of any endofunctor on $\catC $.	
\end{definition}

\subsection{$\mu\nu$-polynomial preserving functors} 
Finally, we can introduce the concept of a structure-preserving functor for $\mu\nu$-polynomials.
\begin{definition}
A functor $G: \catD\to\catC$ \textit{(strictly) preserves $\mu\nu$-polynomials} if it strictly preserves finite coproducts, finite products, as well as initial algebras and terminal coalgebras of $\mu\nu$-polynomials.
\end{definition}

\section{An expressive functional language as a source language for AD}\label{sec:source-language}
We describe a source language for our AD code transformations.
We consider a standard total functional programming language with an expressive 
type system, over ground types $\reals^n$ for arrays of real numbers of static 
length $n$, for all $n\in \NN$, and sets $\Op_{n_1,...,n_k}^m$ of primitive operations $\op$, for all $k, m, n_1,\ldots, n_k\in \NN$.
These operations $\op$ will be interpreted as differentiable functions $(\RR^{n_1}\times \cdots\times \RR^{n_k})\To \RR^m$ and the reader can keep the following examples in mind:
\begin{itemize}
  \item constants $\cnst\in \Op_{}^n$ for each $c\in \RR^n$, for which 
   we slightly abuse notation and write $\cnst(\tUnit)$ as $\cnst$;
    \item elementwise addition and product $(+),(*)\!\in\!\Op_{n,n}^n$
    and matrix-vector product $(\star)\!\in\!\Op_{n\cdot m, m}^n$;
    \item operations for summing all the elements in an array: $\tSum\in\Op_{n}^1$;
    \item some non-linear functions like the sigmoid function $\sigmoid\in \Op_{1}^1$.
\end{itemize} 

Its kinds, types and terms are generated by the 
grammar in Fig. \ref{fig:sl-terms-types-kinds}.
\begin{figure}[!ht]
  \framebox{\begin{minipage}{0.98\linewidth}\qquad
\input{sl-terms-types-kinds}
  \end{minipage}}
\caption{\label{fig:sl-terms-types-kinds} Grammar for the kinds, types 
and terms of the source language for our AD transformations.}
\end{figure}
We write $\Delta\vdash\ty:\type$ to specify that the type $\ty$ is \emph{well-kinded} in 
\emph{kinding context} $\Delta$, where $\Delta$ is a list of the form $\alpha_1:\type,\ldots,\alpha_n:\type$.
The idea is that the type variables identifiers $\alpha_1,\ldots, \alpha_n$ can be used in the formation of $\ty$.
These kinding judgements are defined according to the rules displayed in 
Fig. \ref{fig:sl-kind-system}.
\begin{figure}[!ht]
\framebox{\begin{minipage}{0.98\linewidth}\noindent\hspace{-24pt}\input{sl-kind-system}\end{minipage}}
    \caption{Kinding rules for the AD source language.
    Note that we only consider the formation of function types 
    of non-parameterized types (shaded in grey).\label{fig:sl-kind-system}}\;
  \end{figure}
We write $\DGinf{\trm}{\ty}$ to specify that the term $\trm$ is \emph{well-typed} in the \emph{typing context} 
$\Gamma$, where $\Gamma$ is a list of the form $\var_1:\ty_1,\ldots,\var_n:\ty_n$ for variable identifiers $\var_i$ and types 
$\ty_i$ that are well-kinded in kinding context $\Delta$.
These typing judgements are defined according to the rules displayed in 
Fig. \ref{fig:sl-type-system}.
\begin{figure}[!ht]
\framebox{\begin{minipage}{0.98\linewidth}\noindent\hspace{-24pt}\input{sl-type-system}\end{minipage}}
    \caption{Typing rules for the AD source language.\label{fig:sl-type-system}}\;
  \end{figure}
As Fig. \ref{fig:sl-equations} displays, we consider the terms of our language up to the standard $\beta\eta$-theory.
To present this equational theory, we define in Fig. \ref{fig:sl-types-functorial-action}, by induction, some syntactic sugar for the functorial action $\Delta,\Delta'\mid\Gamma,\var:\subst{\ty}{\sfor{\tvar}{\ty[2]}}\vdash \subst{\ty}{\sfor{\tvar}{\var\vdash \trm}} :\subst{\ty}{\sfor{\tvar}{\ty[3]}}$ in argument $\tvar$ 
of parameterized types $\Delta,\tvar:\type\vdash \ty:\type$ on terms $\Delta'\mid\Gamma,x:\ty[2]\vdash\trm:\ty[3]$.
\begin{figure}[!ht]
  \framebox{\begin{minipage}{0.98\linewidth}\hspace{-24pt} 
\input{sl-types-functorial-action}
\end{minipage}}
\caption{\label{fig:sl-types-functorial-action} Functorial action $\Delta,\Delta'\mid\Gamma,\var:\subst{\ty}{\sfor{\tvar}{\ty[2]}}\vdash \subst{\ty}{\sfor{\tvar}{\var\vdash \trm}} :\subst{\ty}{\sfor{\tvar}{\ty[3]}}$ in argument $\tvar$ 
of parameterized types $\Delta,\tvar:\type\vdash \ty:\type$ on terms $\Delta'\mid\Gamma,x:\ty[2]\vdash\trm:\ty[3]$ of the source language.}
\end{figure}%

We employ the usual conventions of free and bound variables and write
 $\subst{\ty}{\sfor{\tvar}{\ty[2]}}$ for the capture-avoiding substitution of
 the type $\ty[2]$ for the identifier $\tvar$ in $\ty$ (and similarly, 
 $\subst{\trm}{\sfor{\var}{\trm[2]}}$ for the capture-avoiding substitution of the term 
 $\trm[2]$ for the identifier $\var$ in $\trm$).
 We define make liberal use of the standard syntactic sugar $\letin{\tPair{\var}{\var[2]}}{\trm}{\trm[2]}\defeq \letin{\var[3]}{\trm}{\letin{\var}{\tFst\var[3]}{\letin{\var[2]}{\tSnd\var[3]}{\trm[2]}}}$.

  \begin{figure}[!ht]
    \framebox{\begin{minipage}{0.98\linewidth}\hspace{-24pt} \input{sl-equations}
   \end{minipage}}
   \caption{We consider the standard $\beta\eta$-laws above for our language.
   We write $\freeeq{\var_1,\ldots,\var_n}$ to indicate that the variables $\var_1,\ldots,\var_n$ need to be fresh in the left hand side.
   Equations hold on pairs of terms of the same type.
   As usual, we only distinguish terms up to $\alpha$-renaming of bound variables.\label{fig:sl-equations}\;
  }
   \end{figure}

This standard language is equivalent to the freely generated bicartesian closed category 
$\Syn$ with $\mu\nu$-polynomials on the directed polygraph (computad) given by the ground types $\reals^n$ as objects and primitive operations $\op$ as arrows.
Equivalently, we can see it as the initial category that supports tuple types, function types, sum types, inductive and coinductive types and primitive types 
$\Ty=\set{\reals^n\mid n\in\NN}$ and primitive operations $\Op(\reals^{n_1},\ldots,\reals^{n_k};\reals^m)=\Op_{n_1,\ldots,n_k}^m$ (in the sense of \S\ref{sec:background-categorical-semantics}).
$\Syn$ effectively represents programs as (categorical) combinators, also known as ``point-free style''
in the functional programming community.
Concretely, $\Syn$ has types as objects,
homsets $\Syn(\ty,\ty[2])$ consist of $(\alpha)\beta\eta$-equivalence classes of terms 
$\cdot\mid\var:\ty\vdash \trm:\ty[2]$,
identities are $\cdot\mid\var:\ty\vdash\var:\ty$,
and the composition of $\cdot\mid\var:\ty\vdash \trm:\ty[2]$ and $\cdot\mid\var[2]:\ty[2]\vdash \trm[2]:\ty[3]$ is given by $\cdot\mid\var:\ty\vdash \letin{\var[2]}{\trm}{\trm[2]}:\ty[3]$.
\begin{corollary}[Universal property of $\Syn$] \label{cor:universal-property-of-source-language}
Given any bicartesian closed category with $\mu\nu$-polynomials $\catC$, any consistent assignment of $F(\reals^n )\in\ob\catC$ and $F(\op)\in \catC(F(\reals^{n_1})\times \cdots\times F(\reals^{n_k}), F(\reals^m))$ for $\op\in\Op_{n_1,\ldots,n_k}^m$ extends to a unique $\mu\nu$-polynomial preserving bicartesian closed functor $F:\Syn\to\catC$.
\end{corollary}

\section{Modelling expressive functional languages in Grothendieck constructions}\label{sec:grothendieck-constructions}
In this section, we present a novel construction of categorical models (in the sense of \S\ref{sec:background-categorical-semantics})
$\Sigma_\catC\catL$ and $\Sigma_\catC\catL^{op}$
of expressive functional languages (like our AD source language of \S\ref{sec:source-language}) in $\Sigma$-types of suitable 
indexed categories $\catL:\catC^{op}\to \Cat$.
In particular, the problem we solve in this section is to 
identify suitable sufficient conditions to put on an indexed category 
$\catL:\catC^{op}\to \Cat$, whose base category we think of 
as the semantics of a cartesian type theory and whose 
fibre categories we think of as the semantics of a dependent linear type theory, such that $\Sigma_\catC\catL$ and $\Sigma_\catC\catL^{op}$ are categorical models of expressive functional languages in this sense.
We call such an indexed category a \emph{$\Sigma$-bimodel of} language feature 
$X$ if it satifies our sufficient conditions for $\Sigma_\catC\catL$ and $\Sigma_\catC\catL^{op}$ to be categorical models of language feature $X$.

This abstract material in many ways forms the theoretical crux of this paper.
We consider two particular instances of this idea later:
\begin{itemize}
\item the case where $\catL$ is the syntactic category $\LSyn:\CSyn^{op}\to \Cat$ of a suitable target language for AD translations (\S\ref{sec:target-language}); the universal property of the source language $\Syn$ then yields unique structure preserving functors $\Dsynsymbol:\Syn\to\Sigma_{\CSyn}\LSyn$ and $\Dsynrevsymbol:\Syn\to\Sigma_{\CSyn}\LSyn^{op}$  implementing forward and reverse mode AD; 
\item the case where $\catL$ is the indexed category of families of real vector spaces $\FVect:\Set^{op}\to \Cat$ (\S\ref{sec:concrete-models}); this gives a concrete denotational semantics to the target language, which we use in the correctness proof of AD.
\end{itemize}

\subsection{Basics: the categories $\Sigma_\catC \catL$ and $\Sigma_\catC \catL^{\op}$}
Recall that for any strictly indexed category, i.e. a (strict) functor $\catL:\catC^{\op}\to\Cat$, 
we can consider its total category (or Grothendieck construction) $\Sigma_\catC \catL$,
which is a fibred category over $\catC$  (see \citep[sections A1.1.7, B1.3.1]{johnstone2002sketches}).
We can view it as a $\Sigma$-type of categories, which 
generalizes the cartesian product.
Further, given a strictly indexed category $\catL:\catC^{\op}\to \Cat$,
we can consider its fibrewise dual category $\catL^{\op}:\catC^{\op}\to \Cat$,
which is defined as the composition $\catC^{\op}\xto{\catL}\Cat\xto{\op}\Cat$, where $\op$ is defined by 
$A\mapsto A^\op $.
Thus, we can apply the same construction to $\catL^{\op}$ to obtain a category~$\Sigma_{\catC}\catL^{\op}$.

Concretely, $\Sigma_\catC \catL$ is the following category:
\begin{itemize}
    \item objects are pairs $(W,w)$ of an object $W$ of $\catC$ and an object $w$ of $\catL(W)$;
    \item morphisms $(W,w)\to (X,x)$ are pairs $(f, \ff{f} )$ with $f :W\to X$ in $\catC$ and $\ff{f} : w \to \catL(f)(x)$ in  $\catL(W)$;
    \item identities $\id[(W,w)]$ are $(\id[W],\id[w])$;
    \item composition of $(W,w)\xto{(f,\ff{f})}(X,x)$
    and $(X,x)\xto{(g,\ff{g})}(Y,y)$ is given by
    $$(g\circ f, \catL(f)(\ff{g} ) \circ \ff{f} ) .$$
    \end{itemize}
    Concretely, $\Sigma_{\catC}\catL^{\op}$ is the following category:
\begin{itemize}
    \item objects are pairs $(W, w)$ of an object $W$ of $\catC$ and an object $w$ of $\catL(W)$;
    \item morphisms $(W,w)\to (X,x)$ are pairs $(f, \ff{f} )$ with $f :W\to X$ in $\catC$ and $\ff{f} :  \catL(f)(x)\to w $ in  $\catL(W)$;
    \item identities $\id[(W,w)]$ are $(\id[W],\id[w])$;
  \item composition of $(W,w)\xto{(f,\ff{f})}(X,x)$
and $(X,x)\xto{(g,\ff{g})}(Y,y)$ is given by
$$(g\circ f, \ff{f}\circ \catL(f)(\ff{g} ) ) .$$
\end{itemize}

\subsection{Products in total categories}\label{subsect:cartesian-structure-sigmatypes-Grothendieck-Construction}
We start by studying the cartesian structure of $\Sigma_\catC\catL$. We refer to \citep{MR0213413} for a basic reference for fibrations/indexed categories and properties of the total category.

\begin{definition} 
A strictly indexed category $\catL$ has \emph{strictly indexed finite (co)products} if 
\begin{enumerate}[i)]
	\item each fibre $\catL(C)$ has chosen finite (co)products $(\times , \terminal)$ (respectively, $(\sqcup , \initial)$);
	\item change of base strictly preserves these (co)products in the sense that $\catL(f)$ preserves finite products (respectively, finite coproducts) for all morphisms $f$ in $\catC$.
\end{enumerate}
\end{definition} 

We \textit{recall} the well-known fact that
$\Sigma_\catC \catL$  ($\Sigma_\catC \catL^{\op}$) has finite products if $\catC $ has finite products and 
$\catL $ has indexed finite products (coproducts).

\begin{proposition}[Cartesian structure of $\GrothC \catL $]\label{prop:grothendieck-products-covariant}
	Assuming that $\catC$ has finite products $(\terminal,\times)$ 
	and $\catL$ has indexed finite products $(\terminal,\times)$, we have that
	$\Sigma_{\catC}\catL$ has (fibred) terminal object 
	$\terminal =\left(\terminal,\terminal\right)$ and (fibred) binary product 
	$(W,w)\times (Y,y)=(W\times Y,\catL(\pi_1)(w)\times  \catL(\pi_2)(y))$.
	\end{proposition}
	\begin{proof}
		We have (natural) bijections 
		\vspace{-2pt}\\
		\resizebox{\linewidth}{!}{\parbox{\linewidth}{
		\begin{align*}
		&\Sigma_{\catC}\catL((X,x), (\terminal,\terminal))\\
		&=\Sigma_{f\in \catC(X,\terminal)} \catL(X)(x ,\catL(f)(\terminal)) \explainr{by definition}\\
		&\cong
		\Sigma_{f\in \catC(X,\terminal)} \catL(X)(x,\terminal)\explainr{indexed $\terminal$}\\
		&\cong\terminal\times \terminal\explainr{$\terminal$  terminal in $\catC$ and $\catL(X)$}\\
		&\cong \terminal
		\end{align*} \small 
		\begin{align*}
		&\Sigma_{\catC}\catL\left( (X,x),(W\times Z , \catL(\pi_1)(w)\times  \catL(\pi_2)(z))\right) \\
		&=\Sigma_{\sPair{f}{g}\in \catC ( X, W\times Y )} \catL( X )( x,\catL\sPair{f}{g}(\catL(\pi_1)( w )\times  \catL(\pi_2)( z ))) \explainr{by definition}\\
		&\cong\Sigma_{\sPair{f}{g}\in \catC( X , W\times Z )} \catL( X )(x ,\catL\sPair{f}{g}\catL(\pi_1)( w )\times \catL\sPair{f}{g} \catL(\pi_2)( z ))\explainr{indexed $\times$}\\
		&=\Sigma_{\sPair{f}{g}\in \catC( X , W\times Z )} \catL(X)( x ,\catL(f)(w)\times \catL(g)(z))\explainr{functoriality $\catL$}\\
		&\cong \Sigma_{\sPair{f}{g}\in \catC( X, W\times Z )} \catL(X)(x,\catL(f)(w))\times 
		\catL(X)(x, \catL(g)(z))\explainr{$\times$ product in $\catL(A_1)$}\\
		&\cong \Sigma_{f\in \catC(X, W)}\Sigma_{g\in \catC(X, Z)} \catL(X)(x,\catL(f)(w))\times 
		\catL(X)(x, \catL(g)(z))\explainr{$\times$ product in $\catC$}\\
		&\cong \left(\Sigma_{f\in \catC( X, W)}\catL(X)(x,\catL(f)(w))\right) \times \left(\Sigma_{g\in \catC(X, Z)} 
		\catL(X)(x, \catL(g)(z))\right)\explainr{Beck-Chevalley for $\Sigma$ in $\Set$}\\
		&=\Sigma_{\catC}\catL((X,x),(W,w))\times \Sigma_{\catC}\catL((X,x), (Z,z)).
	\end{align*}
		}}\end{proof}\normalsize
	In particular, finite products in $\Sigma_{\catC}\catL$ are fibred in the sense that the projection functor $\Sigma_{\catC}\catL\to \catC$ preserves them, on the nose.
	Codually, we have:
\begin{proposition}[Cartesian structure of $\GrothC \catL ^\op $]\label{theo:grothendieck-products-contravariant}
	Assuming that $\catC$ has finite products $(\terminal,\times)$ 
	and $\catL$ has indexed finite coproducts $(\initial,\sqcup )$, we have that
	$\Sigma_{\catC}\catL^{\op}$  has (fibred) terminal object 
	$\terminal=(\terminal,\initial )$ and (fibred) binary product 
	$(W,w)\times (Y,y)=(W\times Y,\catL(\pi_1)(w)\sqcup  \catL(\pi_2)(y))$.
\end{proposition}

That is, in our terminology, $\catL:\catC^{op}\to \Cat$ is a \emph{$\Sigma$-bimodel of 
tuple types} if $\catC$ has chosen finite products and $\catL$ has finite strictly indexed products and coproducts.
	
We will, in particular, apply the results above in the situation where $\catL$ has indexed finite biproducts in the sense of Definition~\ref{def:strictly-indexed-biproducts}, in which case the finite product structures of $\Sigma_\catC\catL$ and $\Sigma_\catC\catL^{op}$ coincide.

\begin{definition}[Strictly indexed finite biproducts]\label{def:strictly-indexed-biproducts}
A category with finite products and coproducts is \textit{semi-additive} if the binary coproduct functor is naturally isomorphic to the binary product functor; see, for instance, \citep{MR2864762, arXiv:1711.02051}.	In this case, the product/coproduct is called \textit{biproduct}, and the biproduct structure 
is denoted by $(\times, \terminal) $ or $(+, \initial)$.

A strictly indexed category $\catL$ has \emph{strictly indexed finite biproducts}  if 
	\begin{enumerate}[--]
		\item $\catL $ has strictly indexed finite products and coproducts;
		\item each fibre $\catL(C)$ is semi-additive.
	\end{enumerate}
\end{definition}

\subsection{Generators}\label{subsect:generators-Grothendieck-construction}
In this section, we establish the obvious sufficient (and necessary)
conditions for $\Sigma_\catC\catL$ and $\Sigma_\catC\catL^{op}$ to model 
primitive types and operations in the sense of \S\ref{sec:background-categorical-semantics}.
These conditions are an immediate consequence of the structure of $\Sigma_\catC\catL$ and $\Sigma_\catC\catL^{op}$ as cartesian categories.

\begin{definition}
We say that $\catL:\catC^{op}\to \Cat$ is a \emph{$\Sigma$-bimodel of primitive types $\Ty$ and operations $\Op$} if
\begin{itemize}
	\item for all $T\in\Ty$, we have a choice of objects $C_T\in \ob\catC$ and $L_T,L'_T\in \ob\catL(C_T)$;
	\item  
	for all $\op\in \Op(T_1,\ldots, T_n; S)$, we have a choice of morphisms
	\begin{align*}
		&f_{\op}\in \catC(C_{T_1}\times \ldots\times C_{T_n}, C_S)\\
		&g_{\op}\in \catL(C_{T_1}\times \ldots\times C_{T_n})(\catL(\pi_1)(L_{T_1})\times\cdots\times \catL(\pi_n)(L_{T_n}), \catL(f_{\op})(L_S))\\
		&g'_\op\in \catL(C_{T_1}\times \ldots\times C_{T_n})(\catL(f_{\op})(L'_S),\catL(\pi_1)(L'_{T_1})\sqcup\cdots\sqcup \catL(\pi_n)(L'_{T_n})).
	\end{align*}
\end{itemize}
We say that such a model has \emph{self-dual primitive types} in case $L_T=L'_T$ for all $T\in\Ty$.
\end{definition}

\subsection{Cartesian closedness of total categories}\label{subsect:closed-structure-sigmatypes-Grothendieck-Construction}
The question of Cartesian closure of the categories $\Sigma_\catC \catL$ and $\Sigma_\catC \catL^{op}$ 
is a lot more subtle.
In particular, the formulas for exponentials tend to involve $\Pi$- and $\Sigma$-types, hence we need to recall some 
definitions from categorical dependent type theory.
As also suggested by \citep{kerjean:hal-03123968}, these formulas relate closely to the Diller-Nahm variant \citep{diller1974variante,Hyland02,glehnmoss2018}
of the Dialectica interpretation \citep{godel1958bisher} and \citep{altenkirch2010higher}'s formula for higher-order containers.
We plan to explain this connection in detail in future work as it would form a distraction from the point of the current paper.

We use standard definitions from the semantics of dependent 
type theory and the dependently typed enriched effect calculus.
An interested reader can find background on this material in \citep[Chapter 5]{vakar2017search}
and \citep{ahman2016dependent}.
We briefly recalling some of the usual vocabulary~\citep[Chapter 5]{vakar2017search}.
\begin{definition}
Given an indexed category $\catD:\catC^{op}\to \Cat$, we say:
\begin{itemize}
	\item it satisfies the \emph{comprehension axiom} if: $\catC$ has a chosen terminal object $\terminal$;
	$\catD$ has strictly indexed terminal objects $\terminal$ (i.e. chosen terminal objects $\terminal\in\catD(X)$, such 
	that $\catD(g)(\terminal)=\terminal\in\catD (W) $ for all $g:W\to X$ in $\catC$);  and, for each object $\left( X, x\right)\in \Sigma_\catC \catD $, the functor  
	\begin{align*}
	\repP{(X,x)}: (\catC/X)^{op} &\to \Set \\
	\left( W, W\xto{f}X \right) & \mapsto \catD(W)(\terminal, \catD(f)(x))
	\end{align*}
	are representable by an object $\left( X.x, X.x\xto{\depproj{X}{x}} X\right) $ of $\catC/X$:
	\begin{align*}
	\repP{(X,x)}\left( W, W\xto{f}X \right) = \catD(W)(\terminal, \catD(f)(x))  &\cong \catC/X\left( \left( W, f \right) ,\left( X.x, \depproj{X}{x}\right) \right)\\
	b&\mapsto \sPair{f}{b}.
	\end{align*}
	We write $\depv{X}{x}$ for the unique element of $\catD(X.x)(\terminal, \catD(\depproj{X}{x})(x))$ such that $\sPair{\depproj{X}{x}}{\depv{X}{x}}=\id[\depproj{X}{x}]$ (the universal element of the representation). 
	
	Furthermore, given $f: W\to X $, we write $\depq{f}{b}$ for the unique morphism $\sPair{f\circ \depproj{W}{\catD(f)(x)}}{\depv{W}{\catD(f)(x)}}$ making the  square below a pullback:  \\
\begin{center} 
\begin{tikzcd}
	W.\catD(f)(x) \arrow[r, "\depq{f}{x}"] \arrow[d, "\depproj{X}{\catD(f)(x)}"'] & X.x \arrow[d, "\depproj{X}{x}"] \\
	W \arrow[r, "f"']                                                              & X                             
\end{tikzcd}
\end{center} 	
We henceforth call such squares $\mathbf{p}$-squares; 
	\item it \emph{supports $\Sigma$-types} if we have left adjoint functors $\Sigma_w\dashv \catD(\depproj{W}{w}):\catD(W.w)\leftrightarrows \catD(W)$
	satisfying the left Beck-Chevalley condition for $\mathbf{p}$-squares w.r.t. $\catD $ (this means that $\catD(f) \circ \left( \Sigma_{\catD(f)(x)} \to \Sigma_x\right) \circ \catD(\depq{f}{x}) $ are the identity);
	\item it \emph{supports $\Pi$-types} if $\catD^{op}$ supports $\Sigma$-types; explicitly, that is the case iff we have right adjoint functors $\catD(\depproj{W}{w})\dashv \Pi_w:\catD(W)\leftrightarrows \catD(W.w)$ satisfying the right Beck-Chevalley condition for $\mathbf{p}$-squares in the sense that the canonical maps $\Pi_{\catD(f)(x)} \circ \left( \catD(f)\to\catD(\depq{f}{x})\right) \circ \Pi_x$ are the identity.
\end{itemize}
\end{definition} 
\begin{definition} 
In case $\catD:\catC^{op}\to \Cat$ satisfies the comprehension axiom, we say that
\begin{itemize}
	\item it satisfies \emph{democratic comprehension} if the comprehension functor 
	\begin{align*}
	\catD(W)(w',w) & \xto{\depproj{W}{-}} \catC/W\left( \left( W.w', \depproj{W}{w'}\right) , \left( W.w, \depproj{W}{w}\right) \right)\\
	d & \mapsto \sPair{\depproj{W}{w'}}{\catD(\depproj{W}{w'})(d)\circ \depv{W}{w'}}
	\end{align*}
	defines an isomorphism of categories $\catD(\terminal )\cong \catC/\terminal \cong \catC$;
	\item it satisfies \emph{full/faithful comprehension} if the comprehension functor 
	is full/faithful;
	\item it \emph{supports (strong) $\Sigma$-types} (i.e. $\Sigma$-types with a dependent elimination rule,
	which in particular makes $\catD$ support $\Sigma$-types) if dependent projections compose: for all triple $\left( W, w , s\right)$ where $W\in\catC$, $w\in \objects\left(\catD(W)\right) $ and $s\in \objects\left( \catD(W.w)\right) $, we have 
	$$\depproj{W}{w}\circ \depproj{W.w}{s}\cong\depproj{W}{\Sigma_w s};$$
	then, in particular, $W.\Sigma_w s\cong W.w.s$; further, we have projection morphisms $\pi_1\in\catD(W)(\Sigma_w s, w)$
	and $\pi_2\in \catD(W.w)(\terminal, s)$;
\end{itemize}
\end{definition} 
\begin{remark}[$\Sigma$- and $\Pi$- as dependent product and function types]
In case $\catD$ satisfies fully faithful comprehension,
\begin{itemize}
\item $\Sigma_w \catD(\depproj{W}{w})(v)$ gives the categorical product $w\times v$ of $w$ and $v$ in $\catD(W)$;
\item $\Pi_w \catD(\depproj{W}{w})(v)$ gives the categorical exponential $w\Rightarrow v$ of $w$ and $v$ in $\catD(W)$.
\end{itemize}
\end{remark}
\begin{definition}[$\Sigma$-bimodel for function types]
We call a strictly indexed category $\catL:\catC^{op}\to\Cat$ a \emph{$\Sigma$-bimodel for function types} if it is a biadditive model of the dependently typed enriched effect calculus in the sense that it comes equipped with
\begin{enumerate}[$\catL$A)]
\item a model of cartesian dependent type theory in the sense of a 
strictly indexed category $\catC':\catC^{op}\to \Cat$ that satisfies full, faithful, democratic comprehension with $\Pi$-types and strong $\Sigma$-types;
\item strictly indexed finite biproducts in the sense of Definition~\ref{def:strictly-indexed-biproducts} in $\catL$;
\item $\Sigma$- and $\Pi$-types in $\catL$;
\item a strictly indexed functor $\multimap: \catL^{\op}\times\catL\to\catC'$ and a natural isomorphism
$$
\catL(W)(w,x)\cong \catC'(A)(\terminal, w\multimap x).
$$
\end{enumerate}
\end{definition}

We can immediately note that our notion of $\Sigma$-bimodel of function types is also a $\Sigma$-bimodel of tuple types.
Indeed, strong $\Sigma$-types and comprehension give us, in particular, chosen finite products in $\catC$.

We next show why this name is justified: we show that the Grothendieck construction of a
$\Sigma$-bimodel of function types is cartesian closed.\footnote{This is a generalization of the proof given in   \citep{vakar2020reverse}, where the result is established for locally indexed categories.}

In the following, we slightly abuse notation to aid legibility:
\begin{itemize}
\item denoting by $\morterminal{W}: W\to\terminal $ the only morphism, we will sometimes conflate $Z\in\objects\catC'(\terminal )$ and $\terminal .Z\in\ob\catC$ as well as 
$f\in \catC'(W)(\terminal, \catC'(\morterminal{W} )(Z))$ and $\sPair{\morterminal{W} }{f}\in \catC(W, \terminal .Z)$); this is justified by the democratic comprehension axiom;
\item we will sometimes simply write $z$ for $\catD(\depproj{W}{w})(z)$ where the weakening map $\catD(\depproj{W}{w})$ is clear from context.
\end{itemize}
Given $X, Y\in \catC$ we will write $\evf$ for the obvious $\catC$-morphism
$$
\evf:\Pi_{X} \Sigma_{Y}Z.X\to Y,
$$
that is, the morphism obtained as the composition (where we write $\pi_1$ for the projection $\Sigma_{Y}Z\to Y$)
$$
\Pi_{X}\Sigma_{Y}Z.X\cong (\Pi_{X}\Sigma_{Y}Z)\times X\xto{(\Pi_{X}\pi_1)\times X}(\Pi_{X}Y)\times X\cong (X\Rightarrow Y)\times X\xto{\mathrm{ev}}Y
$$
With these notational conventions in place, we can describe the cartesian closed structure of Grothendieck constructions.

\begin{therm}[Exponentials of the total category]\label{theo:grothendieck-ccc-covariant}
	For a $\Sigma$-bimodel $\catL$ for function types,
	$\Sigma_{\catC}\catL$ has exponential $$(X,x)\Rightarrow (Y,y)=
	(\Pi_{X}\Sigma_{Y} \catL(\pi_1)(x)\multimap \catL(\pi_2)(y),
	\Pi_{X} \catL(\evf)(y)).$$
\end{therm}
\begin{proof}
	We have (natural) bijections 
	\vspace{-2pt}\\
	\resizebox{\linewidth}{!}{\parbox{\linewidth}{
		\begin{align*}
			&\Sigma_{\catC}\catL((W,w)\times (X,x), (Y,y)) =\\
			&=\Sigma_{\catC}\catL((W\times X,\catL(\pi_1)(w)\times \catL(\pi_2)(x)), (Y,y)) \explainr{by Prop.~\ref{prop:grothendieck-products-covariant}}\\
			&=\Sigma_{f\in \catC(W\times X, Y)}\catL(W\times X)(\catL(\pi_1)(w)\times \catL(\pi_2)(x), \catL(f)(y)) \explainr{by definition}\\
			&\cong\Sigma_{f\in \catC(W\times X, Y)}\catL(W\times X)(\catL(\pi_1)(w), \catL(f)(y))\times \catL(W\times X)(\catL(\pi_2)(x), \catL(f)(y))\explainr{$\times$ coproduct in $\catL(W\times X)$}\\
			&\cong\Sigma_{f\in \catC(W\times X, Y)}\catL(W)(w, \Pi_{X}  \catL(f)(y))\times \catL(W\times X)(\catL(\pi_2)(x), \catL(f)(y))\explainr{$\Pi $-types in $\catL$}\\
			&\cong\Sigma_{f\in \catC(W\times X, Y)}\catL(W)(w, \Pi_{X}  \catL(f)(y))\times \catC'(W\times X)(\terminal, \catL(\pi_2)(x)\multimap\catL(f)(y))\explainr{$\multimap$-types in $\catC'$}\\
			&\cong \Sigma_{\sPair{f}{g}\in \Sigma_{f\in\catC(W\times X, Y)} \catC'(W\times X)(\terminal, \catL(\pi_2)(x)\multimap\catL(f)(y))}\catL(W)(w, \Pi_{X}  \catL(f)(y))\explainr{$\Sigma$-types in $\Set$}\\
			&\cong \Sigma_{\sPair{f}{g}\in \Sigma_{f\in\catC'(W\times X)(\terminal, Y)} \catC'(W\times X)(\terminal, \catL(\pi_2)(x)\multimap\catL(f)(y))}\catL(W)(w, \Pi_{X}  \catL(f)(y))\explainr{comprehension}\\
			&\cong \Sigma_{\sPair{f}{g}\in \catC'(W\times X)(\terminal, \Sigma_{Y}\catL(\pi_2\circ\pi_1)(x)\multimap \catL(\pi_2)(y))}\catL(W)(w, \Pi_{X}  \catL(f)(y))\explainr{strong $\Sigma$-types in $\catC'$}\\
			&=\Sigma_{\sPair{f}{g} \in \catC'(W\times X)(\terminal, \Sigma_{Y} \catL(\pi_2\circ\pi_1)(x)\multimap \catL(\pi_2)(y))}\catL(W)(w, \Pi_{X} \catL(\evf\circ \sPair{\sPair{f}{g}}{\pi_2})((y)))\explainr{definition $\evf$}\\
			&=\Sigma_{\sPair{f}{g} \in \catC'(W\times X)(\terminal, \Sigma_{Y} \catL(\pi_2\circ\pi_1)(x)\multimap \catL(\pi_2)(y))}\catL(W)(w, \Pi_{X} \catL(\sPair{\sPair{f}{g}}{\pi_2})(\catL(\evf)(y)))\explainr{functoriality of $\catL$}\\
			&=\Sigma_{\sPair{f}{g} \in \catC'(W\times X)(\terminal, \Sigma_{Y} \catL(\pi_2\circ\pi_1)(x)\multimap \catL(\pi_2)(y))}\catL(W)(w, \catL(\sPair{f}{g})(\Pi_{X} \catL(\evf)(y)))\explainr{Beck-Chevalley for $\Pi $-types}\\
			&\cong\Sigma_{h \in \catC'(W\times X)(\terminal, \Sigma_{Y} \catL(\pi_2\circ\pi_1)(x)\multimap \catL(\pi_2)(y))}\catL(W)(w, \catL(h)(\Pi_{X} \catL(\evf)(y)))\explainr{strong $\Sigma$-types in $\catC'$}\\
			&=\Sigma_{h \in \catC'(W\times X)(\catL(\pi_1)(\terminal), \Sigma_{Y} \catL(\pi_2\circ\pi_1)(x)\multimap \catL(\pi_2)(y))}\catL(W)(w, \catL(h)(\Pi_{X} \catL(\evf)(y)))\explainr{indexed $\terminal$ in $\catC'$}\\
			&\cong\Sigma_{h \in \catC'(W)(\terminal, \Pi_{X}\Sigma_{Y} \catL(\pi_2\circ\pi_1)(x)\multimap \catL(\pi_2)(y))}\catL(W)(w, \catL(h)(\Pi_{X} \catL(\evf)(y)))\explainr{$\Pi$-types in $\catC'$}\\
			&\cong\Sigma_{h \in \catC(W, \Pi_{X}\Sigma_{Y} \catL(\pi_1)(x)\multimap \catL(\pi_2)(y))}\catL(W)(w, \catL(h)(\Pi_{X} \catL(\evf)(y)))\explainr{comprehension}\\
			&=\Sigma_\catC \catL((W, w),(\Pi_{X}\Sigma_{Y} \catL(\pi_1)(x)\multimap \catL(\pi_2)(y),
			\Pi_{X} \catL(\evf)(y)))\\
			&=\Sigma_\catC \catL((W, w),(X, x)\Rightarrow (Y, y)).
		\end{align*}
}}\end{proof}
\noindent Codually, we have:

	\begin{therm}\label{theo:grothendieck-ccc-contravariant}
		For a $\Sigma$-bimodel $\catL$ for function types,
    $\Sigma_{\catC}\catL^{\op}$ has exponential $$(X,x)\Rightarrow (Y,y)=
    (\Pi_{X}\Sigma_{Y} \catL(\pi_2)(y)\multimap \catL(\pi_1)(x),
    \Sigma_{X} \catL(\evf)(y)).$$
	\end{therm}

Note that these exponentials are not fibred over $\catC$ in the sense that the projection functors $\Sigma_\catC \catL\to \catC$ and $\Sigma_\catC \catL^{op}\to \catC$ are generally not cartesian closed functors.
This is in contrast with the interpretation of all other type formers we consider in this paper.

 \subsection{Coproducts in total categories}\label{sec:appendix-IndedxedCategories-Products-Corpoducts}
We, now, study the coproducts in the total categories $\GrothC\catL $ and $\GrothC\catL ^\op $. We are particularly interested in the case of \textit{extensive indexed categories}, a notion introduced in \ref{subsect:coproduct-structure}.
For future reference, we start by recalling the general case: see, for instance, \citep{MR0213413} for a basic reference on properties of the total categories. 
 
 \begin{proposition}[Initial object in $\GrothC \catL $]\label{prop:initialobject-grothendieckconstruction}
 	Let $\catL : \catC ^\op\to\Cat $ be a strictly indexed category. We assume that
 	\begin{enumerate}[i)]
 		\item $\catC $ has initial object $\initialobject $;
 		\item $\catL (\initialobject ) $ has initial object, denoted, by abuse of language, by $\initialobject $.
 	\end{enumerate}	
 	In this case, $\left( \initialobject  ,\initialobject   \right) $
 	is the initial object of $\GrothC \catL $.
 \end{proposition}	
 \begin{proof}
 	Assuming the hypothesis above, given any object $(Y,y)\in\GrothC\catL $,
 	\begin{align*}
 		& \displaystyle  \GrothC\catL \left( (\initialobject ,\initialobject ), (Y,y) \right) \\
 		&= \displaystyle  \displaystyle\coprod _{n\in\catC (\initialobject , Y) } \catL ( \initialobject ) ( \initialobject , \catL (n) (y)  )   \explainr{by definition}\\
 		&  \cong \displaystyle  \displaystyle\coprod _{n\in\catC (\initialobject , Y) } \terminal   \explainr{$\initialobject$ initial in $\catL ( \initialobject )$} \\
 		& \cong \terminal . \explainr{$\initialobject$ initial in $\catC$ }   
 	\end{align*}
 \end{proof}

 \begin{proposition}[Coproducts in $\GrothC \catL $]\label{prop:corpoducts-in-the-total-category}
 	Let $\catL : \catC ^\op\to\Cat $ be a strictly indexed category. We assume that
 	\begin{enumerate}[i)]
 		\item $ ((W_i,w_i)) _{i\in I} $ is family of objects of $\GrothC\catL$;
 		\item the category $\catC $ has the coproduct
 		\begin{equation}
 			\left(\begin{tikzcd}
 				W_t\arrow[r, "{\ic _ {W_t}}"] & \displaystyle\coprod _{i\in I } W_i                           
 			\end{tikzcd}\right) _{t\in I }
 		\end{equation}
 		of the objects in $ \left( (W_i,w_i)\right) _{i\in I} $;
 		\item there is an adjunction $\catL (\ic _ {W_i} )! \dashv \catL (\ic _ {W_i} ) $ for each $i\in I$;
 		\item $\catL \left( \displaystyle\coprod _{i\in I } W_i  \right) $ has the coproduct $\displaystyle \coprod _{i\in I} \catL (\ic _ {W_i} )! (w_i) $
 		of the objects  $\left( \catL (\ic _ {W_i} )! (w_i)\right) _{i\in I} $.
 	\end{enumerate}	
 	In this case, $$\left( \displaystyle\coprod _{i\in I } W_i  ,\quad \displaystyle\coprod _{i\in I } \catL (\ic _ {W_i} )! (w_i)   \right) $$
 	is the coproduct of the objects $ \left( (W_i,w_i )\right) _{i\in I} $ in $\GrothC \catL $.
 \end{proposition}	
 \begin{proof}
 	Assuming the hypothesis above, given any object $(Y,y)\in\GrothC\catL $,	
 	\begin{align*}
 		& \displaystyle \prod _{i\in I }  \GrothC\catL \left( (W_i,w_i), (Y,y) \right) \\
 		& = \displaystyle \prod _{i\in I } \left( \displaystyle\coprod _{n\in\catC (W_ i, Y) } \catL ( W_i ) (w_i, \catL (n) (y)  ) \right)  \explainr{by definition}\\
 		&  \cong \displaystyle \coprod _{(n_i)_{i\in I}\in \prod _{i\in I} \catC (W_ i, Y) } \left( \displaystyle\prod _{ i\in I } \catL ( W_i ) (w_i, \catL (n _i ) (y)  ) \right)  \explainr{distributivity} \\
 		& \cong\displaystyle \coprod _{h\in  \catC (\coprod _{i\in I} W_ i, Y) } \left( \displaystyle\prod _{ i\in I } \catL ( W_i ) (w_i, \catL (h\circ \ic _ {W_i} ) (y)  ) \right)  \explainr{coprod. univ. property}\\
 		&\cong \displaystyle \coprod _{h\in  \catC (\coprod _{i\in I} W_ i, Y) } \left( \displaystyle\prod _{ i\in I } \catL ( W_i ) \left( w_i, \catL (\ic _ {W_i})\circ \catL (h ) (y)  \right) \right) \\
 		& \cong  \displaystyle \coprod _{h\in  \catC (\coprod _{i\in I} W_ i, Y) } \left( \displaystyle\prod _{ i\in I } \catL \left( \coprod _{i\in I} W_i \right) \left( \catL (\ic _ {W_i}) ! (w_i), \catL (h ) (y)  \right) \right) \explainr{adjunctions}\\
 		& \cong  \displaystyle \coprod _{h\in  \catC (\coprod _{i\in I} W_ i, Y) } \left( \displaystyle \catL \left( \coprod _{i\in I} W_i \right) \left( \coprod _{ i\in I }\catL (\ic _ {W_i}) ! (w_i), \catL (h ) (y)  \right) \right) \explainr{coprod. univ. property}\\
 		& = \GrothC\catL \left( \left( \coprod _{i\in I} W_ i , \coprod _{ i\in I }\catL (\ic _ {W_i}) ! (w_i) \right), (Y,y) \right) . \explainr{coprod. univ. property}   
 	\end{align*}
 \end{proof}
 
 Codually, we get results on the initial objects and coproducts in the category $\GrothC \catL ^\op $ below.
 
 \begin{corollary}[Initial object in $\GrothC \catL ^\op $]\label{coro:initialobject-in-the-total-category-RAD}
 	Let $\catL : \catC ^\op\to\Cat $ be a strictly indexed category. We assume that
 	\begin{enumerate}[i)]
 		\item $\catC $ has initial object $\initialobject $;
 		\item $\catL (\initialobject ) $ has terminal object $\terminal $.
 	\end{enumerate}	
 	In this case, $\left( \initialobject  ,\terminal  \right) $
 	is the initial object of $\GrothC \catL $.
 \end{corollary}

 \begin{corollary}[Coproducts in $\GrothC \catL ^\op $]\label{coro:coproducts-in-the-total-category-RAD}
 	Let $\catL : \catC ^\op\to\Cat $ be a strictly indexed category. We assume that
 	\begin{enumerate}[i)]
 		\item $ ((W_i,w_i)) _{i\in I} $ is family of objects of $\GrothC\catL$;
 		\item the category $\catC $ has the coproduct
 		\begin{equation}
 			\left(\begin{tikzcd}
 				W_t\arrow[r, "{\ic _ {W_t}}"] & \displaystyle\coprod _{i\in I } W_i                           
 			\end{tikzcd}\right) _{t\in I }
 		\end{equation}
 		of the objects in $ \left( (W_i,w_i)\right) _{i\in I} $;
 		\item there is an adjunction $\catL (\ic _ {W_i} )\dashv \catL (\ic _ {W_i} )^\ast $ for each $i\in I$;
 		\item $\catL \left( \displaystyle\coprod _{i\in I } W_i  \right) $ has the product $\displaystyle \prod _{i\in I} \catL (\ic _ {W_i} )^\ast (w_i) $
 		of the objects $\left( \catL (\ic _ {W_i} )^\ast (w_i)\right) _{i\in I} $.
 	\end{enumerate}	
 	In this case, $$\left( \displaystyle\coprod _{i\in I } W_i  ,\quad \displaystyle\prod _{i\in I } \catL (\ic _ {W_i} )^\ast (w_i)   \right) $$
 	is the coproduct of the objects $ \left( (W_i,w_i )\right) _{i\in I} $ in $\GrothC \catL ^\op $.
 \end{corollary}

\subsection{Extensive indexed categories and coproducts in total categories}\label{subsect:coproduct-structure}
 We introduce a special property that fits our context well. We call this property \textit{extensivity} because it generalizes the concept of 
\textit{extensive categories} (see \ref{subsect:parameterized-initial-algebras-coalgebras} for the notion of extensive category). 

As we will show, the property of \textit{extensivity} is a crucial requirement for our models. One significant advantage of this property is that it allows us to easily construct coproducts in the total categories, even under lenient conditions. We demonstrate this in Theorem~\ref{coro:cocartesianstructure-in-the-cocartesian-csategory}.

\begin{itemize}
	\item We assume that the category $\catC $ has \textit{finite coproducts}. Given $  W, X\in\catC $, we denote by
	\begin{equation}
	\begin{tikzcd}
	W \arrow[rr, "{\ic_1 = \ic _ W}"] && W\sqcup X && X \arrow[swap, ll, "{\ic_2 = \ic _X}"]                             
	\end{tikzcd}
	\end{equation}
	the coproduct (and coprojections) in $\catC $, and by $\initial $ the initial object of $\catC $. 
\end{itemize}

\begin{definition}[Extensive indexed categories]\label{definition:extesive-indexed-categories}
We call	an indexed category $\catL : \catC ^\op \to \Cat $  \textit{extensive}
if, for any $(W,X)\in\catC\times\catC $, the unique functor 
    \begin{equation}\label{eq:definition-of-extensivity-equivsalence}
	\begin{tikzcd}
	\catL (W\sqcup X)  \arrow[rrr, "{\left(\catL (\ic _ W), \catL (\ic _ X)\right)}"] &&& \catL (W) \times \catL (X)                              
	\end{tikzcd}
	\end{equation}
induced by the functors
\begin{equation}
\begin{tikzcd}
\catL (W) && \catL (W\sqcup X) \arrow[rr, "{\catL (\ic _ X) }"]\arrow[swap, ll, "{\catL (\ic _W) }"]  && \catL (X)                             
\end{tikzcd}
\end{equation}
is an equivalence. In this case, for each $(W,X)\in\catC\times\catC $, we denote by
\begin{equation}\label{eq:equivalence-extensive-indexed-category} 
	\equivalenceextensive ^{(W,X)} : \catL (W)\times \catL (X)\to \catL (W\sqcup  X)
\end{equation} 
 an inverse equivalence
of ${\left(\catL (\ic _ W), \catL (\ic _ X)\right)} $. 
\end{definition}	

Since the products of $\catC ^\op $ are the coproducts of $\catC $,
the extensive condition described above is equivalent to say that the (pseudo)functor  $\catL : \catC ^\op \to \Cat $ preserves binary (bicategorical) products (up to equivalence). 

Since our cases of interest are strict, this leads us to consider \textit{strict extensivity}, that is to say, \textit{whenever we talk about extensive strictly indexed categories, we are assuming that \eqref{eq:definition-of-extensivity-equivsalence} is invertible.} 
In this case, it is even clearer that extensivity coincides with the well-known notion of preservation of binary products.
\begin{lemma}[Extensive strictly indexed categories]
	Let  $\catL : \catC ^\op \to \Cat $ be an indexed category. $\catL $ is strictly extensive 
	if, and only if,  $\catL $ is a functor that preserves binary products.
\end{lemma} 

Recall that, in general, \textit{preservation of binary products implies preservation of preterminal objects}; see, for instance, \citep[Remark~4.14]{zbMATH07629358}. 
Lemma \ref{lem:preservation-of-terminal-object} 
is the appropriate analogue of this observation suitably applied to the context of extensive indexed categories. Moreover, Lemma~\ref{lem:preservation-of-terminal-object} can be seen as a generalization of \citep[Proposition~2.8]{MR1201048}.

\begin{lemma}[Preservation of terminal objects]\label{lem:preservation-of-terminal-object}
	Let $ \catL : \catC ^\op \to \Cat $ be an extensive indexed category
	which is not (naturally isomorphic to the functor) constantly equal to $\initial $. The unique 
	functor 
	\begin{equation}\label{eq:basic-observation-preservation-of-terminal}
	\catL (\initial )\to \terminal
	\end{equation}
	  is an equivalence. If, furthermore, 
	\eqref{eq:definition-of-extensivity-equivsalence} is an isomorphism, then \eqref{eq:basic-observation-preservation-of-terminal} is invertible.
\end{lemma}
\begin{proof}
Firstly, given any $X\in\catC $ such that $\catL (X) $ is not (isomorphic to) the initial object of $\Cat $, we have that $\catL ( i_X : \initial\to X ) $ is a functor from $\catL (X) $ to
$\catL (\initial ) $. Hence $\catL (\initial ) $ is not isomorphic to the initial category as well. 	

Secondly, since	$\ic _{\initial } : \initial \to \initial\sqcup \initial $ is an isomorphism,
 $\left(\catL (\ic _ {\initial} ), \catL (\ic _ \initial )\right)$ is an equivalence and
\begin{equation}
\begin{tikzcd}
\catL ( \initial\sqcup \initial ) \arrow[rrrrrr, bend right=15, swap, "{\catL (\ic _ {\initial} )}"] \arrow[rrr, "{\left(\catL (\ic _ {\initial} ), \catL (\ic _ \initial )\right)}"] &&& \catL (\initial ) \times \catL (\initial ) \arrow[rrr, "{ \pi _{\catL(\initial )} }"]    &&&  \catL (\initial ) ,                         
\end{tikzcd} 
\end{equation}	
 we conclude that
	 $ \pi _{\catL(\initial )} $ is an equivalence. 	  
	 This proves that 
	$\catL (\initial ) \to \terminal $ is an equivalence by Appendix \ref{sec:appendix-pseudoterminal-in-Cat}, Lemma \ref{lem:pseudopreterminal-Cat}.
\end{proof}

We proceed to study the cocartesian structure of $\GrothC \catL $ (and $\GrothC \catL ^\op $) when $\catL $ is extensive. We start by proving  in Theorem \ref{theo:adjoint-to-get-coproduct} that, in the case of extensive indexed categories, the hypothesis of Proposition \ref{prop:initialobject-grothendieckconstruction} always holds.

\begin{therm}\label{theo:adjoint-to-get-coproduct}
	Let $\catL : \catC ^\op \to \Cat $ be an extensive (strictly) indexed category. Assume that $X$ is an object of $\catC $ such that 
	$\catL (X) $ has initial object $\initialobject $. In this case, for any $W\in \catC $,
	we have an adjunction 
\begin{equation}
\begin{tikzcd}
\catL (W\sqcup X) \arrow[bend right=15, rrrr, swap, "{\catL (\ic _ W)}"] &&\bot && \catL (W)  \arrow[bend right=15, llll, swap, "{\equivalenceextensive ^{(W,X)}\circ \left(\id _{\catL (W)}, \initialobject \right) }"]                            
\end{tikzcd}
\end{equation}
in which, by abuse of language, $\initialobject : \catL (W)\to \catL (X) $ is the functor constantly equal to $\initialobject $. 	
Dually, we have an adjunction
\begin{equation}
\begin{tikzcd}
\catL (W) \arrow[bend right=15, rrrr, swap, "{\equivalenceextensive ^{(W,X)}\circ \left(\id _{\catL (W)}, \terminal \right) }"] &&\bot && \catL (W\sqcup X)  \arrow[bend right=15, llll, swap, "{\catL (\ic _ W)}"]                            
\end{tikzcd}
\end{equation}
provided that $\catL (X) $ has terminal object $\terminal $ and, by abuse of language, we denote by $\terminal : \catL (W)\to \catL (X) $ the functor constantly equal to $\terminal $.
\end{therm}
\begin{proof}
	Assuming that $\catL (X) $ has initial object $\initialobject $, we have the adjunction
\begin{equation}
\begin{tikzcd}
\catL (W) \times \catL (X) \arrow[bend right=15, rrrr, swap, "{\pi _{\catL (W)}}"] &&\bot && \catL (W)\arrow[bend right=15, llll, swap, "{\left(\id _{\catL (W)}, \initialobject \right) }"]                            
\end{tikzcd}
\end{equation}
whose unit is the identity and counit is pointwise given by 
$ \varepsilon _{(w,x)} =  (\id _w, \initialobject \to x ) $. Therefore we have the composition of adjunctions 
\begin{equation*}
\begin{tikzcd}
\catL (W\sqcup X)\arrow[bend right=35, rrrrrr, swap, "{\catL (\ic _ W)}"]\arrow[bend right=15, rr, swap, "{\left(\catL (\ic _ W), \catL (\ic _ X)\right)}"] &\bot &\catL (W) \times \catL (X) \arrow[bend right=15, ll, swap, "{\equivalenceextensive ^{(W,X)} }"]\arrow[bend right=15, rrrr, swap, "{\pi _{\catL (W)}}"] &&\bot && \catL (W)  .\arrow[bend right=15, llll, swap, "{\left(\id _{\catL (W)}, \initialobject \right) }"]\arrow[bend right=35, llllll, swap, "{\equivalenceextensive ^{(W,X)}\circ \left(\id _{\catL (W)}, \initialobject \right) }"]                            
\end{tikzcd}
\end{equation*}
\end{proof}

\begin{corollary}[Cocartesian structure of $\GrothC \catL $]\label{coro:cocartesianstructure-in-the-cocartesian-csategory}
	Let $\catL : \catC ^\op \to \Cat $ be an extensive strictly indexed category, with initial objects    
	$\initialobject\in\catL (W)$ for each $W\in \catC $. In this case, the category $\GrothC \catL $ 
	has initial object $\initialobject = (\initialobject , \initialobject )\in \GrothC \catL$, and	
	 (fibred) binary coproduct given by $(W, w)\sqcup (X,x) = \left( W\sqcup X, \equivalenceextensive ^{(W,X)} (w,x)  \right) $.
\end{corollary}
\begin{proof}
	In fact, by Proposition \ref{prop:initialobject-grothendieckconstruction}, we have that $(\initialobject , \initialobject )$ is the initial object of $\GrothC \catL $.	
	Moreover, given $\left( (W,w), (X,x)\right) \in \GrothC\catL  \times \GrothC\catL $, we have that
\begin{eqnarray*}
\equivalenceextensive ^{(W,X)}\circ \left(\id _{\catL (W)}, \initialobject \right) = \catL (\ic _W )! &\dashv &\catL (\ic _W )\\
\equivalenceextensive ^{(W,X)}\circ \left( \initialobject , \id _{\catL (X)}\right) = \catL (\ic _X )! &\dashv &\catL (\ic _X )
\end{eqnarray*}	 
by Theorem  \ref{theo:adjoint-to-get-coproduct}. Therefore we get that
\begin{align*}
&  \displaystyle \left( W\sqcup X, \equivalenceextensive ^{(W,X)}\left( w , x\right)  \right) \\
& \cong \displaystyle \left( W\sqcup X, \equivalenceextensive ^{(W,X)} \left( w, \initialobject \right) \sqcup \equivalenceextensive ^{(W,X)} \left( \initialobject , x\right)  \right) \explainr{$\equivalenceextensive ^{(W,X)}$ preserves coproducts} \\
& \cong \displaystyle \left( W\sqcup X, \equivalenceextensive ^{(W,X)}\circ \left(\id _{\catL (W)}, \initialobject \right) (w)\sqcup \equivalenceextensive ^{(W,X)}\circ \left( \initialobject , \id _{\catL (X)}\right) (x)  \right) \\
& \cong \displaystyle \left( W\sqcup X, \catL (\ic _W )! (w)\sqcup \catL (\ic _X )! (x)  \right) \explainr{Theorem~\ref{theo:adjoint-to-get-coproduct}}\\
& \cong \displaystyle (W,w)\sqcup (X,x).\explainr{Proposition~\ref{prop:corpoducts-in-the-total-category}}
\end{align*}
\end{proof}
In particular, finite coproducts in $\Sigma_{\catC}\catL$ are fibred in the sense that the projection functor $\Sigma_{\catC}\catL\to \catC$ preserves them, on the nose.

Codually, we have: 
\begin{corollary}[Cocartesian structure of $\GrothC \catL ^\op $]\label{coro:cocartesian-structure-GrothCL-contravariant}
	Let $\catL : \catC ^\op \to \Cat $ be an extensive strictly indexed category, with terminal objects
	$\terminal\in\catL (W)$ for each $W\in \catC $. In this case, the category $\GrothC \catL ^\op $ 
	has (fibred) initial object $\initialobject = (\initialobject , \terminal )\in \GrothC \catL ^\op$, and	
	(fibred) binary coproduct given by 
	\begin{equation}
	(W, w)\sqcup (X,x) = \left( W\sqcup X, \equivalenceextensive ^{(W,X)} (w,x)  \right).
	\end{equation}
\end{corollary}

\begin{definition}[$\Sigma$-bimodel for sum types]
	A strictly indexed category $\catL :\catC ^\op \to\Cat $ is a $\Sigma$-bimodel for sum types
	if $\catL $ is an \textit{extensive strictly indexed category} such that $\catL (W)$
	has initial and terminal objects.
\end{definition}


\subsection{Distributive property of the total category}
We refer the reader to \citep{MR1201048, MR2864762} for the basics on distributive categories. 

As we proved, $\GrothC \catL $ is bicartesian closed  provided that $\catL :\catC ^\op\to\Cat  $ is $\Sigma$-bimodel for function types and sum types. Therefore, in this setting, we 
get that $\GrothC \catL $ is distributive.

However, even without the assumptions concerning closed structures, whenever
we have a $\Sigma$-bimodel for sum types, we 
can inherit distributivity from $\catC$. Namely,
we have Theorem \ref{theo:distributive-property-total-category}.

Recall that a category $\catC $ with finite products and coproducts is a \textit{distributive category} if, for each triple $\left( W, Y, Z\right) $ of objects in $\catC $, the canonical morphism 
\begin{equation}\label{eq:isomorphism-distributive-category}
\left< W\times \ic _ Y ^{Y\sqcup Z}, W\times\ic _Z^{Y\sqcup Z} \right> : \left( W\times Y\right) \sqcup \left(  W\times Z\right)\rightarrow W\times \left( Y\sqcup Z \right) ,        
\end{equation}
induced by $W\times \ic _ {Y}  $ and $W\times \ic _ {Z}  $,
is invertible.
It should be noted that, in a such a distributive category $\catC$, for any such a triple $\left( W, Y, Z\right) $ of objects in $\catC $, the diagram 
\begin{equation*}
	\begin{tikzpicture}[x=6cm, y=3cm]
		\node (a) at (0,2) {$ W\times\left(  Y\sqcup Z \right)   $};
		\node (b) at (0, 1) {$ \left( W\times Y \right)\sqcup \left( W\times Z \right)     $ };
		\node (c) at (1,1) {$ W  $ };
		\node (d) at (-1,1) {$ \left( Y\sqcup Z\right)  $ };
		\draw[->] (a)--(c) node[midway,above right] {$ \pi _W ^{W\times \left( Y\sqcup Z\right) }  $};
		\draw[->] (b)--(c) node[midway,below] {$ \left< \pi _W ^{W\times Y}, \pi _W ^{W\times Z}  \right>    $};
		\draw[->] (b)--(a) node[midway,left] {$ \cong    $} node[midway,right] {$ \left< W\times \ic _ Y, W\times\ic _Z \right>    $};
		\draw[->] (a)--(d) node[midway,above] {$ \pi _{\left( Y\sqcup Z\right)} ^{W\times \left( Y\sqcup Z\right) }  $};
		\draw[->] (b)--(d) node[midway,above] {$ \scriptscriptstyle\left< \ic _Y \circ \pi _Y ^{W\times Y}, \ic _ Z \circ \pi _Z ^{W\times Z}  \right>    $} node[midway,below] {$ \pi _Y ^{W\times Y}\sqcup \pi _Z ^{W\times Z}     $};
	\end{tikzpicture} 
\end{equation*} 
commutes. Therefore we have:
\begin{lemma}
	Let $\catL : \catC ^\op\to\Cat $ 
	be an extensive strictly indexed category, in which $\catC $ is a distributive category.
	For each triple $\left( W, Y, Z\right) $ of objects in $\catC $,	the diagrams 
\begin{equation}\label{eq:piW-distributiveproperty}
		\begin{tikzpicture}[x=0.4cm, y=0.3cm]
			\node (a) at (0,16) {$ \catL \left(W\times\left(  Y\sqcup Z \right)\right)    $};
			\node (b) at (0, 0)  {$ \catL  \left( W\times Y \right) \times \catL   \left( W\times Z \right)      $ };
			\node (c) at (16,8) {$ \catL\left( W\right)   $ };
			\node (d) at (0,8) {$ \catL \left( \left( W\times Y \right)\sqcup \left( W\times Z \right) \right)     $ };
			\draw[->] (a)--(d) node[midway, right] {$ \cong    $} node[midway,left] {$ \catL \left( \left< W\times \ic _ Y, W\times\ic _Z \right> \right)     $};
			\draw[<-]  (d)  to[bend right=25] node[midway,left] {$ \equivalenceextensive ^{(W\times Y ,W\times Z)}   $} (b);
			\draw[<-]  (b)  to[bend right=25] node[midway,right] {$\scriptscriptstyle  \left( \catL \left( \ic _ {W\times Y} \right)  , \catL \left( \ic _ {W\times Z} \right) \right)   $} (d);
			\draw[->] (c)--(d) node[midway,above] {$ \catL \left(\left< \pi _W ^{W\times Y}, \pi _W ^{W\times Z}  \right> \right)   $};
			\draw[->]  (c)  to[bend right=20] node[midway,above right] {$\scriptstyle \catL \left( \pi _W ^{W\times \left( Y\sqcup Z\right) }\right)  $} (a);
			\draw[->]  (c)  to[bend left=20] node[midway,below right]  {$ \scriptstyle\left( \catL \left( \pi _W ^{W\times Y} \right)  , \catL \left( \pi _W ^{W\times Z} \right) \right)    $} (b);
		\end{tikzpicture} 
\end{equation} 
\begin{equation}\label{eq:piY+Z-distributiveproperty}
	\begin{tikzpicture}[x=0.4cm, y=0.3cm]
		\node (a) at (0,16) {$ \catL \left(W\times\left(  Y\sqcup Z \right)\right)    $};
		\node (b) at (0, 0)  {$ \catL  \left( W\times Y \right) \times \catL   \left( W\times Z \right)      $ };
		\node (c) at (16,8) {$ \catL\left( Y\sqcup Z \right)   $ };
		\node (d) at (0,8) {$ \catL \left( \left( W\times Y \right)\sqcup \left( W\times Z \right) \right)     $ };
		\node (e) at (16,0) {$ \catL\left( Y\right) \times \catL \left(  Z \right)   $ };
		\draw[->] (a)--(d) node[midway, right] {$ \cong    $} node[midway,left] {$ \catL \left( \left< W\times \ic _ Y, W\times\ic _Z \right> \right)     $};
		\draw[<-]  (d)  to[bend right=25] node[midway,left] {$ \equivalenceextensive ^{(W\times Y ,W\times Z)}   $} (b);
		\draw[<-]  (b)  to[bend right=25] node[midway,right] {$\scriptscriptstyle  \left( \catL \left( \ic _ {W\times Y} \right)  , \catL \left( \ic _ {W\times Z} \right) \right)   $} (d);
		\draw[->] (c)--(d) node[midway,above] {$ \catL \left(\left< \pi _Y ^{W\times Y}, \pi _Z ^{W\times Z}  \right> \right)   $};
		\draw[->]  (c)  to[bend right=20] node[midway,above right] {$\scriptstyle \catL \left( \pi _ {\left( Y\sqcup Z\right) }^{W\times \left( Y\sqcup Z\right) }\right)  $} (a);
		\draw[->]  (e)  to[bend right=20] node[midway,right]  {$  \equivalenceextensive ^{(Y , Z)}  $} (c);
		\draw[->]  (c)  to[bend right=20] node[midway,left]  {$\scriptscriptstyle  \left( \catL \left( \ic _ {Y} \right)  , \catL \left( \ic _ {Z} \right) \right)   $} (e);
		\draw[->] (e)--(b) node[midway,below] {$ \catL \left( \pi _Y^{W\times Y}\right) \times \catL \left( \pi _Z^{W\times Z} \right)   $};
	\end{tikzpicture} 
\end{equation} 
commute.	 
\end{lemma}

\begin{therm}\label{theo:distributive-property-total-category}
	Let $\catL : \catC ^\op\to\Cat $ be $\Sigma$-bimodel for sum and tuple types, in which $\catC $ is a distributive category. In this setting, the category $\GrothC\catL $ 
	is a distributive category. 
\end{therm}	
\begin{proof}
By Proposition \ref{prop:grothendieck-products-covariant} and  Corollary \ref{coro:cocartesianstructure-in-the-cocartesian-csategory}, we have that 
$\GrothC\catL $ indeed has finite coproducts and finite products.

	Let $\catD $ be a category with finite coproducts and products.	
	A category is distributive if the canonical morphisms \eqref{eq:isomorphism-distributive-category} are invertible. However, by \citep[Theorem~4]{MR2864762}, the existence of any natural isomorphism $\left( W\times Y\right) \sqcup \left(  W\times Z\right)\cong W\times \left( Y\sqcup Z \right)$ implies that $\catD $ distributive. Hence, we proceed to prove below that such a natural isomorphism exists in the case of $\GrothC \catL $, leaving the question of canonicity omitted.

We indeed have the natural isomorphisms in $\left( \left( W, w\right) , \left( Y, y\right) , \left( Z, z\right)\right) \in \GrothC\catL \times\GrothC\catL\times \GrothC\catL  $
\begin{align*}
	&\left( W, w\right) \times\left( \left( Y, y\right)\sqcup \left( Z, z\right) \right)    \\
	&\cong \left( W, w\right) \times \left( Y\sqcup Z , \equivalenceextensive ^{(Y , Z)} (y, z)  \right)  \explainr{Corollary~\ref{coro:cocartesianstructure-in-the-cocartesian-csategory}} \\
	&\cong \left( W\times\left(Y\sqcup Z\right) , \catL (\pi _W ) (w) \times\catL (\pi _{Y\sqcup Z})\equivalenceextensive ^{(Y , Z)} (y, z)  \right) ,  \explainr{Proposition~\ref{prop:grothendieck-products-covariant} }
\end{align*}
which, by the distributive property of $\catC $, is (naturally) isomorphic to
\begin{equation}\label{eq:step1-distributive-grothcL}
\left( \left(W\times Y\right)\sqcup  \left(W\times Z\right), \catL \left( \left< W\times \ic _ Y, W\times\ic _Z \right> \right)    \left( \catL (\pi _W ) (w) \times\catL (\pi _{Y\sqcup Z})\equivalenceextensive ^{(Y , Z)} (y, z)\right)  \right) . 
\end{equation}
Moreover, we have the natural isomorphisms
\begin{align*}
	&\scriptstyle\catL \left( \left< W\times \ic _ Y, W\times\ic _Z \right> \right)    \left( \catL (\pi _W ) (w) \times\catL (\pi _{Y\sqcup Z})\equivalenceextensive ^{(Y , Z)} (y, z)\right)    \\
	& \scriptstyle\cong \catL \left( \left< W\times \ic _ Y, W\times\ic _Z \right> \right)    \left( \catL (\pi _W ) (w) \right)  \times \catL \left( \left< W\times \ic _ Y, W\times\ic _Z \right> \right) \left( \catL (\pi _{Y\sqcup Z})\equivalenceextensive ^{(Y , Z)} (y, z)\right)  \explainr{$ \catL \left( \left< W\times \ic _ Y, W\times\ic _Z \right> \right)$ invertible}\\
	& \scriptstyle = \equivalenceextensive ^{(W\times Y , W\times Z)}   \left( \catL (\pi _W ) (w), \catL (\pi _W ) (w) \right)  \times \catL \left( \left< W\times \ic _ Y, W\times\ic _Z \right> \right) \circ \catL (\pi _{Y\sqcup Z})\circ\equivalenceextensive ^{(Y , Z)} (y, z)  \explainr{Diagram~\eqref{eq:piW-distributiveproperty}}\\
	& \scriptstyle = \equivalenceextensive ^{(W\times Y , W\times Z)}   \left( \catL (\pi _W ) (w), \catL (\pi _W ) (w) \right)  \times \equivalenceextensive ^{(W\times Y , W\times Z)} \left( \catL (\pi _Y ) (y), \catL (\pi _Z ) (z) \right),  \explainr{Diagram~\eqref{eq:piY+Z-distributiveproperty}}
\end{align*}	
which is naturally isomorphic to 
\begin{equation}\label{eq:step2-distributive-grothcL}	
\equivalenceextensive ^{(W\times Y , W\times Z)}   \left( \catL (\pi _W ) (w) \times \catL (\pi _Y ) (y) , \catL (\pi _W ) (w) \times  \catL (\pi _Z ) (z) \right) . 
\end{equation}
since $\equivalenceextensive ^{(W\times Y , W\times Z)} $ is invertible.
Therefore we have the natural isomorphisms
\begin{align*}
	&\left( W, w\right) \times\left( \left( Y, y\right)\sqcup \left( Z, z\right) \right)    \\
	& \scriptstyle\cong \left( \left( W\times Y\right)\sqcup  \left(W\times Z\right), \catL \left( \left< W\times \ic _ Y, W\times\ic _Z \right> \right)    \left( \catL (\pi _W ) (w) \times\catL (\pi _{Y\sqcup Z})\equivalenceextensive ^{(Y , Z)} (y, z)\right)  \right) \explainr{Eq.~\eqref{eq:step1-distributive-grothcL}}\\
	& \scriptstyle \cong \left( \left( W\times Y\right)\sqcup  \left(W\times Z\right), \equivalenceextensive ^{(W\times Y , W\times Z)}   \left( \catL (\pi _W ) (w) \times \catL (\pi _Y ) (y) , \catL (\pi _W ) (w) \times  \catL (\pi _Z ) (z) \right)\right)\explainr{Eq.~\eqref{eq:step2-distributive-grothcL}}\\
	& \scriptstyle \cong \left(  W\times Y, \catL (\pi _W ) (w) \times \catL (\pi _Y ) (y) \right)
	\sqcup
	\left(  W\times Z, \catL (\pi _W ) (w) \times  \catL (\pi _Z ) (z) \right)
	\explainr{Corollary~\ref{coro:cocartesianstructure-in-the-cocartesian-csategory}}\\
	&\left( \left( W, w\right) \times \left( Y, y\right)\right)\sqcup  \left( \left( W, w\right)\times \left( Z, z\right) \right), \explainr{Proposition~\ref{prop:grothendieck-products-covariant}}
\end{align*}	
which completes our proof.	
\end{proof}
Codually, we have:
\begin{therm}\label{theo:distributive-property-total-category-reverse}
	Let $\catL : \catC ^\op\to\Cat $ be a $\Sigma $-bimodel for sum and tuple types, in which $\catC $ is a distributive category. Then we conclude that $\GrothC\catL ^\op $ 
	is a distributive category. 
\end{therm}

\subsection{Extensive property of the total category}
As per the definition provided in \citep[Definition2.1]{MR1201048}, a category $\catC$ is considered extensive if the basic (codomain) indexed category $\catC /- : \catC ^\op \to \Cat$ is an extensive indexed category as introduced at Definition~\ref{definition:extesive-indexed-categories}. Recall that every extensive category is distributive~\citep[Proposition~4.5]{MR1201048}.

The result below also holds for the non-strict scenario.

\begin{therm}\label{theo:extensivity-property-total-category}
	Let $\catL : \catC ^\op\to\Cat $ be an extensive strictly indexed category, in which $\catC $ is an extensive category. Assume that we have initial objects $\initial\in\catL \left( W\right) $. In this case, the category $\GrothC\catL $ is extensive and, hence, distributive.
\end{therm}	
\begin{proof}
	We denote  by $\equivalenceextensive ^{\left( W,X  \right)} _\catL : \catL (W)\times \catL (X)\to \catL(W\sqcup X) $ the isomorphisms of the
	extensive strictly indexed category $\catL $.
	
	The first step is to see that, indeed,  $\GrothC\catL $ has coproducts by Corollary \ref{coro:cocartesianstructure-in-the-cocartesian-csategory}. 
	We, then, note that, for each pair $(W,w)$ and $(X,x) $ of objects in 
	$\GrothC\catL $, we note that, in fact, we have  have that 
	\begin{equation}
	\equivalenceextensive ^{\left( (W,w), (X,x)  \right)} _ {\GrothC\catL /-} :
	 \GrothC\catL /(W,w) \times  \GrothC\catL /(X,x) \to \GrothC\catL /\left( (W,w)\sqcup (X,x)  \right) 
	\end{equation}
defined by the coproduct of the morphisms is an equivalence. Explicitly, given objects $A=  \left((W _0 ,w _0 ) ,   ( f: W_0\to W, f' : w_0\to \catL\left( f \right) w  ) \right) $ of $\GrothC\catL /(W,w) $ and 
 $B =  \left((X _0 ,x _0 ) ,   ( g: X_0\to X, g' : x_0\to \catL\left( g \right) x  ) \right) $ of $\GrothC\catL /(X,x)$,
 $ \equivalenceextensive ^{\left( (W,w), (X,x)  \right)} _ {\GrothC\catL /-}\left( A, B\right) $
 is given by 
 $$\left( \left( W_0\sqcup X_0, \equivalenceextensive ^{\left( W, X \right) } _{\catL } (w_0, x_0)        \right) , \left(   f\sqcup g:  W_0\sqcup X_0\to W\sqcup X,  \equivalenceextensive ^{\left( W,X  \right)} _\catL  \left(f', g'\right)                    \right)            \right)   $$ 
which is clearly an equivalence given that the functor $\left( (W_0, f), (X_0, g) \right)\mapsto \left( W_0\sqcup X_0, f\sqcup g\right)  $ is an equivalence
$\catC / W\times \catC / X \to \catC / W\sqcup X $.
\end{proof}

\begin{therm}\label{theo:extensivity-property-total-category-reverse}
	Let $\catL : \catC ^\op\to\Cat $ be an extensive strictly indexed category, in which $\catC $ is an extensive category. Assume that we have terminal objects $\terminal\in\catL \left( W\right) $. In this case, the category $\GrothC\catL ^\op $ is extensive and, hence, distributive.
\end{therm}

It is worth mentioning that \textit{free cocompletions under (finite) coproducts} are extensive, as shown in \citep[Proposition~2.4]{MR1201048} for the infinite case. This implies that freely generated models on languages featuring variant types are extensive. Therefore, having an extensive base category $\catC$ is a common occurrence in our setting.

\subsection{Strictly indexed categories and split fibrations}
\label{sec:appendix-SplitFibrations-vs-IndedxedCategories}
Before we specialize to our setting of $\mu\nu$-polynomials, we need to establish and prove general results on parameterized initial algebras (and terminal coalgebras) in the total category of a split fibration (see \ref{subsect:general-result-initial-algebras-total-categories} and \ref{sect:appendix-coinductivetypes-grothendieckconstruction}).

In order to talk about these results, we need to talk about \textit{strictly indexed functors} and \textit{split fibration functors} and the one-to-one correspondence between them. For this purpose, we shortly recall the equivalence between strict indexed categories and split fibrations below.

\begin{definition}[Strictly indexed functor]
	Let $\catL ' : \catD ^\op\to \Cat $ and $ \catL  : \catC ^\op\to \Cat $ be two strictly indexed categories.
	A \textit{strictly indexed functor} between $\catL ' $ and $\catL $ consists of a pair $(\overline{H}, h ) $
	in which $\iH  : \catD \to \catC $ is a functor and
	\begin{equation} 
		h: \catL ' \longrightarrow \left(\catL  \circ \iH ^\op\right)
	\end{equation}
	is a natural transformation, where  $\iH ^\op $ denotes the image of $\overline{H} $ by $\op $.
	Given two strictly indexed functors $(\iE ,e) :\catL ''\to \catL '  $ and $(\overline{H} , h) : \catL '\to \catL  $, the composition is given by 
	\begin{equation}
		\left( \overline{HE}, (h _ {\iE ^{\op }})\cdot e: \catL '' \longrightarrow \left(\catL  \circ \left(\overline{HE}\right) ^{\op }\right) \right) .
	\end{equation}
	Strictly indexed categories and strictly indexed functors do form a category, denoted herein by $\GrothInd $.
\end{definition}
It is well known that the Grothendieck construction provides an equivalence between indexed categories
and fibrations. Restricting this to our setting, we get the equivalence
\begin{eqnarray*}
	\int : & \GrothInd &\to \GrothSpFib\\
	& \catL : \catC ^\op\to \Cat & \mapsto \left(\mathsf{P}_{\catL } : \GrothC \catL\to \catC\right)\\
	& (\iE, e) & \mapsto (E, \iE ) 
\end{eqnarray*}
between the category of strictly indexed categories (with strictly indexed functors) and the category of
(Grothendieck) 
split fibrations.

Although not necessary to your work, we refer to \citep{MR0213413} and \citep[Theorem~1.3.6]{johnstone2002sketches} for further details. 
We explicitly state the relevant part of this result below.

\begin{proposition}\label{prop:GrothedieckSplitFibrations-vs-IndexedCategories}
	Given two strictly indexed categories, $\catL ' : \catD ^\op\to \Cat$ and $\catL  : \catC ^\op\to \Cat $, there is a bijection between strictly indexed functors
	\begin{equation*}
		\left(\iH : \catD \to \catC , h: \catL ' \longrightarrow \left(\catL  \circ \iH ^\op\right)\right) : \catL '\to \catL 
	\end{equation*}
	and pairs $ (H, \iH ) $ in which $H:\GrothD \catL ' \to\GrothC \catL$ is a functor
	satisfying the following two conditions. 
	\begin{enumerate}[$\Sigma$A)]
		\item The diagram
		\begin{equation}\label{eq:1-condition-morphism-split-fibrations}	
			\begin{tikzcd}
				\GrothD \catL ' \arrow[rrr, "{H}"] \arrow[swap,d, "{\mathsf{P}_{\catL '}}"] &&& \GrothC \catL \arrow[d, "{\mathsf{P}_{\catL}}"] \\
				\catD  \arrow[swap, rrr, "{\iH}"] &&& \catC                              
			\end{tikzcd}
		\end{equation} 
		commutes.
		\item For any morphism $(f: X\to Y, \id : \catL ' (f) (y) \to  \catL ' (f) (y) ) $ between $(X, \catL ' (f) (y) )$ and $(Y, y) $ in  $\GrothD \catL '$, 
		\begin{equation}\label{eq:2-condition-morphism-split-fibrations}	
			H(f, \id) = ( \overline{H} (f) , \id ) : H(X, \catL ' (f) (y) )\to H(Y, y ).
		\end{equation}
	\end{enumerate}  
\end{proposition}
\begin{proof}
	Although, as mentioned above, this result is just a consequence of the well known result about the equivalence
	between indexed categories and fibrations, we recall below how to construct the bijection.
	
	For each strictly indexed functor $(\overline{H}, h) : \catL ' \to \catL $, we define
	\begin{equation}\label{eq:fibration-to-indexed-category}
		H (f: X\to Y , f': x\to \catL '(f) y ) := (\overline{H} (f), h_X (f') ).
	\end{equation}
	Reciprocally, given a pair $(H, \overline{H}) $ satisfying  \eqref{eq:1-condition-morphism-split-fibrations} and \eqref{eq:2-condition-morphism-split-fibrations}, we define
	\begin{equation}
		h_X (f': w\to x ) := H \left( (\id_ X , f') : (X, w)\to (X, x)    \right) 
	\end{equation}	
	for each object $X\in \catD $ and
	each morphism $f': w\to x $ of $\catL '(X)$.
\end{proof}
\begin{definition}[Split fibration functor]		
	A pair $(H, \overline{H} ): \mathsf{P}_{\catL '}\to \mathsf{P}_{\catL } $ satisfying \eqref{eq:1-condition-morphism-split-fibrations} and \eqref{eq:2-condition-morphism-split-fibrations} is herein called a
	\textit{split fibration functor}. Whenever it is clear from the context, we omit the split fibrations 
	$\mathsf{P}_{\catL '}$, $\mathsf{P}_{\catL }$, and the functor $\overline{H} $.
\end{definition}
Following the above, given a 
strictly indexed functor $(\iH , h) : \catL '\to \catL $, we denote
\begin{eqnarray*}
	\int \catL & = & \left( \mathsf{P}_{\catL } : \GrothC \catL\to \catC\right) \\ 
	\int \left( \iH , h \right)  & = &   \left( H ,\iH \right)
\end{eqnarray*}
in which $H \left( f : X\to Y , f': x\to \catL (f) (y) \right) = (\iH (f), h_X (f') ) $.

Let $\catL ': \catD ^\op \to \Cat $ and $\catL  : \catC ^\op \to \Cat $  be strictly indexed categories. We denote by
$ \catL '\prodstrict\catL   $ the product of the strict indexed categories in $\GrothInd $. Explicitly,  
\begin{eqnarray*}
	\catL '\prodstrict\catL  : &(\catD\times \catC ) ^\op & \to \Cat\\
	& (X,Y ) & \mapsto \catL ' (X) \times \catL (Y)\\
	& (f, g ) & \mapsto \catL ' (f)\times \catL  (g) .
\end{eqnarray*}
It should be noted that 
\begin{equation}
	\left(\int \catL '\prodstrict\catL   \right)  \cong \left(\int \catL '\right)\times  \left( \int  \catL   \right)  =  \left(\mathsf{P}_{\catL ' }\times\mathsf{P}_{\catL }  : \left( \GrothD \catL '\right)\times\left(\GrothC \catL \right)  \to \catD\times\catC \right) ,
\end{equation}
which means that the product in $ \GrothSpFib $ coincides with the usual product
of functors $\mathsf{P}_{\catL }\times\mathsf{P}_{\catL '}$. Moreover, given 
indexed functors $(\overline{H}, h) : \cat{H}\to \cat{H}' $ and  $(\overline{E}, e) : \catL\to \catL ' $, we have that
$$ (\overline{H}, h)\prodstrict (\overline{E}, e)  = \left( \overline{H}\times \overline{E}, h\times e   \right) $$
and, since the product of split fibrations is given by the usual product of functors,
\begin{equation}
	\int \left( (\overline{H}, h)\prodstrict (\overline{E}, e)\right) = 
	\left(\int (\overline{H}, h)\right) \times  \left( \int (\overline{E}, e)\right) .
\end{equation}  

Codually, given a strictly indexed category $ \catL : \catC ^\op\to \Cat $, we have the Grothendieck codual construction 
\begin{center}
	$\int ^\op \catL   =  \left( \mathsf{P}_{\catL ^\op } : \GrothC \catL ^\op\to \catC\right)$, \qquad  $\int ^\op \left( \iH , h \right)   =    \left( H ,\iH \right)$
\end{center}
in which $H \left( f : X\to Y , f': \catL (f) (y)\to x \right) = (\iH (f), h_X (f') ) $. This construction
gives an equivalence between the indexed categories and split op-fibrations (if we consider the opposite of $ \int ^\op \catL   $). We, of course, have the codual observations above.


\subsection{General result on initial algebras in total categories}
\label{subsect:general-result-initial-algebras-total-categories}
In order to study the $\mu\nu$-polynomials of total categories in our setting in \ref{subsect:parameterized-initial-algebras-coalgebras}, 
we start by establishing general results about parameterized initial algebras 
in the Grothendieck construction of split fibrations. More precisely, in Theorem \ref{theo:parameterized-initial-algebras-for-indexed-functors},  we investigate when 
a total category $ \GrothC \catL $ has the parameterized initial algebra of a split fibration functor
\begin{equation}\label{eq:strictly-indexed-functor-split-fibration-functor}
H : \left( \GrothD \catL '\right) \times \left( \GrothC \catL \right)  \to \GrothC\catL .
\end{equation} 
We start by studying initial algebras os strictly indexed endofunctors:

\begin{therm}[Initial algebras of strictly indexed endofunctors]\label{theo:fundamental-theorem-on-initial-algebra-semantics=Grothendieck}
	Let  
	$(  \iE , e )  $
	be a strictly indexed endofunctor on $\catL : \catC ^\op\to\Cat $
	and $E: \GrothC\catL \to \GrothC\catL $ the corresponding split fibration endofunctor. 
	Assume that
	\begin{enumerate}[$\mathfrak{e}$1)]
		\item the initial $\iE$-algebra $\left( \mu \iE,\, \ind _{\iE}  \right) $ exists;
		\item the initial $\left( \catL(\ind _{\iE} )^{-1} \iem \right) $-algebra $\left( \mu \left( \catL (\ind _{\iE } )^{-1} \iem \right),\, \ind _{\left(\catL (\ind _{\iE } )^{-1} e_{\mu \iE }\right) } \right)  $ exists.
	\end{enumerate}
	Denoting by 
	$\ie $ the endofunctor $\catL (\ind _{\iE } )^{-1} e_{\mu \iE } $ on $\catL (\mu \iE ) $,
	the initial $E$-algebra exists and is given by
	\begin{equation}\label{eq:initial-algebra-structure-splitendofunctor}
		\mu E = \left( \mu \iE ,\, \mu \ie \right), \qquad \ind _E = \left( \ind _{\iE },\, \catL (\ind _{\iE } ) \left(\ind _{\ie } \right) \right).
	\end{equation}
	Moreover, 
	for each $E$-algebra 
	$$\left( (Y,y),\, (\xi ,\xi '): E(Y,y)\to (Y,y)  \right) = \left(  (Y,y), \left( \xi : \iE (Y)\to Y, \xi ': e_Y(y)\to \catL (\xi ) (y)  \right) \right), $$ 
	we have that
	\begin{equation}\label{eq:fold-initialalgebras-grothendieckconstruction-endofunctor}
		\fold _{E} \left( \xi , \xi ' \right)  = \left( \fold _ {\iE } \xi , \, \, 
		\fold _{ \ie  } \left(  
		\catL \left(\iE ( \fold _{\iE }\, \xi   )\cdot \ind _{\iE } ^{-1} \right)
		(\xi ')\right)  \right).
	\end{equation}
\end{therm}

\begin{proof}
	In fact, under the hypothesis above, given an $E$-algebra
	\begin{equation*}
		\left( \xi : \iE (Y)\to Y, \xi ': e_Y(y)\to \catL (\xi ) (y)  \right)
	\end{equation*}
	on $(Y, y) $, 
	we have that there is a unique morphism	
	\begin{equation*}
		\left(\fold _{ \ie } \, \catL \left(\iE ( \fold _{\iE } \xi   )\cdot \ind _{\iE } ^{-1} \right)
		(\xi ')\right)  : \mu\ie \to 
		\catL\left( \fold _{\iE }  \xi \right)(y) 
	\end{equation*}
	in $\catL (\mu \iE ) $ such that
	\begin{equation*}
		\begin{tikzcd}
			\ie (\mu\ie) \arrow[rrrrr, "{\ie\left(\fold _{ \ie } \, \catL \left(\iE ( \fold _{\iE } \xi   )\cdot \ind _{\iE }^{-1} \right)(\xi ')\right)}"] \arrow[swap,ddd, "{\ind _{\ie } }"] &&&&& 
			\ie\circ \catL\left(\fold _{\iE } \xi \right) (y)\arrow[d,equal]\\
			&&&&& \catL \left(\iE ( \fold _{\iE } \xi   )\cdot \ind _{\iE }^{-1} \right)\circ e_Y (y) \arrow[d, "{\catL \left(\iE ( \fold _{\iE } \xi   )\cdot \ind _{\iE }^{-1} \right)(\xi ')}"]\\
			&&&&& \catL \left(\xi \cdot \iE ( \fold _{\iE } \xi   )\cdot \ind _{\iE }^{-1} \right) (y)\arrow[d, equal]
			\\
			\mu \ie  \arrow[swap, rrrrr, "{\left(\fold _{ \ie } \, \catL \left(\iE ( \fold _{\iE } \xi   )\cdot \ind _{\iE }^{-1} \right)(\xi ')\right)}"] &&&&& 
			\catL \left(\fold _{\iE } \xi\right)( y )                             
		\end{tikzcd}
	\end{equation*} 
	commutes. Since $\catL (\ind _{\iE})$ is invertible, this implies that 
	\begin{equation*}
		\left(	\fold _{ \ie } \, \catL \left(\iE ( \fold _{\iE } \xi   )\cdot \ind _{\iE } ^{-1} \right)
		(\xi ') \right)  : \mu\ie \to 
		\catL\left( \fold _{\iE }  \xi \right)(y) 
	\end{equation*}
	is the unique morphism in $\catL \left( \iE (\mu \iE ) \right)$ such that
	\begin{equation*}
		\begin{tikzcd}
			\iem (\mu\ie) \arrow[rrrrrr, "{\iem\left(\fold _{ \ie } \, \catL \left(\iE ( \fold _{\iE } \xi   )\cdot \ind _{\iE }^{-1} \right)(\xi ')\right)}"] \arrow[swap,ddd, "{\catL (\ind _{\iE}) (\ind _{\ie }) }"] &&&&&& \iem\circ \catL\left(\fold _{\iE } \xi \right) (y)\arrow[d,equal]\\
			&&&&&& \catL \left(\iE ( \fold _{\iE } \xi )\right)\circ e_Y (y) \arrow[d, "{\catL \left(\iE ( \fold _{\iE } \xi   )\right)(\xi ')}"]\\
			&&&&&& \catL \left( \iE ( \fold _{\iE } \xi   ) \right)\circ \catL \left(\xi \right) (y)\arrow[d, equal]
			\\
			\catL (\ind _{\iE}) (\mu \ie )  \arrow[swap, rrrrrr, "{\catL (\ind _{\iE})\left(\fold _{ \ie } \, \catL \left(\iE ( \fold _{\iE } \xi   )\cdot \ind _{\iE }^{-1} \right)(\xi ')\right)}"] &&&&&& \catL \left(\ind _{\iE}\right)\circ \catL \left(\fold _{\iE } \xi\right)( y )                             
		\end{tikzcd}
	\end{equation*} 
	commutes.	
	Finally, by the above and the universal property of $\fold _{\iE} \xi$, this completes the proof that  
	\begin{equation}
		\mathfrak{u} = \left( \fold _{\iE} \xi, \,   \left(\fold _{ \ie } \, \catL \left(\iE ( \fold _{\iE } \xi   )\cdot \ind _{\iE } ^{-1} \right) (\xi ')\right)\right)
	\end{equation}
	is the unique morphism in $\GrothC\catL $ such that
	$$  (\xi , \xi ')\circ E( \mathfrak{u} ) = \mathfrak{u} \circ \left( \ind _{\iE },\, \catL (\ind _{\iE } ) \left(\ind _{\ie } \right) \right) . $$
	This completes the proof that $\left( (\mu\iE , \mu\ie), \left( \ind _{\iE },\, \catL (\ind _{\iE } ) \left(\ind _{\ie } \right) \right)  \right) $ is the initial object of $E\AAlg$, and that $\fold _{E} ((Y,y), (\xi, \xi ') ) = \mathfrak{u} $.
\end{proof}

Let $\catL : \catC ^\op \to \Cat $,  $\catL ': \catD ^\op \to \Cat $  be  strictly indexed categories as above. We denote by
$\catL '\prodstrict \catL  : \left( \catD\times\catC \right) ^\op \to \Cat  $ the product of the indexed categories (see \ref{sec:appendix-SplitFibrations-vs-IndedxedCategories}). 
An object of  $ \Sigma _{\catD\times \catC } \left( \catL '\prodstrict \catL\right)  \cong \left( \GrothD \catL '\right) \times \left( \GrothC \catL \right)  $ can be seen as a quadruple $\left( (X, x) , (W, w)\right)   $ 
in which $x \in\catL '(X) $ and $ w\in \catL (W) $. Moreover, a morphism between objects 
$\left( (X_0 , x _0) , (W _0, w_ 0 )\right) $ and $\left( (X_1 , x_1 ) , (W_1,  w _1)\right)  $ consists of a quadruple  $\left(  (f, f')   , (g, g')  \right) $ in which 
$ (f, g) : ( X _ 0 , W _ 0 )\to ( X_1 , W_1 ) $  is a morphism in $\catD\times \catC  $, and  
$ (f', g') : (x_0, w_0)\to \left( \catL ' (f)( x_1 ), \catL  (g)( w_1 )\right) $
is a morphism in $\catL ' ( X_0 )\times \catL  ( W_0 ) $. 

Given a strictly indexed functor $(\iH, h) : \catL '\prodstrict \catL  \to\catL  $ and  an object $(X,x)  $ of $\left( \GrothD \catL '\right) $, we can consider the restriction $( \iH ^X, h ^{(X,x)} ) $ in which $\iH ^X = \iH (X,-)$ and $h^{(X,x)} : \catL\longrightarrow \left( \catL\circ \iH ^X \right)  $    is pointwise defined by
\begin{eqnarray*}
	h^{(X,x)}_Y : &\catL (Y) &\to \catL\circ \iH ^X (Y)\\
	& f' : y\to z & \mapsto h_{(X,Y)} (x, f' )
\end{eqnarray*}
in which we denote by $(X,Y) \in \catD \times \catC  $. 
To be consistent with the notation previously introduced (in Proposition \ref{prop:general_parameterized_initial_algebras}), we also denote by $h_{(X,Y)}^x $ the morphism $h^{(X,x)}_Y $ above.

As a consequence of Theorem \ref{theo:fundamental-theorem-on-initial-algebra-semantics=Grothendieck}, we have that, under suitable conditions, parameterized 
initial algebras of split fibration functors are split fibration functors; namely, we have:

\begin{therm}[Parameterized initial algebras are split fibration functors]\label{theo:parameterized-initial-algebras-for-indexed-functors}
	Let  
	$(  \iH , h )  $
	be a strictly indexed functor from 
	$\catL '\prodstrict \catL  : \left( \catD\times\catC \right) ^\op \to \Cat  $ to $\catL : \catC ^\op\to\Cat $,
	and $$ H : \left( \GrothD \catL '\right) \times \left( \GrothC \catL \right)  \to \GrothC\catL $$ the corresponding split fibration functor. 
	Assume that:
	\begin{enumerate}[$\mathfrak{h}$1)]
		\item for each object $X$ of $\catD $, the initial $\iH ^X $-algebra $\left( \mu \iH ^X,\, \ind _{\iH ^X}  \right) $ exists;
		\item for each object $(X,x) $ in $\GrothD\catL '  $, denoting by 
		$\ih _X  $ the functor 
		\begin{equation}\label{eq: concise-notation-parameterized-types-fibers}
			\catL(\ind _{\iH ^X} )^{-1} \ihmXx : \catL '( X )  \times \catL(\mu \iH ^X ) \to\catL(\mu \iH ^X )     
		\end{equation}
		is such that the initial $\ih _X ^x $-algebra $\left( \mu \ih _X^x   ,\, \ind _{\ih _X^x}  \right)  $ exists;
		\item for each morphism $ g : X\to Y $ in $\catD $ and $ y\in \catL ' (Y) $, Eq.~\eqref{eq:additionalhypothesis-muH-GrothendieckFibration} holds.
		\begin{equation}\label{eq:additionalhypothesis-muH-GrothendieckFibration}
			\catL \left( \mu\iH (g) \right) (\ind _{\ih _ Y^y } ) =  \ind _ {\ih _ X^{\catL ' (g) (y) }} 
		\end{equation}
	\end{enumerate}
	In this setting, the parameterized initial algebra $\mu H: \GrothD \catL '  \to \GrothC\catL $   exists and is a split fibration functor.
\end{therm}
\begin{proof}
	Assuming the hypothesis,
	we conclude that, for each $(X,x)$ in $\GrothD \catL ' $, the category  
	$\GrothC\catL $ has the initial $H^{(X,x)} $-algebra, by Theorem \ref{theo:fundamental-theorem-on-initial-algebra-semantics=Grothendieck}. Hence we have that 
	$$ \mu H : \GrothD \catL '\to \GrothC\catL $$ 
	exists by Proposition \ref{prop:general_parameterized_initial_algebras}.	
	More precisely, given a morphism $(f, f') : (X,x)\to (Y,y) $ in $\GrothD\catL ' $, we compute  $\mu H (f, f') $ below. 
	\begin{align*}
		&\mu H (f, f')\\
		&= \fold _{H^{(X,x)} }\left( \ind_{H^{(Y,y)} }\circ  H\left( (f, f' ), \mu H ^{(Y,y)}\right)\right) \explainr{Proposition~\ref{prop:general_parameterized_initial_algebras}}\\
		&= \fold _{H^{(X,x)} }\left( \left( \ind _{\iH ^Y }, \catL (\ind _{\iH ^Y } ) ( \ind _{ \ih _ Y^y } ) \right)  \circ  H\left( (f, f') , \mu H ^{(Y,y)} \right)\right) \explainr{  Eq.~\eqref{eq:initial-algebra-structure-splitendofunctor}}\\
		&= \fold _{H^{(X,x)} }\left( \left( \ind _{\iH ^Y }, \catL (\ind _{\iH ^Y } ) ( \ind _{ \ih _ Y ^y } ) \right)  \circ \left( \iH (f, \mu \iH ^{Y}), h_{ (X, \mu\iH ^Y ) } ( f', \mu\ih _ Y^y )   \right) \right) \explainr{indexed functor}\\
		&= \fold _{H^{(X,x)} }\left( \ind _{\iH ^Y }\circ \iH (f, \mu \iH ^{Y}), \catL \left(\ind _{\iH ^Y }\circ \iH (f, \mu \iH ^{Y}) \right) ( \ind _{ \ih _ Y^y } )  \right.\\ &\hspace{220pt}\left.
		\circ\left( h_{ (X, \mu\iH ^Y ) } ( f', \mu\ih _ Y^y )   \right) \right) \explainr{composing} 
	\end{align*}
	which, by denoting
	$\xi = \ind _{\iH ^Y }\circ \iH (f, \mu \iH ^{Y}) $ and 	
	$\xi ' =  \catL \left(\xi \right) ( \ind _{ \ih _ Y^y } ) \circ\left( h_{ (X, \mu\iH ^Y ) } (f', \mu\ih _ Y^y )   \right) $,
	is equal to
	\begin{align*}
		&\fold _{H^{(X,x)} }\left( \ind _{\iH ^Y }\circ \iH (f, \mu \iH ^{Y}) , \xi ' \right)\\
		&= \left( \fold _{\iH ^X } \left(\ind _{\iH ^Y }\circ \iH (f, \mu \iH ^{Y})\right), \,  \left(\fold _{ \ih _X^x } \, \catL \left(\iH ^X ( \fold _{\iH ^X } \xi   )\cdot \ind _{\iH  ^X } ^{-1} \right) (\xi ')\right)\right) \explainr{  Eq.~\eqref{eq:fold-initialalgebras-grothendieckconstruction-endofunctor}}\\
		&= \left( \mu\iH (f),  \,   \left(\fold _{ \ih _X^x } \, \catL \left(\iH ^X ( \fold _{\iH ^X } \xi   )\cdot \ind _{\iH  ^X } ^{-1} \right) (\xi ')\right)  
		\right) . \explainr{  Proposition~\ref{prop:general_parameterized_initial_algebras}}
	\end{align*}
	The above shows that 
	\begin{equation}\label{eq:first-eq-muH-splitfibration}
		\mu H (f,f') = \left( \mu\iH (f),  \,   \left(\fold _{ \ih _X^x } \, \catL \left(\iH ^X ( \fold _{\iH ^X } \xi   )\cdot \ind _{\iH  ^X } ^{-1} \right) (\xi ')\right) \right) .
	\end{equation}
	
	Now, we can proceed to prove that $ \mu H $ is actually a split fibration functor. Firstly, 
	by Equation \eqref{eq:first-eq-muH-splitfibration}, we have that 
	\begin{equation}
		\begin{tikzcd}
			\GrothD \catL '  \arrow[rrrrr, "{\mu H}"] \arrow[swap,d, "{\mathsf{P}_{\catL ' } }"] &&&&& \GrothC \catL \arrow[d, "{\mathsf{P}_{\catL}}"] \\
			\catD \arrow[swap, rrrrr, "{\mu\iH}"] &&&&& \catC                              
		\end{tikzcd}
	\end{equation} 
	commutes. 
	
	Let 
	$ \left( g , \id \right) : \left( X, \catL ' (g) (y) \right) \to
	\left( Y , y \right) $
	be a morphism in $ \left( \GrothD \catL '\right)  $. 
	Denoting, again, 
	\begin{center}  
		$\xi = \ind _{\iH ^Y }\circ \iH (g, \mu \iH ^{Y}) $\qquad and\qquad 	
		$\xi ' =  \catL \left(\xi \right) ( \ind _{ \ih _ Y^y } ) \circ\left( h_{ (X, \mu\iH ^Y ) } ( \id , \mu\ih _ Y^y )   \right) $,
	\end{center} 
	we have that
	\begin{align*}
		& \left(\fold _{ \ih _X^{\catL ' (g) (y)} } \, \catL \left(\iH ^X ( \fold _{\iH ^X } \xi   )\cdot \ind _{\iH  ^X } ^{-1} \right) (\xi ')\right) \\
		& =  \left(\fold _{ \ih _X^ {\catL ' (g) (y)} } \, \catL \left( \xi\cdot \iH ^X ( \fold _{\iH ^X } \xi   )\cdot \ind _{\iH  ^X } ^{-1} \right) (\ind _{ \ih _ Y^y } )\right) \explainr{$h_{ (X, \mu\iH ^Y ) } ( \id , \mu\ih _ Y^y ) = \id $}  \\
		& =   \left(\fold _{ \ih _X^{\catL ' (g) (y)} } \, \catL \left(  ( \fold _{\iH ^X } \xi   )\cdot \ind _{\iH  ^X } \cdot \ind _{\iH  ^X } ^{-1} \right) (\ind _{ \ih _ Y^y } )\right) \explainr{ $ \fold _{\iH ^X } \xi $ } \\
		& =  \left(\fold _{ \ih _X^ {\catL ' (g) (y)} } \, \catL \left(  \mu\iH (g)   \right) (\ind _{ \ih _ Y^y } )\right) \explainr{ Proposition~\ref{prop:general_parameterized_initial_algebras}} \\
		&  =  \id _{\mu \ih _X^ {\catL ' (g) (y)} } \explainr{ Eq.~\eqref{eq:additionalhypothesis-muH-GrothendieckFibration} }
	\end{align*} 
	By Equation \eqref{eq:first-eq-muH-splitfibration}, the above proves that
	$$ \mu H \left( g , \id \right) = \left( \mu\iH (g), \id \right) $$
	and, hence, we completed the proof that $\mu H $ is a split fibration functor.
\end{proof}

We can, then, reformulate our result in terms of the existence
of parameterized initial algebras in the base category and in the fibers. That is to say, we have:

\begin{therm}[Parameterized initial algebras are strictly indexed functors]\label{coro:parameterized-initial-algebras-for-indexed-functos}
	Let  $(  \iH , h )  $ be a strictly indexed functor from $\catL '\prodstrict \catL : ( \catD\times \catC ) ^\op \to\Cat $ to $\catL : \catC ^\op\to\Cat $, and $H:\left(\GrothD\catL '\right)\times \left(\GrothC\catL \right)   \to \GrothC\catL $ the corresponding split fibration functor. Assume that:	
	\begin{enumerate}[$\mathfrak{h}$1)]
		\item the parameterized initial algebra $\mu\iH : \catD\to \catC $ exists;
		\item for any $X\in \catD $, the parameterized initial algebra $ \mu \ih _X $ exists;
		\item for each morphism $ g : X\to Y $ in $\catD $ and $ y\in Y $, 	
	Eq.~\eqref{coro:strange-condition-preservation} holds.
		\begin{equation}\label{coro:strange-condition-preservation}
			\catL \left( \mu\iH (g) \right) (\ind _{\ih _ Y^y } ) =  \ind _ {\ih _ X^{\catL ' (g) (y) }} 
		\end{equation} 		
	\end{enumerate}	
	In this setting, the parameterized initial algebra $$\mu H: \GrothD\catL ' \to \GrothC\catL $$   is a split fibration functor coming from the strictly indexed functor  $\left(\mu \iH ,  \mu\left( \ih _ {(-)} \right) \right) $ in which, for each $X\in\catD $,
	\begin{equation}
		\mu\left( \ih _ {(X)} \right)  = \mu\ih _ {X} = \mu\left( \catL(\ind _{\iH ^X} )^{-1} \ihmXx \right) : \catL '(X) \to\catL (\mu \iH ^X ) .  
	\end{equation}	 
\end{therm}
\begin{proof}
	By Theorem \ref{theo:parameterized-initial-algebras-for-indexed-functors} (Eq.~\eqref{eq:first-eq-muH-splitfibration}) and Proposition \ref{prop:GrothedieckSplitFibrations-vs-IndexedCategories} (Eq.~\eqref{eq:fibration-to-indexed-category}),
	we have that $$ \mu H: \GrothD\catL ' \to \GrothC\catL $$
	comes from the indexed category $ (\mu\iH , \mathfrak{h} ) $
	in which, for each $X\in\catD  $ and each morphism $ f': x\to w $
	in $ \catL ' (X) $, 
	\begin{align*}
		& \mathfrak{h} _ X (f') \\
		& = \mu H ( \id _X , f' ) \\
		& = \left( \id _{\mu \iH ^X } ,  \,  \fold _{ \ih _X^x } \,   \left( \ind _{ \ih _ X^w }  \circ \catL \left( \ind _{\iH  ^X } ^{-1} \right) \left( h_{ (X, \mu\iH ^X ) } (f', \mu\ih _ X^w )   \right) 
		\right) \right)  \explainr{ Eq.~\eqref{eq:first-eq-muH-splitfibration} } \\
		& = \left( \id _{\mu \iH ^X } , \,   \fold _{ \ih _X^x } \,  \left( \ind _{ \ih _ X^w }  \circ  \ih _ X \left(  f', \mu\ih _ X^w    \right) 
		\right) \right) \\
		& = \left( \id _{\mu \iH ^X } , \,   \mu  \ih _X (f')   \right) \explainr{Proposition~\ref{prop:general_parameterized_initial_algebras} }
	\end{align*} 
\end{proof}

Finally, for strictly indexed categories respecting initial algebras (see Definition \ref{def:indexed-category-respecting-initial-algebras}), we get a cleaner version of Theorem  \ref{coro:parameterized-initial-algebras-for-indexed-functos} below.

\begin{corollary}[Parameterized initial algebras and strictly indexed categories respecting initial algebras]\label{coro:finalresult-of-initial-algebras-fibrations}
	Let  $(  \iH , h )  $ be a strictly indexed functor from $\catL '\prodstrict\catL  :  ( \catD\times\catC ) ^\op \to\Cat $ to $\catL : \catC ^\op\to\Cat $, and $H:\left(\GrothD\catL '\right)\times \left(\GrothC\catL \right)   \to \GrothC\catL $ the corresponding split fibration functor. Assume that:	
	\begin{enumerate}[$\mathfrak{h}$1)]
		\item $\catL$ respects initial algebras;
		\item the parameterized initial algebra $\mu\iH : \catD \to \catC $ exists;
		\item for any $X\in \catD $, the parameterized initial algebra $ \mu \ih _X $ exists.
	\end{enumerate}	
	In this setting, the parameterized initial algebra $$\mu H: \GrothD\catL '  \to \GrothC\catL $$   is a split fibration functor coming from the strictly indexed functor  $\left(\mu \iH ,  \mu\left( \ih _ {(-)} \right) \right) $ in which, for each $X\in\catD $,
	\begin{equation}
		\mu\left( \ih _ {(X)} \right)  = \mu\ih _ {X} = \mu\left( \catL(\ind _{\iH ^X} )^{-1} \ihmXx \right) : \catL ' (X ) \to\catL (\mu \iH ^X ) .  
	\end{equation}	 
\end{corollary}
\begin{proof}
	By Theorem \ref{coro:parameterized-initial-algebras-for-indexed-functos}, it is enough to show that
	Equation \eqref{coro:strange-condition-preservation} holds whenever $\catL $ respects initial algebras.
	
	We have that, for any morphism $g : X\to Y $ in $\catD $, and each $y\in \catL '(Y) $, by the naturality of $h: \catL ' \prodstrict \catL  \longrightarrow \left(\catL  \circ \iH ^\op\right) $
	and the definition of $\mu\iH (g) $, the squares 
	\begin{equation*}
		\begin{tikzcd}
			\catL \left( \mu \iH ^Y \right)  \arrow[rrr, "{\catL \left(\mu \iH (g)\right) }"] \arrow[swap,d, "{\left( y, \id _{\catL ( \mu \iH ^Y ) }  \right) }"] &&& \catL ( \mu \iH ^X ) \arrow[d,"{\left(\catL ' (y), \id _{\catL ( \mu \iH ^X ) }  \right) }"] \\
			\catL ' \left( Y  \right) \times \catL \left( \mu\iH ^Y \right)  \arrow[d, swap, "{h_{\left( Y, \mu\iH ^Y \right) }}"] \arrow[rrr, "{\catL ' \left( g  \right)\times \catL \left(  \mu\iH (g) \right) }"] &&& \catL ' \left( X  \right) \times \catL \left( \mu\iH ^X \right)\arrow[d, "{h_{\left( X, \mu\iH ^X \right) }}"] \\
			\catL \left( \iH \left( Y, \mu\iH ^Y \right)  \right)\arrow[d, swap, "{\catL\left( \ind _{\iH ^Y } \right) ^{-1}}"] \arrow[rrr, "{\catL\left( \iH \left( g, \mu \iH (g) \right)  \right) }"] &&&\catL \left( \iH \left( X, \mu\iH ^X \right)  \right)\arrow[d, "{\catL\left( \ind _{\iH ^X } \right) ^{-1}}"] \\
			\catL \left( \mu\iH ^Y\right) \arrow[rrr, swap, "{\catL\left( \mu\iH (g)\right) }"] &&&\catL \left( \mu\iH ^X\right)                     
		\end{tikzcd}
	\end{equation*} 
	commute. Thus, we get that 	
	\begin{align*}
		& \catL\left( \mu\iH (g)\right)\circ\ih _ {Y}^y \\
		& = \catL\left( \mu\iH (g)\right)\circ\ih _ {Y}\circ \left( y, \id _{\catL ( \mu \iH ^Y ) }  \right) \\ 
		& = \catL\left( \mu\iH (g)\right)\circ\catL\left( \ind _{\iH ^Y } \right) ^{-1} \circ h_{\left( Y, \mu\iH ^Y \right)}\circ \left( y, \id _{\catL ( \mu \iH ^Y ) }  \right)   \\
		& = \catL\left( \ind _{\iH ^X } \right) ^{-1} \circ h_{\left( X, \mu\iH ^X \right)}\circ \left(\catL ' (y), \id _{\catL ( \mu \iH ^X ) }  \right)\circ\catL\left( \mu\iH (g)\right)  \\
		& =\ih _ X ^{\catL ' (y) } \circ\catL\left( \mu\iH (g)\right)  .
	\end{align*}	
	Therefore, assuming that $\catL $ respects initial algebras, we conclude that
	\begin{equation*}
		\catL \left( \mu\iH (g) \right) (\ind _{\ih _ Y^y } ) =  \ind _ {\ih _ X^{\catL ' (g) (y) }} 
	\end{equation*}
	holds. That is to say \eqref{coro:strange-condition-preservation} holds for any $g : X\to Y $ in $\catD $ and any $y\in\catL ' (Y)$. This completes the proof by Theorem \ref{coro:finalresult-of-initial-algebras-fibrations}.	
\end{proof}

\subsection{General result on terminal coalgebras in total categories} \label{sect:appendix-coinductivetypes-grothendieckconstruction}
Analogously to the case of initial algebras above, 
in order to give basis for our study in \ref{subsect:parameterized-initial-algebras-coalgebras}, 
we investigate the general case of parameterized terminal coalgebras of split fibration functors 
like in \eqref{eq:strictly-indexed-functor-split-fibration-functor}. 

Definition \ref{def:preservation-of-initial-algebras} on initial algebra preserving functors plays a central role in Theorem \ref{theo:fundamental-theorem-on-terminal-coalgebra-Grothendieck}. Specifically, we use this definition in the context of indexed categories, where we define:

\begin{definition}[Initial-algebra-respecting]\label{def:indexed-category-respecting-initial-algebras}\label{def:indexed-category-respecting-terminal-coalgebras}
	A strictly indexed category $\catL : \catC ^\op\to \Cat $ \textit{respects initial algebras} if $\catL (f) $ strictly preserves initial
	algebras for any morphism $f$ of $\catC $.\footnote{We could have allowed non-strict preservation but, in our context, it is more practical to keep things as strict as possible.} 
	
	Dually, $\catL : \catC ^\op\to \Cat $ \textit{respects terminal coalgebras } if $\catL (f) $ strictly preserves terminal coalgebras for any morphism $f$ of $\catC $.
\end{definition}

\begin{therm}[Terminal coalgebras of strictly indexed endofunctors]\label{theo:fundamental-theorem-on-terminal-coalgebra-Grothendieck}
	Let  
	$(  \iE , e )  $
	be a strictly indexed endofunctor on $\catL : \catC ^\op\to\Cat $
	and $E: \GrothC\catL \to \GrothC\catL $ the corresponding split fibration endofunctor. 
	Assume that:
	\begin{enumerate}[$\mathfrak{e}$1)]
		\item $\catL $ respects terminal coalgebras;
		\item the terminal $\iE$-coalgebra $\left( \nu \iE,\, \coind _{\iE}  \right) $ exists;
		\item the terminal $\left( \catL(\coind _{\iE} ) \ienu \right) $-coalgebra $\left( \nu \left( \catL(\coind _{\iE} ) \ienu\right),\, \coind _{\catL(\coind _{\iE} ) \ienu } \right)  $ exists.
	\end{enumerate}
	Denoting by $\ienn $ the endofunctor $\catL \left( \coind _{\iE } \right) \ienu $ on $\catL (\nu \iE ) $,
	the terminal $E$-coalgebra exists and is given by
	\begin{equation}\label{eq:terminal-coalgebra-structure-splitendofunctor}
		\nu E = \left( \nu \iE ,\,\nu \ienn \right), \qquad \coind _E = \left( \coind _{\iE },\, \coind _{\ienn } \right).
	\end{equation}
	Moreover, 
	for each $E$-coalgebra 
	$$\left( (Y,y),\, (\xi ,\xi '): (Y,y)\to E (Y,y)  \right) = \left(  (Y,y), \left( \xi : Y\to \iE (Y), \xi ': y\to \catL (\xi ) e_Y (y)  \right) \right), $$ 
	we have that
	\begin{equation}\label{eq:unfold-terminalcoalgebras-grothendieckconstruction-endofunctor}
		\unfold _{E} \left( \xi , \xi ' \right)  = \left( \unfold _{\iE }  \xi , \, \,  
		\unfold _{\catL \left( \xi  \right) e_Y } \xi ' \right)  .
	\end{equation}
\end{therm}

\begin{proof}
	Under the hypothesis above, given an $E$-coalgebra
	$$ \left( \xi : Y\to \iE (Y), \xi ': y\to \catL (\xi ) e_Y (y)  \right)  $$
	on $ ( Y , y ) $, we have that the diagram
	\begin{equation*}
		\begin{tikzcd}
			\catL \left( \nu \iE \right)\arrow[swap,dddd, bend right =80, "{ \ienn }"]  \arrow[rrr, "{\catL \left(\unfold _{\iE }  \xi \right) }"] \arrow[swap,dd, "{ \ienu }"] &&& \catL \left( Y \right)  \arrow[dd,"{ e_Y }"] \\
			&&&\\
			\catL \left( \iE \left( \nu \iE \right) \right) \arrow[dd, swap, "{ \catL \left(  \coind _{\iE } \right)  }"] \arrow[rrr, "{\catL \left( \iE\left(\unfold _{\iE }  \xi \right) \right) }"] &&&\catL \left( \iE \left( Y\right) \right) \arrow[dd, "{ \catL \left(  \xi \right) }"] \\
			&&&\\
			\catL \left( \nu \iE \right) \arrow[rrr, swap, "{\catL \left(\unfold _{\iE }  \xi \right) }"] &&&\catL \left( Y \right)                    
		\end{tikzcd}
	\end{equation*} 
	commutes. Thus, since  $\catL $ respects terminal coalgebras, we have that 
	\begin{equation*}
		\left( \catL \left(\unfold _{\iE }  \xi \right) \left( \nu \ienn\right)  ,\, \catL \left(\unfold _{\iE }  \xi \right) \left( \coind _{\ienn }\right) \right)  
	\end{equation*}
	is the terminal $ \catL \left(  \xi \right) e_Y $-coalgebra.
	Therefore, we have that 
	\begin{equation*}
		\unfold _{\catL \left( \xi \right) e_Y   }  \xi ' : y\to\catL \left(\unfold _{\iE }  \xi \right) \left( \nu \ienn\right) 
	\end{equation*}
	is the unique morphism of $\catL (Y) $ such that   
	\begin{equation*}
		\begin{tikzcd}
			y \arrow[dd, swap, "{ \xi ' }"] \arrow[rrrrr, "{\unfold _{\catL \left( \xi \right) e_Y   }  \xi '  }"] &&&&&\catL \left(\unfold _{\iE }  \xi \right) \left( \nu \ienn\right)  \arrow[dd, "{ \catL \left(\unfold _{\iE }  \xi \right) \left( \coind _{\ienn }\right) }"] \\
			&&&&&\\
			\catL \left(  \xi \right) e_Y\left( y \right)  \arrow[rrrrr, swap, "{  \catL \left(  \xi \right) e_Y \left( \unfold _{\catL \left( \xi \right) e_Y   }  \xi ' \right)  }"] &&&&& \catL \left(  \xi \right) e_Y \catL \left(\unfold _{\iE }  \xi \right) \left( \nu \ienn\right)                  
		\end{tikzcd}
	\end{equation*} 
	which shows that 
	$$ \left( \unfold _{\iE }  \xi , \, \,  
	\unfold _{\catL \left( \xi  \right) e_Y } \xi ' \right)  : (Y,y)\to E(Y,y) = \left( \iE (Y), e_Y (y) \right)
	$$
	is the unique morphism of $ \GrothC\catL  $
	such that 
	\begin{equation*}
		\begin{tikzcd}
			\left( Y, y \right) \arrow[dd, swap, "{ \left( \xi, \xi '\right) }"] \arrow[rrrrr, "{\left( \unfold _{\iE }  \xi , \, \, \unfold _{\catL \left( \xi  \right) e_Y } \xi ' \right)   }"] &&&&& \left( \nu \iE ,\,\nu \ienn \right) \arrow[dd, "{ \left( \coind _{\iE },\, \coind _{\ienn } \right) }"] \\
			&&&&&\\
			E\left( Y, y \right) = \left( \iE (Y), e_Y (y) \right)\arrow[rrrrr, "{  \left( \iE \left( \unfold _{\iE }  \xi \right) ,  e_Y \left( \unfold _{\catL \left( \xi  \right) e_Y } \xi '\right)  \right)  }"]  \arrow[rrrrr, swap, "{  E\left( \unfold _{\iE }  \xi , \, \, \unfold _{\catL \left( \xi  \right) e_Y } \xi ' \right)  }"] &&&&& E\left( \nu \iE ,\,\nu \ienn \right) = \left( \iE\left( \nu \iE\right) , \ienn \left( \nu \ienn\right)  \right)                 
		\end{tikzcd}
	\end{equation*} 
	commutes. This completes the proof that $\nu E = \left( \nu \iE ,\,\nu \ienn \right) $ 
	is the terminal $E$-coalgebra.
\end{proof}

\begin{therm}[Parameterized terminal coalgebras are strictly indexed functors]\label{theo:parameterized-terminal-coalgebras-for-indexed-functors}
	Let  
	$(  \iH , h )  $
	be a strictly indexed functor from 
	$\catL '\prodstrict \catL  :  \left( \catD\times \catC \right)  ^\op \to\Cat $ to $\catL : \catC ^\op\to\Cat $,
	and $H : \left(\GrothD\catL '\right) \times\left(\GrothC\catL \right)  \to \GrothC\catL $ the corresponding split fibration functor. 
	Assume that
	\begin{enumerate}[$\mathfrak{h}$1)]
		\item $\catL $ respects terminal coalgebras;
		\item for each object $X$ of $\catC $, the terminal $\iH ^X $-coalgebra $\left( \nu \iH ^X,\, \coind _{\iH ^X}  \right) $ exists;
		\item for each object $(X,x) $ in $\GrothD\catL ' $, denoting by 
		$\ihnu _X  $ the functor 
		\begin{equation}
			\catL(\coind _{\iH ^X} ) \ihnXx : \catL ' (X) \times \catL(\nu \iH ^X ) \to\catL(\nu \iH ^X )     
		\end{equation}
		is such that the terminal $\ihnu _X ^x $-coalgebra $\left( \nu \ihnu _X^x   ,\, \coind _{\ihnu _X^x}  \right)  $ exists.
	\end{enumerate}
	In this setting, the parameterized terminal coalgebra $$\nu H: \GrothD\catL ' \to \GrothC\catL $$   is a split fibration functor coming from the strictly indexed functor  $\left(\nu \iH ,  \nu\left( \ihnu _ {(-)} \right) \right) $ in which, for each $X\in\catD $,
	\begin{equation}
		\nu\left( \ihnu _ {(X)} \right)  = \nu\ihnu _ {X} = \nu\left( \catL(\coind _{\iH ^X} ) \ihnXx \right) :  \catL ' (X) \to\catL (\nu \iH ^X ) .  
	\end{equation}
\end{therm}
\begin{proof}
	Assuming the hypothesis, 
	we conclude that, for each $(X,x)$ in $ \GrothD\catL ' $, 
	$\GrothC\catL $ has the terminal $H^{(X,x)} $-coalgebra by Theorem \ref{theo:fundamental-theorem-on-terminal-coalgebra-Grothendieck}. Hence, by Proposition \ref{prop:general_parameterized_initial_algebras}, we have that 
	$$ \nu H : \GrothD\catL '\to \GrothC\catL $$ 
	exists.	
	More precisely, given a morphism $(f, f') : (X,x)\to (Y,y) $ in $\GrothD\catL '$, we compute  $\nu H (f, f') $ below. 
	\begin{align*}
		&\nu H (f, f')\\
		&= \unfold _{H^{(Y,y)} }\left(  H\left( (f, f') , \nu H ^{(X,x)}\right) \circ \coind_{H^{(X,x)} }     \right) \explainr{Proposition~\ref{prop:general_parameterized_initial_algebras}}\\
		&= \unfold _{H^{(Y,y)} }\left(   H\left( (f, f') , \nu H ^{(X,x)} \right) \circ  \left( \coind _{\iH ^X }, \,   \coind _{ \ihnu _ X^x }  \right)  \right) \explainr{  Eq.~\eqref{eq:terminal-coalgebra-structure-splitendofunctor}}\\
		&= \unfold _{H^{(Y,y)} }\left( \left( \iH (f, \nu \iH ^{X}), h_{ (X, \nu\iH ^X ) } ( f', \nu\ihnu _ X^x )   \right) \circ    \left( \coind _{\iH ^X }, \,   \coind _{ \ihnu _ X^x }  \right)        \right) \explainr{hypothesis}\\
		&= \unfold _{H^{(Y,y)} }\left( \iH (f, \nu \iH ^{X})\circ\coind _{\iH ^X }   ,\, \catL \left(\coind _{\iH ^X } \right)  \left( h_{ (X, \nu\iH ^X ) } ( f', \nu\ihnu _ X^x )\right)  \,  \circ\coind _{ \ihnu _ X^x } \right) \explainr{composing} \\
		&= \unfold _{H^{(Y,y)} }\left( \iH (f, \nu \iH ^{X})\circ\coind _{\iH ^X }   ,\, \ihnu _X \left( f', \nu\ihnu _ X^x \right)    \circ\coind _{ \ihnu _ X^x } \right) \explainr{definition of $\ihnu _ X$}\\
		& =  \left( \unfold _{\iH ^Y } \left( \iH (f, \nu \iH ^{X})\circ\coind _{\iH ^X }\right)   , \,\unfold _{\ihnu_Y ^y } \left( \ihnu _X \left( f', \nu\ihnu _ X^x \right)    \circ\coind _{ \ihnu _ X^x }\right) \right)\explainr{Eq.~\eqref{eq:unfold-terminalcoalgebras-grothendieckconstruction-endofunctor}}\\
		& = \left( \nu \iH (f), \nu \ihnu _Y (f')    \right) \explainr{Proposition~\ref{prop:general_parameterized_initial_algebras}}
	\end{align*}
	Since $\nu H (f, f') = \left( \nu \iH (f), \nu \ihnu _Y (f')    \right) $,  clearly, then, the pair $ \left( \nu H, \nu \iH \right) $ 
	satisfies 
	Eq.~\eqref{eq:1-condition-morphism-split-fibrations} and Eq.~\eqref{eq:2-condition-morphism-split-fibrations} of Proposition 	\ref{eq:1-condition-morphism-split-fibrations}. Moreover, $\nu H $ comes from the  
	strictly indexed functor  $\left(\nu \iH ,  \nu\left( \ihnu _ {(-)} \right) \right) $.
\end{proof}


\subsection{$\mu\nu$-polynomials in total categories}\label{subsect:parameterized-initial-algebras-coalgebras}

We examine the existence of \textit{$\mu\nu$-polynomials} in $\GrothC\catL $ and $\GrothC\catL ^\op $. In order to do so,
we employ the results and terminology established in \ref{theo:fundamental-theorem-on-initial-algebra-semantics=Grothendieck} and \ref{sect:appendix-coinductivetypes-grothendieckconstruction} 

Making use of Definition~\ref{definition:extesive-indexed-categories} and Definition~\ref{def:strictly-indexed-biproducts}, we introduce the following concept to provide support for our definition of $\Sigma$-bimodel for inductive and coinductive types:

\begin{definition}[$\mnPoly _ \catL $]\label{def:mnpolyL}
	Let $\catC $ be a category with $\mu\nu $-polynomials, and  $\catL : \catC ^\op \to\Cat $ an extensive strictly indexed category with strictly indexed finite biproducts. We define the category  $ \mnPoly_\catL  $ as the smallest subcategory of $\Cat $ satisfying 
	the following.
	\begin{enumerate}
		\item[O)] The objects are defined inductively by:
	\begin{enumerate}[O1)]
		\item the terminal category $\terminal $ is an object of $\mnPoly _ \catL $;  
		\item if $\catD $ and $\catD ' $ are objects of  $\mnPoly _ \catL $, then so is $\catD \times \catD ' $;  
		\item for each object $W\in\catC $, the category $\catL (W) $ is an object of $\mnPoly _ \catL $.
	\end{enumerate}
       \item[M)] The morphisms satisfy the following properties:
	\begin{enumerate}[M1)]
	\item for any object $\catD $ of $\mnPoly _ \catL $, the unique functor	$\catD \to \terminal $ is a morphism of $\mnPoly _ \catL $;
	\item for any object $\catD $ of $\mnPoly _ \catL $, all the functors $\terminal \to \catD $ are morphisms of $\mnPoly _ \catL $;
	\item for each $(W, X)\in \catC\times \catC $, the projections $\pi _1 : \catD \times \catD ' \to \catD $ and $\pi_2 : \catD \times \catD ' \to \catD '$
	are morphisms of  $\mnPoly _ \catL $;
	\item  for each $W\in\catC $, the biproduct $ +  : \catL (W) \times\catL (W)  \to \catL (W) $ is a morphism of  
	$\mnPoly _ \catL $;
	\item for each $(W, X)\in \catC\times \catC $, the functor 
	\begin{equation*}
		\equivalenceextensive ^{(W,X)} : \catL (W)\times \catL (X)\to \catL (W\sqcup  X) 
	\end{equation*}
	of the extensive structure (see \eqref{eq:equivalence-extensive-indexed-category})
	is a morphism of $\mnPoly _ \catL $;
	\item given an object $\catD $ of $\mnPoly _\catC $, a morphism $ \iH  : \catD  \times \catC    \to \catC    $ of $\mnPoly _ \catC $ and any object $X\in \catD ' $,  
	\begin{eqnarray*} 
		\catL (\ind _{\iH  ^X} )^{-1} : & \catL \left(  \iH  ^X \left( \mu \iH  ^ X \right) \right) &\to  \catL \left( \mu \iH  ^X \right) ,\\ 
		\catL (\coind _{\iH ^X } ) : &\catL \left(  \iH ^X \left( \nu \iH ^X \right) \right) & \to  \catL \left( \nu \iH ^X \right) 
	\end{eqnarray*}
	are morphisms of  $\mnPoly _ \catL $; \label{def:mnpolyL:munuHconditionformnPolyL}
		\item for each $(W, X)\in \catC\times \catC $, the functors induced by the projections
	\begin{equation*} 
		\catL (\pi _1 ) :\catL \left( W \right)\to  \catL \left( W\times X \right) , \qquad \catL (\pi _2 ) :\catL \left( X \right)\to  \catL \left( W\times X \right) 
	\end{equation*} 
	are morphisms of $\mnPoly _ \catL $;
	\item if $E: \catD   \to \catD  '   $ and $J : \catD   \to \catD ''    $ are morphisms of $\mnPoly _ \catL $, then so is $(E,J) :\catD    \to \catD ' \times\catD ''   $;
	\item if $\catD ', \catD $ are objects of $\mnPoly _ \catL $, $h: \catD '\times \catD  \to\catD    $ is a morphism of $\mnPoly _ \catL $  and $\mu h : \catD ' \to \catD   $ exists,
	then $\mu h $ is a morphism of  $\mnPoly _ \catL $;
	\item  if $\catD ', \catD $ are objects of $\mnPoly _ \catL $, $h: \catD '\times \catD  \to\catD    $ is a morphism of $\mnPoly _ \catL $  and $\nu h : \catD '\to \catD   $ exists,
	then $\nu h $ is a morphism of  $\mnPoly _ \catL $.
	\end{enumerate}
	\end{enumerate}
\end{definition}

Having established the previous definition, we can now introduce the notion of a $\Sigma$-bimodel for inductive and coinductive types:

\begin{definition}[$\Sigma$-bimodel for inductive and coinductive types]\label{def:suitL-indexedcategory}
We say that $\catL : \catC ^\op\to\Cat $ is a 
\emph{$\Sigma$-bimodel for inductive and coinductive types}  (or, for short, a \textit{$\suitL $-indexed category}) if: 
\begin{enumerate}[$\suitL$1)]
	\item $\catL $ is a strictly indexed category;
	\item $\catC $ has $\mu\nu$-polynomials (Definition~\ref{def:basic-definition-munupolynomials});
			\item  $\catL : \catC ^\op\to\Cat $ has strictly indexed finite biproducts (Definition~\ref{def:strictly-indexed-biproducts});
 	\item $\catL $ is extensive (Definition~\ref{definition:extesive-indexed-categories});
 	\item $\catL $ respects terminal coalgebras and initial algebras (Definition~\ref{def:indexed-category-respecting-initial-algebras});
	\item whenever $\catD $ is an object of $\mnPoly _ \catL $ and $e: \catD  \to\catD    $ is a morphism of $\mnPoly _ \catL $, $\mu e   $ and $\nu e $ exist.
\end{enumerate} 
\end{definition}

\begin{lemma}\label{lem:munucompleteness-mnPolyL}
Let $\catL : \catC ^\op \to \Cat $ be a $\suitL$-indexed category. If $\catD , \catD '$ are objects of  $\mnPoly _ \catL $ then, whenever $h : \catD '\times \catD \to \catD $ is a morphism of   $\mnPoly _ \catL $, 
\begin{center}
$\mu h : \catD '\to\catD $\qquad and\qquad $\nu h : \catD '\to \catD $
\end{center}
exist.
\end{lemma}
\begin{proof}
By Proposition \ref{prop:general_parameterized_initial_algebras}, it is enough to show that, for each $x\in \catD ' $, $\mu h^x $ and $\nu h^x $ exist.

In fact, denoting by $x: \terminal \to \catD '$ the functor constantly equal to $x\in\catD ' $, the functor $h^x $ is the composition below.
\begin{equation*}
\begin{tikzcd}
\catD \arrow[rr, "{ \left( 1, \id _{\catD} \right) }"] 
\arrow[swap,bend right=15, rrrrrrrrrr, "{ h^x }"] && 
\terminal\times\catD 
\arrow[rrrr, "{ \left( x\circ \pi _1 , \id _{\catD}\circ\pi _2 \right) }"] 
&&&&\catD '\times\catD \arrow[rrrr, "{ h }"] &&&& \catD
\end{tikzcd}
\end{equation*}
Since all the horizontal arrows above are morphisms of  $ \mnPoly _ \catL $, we conclude that $h^x $ is an endomorphism of $ \mnPoly _ \catL $. Therefore, since $\catL $ is a
$\suitL$-indexed category,
$\mu h^x $ and $\nu h ^x $ exist.
\end{proof}

\begin{definition}[$\suitHL $-indexed category and indexed functor]\label{def:suitHL-indexed category and indexed functor}
Let $\catL : \catC ^\op \to \Cat $, $ \catL ': \catD ^\op \to \Cat $ be strictly indexed categories.
We say that $\catL '$ is a \textit{$\suitHL $-indexed category} if: 
\begin{enumerate}[label=$\suitHL$\arabic*),series=suitHL]
\item \label{1def:suitHL-indexed category and indexed functor} $ \cat D $ is an object of $\mnPoly _{\catC} $; 
\item \label{2def:suitHL-indexed category and indexed functor} $\catL '(W) $ is an object of $\mnPoly _{\catL } $  for any $W$ in $\catD $.
\end{enumerate}
A strictly indexed functor  $(  \iH , h )  $ between 
$\catL ': \catD ^\op \to \Cat  $ and $\catL '' : \catE ^\op \to \Cat $ is a \textit{$\suitHL $-indexed functor} if:
\begin{enumerate}[resume*=suitHL]
	\item $\catL '  , \catL ''  $ are $\suitHL $-indexed categories;
	\item $\iH : \catD \to\catE $ is a morphism of $\mnPoly _ { \catC } $;
	\item for each 	$X \in \catD  $,  $ h_{X} : \catL '\left( X \right) \to \catL ''  \circ \iH (X)$	is a morphism of $\mnPoly _{\catL } $.
\end{enumerate}
\end{definition} 

\begin{therm}\label{theo:split-fibration-functor-munu}
Let $\catL ': \catD ^\op \to \Cat $ be a strictly indexed category and 
$\catL : \catC ^\op \to \Cat $ a $\suitL$-indexed category. 
Assume that
$(  \iH , h )  $ is a $\suitHL $-indexed functor,  and $H:\GrotED\left( \catL '\prodstrict \catL  \right)\cong \left(\GrothE\catL '\right)\times \left(\GrothD\catL  \right)   \to \GrothD\catL  $ is the corresponding split fibration functor.  We have that:
 \begin{enumerate}[i)]
	\item $\mu H : \GrothE\catL ' \to \GrothD\catL  $ exists and
	is the split fibration functor induced by the $\suitHL $-indexed functor
	\begin{equation}\label{eq:muH-forthePolynomials}
		\left( \mu\iH :  \catE\to \catD  , \,  \mu\left( \ih _ {(-)} \right)  \right) 
	\end{equation}
	in which 
	\begin{equation}
		\mu\left( \ih _ {(X)} \right)  = \mu\ih _ {X} = \mu\left( \catL(\ind _{\iH ^X} )^{-1} \ihmXx \right) :  \catL ' ( X ) \to\catL  (\mu \iH ^X ) .  
	\end{equation}	
	\item $\nu H :  \GrothE \catL ' \to \GrothD\catL  $ exists and
	is the split fibration functor induced by the $\suitHL $-indexed functor
	\begin{equation}\label{eq:nuH-forthePolynomials} 
		\left(\nu \iH :  \catE\to \catD  , \,  \nu\left( \ihnu _ {(-)} \right) \right) 
	\end{equation} 
	in which 
	\begin{equation}
		\nu\left( \ihnu _ {(X)} \right)  = \nu\ihnu _ {X} = \nu\left( \catL(\coind _{\iH ^X} ) \ihnXx \right) : \catL '' (X )  \to\catL ' (\nu \iH ^X ) .  
	\end{equation} 
\end{enumerate}
Furthermore, both $\mu H $ and $\nu H $ are  $\suitHL $-indexed functors.
\end{therm}
\begin{proof}
	Since $\catC $ has $\mu\nu $-polynomials, $\catD $ is an object
	of $\mnPoly _ {\catC } $ and $\iH $ is a morphism of $\mnPoly _{\catC} $, we have that 
	$\mu\iH $ and $\nu\iH $ exist by Lemma \ref{lem:about-completeness-of-munupolynomials} (and, hence, are morphisms in $ \mnPoly _{\catC } $). Moreover,
we have that 	$\catL(\coind _{\iH ^X} ) $ and $ \catL(\ind _{\iH ^X} )^{-1} $
are morphisms of $\mnPoly _ {\catL} $ by \ref{def:mnpolyL:munuHconditionformnPolyL} of Definition \ref{def:mnpolyL}. 

For any $X\in\catD $, since $\left( \iH, h\right) $ is a $\suitHL$-indexed functor, we have that,  $\catL '\left( X\right) $  is an object of $\mnPoly _{\catL} $ and
\begin{eqnarray*}
\ihmXx  : & \catL '\left( X\right)\times \catL \left(\mu\iH ^X \right) &\to \catL \circ \iH \left( X, \mu\iH ^X \right) \\
\ihnXx : & \catL '\left( X\right)\times \catL \left(\nu\iH ^X \right) & \to \catL \circ \iH \left( X, \nu\iH ^X \right)  
\end{eqnarray*}
are morphisms of $\mnPoly _{\catL } $. 

We conclude, then, that the compositions
\begin{eqnarray*}
\ih _X = \catL(\ind _{\iH ^X} ) ^{-1} \ihmXx  : & \catL '\left( X\right)\times \catL \left(\mu\iH ^X \right) & \to \catL \left(\mu\iH ^X \right) \\
\ihnu _X = \catL(\coind _{\iH ^X} ) \ihnXx  : & \catL '\left( X\right)\times \catL \left(\nu\iH ^X \right) & \to \catL \left(\nu\iH ^X \right)
\end{eqnarray*}
are also morphisms of $\mnPoly _ {\catL} $. Thus, we have that $\mu \ih _ X $ and $\nu \ihnu _X  $ exist (and are morphisms of $\mnPoly _{\catL } $) by Lemma \ref{lem:munucompleteness-mnPolyL}.

Finally, since $\catL $ respects initial algebras and terminal coalgebras, we have that
$(  \iH , h )  $ satisfies the hypotheses of Corollary \ref{coro:finalresult-of-initial-algebras-fibrations} and Theorem \ref{theo:parameterized-terminal-coalgebras-for-indexed-functors}.  Therefore $\mu H $ and $\nu H $ exist and are induced by 
\eqref{eq:muH-forthePolynomials} and \eqref{eq:nuH-forthePolynomials} respectively.

The fact that \eqref{eq:muH-forthePolynomials} and \eqref{eq:nuH-forthePolynomials}
are also $\suitHL $-indexed functors follows from the fact that $\catL '$ is a $\suitHL $-indexed category by hypothesis, $ \mu \iH $ is a morphism of $\mnPoly _{\catC } $ (as observed above)
and $\mu \ih _ X, \nu \ihnu _X  $ are morphisms of $\mnPoly _{\catL} $ (also observed above).
\end{proof}
In particular, we see that initial algebras and terminal coalgebras of $\mu\nu$-polynomials in $\Sigma_\catC \catL$ (and, codually, $\Sigma_\catC \catL ^{op}$) are fibred over $\catC$.

Before proving Theorem \ref{theo:Inductive-Coinductive-Completeness-TotalCategory}, our main theorem about $\mu\nu$-polynomials in $\GrothC\catL $, we 
prove Lemma \ref{lem:the-bijection-between-mnPoly-and-indexedcategories} which
establishes a bijection 
between objects of $\mnPoly _{\GrothC\catL } $ and indexed categories.

\begin{definition}
Let $\catL : \catC ^\op \to \Cat $ be a strictly indexed category. We inductively define the 
set $\Produ \catL  $ of indexed categories as follows:
\begin{enumerate}[$\Produ \catL$1.]
\item the terminal indexed category $\terminal : \terminal \to \Cat $ belongs to $\Produ \catL  $;
\item $\catL $ belongs to $\Produ \catL  $;
\item if $\catL '$ and $\catL '' $ belong to $\Produ \catL  $, then 
$\left( \catL '\prodstrict \catL ''\right) \in \Produ \catL  $.  
\end{enumerate}
\end{definition}

\begin{lemma}\label{lem:polynomials-are-suitHL-indexed-categories}
	Let $\catL : \catC ^\op\to\Cat $ be a strictly indexed category.  Then all the elements of $\Produ\catL $ are $\suitHL $-indexed categories.
\end{lemma}
\begin{proof}
	The terminal indexed category $\terminal : \terminal\to\Cat $ is 
	a $\suitHL $-indexed category since $\terminal\in\mnPoly _{\catC } $ and $\terminal\in\mnPoly _{\catL } $.
	Furthermore, $\catL : \catC ^\op \to\Cat $ is a $\suitHL $-indexed category by the definition of $\mnPoly _{\catL }$. 
	
	Finally, if $\catL ': \catD ^\op \to\Cat $ and $\catL '' :\catE ^\op \to\Cat $ are $\suitHL $-indexed categories, then:
	\begin{enumerate}[--]
		\item we have that $\left( \catD, \catE \right)\in\mnPoly _{\catC }\times  \mnPoly _{\catC } $. Thus
		\begin{equation}\label{1lem:polynomials-are-suitHL-indexed-categories}
			\left( \catD \times \catE\right)\in  \mnPoly _{\catC } ;
		\end{equation}
		\item for any $\left( W, W'\right) \in \catD\times\catE $, 
		the categories $\catL ' ( W ) $ and $\catL '' (W ') $ are objects of $\mnPoly _{\catL} $. Thus
		\begin{equation}\label{2lem:polynomials-are-suitHL-indexed-categories}
			\catL '\prodstrict \catL '' \left(  W, W' \right) = \catL ' (W) \times\catL '' (W ') \in \mnPoly _{\catL } .  
		\end{equation} 
	\end{enumerate}
	By \eqref{1lem:polynomials-are-suitHL-indexed-categories} and \eqref{2lem:polynomials-are-suitHL-indexed-categories}, we conclude that
	$\catL '\prodstrict \catL '' : \left( \catD \times \catE\right) ^\op \to \Cat $ 
	is a $\suitHL$-indexed category.
\end{proof}

\begin{lemma}\label{lem:the-bijection-between-mnPoly-and-indexedcategories}
Let $\catL : \catC ^\op \to \Cat $ be a strictly indexed category.
The function 
\begin{equation} 
\inversebijpoly : \objects{\left(\mnPoly _{\GrothC\catL }\right)}\to\Produ  \catL 
\end{equation}
 inductively defined by
\ref{def:delta-1}, \ref{def:delta-2} and \ref{def:delta-3} is a bijection.
\begin{enumerate}[$\inversebijpoly $1. ]
\item terminal respecting:
$\inversebijpoly \left( \terminal\right) \ceq  \left( \terminal : \terminal \to\Cat \right)$;\label{def:delta-1}
\item basic element: $ \inversebijpoly \left(  \GrothC\catL\right)  \ceq \left( \catL : \catC ^\op\to\Cat\right) $;\label{def:delta-2} 
\item product respecting: given $\left( \catD , \catD '\right) \in \mnPoly _{\GrothC\catL }\times \mnPoly _{\GrothC\catL } $,
$$ \inversebijpoly \left( \catD\times\catD ' \right) \ceq \inversebijpoly \left( \catD\right) \prodstrict\inversebijpoly \left( \catD ' \right) . $$\label{def:delta-3}
\end{enumerate}
\end{lemma}
\begin{proof} 
The inverse of $\inversebijpoly  $ is clearly given by the Grothendieck construction. More precisely, the inverse is denoted herein by $\bijpoly  $ and can be inductively defined as follows:
\begin{enumerate}[$\bijpoly $1) ]
\item terminal respecting:
$\bijpoly \left( \terminal : \terminal \to\Cat \right) \ceq  \terminal $;
\item basic element: $ \bijpoly \left(  \catL : \catC ^\op\to\Cat\right)  \ceq \GrothC\catL $; 
\item product respecting: given $\left( \catL ' : \catD ^\op\to\Cat  , \catL '' : \catE^\op \to\Cat\right) \in \Produ \catL \times \Produ\catL $,
$$ \bijpoly \left( \catL '\prodstrict\catL '' \right) \ceq \bijpoly \left( \catL '\right) \times\bijpoly \left( \catL '' \right) . $$
\end{enumerate}
By the inductive definitions of the sets $\objects\left(\mnPoly _{\GrothC\catL }\right) $ and
$\Produ\catL $, we conclude that 
\begin{center}
$\bijpoly\circ \inversebijpoly = \id _{\objects\left( \mnPoly _{\GrothC\catL }\right)} $\qquad and\qquad $\inversebijpoly\circ \bijpoly = \id _{\Produ  \catL } $.
\end{center}
\end{proof}

\begin{lemma}\label{lem:subcategory-suitHL-produL}
	Let $\catL : \catC ^\op \to \Cat $ be a strictly indexed category. The objects of  $ \mnPoly _{\GrothC\catL } $ with the functors that are induced by 
	$\suitHL $-indexed functors between objects of $\Produ\catL$ form a subcategory of $\Cat $.
\end{lemma}	
\begin{proof}
	Let $\catA $ be an object of $\mnPoly _{\GrothC\catL }  $.
	By Lemma \ref{lem:the-bijection-between-mnPoly-and-indexedcategories}, we have the associated strictly indexed category
	\begin{equation*}
		\inversebijpoly\left( \catA \right) = \catL ':  \catD ^\op  \to \Cat .
	\end{equation*}
	The identity $\id _{\catA } $ on $\catA $ clearly comes from the identity 
	$$\left( \id_{\catD} :  \catD\to\catD , \id   \right) : \catL '\to\catL ' $$
	which is a $\suitHL $-indexed category, since $\catL '$ is a $\suitHL$-indexed category
	by Lemma \ref{lem:polynomials-are-suitHL-indexed-categories}.
	
	Finally, if $E:\catA \to \catA ' $ and $H : \catA '\to\catA '' $ are functors induced, respectively, 
	by the $\suitHL $-indexed functors 
	\begin{center}
		$\left( \iE , e \right) : \catL ' \to\catL '' $ \qquad and \qquad $\left( \iH , h \right) : \catL '' \to\catL ''' $,
	\end{center} 
	then $H\circ E $ is induced by the composition 
	$$ \left( \iH\circ \iE , h_{\iE ^{\op}} \circ e    \right) $$
	which is a $\suitHL$-indexed functor as well, since $\iH $, $\iE $ are morphisms of $\mnPoly _{\catC} $ and, for any $W\in\catD $,  $h_{\iE (W) } $ and   $e _W $ are morphisms of $\mnPoly _ \catL $. 
\end{proof}

\begin{definition}
	We denote by $\overline{\mnPoly _{\GrothC \catL } } $ the category defined in Lemma \ref{lem:subcategory-suitHL-produL}.
\end{definition}

\begin{therm}\label{theo:Inductive-Coinductive-Completeness-TotalCategory}
	Let $\catL : \catC ^\op \to \Cat $ be a $\suitL$-indexed category.
	The category $\GrothC\catL $ has $\mu\nu$-polynomials.
\end{therm} 
\begin{proof}
By Theorem \ref{theo:split-fibration-functor-munu}, since $\catL$ is a $\suitL $-indexed category, any endomorphism $E : \GrothC \catL\to \GrothC\catL $ of the subcategory 
$\overline{\mnPoly _{\GrothC \catL } } $ has an initial algebra and a terminal coalgebra. 
Therefore, in order to complete the proof, it is enough to show that the morphisms of
$\overline{\mnPoly _{\GrothC \catL } } $ satisfy the inductive properties of Definition \ref{def:basic-definition-munupolynomials}.

Let $\catA $, $\catA '$ and $\catA '' $ be objects of $\mnPoly _{\GrothC\catL }  $.
By Lemma \ref{lem:the-bijection-between-mnPoly-and-indexedcategories}, we have the associated strictly indexed categories
\begin{eqnarray*}
 \inversebijpoly\left( \catA \right) = \catL ':  &\catD ^\op & \to \Cat ,\\ 
 \inversebijpoly\left( \catA '\right) = \catL '' : & \catE ^\op & \to\Cat ,\\  
 \inversebijpoly\left( \catA '' \right)  = \catL ''': &\catF ^\op & \to\Cat .
\end{eqnarray*}
Recall that $\catL ', \catL ''$ and $\catL ''' $ are $\suitHL$-indexed categories by 
Lemma \ref{lem:polynomials-are-suitHL-indexed-categories}.

\begin{enumerate}[(A)]
	\item The unique functor $\catA\to\terminal $ is induced by the unique 
	indexed functor 
		$$ \left( \catD\to\terminal ,\, \left( \catL '\left(W\right) \to\terminal   \right) _{W\in \catD }\right) $$
	between $\catL $ and the terminal indexed category $\terminal : \terminal\to\Cat $. Since 
	 $	\catD\to\terminal $ is a morphism of  $\mnPoly _{\catC } $ and, for any $W \in\catD  $,  $ \catL '\left(W\right) \to\terminal $  is a morphism of $\mnPoly _{\catL } $, we have that the unique indexed functor is a $\suitHL$-indexed functor.	
	\item Given a functor $ F : \terminal\to\catA\cong \GrothC\catL ' $, it corresponds to
	an object $\left( W\in\catD , x\in\catL '(W) \right) \in \GrothC\catL ' $. In other words, $F$ is induced by the strictly indexed functor
	$$\left( W: \terminal \to \catD , w : \terminal \to\catL '(W)\right)  $$
	in which $W$ and $w$ denote the obvious functors. Since any functor $\terminal \to \catD $ is a morphism of $\mnPoly _{\catC} $ and, for any $W\in\catD $, any functor $\terminal \to\catL '(W) $ is a morphism of   $\mnPoly _{\catL} $, we have that $\left( W: \terminal \to \catD , w : \terminal \to\catL '(W)\right)  $ is a $\suitHL $-indexed functor.
	\item By Proposition \ref{prop:grothendieck-products-covariant}, the binary product $\times : \GrothC\catL\times\GrothC\catL\to\GrothC\catL $ is induced by the 
	strictly indexed functor 
	$$\left( \times : \catC\times\catC\to\catC , \, p \right) : \catL\prodstrict\catL \to \catL  $$
	in which $p_{\left( W, W'\right) } $ is given by the composition	
\begin{equation*}
	\begin{tikzpicture}[x=0.5cm, y=0.5cm]
		\node (a) at (0,0) {$\catL \left( W\right) \times \catL \left( W '\right) $};
		\node (b) at (14,0) {$\catL \left( W\times W ' \right) \times \catL \left( W\times W ' \right) $ };
		\node (c) at (14,-4) {$\catL \left( W\times W ' \right) $};
		\draw[->] (a)--(b) node[midway,above] {$\catL\left( \pi _1 \right) \times \catL\left( \pi _2 \right)   $};
		\draw[->] (b)--(c) node[midway,right] {$ + $};
		\draw[->] (a)--(c) node[midway,below left] {$p_{\left( W, W'\right) }  $};
	\end{tikzpicture} 
\end{equation*}  

It remains to show that $\left( \times : \catC\times\catC\to\catC , \, p \right) $
is a $\suitHL$-indexed functor. Since $\times : \catC\times\catC\to\catC $ is a morphism of $\mnPoly _{\catC} $, it is enough to prove that $p_{(W, W')} $ is a morphism of 
$\mnPoly _  {\catL} $ for any 
$\left( W, W' \right)\in \catC\times\catC $.

Since, for any $\left( W, W' \right)\in \catC\times\catC $, we have that
\begin{center}
$\pi _{\catL \left( W \right) } : \catL \left( W \right) \times \catL \left( W '\right)\to \catL \left( W \right) $, \quad
$\pi _{\catL \left( W '\right) } : \catL \left( W \right) \times \catL \left( W '\right)\to \catL \left( W '\right) $
\end{center}
\begin{center} 
$\catL \left( \pi _1\right) :  \catL \left( W \right)\to \catL \left( W\times W' \right) $,  \quad
$\catL \left( \pi _2\right) :  \catL \left( W '\right)\to \catL \left( W\times W' \right) $  
\end{center}
are morphisms of $\mnPoly _{\catL} $, we conclude that 
$$\left( \catL\left( \pi _1 \right)\circ  \pi _{\catL \left( W \right) }  , \catL\left( \pi _2 \right)\circ \pi _{\catL \left( W '\right) }  \right)  
= \catL\left( \pi _1 \right) \times \catL\left( \pi _2 \right)   
$$
is a morphism of $\mnPoly _{\catL} $. Thus, since 
$\times : \catL\left( W\times W '\right)\times \catL\left( W\times W '\right)\to \catL\left( W\times W '\right) $ is a morphism of $\mnPoly _{\catL} $ as well, we conclude that 
the composition $p_{(W,W')} $
is a morphism of
 $\mnPoly _{\catL} $.
 \item By Corollary \ref{coro:cocartesianstructure-in-the-cocartesian-csategory}, the coproduct $\sqcup :\GrothC\catL\times \GrothC\catL\to\GrothC\catL $
 is induced by the strictly indexed functor
	$$\left( \sqcup : \catC\times\catC\to\catC , \, s \right) : \catL\prodstrict\catL \to \catL  $$
 in which $s_{\left( W, W'\right) } $ is given by the functor
 \begin{equation*}
 	\equivalenceextensive ^{(W,W ')} : \catL (W)\times \catL (X)\to \catL (W\sqcup  X) 
 \end{equation*}
 of the extensive structure (see \eqref{eq:equivalence-extensive-indexed-category})
 is a morphism of $\mnPoly _ \catL $. 
 
 We have that $\left( \sqcup : \catC\times\catC\to\catC , \, s \right) : \catL\prodstrict\catL \to \catL  $ is a $\suitHL$-indexed functor, since $\sqcup : \catC\times\catC\to\catC $ is a morphism of $\mnPoly _{\catC} $ and $\equivalenceextensive ^{(W,W ')} $ is a morphism of $\mnPoly _{\catL} $, for any $\left( W, W'\right)\in\catC\times\catC $.

	\item The projections 
	\begin{equation*}
		\pi _1 : \catA\times \catA ' \to \catA ,\qquad   \pi _2 : \catA\times \catA ' \to \catD '
	\end{equation*} 
    are, respectively, induced by the strictly indexed functors 
    \begin{eqnarray*}
    	\left( \pi _1 : \catD\times\catE\to \catD ,\, \left( \pi_1 : \catL (W)\times\catL (W')\to \catL (W) \right) _{\left( W, W'\right) \in \catD\times \catE }  \right) : &\catL '\prodstrict \catL '' & \to \catL '\\
	\left( \pi _2 : \catD\times\catE\to \catE ,\, \left( \pi_2 : \catL (W)\times\catL (W')\to \catL (W ') \right) _{\left( W, W'\right) \in \catD\times \catE }  \right) : &\catL '\prodstrict \catL ''& \to \catL ''
	\end{eqnarray*}
which are $\suitHL $-indexed functors, since $$\pi _1 : \catD\times\catE\to \catD,\qquad \pi _2 : \catD\times\catE\to \catE $$
are morphisms of $\mnPoly _{\catC } $  and, for any $\left( W, W'\right) \in \catD\times \catE$,  
$$ \pi_1 : \catL (W)\times\catL (W')\to \catL (W)  , \qquad \pi_2 : \catL (W)\times\catL (W')\to \catL (W ')
$$
are morphisms of $\mnPoly _{\catL } $. 	
	\item Assuming that $E: \catA\to\catA ' $ and $J : \catA\to\catA '' $ are functors induced by the  $\suitHL $-indexed functors	
	\begin{center}
$\left( \iE , e : \catL '\rightarrow \catL '' \circ \iE ^\op \right) : \catL ' \to \catL ''  $ \qquad and \qquad 	$\left( \iJ , j : \catL '\rightarrow \catL ''' \circ \iJ ^\op \right) : \catL ' \to \catL '''  $,	
	\end{center}	
the functor $ (E, J ) : \catA\to \catA '\times \catA ''  $ is induced by the 
strictly indexed functor
$$\left( \left( \iE , \iJ \right) , (e,j)\right) : \catL '\to \catL ''\prodstrict \catL ''' . $$ which is a $\suitHL$-indexed functor as well since: 
		\begin{enumerate} [--]
			\item $\iE$, $\iJ$ are morphisms of $\mnPoly _{\catC }$ and, hence, so is $\left( \iE , \iJ\right) $; 
			\item $e_W, j_W $ are morphisms of $\mnPoly _{\catL} $ for any $W\in \catD $ and, hence, so is $\left(e_W, j_W \right) $. 
		\end{enumerate}
\end{enumerate}
Finally, assuming that $H : \catA\times\GrothC\catL \to\GrothC\catL $ is a functor induced
by a $\suitHL $-functor
$$\left( \iH , h\right) : \catL '\prodstrict\catL \to \catL , $$
we have, by Theorem \ref{theo:split-fibration-functor-munu}, that
\begin{enumerate}[(A)]
	\setcounter{enumi}{6}
	\item $\mu H $ is induced by the $\suitHL $-indexed functor $$\left( \mu\iH :  \catE\to \catD  , \,  \mu\left( \ih _ {(-)} \right)  \right)  : \catL '\to \catL  .$$
	\item $\nu H $ is induced by the $\suitHL $-indexed functor $$\left(\nu \iH :  \catE\to \catD  , \, \nu\left( \ihnu _ {(-)} \right) \right)  : \catL '\to \catL  .$$
\end{enumerate} 
\end{proof}

Codually, we have:

\begin{therm}\label{theo:Inductive-Coinductive-Completeness-TotalCategory-contravariant}
	Let $\catL : \catC ^\op \to \Cat $ be a $\suitL$-indexed category.
	The category $\GrothC\catL ^\op $ has $\mu\nu$-polynomials.
\end{therm}

\subsection{$\Sigma$-bimodel for function types, inductive and coinductive types}
By Theorem \ref{theo:distributive-property-total-category}, the \textit{Grothendieck construction of any $\Sigma$-bimodel for inductive and coinductive types is distributive}. Moreover, we get the closed structure if 
$\catL $ satisfies the conditions of \ref{subsect:closed-structure-sigmatypes-Grothendieck-Construction}.
More precisely:

\begin{corollary}
	Let $\catL : \catC ^\op \to \Cat $  be a $\Sigma$-bimodel for inductive and coinductive types. The categories $\GrothC\catL  $  and $\GrothC\catL ^\op $ are distributive categories with $\mu\nu$-polynomials.
\end{corollary}

\begin{corollary}\label{coro:suitL-indexed-category-total-category-munu}
	Let $\catL : \catC ^\op \to \Cat $  be a $\Sigma$-bimodel for inductive, coinductive and function types. The categories $\GrothC\catL  $  and $\GrothC\catL ^\op $ are closed categories with $\mu\nu$-polynomials.
\end{corollary}

\section{Linear $\lambda$-calculus as an idealised AD target language}\label{sec:target-language}
We describe a target language for our AD code transformations, a 
variant of the dependently typed enriched effect calculus \citep[Chapter 5]{vakar2017search}.
Its cartesian types, linear types, and terms are generated by the  
grammar of Fig. \ref{fig:sl-terms-types-kinds} and \ref{fig:tl-terms-types-kinds}, making the target language
a proper extension of the source language.
We note that we use a special symbol $\lvar$ for the unique linear identifier.
\begin{figure}[!ht]
  \framebox{\begin{minipage}{0.98\linewidth}
\input{tl-terms-types-kinds}
  \end{minipage}}
\caption{\label{fig:tl-terms-types-kinds} A grammar for the kinds, types and terms of the target language, extending that of Fig. \ref{fig:sl-terms-types-kinds}.}
\end{figure}
We introduce kinding judgements $\Delta\mid \Gamma\vdash \ty:\type$ and 
$\Delta\mid \Gamma\vdash \cty:\ltype$ for cartesian and linear types, where $\Delta=\tvar_1:\type,\ldots,\tvar_n:\type$ is a list of (cartesian) type identifiers and $\Gamma=\var_1:\ty_1,\ldots,\var_n:\ty_n$ is a list of identifiers $x_i$ with cartesian type $\tau_i$.
These kinding judgements are defined according to the rules displayed in 
Fig. \ref{fig:sl-kind-system} and \ref{fig:tl-kind-system}.
\begin{figure}[!ht]
\framebox{\begin{minipage}{0.98\linewidth}\noindent\hspace{-24pt}\input{tl-kind-system}\end{minipage}}
    \caption{Kinding rules for the AD target language that we consider on top of those of Fig. \ref{fig:sl-kind-system}, where our first rule specifies how kinding judgements of the source language imply kinding of types in the target language.
    Observe that, according to the second rule, type variables $\tvar$ from the kinding context $\Delta $ can  be used as a linear 
    type $\ltvar$.
    Note that we only consider the formation of $\Sigma$- and $\Pi$-types 
    and linear function types of non-parameterized types (shaded in grey).\label{fig:tl-kind-system}}\;
  \end{figure}

We use typing judgements $\Delta\mid\Gamma\vdash \trm:\ty$ and $\Delta\mid\Gamma;\lvar:\cty\vdash \trm[2]:\cty[2]$ for terms of well-kinded cartesian types $\Delta\mid\Gamma\vdash \ty:\type$ and linear type $\Delta\mid\Gamma\vdash \cty[2]:\ltype$, where $\Delta=\tvar_1:\type,\ldots,\tvar_n:\type$ is a list of cartesian type identifiers, $\Gamma=\var_1:\ty_1,\ldots,\var_n:\ty_n$ is a list of identifiers $x_i$ of well-kinded cartesian type $\Delta\mid \var_1:\ty_1,\ldots,\var_{i-1}:\ty_{i-1}\vdash\ty_i:\type$ and $\lvar$ is the unique linear identifier of well-kinded linear type $\Delta\mid\Gamma\vdash \cty:\ltype$.
Note that terms of linear type always contain the unique linear identifier $\lvar$ in the typing context.
These typing judgements are defined according to the rules displayed in 
Fig. \ref{fig:sl-type-system}, \ref{fig:tl-type-system1} and 
\ref{fig:tl-type-system2}.

We work with linear operations $\lop\in\LOp_{n_1,...,n_k;
n'_1,\ldots,n'_l}^{m_1,\ldots, m_r}$,
which are intended to represent functions that are linear (in the sense of 
respecting $\zero$ and $+$) in the last $l$ 
arguments but not in the first $k$.
To serve as a practical target language for the automatic derivatives of all programs from the source language, we work with the following linear operations:
for all $\op\in \Op_{n_1,...,n_k}^{m}$, 
\begin{align*}
&D\op\in \LOp_{n_1,...,n_k;n_1,....,n_k}^m
&\transpose{D\op } = \transpose{\left( D\op\right) }\in \LOp_{n_1,...,n_k;m}^{n_1,....,n_k}.
\end{align*}
We will use these linear operations $D\op$ and $\transpose{D\op}$ as the forward and reverse derivatives of the corresponding primitive operations $\op$\footnote{Nothing would stop us from defining the derivative of a primitive operations as a more general term, rather than a linear operation. In fact, that is what we considered in \citep{vakar2021chad,vakar2020reverse}.
However, we believe that treating derivatives of operations as linear operations slightly simplifies the development and is no limitation, seeing that we are free to implement linear operations as we please in a practical AD system.
}.
We 
 write
 $$
\LDomain{\lop}\defeq \reals^{n'_1}\t*\ldots\t* \reals^{n'_l}
\qquad\text{and}\qquad
\LCDomain{\lop}\defeq \reals^{m_1}\t*\ldots\t* \reals^{m_r}
$$
for $\lop\in\LOp_{n_1,...,n_k;
n'_1,\ldots,n'_l}^{m_1,\ldots, m_r}$.
\begin{figure}[!ht]
\framebox{\begin{minipage}{0.98\linewidth}\noindent\hspace{-24pt}\input{tl-type-system1}\end{minipage}}
    \caption{Typing rules for the AD target language that we consider on top of the rules of Fig. \ref{fig:sl-type-system} and \ref{fig:tl-type-system2}.\label{fig:tl-type-system1}}\;
  \end{figure}
  \begin{figure}[!ht]
    \framebox{\begin{minipage}{0.98\linewidth}\noindent\hspace{-24pt}\input{tl-type-system2}\end{minipage}}
        \caption{Typing rules for the AD target language that we consider on top of the rules of Fig. \ref{fig:sl-type-system} and \ref{fig:tl-type-system1}.\label{fig:tl-type-system2}}\;
      \end{figure}

  Fig.  \ref{fig:sl-equations} and \ref{fig:tl-equations} display the equational theory we consider for the terms and types, which we call $(\alpha)\beta\eta+$-equivalence.
  To present this equational theory, we define in Fig. \ref{fig:tl-types-functorial-action}, by induction, some syntactic sugar for the functorial action $\Delta,\Delta'\mid\Gamma;\lvar:\subst{\cty}{\sfor{\tvar}{\cty[2]}}\vdash \subst{\cty}{\sfor{\ltvar}{\lvar\vdash \trm}} :\subst{\cty}{\sfor{\tvar}{\cty[3]}}$ in argument $\ltvar$ 
of parameterized types $\Delta,\tvar:\type\mid\Gamma\vdash \cty:\ltype$ on terms $\Delta'\mid\Gamma;\lvar:\cty[2]\vdash\trm:\cty[3]$.
\begin{figure}[!ht]
  \framebox{\begin{minipage}{0.98\linewidth}\hspace{-24pt} 
\input{tl-types-functorial-action}
\end{minipage}}
\caption{\label{fig:tl-types-functorial-action} Functorial action $\Delta,\Delta'\mid\Gamma;\lvar:\subst{\cty}{\sfor{\ltvar}{\cty[2]}}\vdash \subst{\cty}{\sfor{\ltvar}{\lvar\vdash \trm}} :\subst{\cty}{\sfor{\ltvar}{\cty[3]}}$ in argument $\ltvar$ 
of parameterized types $\Delta,\tvar:\type\mid \Gamma\vdash \cty:\ltype$ on terms $\Delta'\mid\Gamma;\lvar:\cty[2]\vdash\trm:\cty[3]$ of the target language.}
\end{figure}%

  \begin{figure}[!ht]
    \framebox{\begin{minipage}{0.98\linewidth}\hspace{-24pt} \input{tl-equations}
   \end{minipage}}
   \caption{Equational rules for the idealised, linear AD language, which we use on top of the 
   rules of Fig. \ref{fig:sl-equations}. 
   In addition to standard $\beta\eta$-rules for $\copower{(-)}{(-)}$- and $\multimap$-types,
   we add rules making $(\zero,+)$ into a commutative monoid on the terms of 
   each linear type as well as rules which say that terms of linear types are homomorphisms in their linear variable.
   Equations hold on pairs of terms of the same type/types of the same kind.
   As usual, we only distinguish terms up to $\alpha$-renaming of bound variables.\label{fig:tl-equations}\;
  }
   \end{figure}

This target language can be viewed as 
defining a strictly indexed category 
$\LSyn:\CSyn^{op}\to \Cat$:
\begin{itemize}
    \item $\CSyn$ extends its full subcategory $\Syn$ with the newly added cartesian types; its objects are cartesian types and $\CSyn(\ty,\ty[2])$ consists of $(\alpha)\beta\eta$-equivalence classes of target language programs $\cdot\mid \var:\ty\vdash \trm:\ty[2]$.
    \item Objects of $\LSyn(\ty)$ are linear types $\cdot\mid\pvar:\ty\vdash \cty[2]:\ltype$ up to $(\alpha)\beta\eta+$-equivalence.
    \item Morphisms in $\LSyn(\ty)(\cty[2],\cty[3])$ are terms
    $\cdot\mid \var:\ty;\lvar:\cty[2]\vdash \trm:\cty[3]$ modulo $(\alpha)\beta\eta+$-equivalence.
    \item Identities in $\LSyn(\ty)$ are represented by the terms 
    $\cdot\mid \var:\ty;\lvar:\cty[2]\vdash\lvar:\cty[2]$.
    \item Composition of $\cdot\mid \var:\ty;\lvar:\cty[2]_1\vdash \trm:\cty[2]_2$
    and $\cdot\mid \var:\ty;\lvar:\cty[2]_2\vdash \trm[2]:\cty[2]_3$ in $\LSyn(\ty)$
    is defined as $\cdot\mid \var:\ty;\lvar:\cty[2]_1\vdash \letin{\lvar}{\trm}{\trm[2]}:\cty[2]_3$.
    \item Change of base $\LSyn(\trm):\LSyn(\ty)\to\LSyn(\ty')$ along 
    $(\cdot\mid\var':\ty'\vdash \trm:\ty)\in \CSyn(\ty',\ty)$ is defined  
    $\LSyn(\trm)(\cdot\mid\var:\ty;\lvar:\cty[2]\vdash \trm[2]:\cty[3])\defeq 
    \cdot\mid \var':\ty';\lvar:\cty[2]\vdash \letin{\var}{\trm}{\trm[2]}~:~\cty[3]
    $.
    \item All type formers are interpreted as one expects based on their notation,
    using introduction and elimination rules for the required structural isomorphisms.
\end{itemize}
\begin{corollary}\label{cor:categorical-structure-target-language}
  $\Sigma_{\CSyn}\LSyn$ and $\Sigma_{\CSyn}\LSyn^{op}$ are both bicartesian closed 
  categories with $\mu\nu$-polynomials.
  \end{corollary}
In fact, $\LSyn:\CSyn^{op}\to \Cat$ is the initial $\Sigma$-bimodel of tuples, self-dual primitive types and primitive operations, function types, sum types and inductive and coinductive types,
in the sense that for any other such a $\Sigma$-bimodel $\catL:\catC^{op}\to \Cat$, we have a unique homomorphic indexed functor $(\overline{F}, f):(\CSyn,\LSyn)\to (\catC,\catL)$.

\begin{corollary}[Concrete semantics of the target language]\label{cor:universal-property-of-target-language}
	Let $\catL:\catC^{op}\to \Cat$ be a $\Sigma$-bimodel for inductive, coinductive and function types.
	Let 
	\begin{enumerate}[a)]
		\item for each $n$-dimensional array $\reals^n\in\Syn$, \label{UP-TARGET1}
		$\overline{F}\left( \reals^n\right)\in\objects\left( \catC\right) $;
		\item \label{UP-TARGET2} for each $n$-dimensional array $\reals^n\in\Syn$, 
		$$\underline{F}\left( \reals^n\right)  \in  \catL \left( \overline{F}\left( \reals^n\right) \right) ;$$
		\item \label{UP-TARGET3} for each primitive $\op\in \Op_{n_1,\ldots, n_k}^m$:
		\begin{enumerate}[i)]
			\item  $\overline{F}\left( \op \right)  :\RR^{n_1}\times\cdots\times \RR^{n_k}\To \RR^m$ is the map in $\Set$ corresponding to the operation that $\op$ intends to implement;
			\item $f_{\op}\in\catFV \left( \overline{F}\left( \reals^{n_1}\right) \times\cdots\times \overline{F}\left( \reals^{n_k}\right) \right) \left( \underline{F}\left( \reals^{n_1}\right) \times\cdots\times \underline{F}\left( \reals^{n_k}\right), \underline{F}\left( \reals^{m}\right)    \right) $ is the family of linear transformations that $D\op$ intends to implement;
			\item $f_\op ^t  \in \catFV \left( \overline{F}\left( \reals^{n_1}\right) \times\cdots\times \overline{F}\left( \reals^{n_k}\right) \right) \left(\underline{F}\left( \reals^{m}\right)  ,\underline{F}\left( \reals^{n_1}\right) \times\cdots\times \underline{F}\left( \reals^{n_k}\right)  \right)
			$ is the family of linear transformations that $\transpose{\left( D\op\right) }$ intends to implement.
		\end{enumerate}
	\end{enumerate}
be an assignment. 
We obtain canonical bicartesian closed functors that preserve $\mu\nu$-polynomials\\
\noindent
\begin{minipage}{.5\linewidth} 
	\begin{equation} \label{eq:canonically-induced-target-language-UP}
		F:  \Sigma_{\CSyn}\LSyn  \to\Sigma_{\catC}\catL 
	\end{equation} 
\end{minipage}
\begin{minipage}{.5\linewidth} 
	\begin{equation} \label{eq:canonically-induced-target-language-reverse-UP}
	\prescript{t}{}{F}:  \Sigma_{\catC }\LSyn^{op} \to\Sigma_{\CSyn}\catL ^{op}
	\end{equation} 
\end{minipage}\\ \\
that extend the assignment given by \ref{UP-TARGET1}, \ref{UP-TARGET2} and \ref{UP-TARGET3}.
\end{corollary}

\section{Novel AD algorithms as source-code transformations}
\label{sec:AD-transformations}
By Corollary \ref{cor:categorical-structure-target-language}, $\Sigma_{\CSyn}\LSyn$ and $\Sigma_{\CSyn}\LSyn^{op}$
are both bicartesian 
closed categories with $\mu\nu$-polynomials.
By the universal property of $\Syn$ established in Corollary \ref{cor:universal-property-of-source-language},
we get unique $\mu\nu$-polynomial preserving bicartesian closed functors $\Dsyn{-}:\Syn\to\Sigma_{\CSyn}\LSyn$ and $\Dsynrev{-}:\Syn\to\Sigma_{\CSyn}\LSyn^{op}$ implementing source-code transformations 
for forward and reverse AD, respectively,
once we fix a compatible definition for the code transformations on primitive types $\reals^n$ and operations $\op$.
\begin{corollary}[CHAD]\label{cor:CHAD-definition}
    Once we fix the derivatives of the ground types and primitive operations of $\Syn $ by defining
        \begin{itemize}
            \item for each $n$-dimensional array $\reals^n\in\Syn$, 
             $\Dsyn{\reals^n }\defeq \left( \reals^n , \creals^n  \right) $ and $\Dsynrev{ \reals^n }\defeq\left( \reals^n , \creals^n  \right) $ in which we think of $\creals^n $ as the associated tangent and cotangent space;
            \item for each primitive $\op\in \Op_{n_1,\ldots, n_k}^m$,  $\Dsyn{\op } \defeq\left( \op , \mathsf{D}\op \right) $ and $\Dsynrev{\op }\defeq \left( \op , \transpose{\mathsf{D}\op} \right) $, in which $\mathsf{D}\op$ and $\transpose{\mathsf{D}\op}$  are the linear operations that implement the derivative and the transposed derivative of $\op $, respectively,
        \end{itemize}
        we obtain unique functors
    \begin{equation}
        \Dsyn{-}:\Syn\to \Sigma_{\CSyn}\LSyn , \qquad \qquad
        \Dsynrev{-}:\Syn\to \Sigma_{\CSyn}\LSyn^{op}
    \end{equation}	
        that extend these definitions such that  $\Dsyn{-}$ and $\Dsynrev{-}$  strictly preserve the bicartesian closed structure and the $\mu\nu$-polynomials. 
    \end{corollary}
\noindent By definition of equality in $\Syn$, $\Sigma_{\CSyn}\LSyn$ and $\Sigma_{\CSyn}\LSyn^{op}$,
these code transformations automatically respect equational
reasoning principles, in the sense that 
$\trm\beeq\trm[2]$ implies that $\Dsyn{\trm}\bepeq\Dsyn{\trm[2]}$ and 
$\Dsynrev{\trm}\bepeq\Dsynrev{\trm[2]}$.
In this section, we detail the implied definitions of $\Dsynsymbol$ and $\Dsynrevsymbol$ as well as their properties.

\sqsubsection{Some notation}
In the rest of this section, we use the following syntactic sugar:
\begin{itemize}
\item a notation for (linear) $n$-ary tuple types: $\tProd{\cty_1}{\ldots}{\cty_n}\defeq \bProd{\bProd{\bProd{\cty_1}{\cty_2}\cdots}{\cty_{n-1}}}{\cty_n}
$;
\item a notation for $n$-ary tuples: $\tTriple{\trm_1}{\cdots}{\trm_n}\defeq \tPair{\tPair{\tPair{\trm_1}{\trm_2}\cdots}{\trm_{n-1}}}{\trm_n}$;
\item given $\Gamma;\lvar:\cty\vdash \trm:\tProd{\cty[2]_1}{\cdots}{\cty[2]_n}$, we write $\Gamma;\lvar:\cty\vdash \tProj{i}(\trm):\cty[2]_i$ for the 
obvious $i$-th projection of $\trm$, which is constructed by repeatedly applying  
$\tFst$ and $\tSnd$ to $\trm$;
\item given $\Gamma;\lvar:\cty\vdash \trm:\cty[2]_i$, we write the $i$-th coprojection $\Gamma;\lvar:\cty\vdash 
\tCoProj{i}(\trm)\defeq \tTuple{\zero,\ldots,\zero,\trm,\zero,\ldots,\zero}
:\tProd{\cty[2]_1}{\cdots}{\cty[2]_n}$;
\item for a list $\var_1,\ldots,\var_n$ of distinct identifiers, we write $\idx{\var_i}{\var_1,\ldots, \var_n}\defeq i$ for the index of the identifier $\var_i$ in this list;
\item a $\mathbf{let}$-binding for tuples: $\pletin{\var}{\var[2]}{\trm}{\trm[2]}\defeq 
\letin{\var[3]}{\trm}{\letin{\var}{\tFst\var[3]}{\letin{\var[2]}{\tSnd\var[3]}{\trm[2]}}},$ where $\var[3]$ is a fresh variable.
\end{itemize}
Furthermore, all variables used in the source code transforms below are assumed to be freshly chosen.

\sqsubsection{Kinding and typing of the code transformations}
We define for each type $\ty$ of the source language:
\begin{itemize}
    \item a cartesian type $\Dsyn{\ty}_1$ of forward mode primals;
    \item a linear type $\Dsyn{\ty}_2$ (with free term variable $\pvar$) of forward mode tangents;
    \item a cartesian type $\Dsynrev{\ty}_1$ of reverse mode primals;
    \item a linear type $\Dsynrev{\ty}_2$ (with free term variable $\pvar$) of reverse mode cotangents.
\end{itemize}
We extend $\Dsyn{-}$ and $\Dsynrev{-}$ to act on typing contexts
$\Gamma=\var_1:\ty_1,\ldots,\var_n:\ty_n$ as
\[\begin{array}{lll}
    \Dsyn{\Gamma}_1\defeq \var_1:\Dsyn{\ty_1}_1,\ldots,\var_n:\Dsyn{\ty_n}_n\qquad
    \;&\text{(a cartesian typing context)}\\
    \Dsyn{\Gamma}_2\defeq \tProd{\subst{\Dsyn{\ty_1}_2}{\sfor{\pvar}{\var_1}}}{\cdots}{\subst{\Dsyn{\ty_n}_2}{\sfor{\pvar}{\var_n}}}\qquad\;&\text{(a linear type)}\\
    \Dsynrev{\Gamma}_1\defeq \var_1:\Dsynrev{\ty_1}_1,\ldots,\var_n:\Dsynrev{\ty_n}_n\qquad&\text{(a cartesian typing context)}\\
    \Dsynrev{\Gamma}_2\defeq \tProd{\subst{\Dsynrev{\ty_1}_2}{\sfor{\pvar}{\var_1}}}{\cdots}{\subst{\Dsynrev{\ty_n}_2}{\sfor{\pvar}{\var_n}}}\qquad&\text{(a linear type)}.\end{array}
\]
Our code transformations are well-kinded in the sense that they 
translate a type $\Delta\vdash\ty:\type$ of the source language into pairs of types of the target language
\[
\begin{array}{lll}
\Delta\mid\cdot \vdash\Dsyn{\ty}_1:\type\\
\Delta\mid\pvar:\Dsyn{\ty}_1\vdash \Dsyn{\ty}_2:\ltype\\
\Delta\mid\cdot\vdash\Dsynrev{\ty}_1:\type\\
\Delta\mid \pvar:\Dsynrev{\ty}_1\vdash \Dsynrev{\ty}_2:\ltype.
\end{array}
\]
Similarly, the functors $\Dsyn{-}:\Syn\to \Sigma_\CSyn\LSyn$ and 
$\Dsynrev{-}:\Syn\to \Sigma_\CSyn \LSyn^{op}$ define for each term 
$\trm$ of the source language and a list $\vGamma$ of identifiers 
that contains at least the free identifiers of $\trm$:
\begin{itemize}
\item a term $\Dsyn[\vGamma]{\trm}_1$ that represents the forward mode primal computation associated with $\trm$; 
\item a term $\Dsyn[\vGamma]{\trm}_2$ that represents the forward mode tangent computation associated with $\trm$;
\item a term $\Dsynrev[\vGamma]{\trm}_1$ that represents the reverse mode primal computation associated with $\trm$;
\item a term $\Dsynrev[\vGamma]{\trm}_2$ that represents the reverse mode cotangent computation associated with $\trm$.
\end{itemize}
These code transformations are well-typed in the sense that a source language term $\trm$ that is typed according to $\Delta\mid \Gamma\vdash \trm:\ty$ is translated into pairs of terms 
of the target language that are typed as follows:
\[
\begin{array}{l}
    \Delta\mid\Dsyn{\Gamma}_1\vdash \Dsyn[\vGamma]{\trm}_1:\Dsyn{\ty}_1\\
    \Delta\mid\Dsyn{\Gamma}_1;\lvar:\Dsyn{\Gamma}_2\vdash \Dsyn[\vGamma]{\trm}_2:\subst{\Dsyn{\ty}_2}{\sfor{\pvar}{\Dsyn[\vGamma]{\trm}_1}}\\
    \Delta\mid\Dsynrev{\Gamma}_1\vdash \Dsynrev[\vGamma]{\trm}_1:\Dsynrev{\ty}_1\\
\Delta\mid\Dsynrev{\Gamma}_1;\lvar:\subst{\Dsynrev{\ty}_2}{\sfor{\pvar}{\Dsynrev[\vGamma]{\trm}_1}}\vdash \Dsynrev[\vGamma]{\trm}_2:\Dsynrev{\Gamma}_2,
\end{array}\]
where $\vGamma$ is the list of identifiers that occurs in $\Gamma$ (that is, $\overline{\var_1:\ty_1,\ldots,\var_n:\ty_n}\defeq \var_1,\ldots,\var_n$).

However, as we noted already in Insight 1 of \S\ref{sec:key-ideas}, 
we often want to share computation between the primal and (co)tangent values,
for reasons of efficiency.
Therefore, we focus instead on transforming a source language term $\Delta\mid \Gamma\vdash \trm:\ty$ into target language terms:
\[
\begin{array}{l}
    \Delta\mid\Dsyn{\Gamma}_1\vdash \Dsyn[\vGamma]{\trm}:\Sigma{\pvar:\Dsyn{\ty}_1}.{\Dsyn{\Gamma}_2\multimap \Dsyn{\ty}_2}\\
    \Delta\mid\Dsynrev{\Gamma}_1\vdash \Dsynrev[\vGamma]{\trm}:\Sigma{\pvar:\Dsynrev{\ty}_1}.{\Dsynrev{\ty}_2\multimap \Dsynrev{\Gamma}_2},
\end{array}\]
where $\Dsyn[\vGamma]{\trm}\bepeq \tPair{\Dsyn[\vGamma]{\trm}_1}{\lfun{\lvar}\Dsyn[\vGamma]{\trm}_2}$ and $\Dsynrev[\vGamma]{\trm}\bepeq \tPair{\Dsynrev[\vGamma]{\trm}_1}{\lfun{\lvar}\Dsynrev[\vGamma]{\trm}_2}$.
While both representations of AD on programs are equivalent in terms of the 
$\beta\eta+$-equational theory of the target language
and therefore for any semantic and correctness purposes, they are meaningfully 
different in terms of efficiency.
Indeed, we ensure that common subcomputations between the primals and (co)tangents are shared via let-bindings in $\Dsyn[\vGamma]{\trm}$ and $\Dsynrev[\vGamma]{\trm}$.

\sqsubsection{Code transformations of primitive types and operations}
We have suitable terms (linear operations) 
\begin{align*}
&\var_1:\reals^{n_1}, \cdots,\var_k: \reals^{n_k}\;\;;\;\; \lvar:\creals^{n_1}\t* \cdots\t* \creals^{n_k} &\hspace{-5pt}&\vdash D\op(\var_1,\ldots,\var_k;\lvar)\hspace{-5pt} &:\; &\creals^m\\
&\var_1:\reals^{n_1}, \cdots,\var_k: \reals^{n_k}\;\;;\;\; \lvar: \creals^m &\hspace{-5pt}&\vdash \transpose{D\op}(\var_1,\ldots,\var_k;\lvar)\hspace{-5pt}& :\;&\creals^{n_1}\t* \cdots\t* \creals^{n_k}
\end{align*} to represent the forward- and 
reverse-mode derivatives of the primitive operations $\op\in\Op_{n_1,...,n_k}^m$.
Using these, we define
\begin{flalign*}
    &\Dsyn{\reals^n}_1 \defeq \reals^n
    && {\Dsyn{\reals^n}_2} \defeq \creals^n\\
    \end{flalign*}
\begin{flalign*}
        &\Dsyn[\vGamma]{\op(\trm_1,\ldots,\trm_k)} \defeq  \ndots\ndots&& \pletin{\var_1}{\var_1'}{\Dsyn[\vGamma]{\trm_1}}{\cdots\pletin{\var_k}{\var_k'}{\Dsyn[\vGamma]{\trm_k}}{\\
        &&&\tPair{\op(\var_1,\ldots,\var_k)}{\lfun\lvar D\op(\var_1,\ldots,\var_n;\tTriple{\lapp{\var_1'}{\lvar}}{\ldots}{\lapp{\var_k'}{\lvar}})}}}
\end{flalign*}
        \begin{flalign*}
            &\Dsynrev{\reals^n}_1 \defeq \reals^n
            && \Dsynrev{\reals^n}_2 \defeq \creals^n\\
        \end{flalign*}
  \begin{flalign*}
            &\Dsynrev[\vGamma]{\op(\trm_1,\ldots,\trm_k)} \defeq  \ndots\ndots\cdot\cdot\cdot&& \pletin{\var_1}{\var_1'}{\Dsynrev[\vGamma]{\trm_1}}{\cdots\\
            &&& \pletin{\var_k}{\var_k'}{\Dsynrev[\vGamma]{\trm_k}}{\\
            &&&\tPair{\op(\var_1,\ldots,\var_k)}{\lfun\lvar \letin{\lvar}{\transpose{D\op}(\var_1,\ldots,\var_k;\lvar)}{\\
            &&&\phantom{\tPair{\op(\var_1,\ldots,\var_k)}{\lfun\lvar }}\lapp{\var_1'}{\tProj{1}{\lvar}}+\cdots+\lapp{\var_k'}{\tProj{k}{\lvar}}}}}}
 \end{flalign*}
 
    For the AD transformations to be correct, it is important that these derivatives of language
    primitives are implemented correctly in the sense that
    $$
\sem{\var_1,\ldots,\var_k;\var[2]\vdash D\op(\var_1,\ldots,\var_k;\lvar)}=D\sem{\op}\qquad \sem{\var_1,\ldots,\var_k;\lvar\vdash \transpose{D\op}(\var_1,\ldots,\var_k;\lvar)}=\transpose{D\sem{\op}}.
    $$
    For example, for elementwise multiplication $(*)\in\Op_{n,n}^n$, we need that
\begin{eqnarray*}
	\sem{D(*)(\var_1,\var_2;\lvar)}((a_1, a_2), (b_1, b_2))&=&a_1 * b_2 + a_2 * b_1;\\
\sem{\transpose{D(*)}(\var_1,\var_2;\lvar)}((a_1, a_2),b)&=&(a_2 * b, a_1 * b).
\end{eqnarray*}
    By Corollary \ref{cor:universal-property-of-source-language}, the extension of the AD transformations $\Dsynsymbol$ and $\Dsynrevsymbol$ 
    to the full source language are now  canonically determined, as the unique 
    $\mu\nu$-polynimials preserving bicartesian closed functors that extend the previous definitions.

    \sqsubsection{Forward-mode CHAD definitions}
We define the types of (forward-mode) primals $\Dsyn{\ty}_1$ and tangents $\Dsyn{\ty}_2$ associated with a type $\ty$ as follows:
\begin{align*}
&\Dsyn{\Unit}_1 \defeq \Unit 
\\
&\Dsyn{\ty\t*\ty[2]}_1 \defeq \Dsyn{\ty}_1\t*\Dsyn{\ty[2]}_1
\\
&\Dsyn{\ty\To\ty[2]}_1 \defeq \Pi{\pvar:\Dsyn{\ty}_1}. 
\Sigma{\pvar':\Dsyn{\ty[2]}_1}.\Dsyn{\ty}_2\multimap 
\subst{\Dsyn{\ty[2]}_2}{\sfor{\pvar}{\pvar'}}
\\
&\Dsyn{\set{\ell_1\ty_1\mid\cdots\mid \ell_n\ty_n}}_1 \defeq \set{\ell_1\Dsyn{\ty_1}_1\mid \cdots \mid\ell_n\Dsyn{\ty_n}_1}
\\
&\Dsyn{\tvar}_1 \defeq \tvar \\
&\Dsyn{\lfp{\tvar}{\ty}}_1\defeq \lfp{\tvar}{\Dsyn{\ty}_1}
\\
&\Dsyn{\gfp{\tvar}{\ty}}_1\defeq \gfp{\tvar}{\Dsyn{\ty}_1}
\\
\\
& {\Dsyn{\Unit }_2} \defeq \lUnit \\
& \Dsyn{\ty\t*\ty[2]}_2 \defeq  \subst{\Dsyn{\ty}_2}{\sfor{\pvar}{\tFst\,\pvar}}\t*\subst{\Dsyn{\ty[2]}_2}{\sfor{\pvar}{\tSnd\,\pvar}}\\
&\Dsyn{\ty\To\ty[2]}_2 \defeq \Pi{\pvar':\Dsyn{\ty}_1}.\subst{\Dsyn{\ty[2]}_2}{\sfor{\pvar}{\tFst\, (\pvar\,\pvar')}}\\
&\Dsyn{\set{\ell_1\ty_1\mid\cdots\mid \ell_n\ty_n}}_2\defeq 
\vMatch{\pvar}{\ell_1\pvar\To \Dsyn{\ty_1}_2\mid\cdots\mid 
\ell_n\pvar\To\Dsyn{\ty_n}_2}\\
& \Dsyn{\tvar}_2\defeq \ltvar\\
&\Dsyn{\lfp{\tvar}{\ty}}_2\defeq \llfp{\ltvar}{\subst{\Dsyn{\ty}_2}{\sfor{\pvar}{\tFold{\pvar}{\var[2]}{\subst{\Dsyn{\ty}_1}{\sfor{\tvar}{\var[2]\vdash \tRoll\var[2]}}}}}}\\
&\Dsyn{\gfp{\tvar}{\ty}}_2\defeq \lgfp{\ltvar}{\subst{\Dsyn{\ty}_2}{\sfor{\pvar}{\tUnroll\pvar}}}\\
\end{align*}
For programs $\trm$, we define we define their efficient CHAD transformation $\Dsyn[\vGamma]{\trm}$ 
 as follows (and we list the less efficient transformations $\Dsyn[\vGamma]{\trm}_1$ and $\Dsyn[\vGamma]{\trm}_2$ that do not share computations between the primals and tangents in Appendix~\ref{sec:inefficient-ad-transformation}):
    \begin{flalign*}
    &\Dsyn[\vGamma]{\var} \defeq  \nndots\nndots\nndots&& \tPair{\var}{\lfun{\lvar}\tProj{\idx{\var}{\vGamma}}(\lvar)}
    \end{flalign*}
    \begin{flalign*}
    &\Dsyn[\vGamma]{\letin{\var}{\trm}{\trm[2]}} \defeq\nndots\nndots\ndots\cdot\cdot &&\pletin{\var}{\var'}{\Dsyn[\vGamma]{\trm}}{ \\
    &&&\pletin{\var[2]}{\var[2]'}{\Dsyn[\vGamma,\var]{\trm[2]}}{\\
    &&&\tPair{\var[2]}{\lfun\lvar
    \lapp{\var[2]'}{\tPair{\lvar}{\lapp{\var'}{\lvar}}}}}}
    \end{flalign*}
    \begin{flalign*}
    &\Dsyn[\vGamma]{\tUnit}  \defeq\nndots\nndots\nndots\ndots\cdot\cdot\cdot\cdot &&\tPair{\tUnit}{\lfun\lvar\tUnit}
    \end{flalign*}
    \begin{flalign*}
    &\Dsyn[\vGamma]{\tPair{\trm}{\trm[2]}} \defeq \nndots\nndots\ndots\ndots&&
    \pletin{\var}{\var'}{\Dsyn[\vGamma]{\trm}}{ \\
    &&&\pletin{\var[2]}{\var[2]'}{\Dsyn[\vGamma]{\trm[2]}}{\\
    &&&\tPair{\tPair{\var}{\var[2]}}{\lfun\lvar \tPair{\lapp{\var'}\lvar}{\lapp{\var[2]'}\lvar}}}}
    \end{flalign*}
    \begin{flalign*}&
    \Dsyn[\vGamma]{\tFst\trm} \defeq \nndots\ndots\ndots\cdot\cdot\cdot\cdot\cdot&&
    \pletin{\var}{\var'}{\Dsyn[\vGamma]{\trm}}
    {\tPair{\tFst\var}{\lfun\lvar \tFst(\lapp{\var'}{\lvar})}}
    \end{flalign*}
    \begin{flalign*}&
    \Dsyn[\vGamma]{\tSnd\trm} \defeq  \nndots\ndots\ndots\cdot\cdot&&
    \pletin{\var}{\var'}{\Dsyn[\vGamma]{\trm}}
    {\tPair{\tSnd\var}{\lfun\lvar \tSnd(\lapp{\var'}{\lvar})}}
    \end{flalign*}
    \begin{flalign*}
        &
    \Dsyn[\vGamma]{\fun\var\trm}\defeq\ndots\ndots\cdot\cdot &&
        \letin{\var[2]}{\fun\var\Dsyn[\vGamma,\var]{\trm}}{\\
            &&&\tPair{\fun\var
            \pletin{\var[3]}{\var[3]'}{\var[2]\,\var}
            \tPair{\var[3]}{\lfun\lvar\lapp{\var[3]'}{\tPair{\zero}{\lvar}}}}{\lfun\lvar\fun\var\lapp{(\tSnd(\var[2]\,\var))}{\tPair{\lvar}{\zero}}}
        }
    \end{flalign*}
    \begin{flalign*}&
    \Dsyn[\vGamma]{\trm\,\trm[2]} \defeq\nndots &&
    \pletin{\var}{\var'_{\text{ctx}}}{\Dsyn[\vGamma]{\trm}}{
    \pletin{\var[2]}{\var[2]'}{\Dsyn[\vGamma]{\trm[2]}}{
    \pletin{\var[3]}{\var'_{\text{arg}}}{\var\,\var[2]}{}}}
    \\&
     && \tPair{\var[3]}{\lfun\lvar (\lapp{\var'_{\text{ctx}}}{\lvar})\,\var[2]+
    \lapp{\var'_{\text{arg}}}{(\lapp{\var[2]'}{\lvar})}}
    \end{flalign*}
    \begin{flalign*}&
    \Dsyn[\vGamma]{\Cns\trm} \defeq \nndots\nndots\ndots\ndots&&
    \pletin{\var}{\var'}{\Dsyn[\vGamma]{\trm}}{\tPair{\Cns\var}{\var'}}
    \end{flalign*}
    \begin{flalign*} &
        \Dsyn[\vGamma]{\vMatch{\trm}{\Cns_1\var_1\To\trm[2]_1\mid\cdots\mid \Cns_n \var_n\To\trm[2]_n}}\defeq \ndots\!\!&&
        \pletin{\var[2]}{\var[2]'}{\Dsyn[\vGamma]{\trm}}{\\ &&&
            \vMatch{\var[2]}{\Cns_1\var_1\To
            \\ &&& \quad \pletin{\var[3]_1}{\var[3]_1'}{\Dsyn[\vGamma,\var_1]{\trm[2]_1}}{\\ &&&
            \quad \tPair{\var[3]_1}{\lfun\lvar 
            \lapp{\var[3]_1'}{\tPair{\lvar}{\lapp{(\letin{\var[2]}{\ell_1\var_1}{\var[2]'})}{\lvar}}}
            }}\\ &&&
            \qquad\qquad\mid\cdots\mid\\ &&&
            \qquad\qquad\Cns_n \var_n\To 
            \\ &&& \quad \pletin{\var[3]_n}{\var[3]_n'}{\Dsyn[\vGamma,\var_n]{\trm[2]_n}}{\\ &&&
            \quad \tPair{\var[3]_n}{\lfun\lvar 
            \lapp{\var[3]_n'}{\tPair{\lvar}{\lapp{(\letin{\var[2]}{\ell_n\var_n}{\var[2]'})}{\lvar}}}
            }}}}
        \end{flalign*}
    \begin{flalign*}
        &
    \Dsyn[\vGamma]{\tRoll\trm} \defeq \nndots\ndots\ndots\cdot\cdot&&\pletin{\var}{\var'}{\Dsyn[\vGamma]{\trm}}{
    \tPair{\tRoll\var}{\lfun\lvar \tRoll (\lapp{\var'}{\lvar})}
    }
    \end{flalign*}
    \begin{flalign*}
        &
    \Dsyn[\vGamma]{\tFold{\trm}{\var}{\trm[2]}}\defeq\ndots &&
    \pletin{\var[2]}{\var[2]'}{\Dsyn[\vGamma]{\trm}}{
    \\ &&&
    \letin{\var[3]}{\fun\var\Dsyn[\var]{\trm[2]}}{\\ &&&
    \tPair{\tFold{\var[2]}{\var}{\tFst(\var[3]\,\var)}}{\\ &&&
    \lfun\lvar\tFold{\lapp{\var[2]'}{\lvar}}{\lvar}{\\ &&& \quad\letin{\var}{\tFold{\var[2]}{\var}{\subst{\Dsyn{\ty}_1}{\sfor{\tvar}{\var\vdash \tFst(\var[3]\,\var)}}}}{\lapp{(\tSnd(\var[3]\,\var))}{\lvar}}}
    }
    }}
    \end{flalign*}
    \begin{flalign*}
    &
    \Dsyn[\vGamma]{\tUnroll\trm} \defeq\nndots\cdot\cdot\cdot\cdot\cdot&&& \pletin{\var}{\var'}{\Dsyn[\vGamma]{\trm}}{
    \tPair{\tUnroll\var}{\lfun\lvar \tUnroll (\lapp{\var'}{\lvar})}
    }
    \end{flalign*}
    \begin{flalign*}&
    \Dsyn[\vGamma]{\tGen{\trm}{\var}{\trm[2]}} \defeq \nndots&&
    \pletin{\var[2]}{\var[2]'}{\Dsyn[\vGamma]{\trm}}{\\ &&&
    \letin{\var[3]}{\fun\var\Dsyn[\var]{\trm[2]}}{\\ &&&
    \tPair{\tGen{\var[2]}{\var}{\tFst(\var[3]\,\var)}}{\\ &&&
    \lfun\lvar\tGen{\lapp{\var[2]'}{\lvar}}{\lvar}{\lapp{(\tSnd(\var[3]\,\var[2]))}{\lvar}}}
    }
    }
    \end{flalign*}

\sqsubsection{Reverse-mode CHAD definitions}
We define the types of (reverse-mode) primals $\Dsynrev{\ty}_1$ and cotangents $\Dsynrev{\ty}_2$ associated with a type $\ty$ as follows:
\begin{align*}
&\Dsynrev{\Unit}_1 \defeq \Unit 
\\
&\Dsynrev{\ty\t*\ty[2]}_1 \defeq \Dsynrev{\ty}_1\t*\Dsynrev{\ty[2]}_1
\\
&\Dsynrev{\ty\To\ty[2]}_1 \defeq \Pi{\pvar:\Dsynrev{\ty}_1}. 
\Sigma{\pvar':\Dsynrev{\ty[2]}_1}.
\subst{\Dsynrev{\ty[2]}_2}{\sfor{\pvar}{\pvar'}}\multimap 
\Dsynrev{\ty}_2
\\
&\Dsynrev{\set{\ell_1\ty_1\mid\cdots\mid \ell_n\ty_n}}_1 \defeq \set{\ell_1\Dsynrev{\ty_1}_1\mid \cdots \mid\ell_n\Dsynrev{\ty_n}_1}
\\
&\Dsynrev{\tvar}_1 \defeq \tvar \\
&\Dsynrev{\lfp{\tvar}{\ty}}_1\defeq \lfp{\tvar}{\Dsynrev{\ty}_1}
\\
&\Dsynrev{\gfp{\tvar}{\ty}}_1\defeq \gfp{\tvar}{\Dsynrev{\ty}_1}
\\
&\\
& \Dsynrev{\Unit }_2 \defeq \lUnit \\
& {\Dsynrev{\ty\t*\ty[2]}_2} \defeq  \subst{\Dsynrev{\ty}_2}{\sfor{\pvar}{\tFst\,\pvar}}\t*\subst{\Dsynrev{\ty[2]}_2}{\sfor{\pvar}{\tSnd\,\pvar}}\\
&\Dsynrev{\ty\To\ty[2]}_2 \defeq  \Sigma{\pvar':\Dsynrev{\ty}_1}.\subst{\Dsynrev{\ty[2]}_2}{\sfor{\pvar}{\tFst\, (\pvar\,\pvar')}}\\
&\Dsynrev{\set{\ell_1\ty_1\mid\cdots\mid \ell_n\ty_n}}_2\defeq 
\vMatch{\pvar}{\ell_1\pvar\To \Dsynrev{\ty_1}_2\mid\cdots\mid 
\ell_n\pvar\To\Dsynrev{\ty_n}_2}\\
& \Dsynrev{\tvar}_2\defeq \ltvar\\
&\Dsynrev{\lfp{\tvar}{\ty}}_2\defeq \lgfp{\ltvar}{\subst{\Dsynrev{\ty}_2}{\sfor{\pvar}{\tFold{\pvar}{\var[2]}{\subst{\Dsynrev{\ty}_1}{\sfor{\tvar}{\var[2]\vdash \tRoll\var[2]}}}}}}\\
&\Dsynrev{\gfp{\tvar}{\ty}}_2\defeq \llfp{\ltvar}{\subst{\Dsynrev{\ty}_2}{\sfor{\pvar}{\tUnroll\pvar}}}\\
\end{align*}
For programs $\trm$, we define their efficient CHAD transformation $\Dsynrev[\vGamma]{\trm}$ 
as follows (and we list the less efficient transformations $\Dsynrev[\vGamma]{\trm}_1$ and $\Dsynrev[\vGamma]{\trm}_2$ that do not share computation between the primals and cotangents in Appendix \ref{sec:inefficient-ad-transformation}):
\begin{flalign*}
&\Dsynrev[\vGamma]{\var} \defeq\nndots\nndots\ndots\ndots\cdot\cdot\cdot &&  \tPair{\var}{\lfun{\lvar} \tCoProj{\idx{\var}{\vGamma}}(\lvar)}
\end{flalign*}
\begin{flalign*}
    &
\Dsynrev[\vGamma]{\letin{\var}{\trm}{\trm[2]}}  
\defeq \nndots\ndots\cdot\cdot\cdot&&
\pletin{\var}{\var'}{\Dsynrev[\vGamma]{\trm}}{\\ &&&
    \pletin{\var[2]}{\var[2]'}{\Dsynrev[\vGamma,\var]{\trm[2]}}{\\ &&&
        \tPair{\var[2]}{\lfun\lvar 
        \letin{\lvar}{\lapp{\var[2]'}{\lvar}}{
            \tFst\lvar+\lapp{\var'}{(\tSnd \lvar)}
        }}
    }}
\end{flalign*}
\begin{flalign*}& 
\Dsynrev[\vGamma]{\tUnit}  \defeq \nndots\nndots\nndots\ndots\cdot\cdot\cdot\cdot && \tPair{\tUnit}{\lfun\lvar\zero}\end{flalign*}
\begin{flalign*}&
\Dsynrev[\vGamma]{\tPair{\trm}{\trm[2]}} \defeq \nndots\nndots\cdot\cdot\cdot&& 
\pletin{\var}{\var'}{\Dsynrev[\vGamma]{\trm}}{ \\ &&&
\pletin{\var[2]}{\var[2]'}{\Dsynrev[\vGamma]{\trm[2]}}{\\ &&&
\tPair{\tPair{\var}{\var[2]}}{\lfun\lvar \lapp{\var'}{(\tFst\lvar)}} + {\lapp{\var[2]'}{(\tSnd \lvar)}}}}
\end{flalign*}
\begin{flalign*}&
\Dsynrev[\vGamma]{\tFst\trm} \defeq\nndots\ndots\ndots\cdot\cdot\cdot\cdot\cdot && 
\pletin{\var}{\var'}{\Dsynrev[\vGamma]{\trm}}
{\tPair{\tFst\var}{\lfun\lvar \lapp{\var'}{\tPair{\lvar}{\zero}}}}
\end{flalign*}
\begin{flalign*}&
\Dsynrev[\vGamma]{\tSnd\trm} \defeq \nndots\ndots\ndots\cdot\cdot && 
\pletin{\var}{\var'}{\Dsynrev[\vGamma]{\trm}}
{\tPair{\tSnd\var}{\lfun\lvar \lapp{\var'}{\tPair{\zero}{\lvar}}}}
\end{flalign*}
\begin{flalign*}&   
\Dsynrev[\vGamma]{\fun\var\trm} \defeq\nndots\ndots\ndots\cdot\cdot\cdot && 
\letin{\var[2]}{\fun\var\Dsynrev[\vGamma,\var]{\trm}}{\\ &&&
\tPair{\fun\var
\pletin{\var[3]}{\var[3]'}{\var[2]\,\var}{
\tPair{\var[3]}{\lfun\lvar\tSnd(\lapp{\var[3]'}{\lvar})}}}
{\\ &&&\lfun\lvar  \tensMatch{\lvar}{\var}{\lvar}{
    \tFst(\lapp{(\tSnd(\var[2]\,\var))}{\lvar})} }
}
\end{flalign*}
\begin{flalign*}&
\Dsynrev[\vGamma]{\trm\,\trm[2]} \defeq\ndots\ndots\cdot\cdot\cdot\cdot\cdot&& 
\pletin{\var}{\var'_{\text{ctx}}}{\Dsynrev[\vGamma]{\trm}}{
\pletin{\var[2]}{\var[2]'}{\Dsynrev[\vGamma]{\trm[2]}}{
\pletin{\var[3]}{\var'_{\text{arg}}}{\var\,\var[2]}{}}}
\\
&&& \tPair{\var[3]}{\lfun\lvar  \lapp{\var'_{\text{ctx}}}{(!\var[2]\otimes \lvar)} + \lapp{\var[2]'}{(\lapp{\var'_{\text{arg}}}{\lvar})}}
\end{flalign*}
\begin{flalign*}&
\Dsynrev[\vGamma]{\Cns\trm} \defeq \nndots\nndots\ndots\ndots&& 
\pletin{\var}{\var'}{\Dsynrev[\vGamma]{\trm}}{\tPair{\Cns\var}{\var'}}\end{flalign*}
\begin{flalign*}&
    \Dsynrev[\vGamma]{\vMatch{\trm}{\ell_1\var_1\To\trm[2]_1\mid\cdots\mid \ell_n\var_n\To\trm[2]_n}} \defeq \ndots\cdot\cdot\cdot\!\!&&
        \pletin{\var[2]}{\var[2]'}{\Dsynrev[\vGamma]{\trm}}{\\ &&&
        \vMatch{\var[2]}{\Cns_1\var_1\To
        \\ &&&\quad \pletin{\var[3]_1}{\var[3]_1'}{\Dsynrev[\vGamma,\var_1]{\trm[2]_1}}{\\ &&&
        \quad \tPair{\var[3]_1}{\lfun\lvar 
        \letin{\lvar}{
            \lapp{\var[3]_1'}{\lvar}
        }{\tFst\lvar+\\ &&&
        \qquad\qquad\lapp{(\letin{\var[2]}{\ell_1\var_1}{\var[2]'})}{(\tSnd\lvar)}}}}\\ &&&
        \qquad\qquad\mid\cdots\mid\\ &&&
        \qquad\qquad\Cns_n \var_n\To 
        \\ &&&\quad \pletin{\var[3]_n}{\var[3]_n'}{\Dsynrev[\vGamma,\var_n]{\trm[2]_n}}{\\ &&&
        \quad \tPair{\var[3]_n}{\lfun\lvar 
        \letin{\lvar}{
            \lapp{\var[3]_n'}{\lvar}
        }{\tFst\lvar+\\ &&&
        \qquad\qquad\lapp{(\letin{\var[2]}{\ell_n\var_n}{\var[2]'})}{(\tSnd\lvar)}}}}}}
    \end{flalign*}
\begin{flalign*}
&
\Dsynrev[\vGamma]{\tRoll\trm} \defeq \nndots\ndots\ndots&& \pletin{\var}{\var'}{\Dsynrev[\vGamma]{\trm}}{
\tPair{\tRoll\var}{\lfun\lvar \lapp{\var'}{(\tUnroll\lvar)}}
}\end{flalign*}
\begin{flalign*}&
\Dsynrev[\vGamma]{\tFold{\trm}{\var}{\trm[2]}}\defeq  \ndots&& 
\pletin{\var[2]}{\var[2]'}{\Dsynrev[\vGamma]{\trm}}{\\ &&&
\letin{\var[3]}{\fun\var\Dsynrev[\var]{\trm[2]}}{\\ &&&
\tPair{\tFold{\var[2]}{\var}{\tFst (\var[3]\,\var)}\\ &&&}{
\lfun\lvar \lapp{\var[2]'}{\tGen{\lvar}{\lvar}{\\ &&&
\quad\letin{\var}{\tFold{\var[2]}{\var}{\subst{\Dsyn{\ty}_1}{\sfor{\tvar}{\var\vdash \tFst(\var[3]\,\var)}}}}{\lapp{(\tSnd(\var[3]\,\var))}{\lvar}}}
}
}
}
}\end{flalign*}
\begin{flalign*}&
\Dsynrev[\vGamma]{\tUnroll\trm} \defeq \nndots\cdot\cdot\cdot\cdot\cdot\cdot\cdot&& \pletin{\var}{\var'}{\Dsynrev[\vGamma]{\trm}}{
\tPair{\tUnroll\var}{\lfun\lvar\lapp{\var'}{(\tRoll\lvar)}}
}\end{flalign*}
\begin{flalign*}&
\Dsynrev[\vGamma]{\tGen{\trm}{\var}{\trm[2]}}\defeq\nndots\cdot\cdot &&
\pletin{\var[2]}{\var[2]'}{\Dsynrev[\vGamma]{\trm}}{\\ &&&
\letin{\var[3]}{\fun\var\Dsynrev[\var]{\trm[2]}}{\\ &&&
\tPair{\tGen{\var[2]}{\var}{\tFst(\var[3]\,\var)}\\ &&&}{
\lfun\lvar\lapp{\var[2]'}{\tFold{\lvar}{\lvar}{
    \lapp{(\tSnd(\var[3]\,\var[2]))}{\lvar}}
}
}
}}
\end{flalign*}

\section{Concrete models}\label{sec:concrete-models}
In order to proceed with our correctness proof of Automatic Differentiation, we need to establish the semantics of the program transformation in our setting. 
In this section, we construct  denotational semantics for the target language.

\subsection{Locally presentable categories and $\mu\nu$-polynomials}\label{subsec:locallypresentable-Set}
We show that any cartesian closed locally presentable category yields a concrete model for the source language. 
The only step needed to establish this fact is to prove that locally presentable categories have $\mu\nu$-polynomials, \textit{cf.}~\citep[Theorem~3.7]{ITA_2002__36_2_195_0}. 
 We establish this result below. We refer the reader to \citep{adamek1994locally,BirdThesis} for basics on 
locally presentable categories.

The first fact to recall is that locally presentable categories are complete (and cocomplete by definition): see, for instance, \citep[pag.~45]{adamek1994locally}. Moreover:

\begin{lemma}\label{lem:right-adjoint-accessible-preserves}
	Let $\catA , \catB $ be locally presentable categories.
\begin{enumerate}[A)] 	
	\item A functor $G: \catA\to\catB $ has a left adjoint  if and only if $G$ is accessible and preserves limits.\label{lem:locally-presentable-adjoint-A}
	\item A functor $F: \catB \to \cat A $ has a right adjoint if and only if $F$ preserves colimits.\label{lem:locally-presentable-adjoint-B}
\end{enumerate} 	
\end{lemma}
\begin{proof}
\ref{lem:locally-presentable-adjoint-A} is \cite[Theorem~1.66]{adamek1994locally}.

Recall that every locally presentable 
is co-wellpowered, see \citep[Theorem~1.58]{adamek1994locally}. By the special adjoint functor theorem~\citep[pag.~129]{zbMATH03367095}, we get that \ref{lem:locally-presentable-adjoint-B} holds. 
\end{proof}

\begin{lemma}\label{lem:accessible-functors-have-initial-algebras-terminal-coalgebras}
Every accessible endofunctor on a locally presentable category has an initial algebra and a terminal coalgebra.
\end{lemma} 
\begin{proof}

Every accessible endofunctor on a locally presentable category has
an initial algebra since we construct the initial algebra via the colimit of the chain $\initialobject\to E\left( \initialobject\right) \to \cdots $, see \citep{adamek1979least}.

If $\catA$ is a locally presentable category, given an endofunctor $E : \catA \to \catA $,
we have that $E\CCoAlg $ is locally presentable. Since the forgetful functor 
$ E\CCoAlg\to  \catA $ is a functor between locally presentable categories that creates colimits, 
we have that it has a right adjoint $R$. Therefore $R(\terminal ) $ is the terminal object of 
$E\CCoAlg $ (terminal coalgebra of $E$), see \citep{MR1228862}. 
\end{proof}

\begin{proposition}\label{prop:locally-presentable-has-munupolynomials}
	If $\catD $ is locally presentable then $\catD $ has $\mu\nu $-polynomials.
\end{proposition}
\begin{proof}
	The terminal category $\terminal $ is a locally presentable category and, if 
	$\catD ' $ and $\catD ''$ are locally presentable categories, then $\catD '\times\catD ''$ is locally presentable as well. Therefore all the objects of $\mnPoly _ \catD $ are locally presentable.
	
	Given locally presentable categories $\catD ', \catD ''$,  the projections $\pi _1 : \catD '\times\catD ''\to \catD '$ and
	$\pi _2 : \catD '\times\catD ''\to \catD ''$ have right (and left) adjoints and, therefore, are accessible.
	
	Moreover, given locally presentable categories $ \catD ', \catD '' ,  \catD '''$, if $E: \catD '   \to \catD ''    $ and $J : \catD '  \to \catD '''    $ are accessible functors, then so is the induced functor $(E,J) :\catD '   \to \catD '' \times \catD '''   $.
	
	Furthermore, $\times : \catD\times \catD \to \catD $ and 	$\sqcup : \catD\times \catD \to \catD $ have, respectively, a left adjoint and a right adjoint. Therefore they are accessible.
	
	Finally, by \citep[Proposition~3.8]{ITA_2002__36_2_195_0}, assuming their existence, $\mu H $ and $\nu H $ are
	accessible whenever $H: \catD '\times\catD\to\catD $ is accessible and $\catD ' $ is locally presentable. 
	
	This completes the proof that all morphisms of $\mnPoly _ \catD $  are accessible. Hence, by
	Lemma \ref{lem:accessible-functors-have-initial-algebras-terminal-coalgebras}, we have that
	all endofunctors in  $\mnPoly _ \catD $  have initial algebras and terminal coalgebras. 
		Therefore $\catD $ has $\mu\nu $-polynomials.
\end{proof}

\begin{remark}[Duality]
Let $\catD$ be a category. 	By a well-known result by Gabriel-Ulmer~\cite[7.13]{zbMATH03356994},
$\catD$ and $\catD ^\op $ are locally presentable if, and only if, $\catD $ is a complete lattice.
Therefore, in general, the property of being locally presentable is not self-dual.
	
As remarked in \ref{remark:sel-duality-munupolynomials}, the property of having $\mu\nu$-polynomials is self-dual. Hence, by Proposition \ref{prop:locally-presentable-has-munupolynomials}, we have that, whenever $\catD ^\op $ is locally presentable, $\catD $ has $\mu\nu$-polynomials.		
\end{remark}

\subsection{$\catLi$, $\catFLi$ and $\Fam{\catLi}$}\label{subsec:locallypresentable-Fam-basic-properties}
Henceforth,  we assume that $\catLi $ is a locally presentable category with biproducts $(+, \initial ) $ that is monadic over $\Set$.
The main examples that we have in mind are the category of real vector spaces $\catLi = \Vect $ and the category of commutative monoids $\catLi = \CMon $.

We consider the indexed category 
\begin{eqnarray}
	\catFLi : &\Set ^\op  & \to \Cat\label{eq:catFli}\\
	& X & \mapsto \Cat\left[ X, \catLi \right] = \catLi ^ X \nonumber\\
	& f : X\to Y & \mapsto \catLi ^f = \Cat \left[ f, \catLi \right] : \catLi ^Y\to \catLi ^X\nonumber
\end{eqnarray}
defined by the composition 
\begin{equation}
	\Set ^\op\rightarrow \Cat ^\op \xrightarrow{\Cat\left[ -, \catLi \right] } \Cat 
\end{equation}
in which $\Cat\left[ -, \catLi \right] = \catLi ^{(-)}$ is the exponential (internal hom) in $\Cat $.
We have that
\begin{equation}
	\displaystyle\GrothSet\catFLi\cong  \Fam{\catLi }, \qquad\qquad \GrothSet\catFLi ^\op\cong  \Fam {\catLi ^\op } 
\end{equation}
where $\Fam {\catLi }$ and $\Fam {\catLi ^\op } $ are, respectively, the free cocompletion under coproducts
of $\catLi $ and of $\catLi ^\op $.  We refer the reader, for instance, to \cite[Section~2]{zbMATH07186728} and \cite[Chapter~6]{zbMATH01577082} for basic facts about free cocompletion under coproducts.

We have the following basic straightforward properties about $\Fam{\catLi} $:

\begin{proposition}\label{theo:Fam-closed-structure}
	Let $\catD$ be a category with biproducts $(+, \initial )$.
		If $\catD $ has (infinite) products, $\Fam{ \catD } $ is cartesian closed. Codually,  if $\catD $ has  (infinite) coproducts, $\Fam{ \catD ^\op } $	is cartesian closed.
\end{proposition}
\begin{proof}
	Namely, given families of objects $\famY : Y\to \catD  , \famZ : Z\to \catD $, we define
\begin{eqnarray}	
	 \mathcal{YZ} : &\Fam{\catD }\left( \left( Y, \famY\right) , \left(Z, \famZ\right) \right) &\to \catD   \\
	 & \left( g : Y\to Z, \left( \alpha _ y : \famY (y)\to \famZ \left( g (y)\right)    \right)_{y\in Y} \right)  & \mapsto \prod _{y\in Y } \famZ\left( g(y)\right) \nonumber\\	
	 &&\\
  	\mathcal{YZ}^t : &\Fam{\catD ^\op }\left( \left( Y, \famY\right) , \left(Z, \famZ\right) \right) &\to \catD   \\
  	& \left( g : Y\to Z, \left( \alpha _ y : \famZ \left( g (y)\right) \to \famY (y)   \right) _{y\in Y} \right)  &  \mapsto \coprod _{y\in Y } \famZ\left( g(y)\right) \nonumber
\end{eqnarray}
The pair $\left(  \Fam{\catD }\left( \left( Y, \famY\right)  , \left(Z, \famZ\right) \right),  \mathcal{YZ} \right) $ 
 is the exponential $\left( Y ,   \famY  \right)\Rightarrow\left( Z ,   \famZ  \right)  $  in $\Fam{ \catD } $, provided that $\catD$ has products. 
 
 Codually, $\left(  \Fam{\catD }\left( \left( Y, \famY\right)  , \left( Z, \famZ\right) \right),  \mathcal{YZ}^t \right) $ 
 is the exponential $\left( Y ,   \famY  \right)\Rightarrow\left( Z ,   \famZ  \right)  $ in $\Fam{ \catD ^\op } $, provided that $\catD $ has coproducts.
\end{proof}

\begin{proposition}\label{theo:Fam-locally-presentable}
	$\Fam{ \catD } $ is locally presentable, whenever $\catD $ is locally presentable.
\end{proposition}
\begin{proof}
Since $\catD $ is cocomplete, $\Fam{\catD}$ is cocomplete (see Lemma \ref{lem:completeness-cocompleteness}). Moreover, 
it is clear that the indexed category defined by $X\mapsto \Cat[X, \catD ] $ satisfies the conditions of 
 \cite[Definition~5.3.1]{zbMATH00044674}, since:
\begin{enumerate}[(1)]
	\item for each $X\in\Set$, $\Cat[X, \catD ]  = \catD ^X $ is locally presentable and, hence, accessible;
	\item for any function $f$, $\Cat[X, \catD ] $ is accessible by Lemma~\ref{lem:right-adjoint-accessible-preserves}, since it has a left adjoint given by the left Kan extension $\lan_f$; see \eqref{eq:pointwise-kan-extensions};
	\item $\Set $ is locally presentable;
	\item $X\mapsto \Cat[X, \catD ] $  preserves any limit of $\Set ^\op $.
\end{enumerate}
Therefore,  $\Fam{ \catD } $ is accessible by \cite[Theorem~5.3.4]{zbMATH00044674}. This completes the proof that $\Fam{ \catD } $  is locally presentable.
\end{proof}

As a consequence, we have that:

\begin{corollary}
	$\Fam{\catLi} $ is cartesian closed and locally presentable and, hence, has $\mu\nu$-polynomials.
\end{corollary} 	

The results proven above do not guarantee that $\Fam{\catLi ^\op} $ has $\mu\nu$-polynomials, since $\Fam{\catLi ^\op} $ is not, generally, locally presentable. However, in \ref{subsec:locallypresentable-Fam}, we show that $\catFLi $ yields a model for the target language and, hence,  $\Fam{\catLi} $ and $\Fam{\catLi ^\op} $ have $\mu\nu$-polynomials (and are cartesian closed).

\subsection{$\catFLi$ is a $\Sigma$-bimodel for inductive and coinductive types}\label{subsec:locallypresentable-Fam}

We establish that $\catFLi : \Set ^\op\to\Cat $  yields a model for the target language in Corollary~\ref{coro:target-language-concrete-model}. By the results of Section \ref{sec:grothendieck-constructions}, this provides proof that $\GrothSet\catFLi \cong \Fam{\catLi} $ and $\GrothSet\catFLi  ^\op\cong \Fam{\catLi ^\op} $ are bicartesian categories with $\mu\nu$-polynomials by Corollary \ref{coro:suitL-indexed-category-total-category-munu}. We start by proving that $\catFLi$ is a $\Sigma $-bimodel for inductive and coinductive types.

Since $\Set $ is locally presentable,
$\Set $ has $\mu\nu $-polynomials by Proposition \ref{prop:locally-presentable-has-munupolynomials}.
Moreover, since $\catLi $ is complete and cocomplete, $\catFLi \left( X\right) = \catLi ^X $
is complete and cocomplete as well; namely, the limits and colimits are constructed pointwise. In particular, 
$\catFLi \left( X\right) = \catLi  ^X $ has biproducts (also constructed pointwise) $(+, \initial ) $.

It should be noted that, for any function $f: X\to Y $ in $\Set $,
we have that
\begin{equation}
	\catLi  ^f = \catFLi \left( f\right) : \Cat\left[ Y, \catLi \right] \to \Cat\left[ X, \catLi \right]
\end{equation}
has a (fully faithful) left adjoint and a  (fully faithful) right adjoint, given by the left and right Kan extensions respectively;\footnote{The basic definition of Kan extension can be found, for instance, in \citep[Chapter~X]{zbMATH03367095}. Although one can verify it directly, \ref{eq:pointwise-kan-extensions} follows from the general result about pointwise Kan extensions; see, for instance, \citep{zbMATH03362059} or \citep[Chapter~4]{zbMATH02172008}.}  
namely, for each
$\famX : X \to \catLi $,
\begin{equation}\label{eq:pointwise-kan-extensions}
	\ran _ f \famX (x) = \prod _{i\in f^{-1}(x) } \mathcal{X} (i), \qquad\qquad \lan _ f \famX (x) = \coprod _{i\in f^{-1}(x) } \famX (i) .
\end{equation}
Therefore, we can conclude that: (1)  $\catFLi \left( f \right) $ preserves limits, colimits and, consequently, biproducts; (2) $\catFLi \left( f \right) $ preserves initial algebras and terminal coalgebras by Theorem \ref{theo:right-adjoint-preserves-terminal-coalgebras}. 
Furthermore,  
$\catFLi \left( f\right) $ \textit{strictly} preserves 
biproducts (and the zero object), initial algebras and terminal coalgebras, provided that $\catLi$ has chosen ones. 

Finally, it is clear that we have the isomorphism
\begin{eqnarray*}
\catFLi \left( X\sqcup Y \right) & = & \Cat\left[ X\sqcup Y, \catLi \right] \\
&\cong  &\Cat\left[ X, \catLi \right] \times \Cat\left[ Y, \catLi \right]  \\
&= &  \catFLi \left( X\right) \times \catFLi \left( Y \right)  
\end{eqnarray*}
and, hence, $\catFLi $ is extensive. Indeed, we have
\begin{equation}
\equivalenceextensive ^{(X,Y)} :  \catFLi \left( X\right) \times \catFLi \left( Y \right) \to \catFLi \left( X\sqcup Y \right) 
\end{equation}
in which  $\equivalenceextensive ^{(X,Y)} \left( \famX , \famY \right) (i) = \famX (i) $ if $i\in X $ and
$\equivalenceextensive ^{(X,Y)} \left( \famX , \famY \right) (j) = \famY (j) $ if $j\in Y $.

\begin{therm}\label{theo:Sigma-bimodel-inductive-coinductive-FAM}
	The strictly indexed category
$\catFLi $ is a $\Sigma$-bimodel for inductive and coinductive types. Therefore $\GrothSet \catFLi $ and $\GrothSet \catFLi ^\op $ have $\mu\nu $-polynomials.
\end{therm}	
\begin{proof}
	
It only remains to prove that all the endomorphisms in  $\mnPoly _ \catFLi $ have initial algebras and terminal coalgebras. In order to do so, by Lemma \ref{lem:accessible-functors-have-initial-algebras-terminal-coalgebras}, it is enough to prove that $\mnPoly _ \catFLi $ is a subcategory of the category of locally presentable categories and accessible functors between them.

The subcategory of locally presentable functors and accessible functors is closed under products. That is to say, if $\catD , \catD '  $ are locally presentable categories and $E,J $ are accessible functors between locally presentable categories, we get that $\terminal , \catD\times \catD '$ are locally presentable categories, 
$(E,J) $ is accessible,  and the projections are accessible (since they have right adjoints).

Moreover, $\catLi ^X $ is locally presentable for any set $X$ since $\catLi $ is locally presentable.
Also, since the biproduct $ + : \catLi ^X\times\catLi ^X\to\catLi ^X $ has a right adjoint, it is accessible.
Furthermore, since it has a right adjoint, we get that 
$ \catLi (f) $ is accessible for any function $f:X\to Y $. 

Finally, by \citep[Proposition~3.8]{ITA_2002__36_2_195_0}, assuming their existence, $\mu h $ and $\nu h $ are
accessible whenever $h: \catD '\times\catD\to\catD $ is accessible and $\catD ', \catD $ are locally presentable categories. 
	
Since isomorphisms between locally presentable categories are accessible, this completes the proof that all functors in $\mnPoly _ \catFLi $  are accessible functors between locally presentable categories. 

Therefore, any endomorphism in $\mnPoly _ \catFLi $  has initial algebra and terminal coalgebra by Lemma~\ref{lem:accessible-functors-have-initial-algebras-terminal-coalgebras}. This completes the proof.
\end{proof}

\subsection{$\catFLi$ is a $\Sigma$-bimodel for function types}
We consider the cartesian dependent type theory $\catFSet : \Set ^\op \to\Cat $, $ X\mapsto \Cat\left[ X, \Set \right] $.
It is well-known that $\catFSet $ satisfies full, faithful, democratic comprehension with $\Pi $-types and strong $\Sigma$-types~\citep{jacobs1999categorical}. 
In this context, we have that $\catFLi $ has $\Pi $- types by \citep[Theorem~5.2.9]{vakar2017search}. Finally, $\catFLi $ indeed has $\Sigma $-types and $\multimap$-types by \citep[Theorem~5.6.3]{vakar2017search}. 

This proves that $\catFLi$ is a $\Sigma$-bimodel for function types. By Theorem \ref{theo:Sigma-bimodel-inductive-coinductive-FAM},
we conclude:
\begin{therm}\label{coro:target-language-concrete-model}
	$\catFLi : \Set ^\op \to \Cat $ yields a $\Sigma$-bimodel for inductive, coinductive, and function types.
\end{therm}

\begin{corollary}\label{coro:source-grothendieck-construction}
	The categories $\Fam{\catLi} $ and $\Fam{\catLi ^\op} $ are bicartesian closed categories with $\mu\nu $-polynomials.
\end{corollary}

\subsection{$\Fam{\catLi} $ and $\Fam{\catLi ^\op} $ are complete and cocomplete}
Concrete models provide a significant advantage in terms of the extra properties they can satisfy, which we leverage in our open semantic logical relations.
In particular, we have:

\begin{lemma} \label{lem:completeness-cocompleteness}
	$\Fam{\catLi} $ and $\Fam{\catLi ^\op} $ are complete and cocomplete.
\end{lemma} 
\begin{proof}
	This is a well known result and, from a fibred perspective, follows from the fact that $\catFLi$ has indexed limits and colimits (and $\Set$ is cocomplete and complete).

We only need, however, the coproducts and pullbacks that we sketch below.

Coproducts: it is clear that $\Fam{\catLi} $ and $\Fam{\catLi ^\op} $ have coproducts, $\Fam{-}$ is the cocompletion under coproducts.
The coproduct of a (possibly infinite) family $\left( W_i, w_i \right) _{i\in L} $ of objects in $\Fam{\catLi} $
(respectively $\Fam{\catLi ^\op} $) is given by the object $\left( \bigsqcup\limits _{i\in L} W_i , \langle w_i \rangle _{i\in L} \right) $ in $\Fam{\catLi} $ (respectively in $\Fam{\catLi ^\op} $), where $\langle w_i \rangle$  denotes 
the family $\bigsqcup\limits _{i\in L} W_i  \to \catLi $ defined by $w_i $ in each component $W_i$.
	
Pullbacks: let $(f, f'): (W,w)\to (Y,y) $ and $(g, g'): (X,x)\to (Y,y) $ be morphisms of $\Fam{\catLi ^\op}$. We consider the pullback $W\times _{(f,g)} X $ of $f$ along $g$, with projections $p_W : W\times _{(f,g)} X\to W $ and $p_X : W\times _{(f,g)} X\to X $. Denoting by $s$ the pushout of \eqref{eq:span-pushout}  in the category $ \catFLi\left( W\times _{(f,g)} X \right) =  \catLi ^{W\times _{(f,g)} X } $, the pullback of $(f, f'): (W,w)\to (Y,y) $ and $(g, g'): (X,x)\to (Y,y) $ in $\Fam{\catLi ^\op}$ is given by $\left( W\times _{(f,g)} X, s \right) $. 
	\begin{equation}\label{eq:span-pushout}
	\begin{tikzpicture}[x=5cm, y=1cm]
		\node (a) at (0,0) {$    y\circ g\circ p_X = y\circ f\circ p_W $};
		\node (b) at (1, 0) { $w\circ p_W   $ };
		\node (c) at (-1,0) {$    x \circ p_X $ };
		\draw[->] (a)--(b) node[midway,above] {$ \catFLi (p_W) (f') $};
		\draw[->] (a)--(c) node[midway,above] {$ \catFLi (p_X) (g') $};
	\end{tikzpicture} 
\end{equation}  
\end{proof}

\section{Concrete denotational semantics for CHAD}\label{sec:concrete-semantics-specification}

In this section, we will establish a concrete denotational semantics for both the source and target languages, and establish CHAD's specification.

\subsection{The concrete model $\Fam{\Set}$ for the source language}\label{subsect:FamSet-concrete-model}
We define a denotational semantics for our source language by interpreting coproducts of Euclidean spaces as families of sets, i.e., we interpret our language in $\Fam{\Set}$. 
This approach offers technical advantages as it is the natural way to interpret functions between sum types in our setting.

Below, we establish some notation to talk about morphisms, objects and coproducts in $\Fam{\Set } $. We start by recalling that 
the category $\Fam{\Set}\simeq \Cat[\mathsf{2}, \Set ] $ is locally presentable (see Proposition \ref{theo:Fam-locally-presentable}). Hence, by Proposition \ref{prop:locally-presentable-has-munupolynomials}, $\Fam{\Set}$ has $\mu\nu $-polynomials.
This proves that $\Fam{\Set}$  is a suitable concrete model for our source language, since $\Fam{\Set}\simeq \Cat[\mathsf{2}, \Set ] $ is cartesian closed.

\begin{proposition}\label{prop:FAMSET-as-CONCRETE-MODEL}
	The category $\Fam{\Set}\simeq \Cat[\mathsf{2}, \Set ] $ is complete, cocomplete, cartesian closed and has $\mu\nu$-polynomials.
\end{proposition}

Henceforth, we use the notation $\left(A_l\right) _{l\in L} = \left(L, A^\ast\right) \in \Fam{\Set}$ to refer to the object of $\Fam{\Set }$ that corresponds to the pair $\left(L, A^\ast\right)$, where $A^\ast$ 
assigns to each $l\in L$ the set $A_l$. This is a standard way to represent families of sets, where the index set $L$ and the set $A_l$ associated to each index $l$ are explicitly given. 

\subsubsection{Morphisms between families of sets}
Recall that a morphism between families $\left(A_l\right) _{l\in L} $ and  $\left(B_i\right) _{i\in I} $ in $\Fam{\Set} $ is a pair $\left( \underline{f} , f \right)$ where $\underline{f}:L\to I $ is a function and $f = \left( f_l: A_l \to B_{\underline{f}(l)} \right) _{l\in L} $ is family of functions. By abuse of language, we often denote such a morphism $\left( \underline{f} , f \right)$  by $f$, keeping $\underline{f} $ implicit.

\subsubsection{Singleton families}\label{subsect:singleton-families}
For a family $\left(A_l\right) _{l\in L} = \left(L, A^\ast\right) \in \sobjects{\Fam{\Set }}$ where 
 $L = \left\{ 0\right\}$ is a singleton, we abuse the notation and write $A_0$ instead of $\left(A_l\right)_{l\in L}$. For example, we use the notation $\RR^n$ to denote the singleton family in $\Fam{\Set }$ whose only object is the set $\RR^n$.
 
In this case, a morphism $f:\RR^n \to \RR^m$ in $\Fam{\Set}$ corresponds to a morphism in $\Set$. More precisely, the functor $\Set \to \Fam{\Set}$ given by $A \mapsto A$ is fully faithful.

\subsubsection{Coproducts of families of sets}
Let $\left( \left(A_{(l,i)}\right)_{l\in L_i}\right) _{i\in I} = \left(L_i, A^\ast_i\right)_{i\in I} $ be a (possibly infinite) family of objects of $\Fam{\Set}$. Recall that the coproduct  $\coprod\limits _{i\in I} \left(L_i, A^\ast_i\right) $
in $\Fam{\Set}$ is given by $\left( \coprod\limits _{i\in I} L_i, \langle A^\ast_i\rangle _{i\in I} \right) $. 

Using the notation established in \ref{subsect:singleton-families}, we see that, for a family of singleton families $\left( A_i\right) _{i\in I}$ in $\Fam{\Set }$, the coproduct $\coprod _{i\in I} A_i $ is the same as the family $\left( A_i\right) _{i\in I}$
considered as an object in $\Fam{\Set}$. Hence, in this context, we often denote  by $ \coprod _{i\in I} A_i $ the object $\left( A_i\right) _{i\in I}$ in $\Fam{\Set }$.

For instance, consider a family of natural numbers $\left( n_i \right) _{i\in I} $, and consider, for each $i\in I $, the object $\RR^{n_i}$ of $\Fam{\Set } $. In this setting, we have that
$ \coprod _{i\in I} \RR ^{n_i} $ is the family $\left( \RR^{n_i} \right) _{i\in I}$.

On one hand, it should be noted that, in this setting, a morphism 
\begin{equation} \label{eq:function-in-famSet}
f: \coprod _{i\in I} \RR ^{n_i} \to \coprod _{j\in J} \RR ^{m_j}
\end{equation} 
in $\Fam{\Set}$ is not the same as a function $\coprod _{i\in I} \RR ^{n_i} \to \coprod _{j\in J} \RR ^{m_j}$ in $\Set $. More precisely, the functor $\coprod : \Fam{\Set}\to \Set $ defined by 
$$\left( \left( A_i\right) _{i\in I} = \coprod _{i\in I} A_i \right)  \mapsto \coprod _{i\in I} A_i $$
is not full.

On the other hand, it is worth noting that there is bijection between morphisms of the form \eqref{eq:function-in-famSet} in $\Fam{\Set}$ and functions $g: \coprod\limits _{i\in I} \RR ^{n_i} \to \coprod _{j\in J} \RR ^{m_j}$ in $\Set $  such that,  for each $i\in I$, there is $j\in J $ such that $ g(\RR ^{n_i})\subset \RR ^{m_j}$.

\subsubsection{Products of families of sets}
Recall that, given  objects $\left(A_l\right)_{l\in L}$ and $\left(B_i\right)_{i\in I}$ of $\Fam{\Set}$, the product $\left(A_l\right)_{l\in L}\times \left(B_i\right)_{i\in I}$  is given by $\left( A_l\times B_i\right) _{(l,i)\in L\times I}$.

\subsection{The concrete model $\catFV$ for the target language} \label{subsect:FVECT-concrete-model}
We provide a denotational semantics for our target language by interpreting spaces of (co)tangent vectors as well as derivatives of differentiable functions in terms of families of vector spaces in \ref{subsect:concrete-semantics-functors-target-language}. To do so, we consider
the indexed category $\catFV : \Set ^\op \to\Cat$ which associates each set $X$ with $\Vect ^X$. 

It should be noted that $\catFV $ is $\catFLi$ as considered in \ref{subsec:locallypresentable-Fam} taking $\catLi = \Vect $.
By Theorem \ref{coro:target-language-concrete-model}:
\begin{corollary}\label{coro:catFV-is-a-model-for-target-language}
	$\catFV : \Set ^\op \to \Cat $ yields a $\Sigma$-bimodel for inductive, coinductive, and function types. Consequently, \begin{equation}\label{eq:FVetFam} 
		\displaystyle\GrothSet\catFV\cong  \Fam{\Vect }, \qquad\qquad \GrothSet\catFV ^\op\cong  \Fam {\Vect ^\op } 
	\end{equation}
	are bicartesian closed and have $\mu\nu$-polynomials. 
\end{corollary} 	
Moreover, by Lemma \ref{lem:completeness-cocompleteness}, we have:
\begin{corollary}\label{coro:catFV-is-complete-cocomplete}
\eqref{eq:FVetFam}  are complete and cocomplete.
\end{corollary}

We recall some basic aspects of \eqref{eq:FVetFam} below.

\subsubsection{Constant families of vector spaces}
We introduce notation for objects in $\Fam{\Vect}$ (and $\Fam{\Vect ^\op}$) that correspond to constant families, which is the case for the semantics of our primitive types in the target language. Given a set $N\in\Set$ and a vector space $V\in\Vect$, we denote the corresponding object as $\left(N,\underline{V}\right)$. Here, $\underline{V}: N\to\Vect$ is the family that is constantly equal to $V$, meaning that $\underline{V}(s) = V$ for all $s\in N$.

\subsubsection{Product of families of vector spaces}
Let $(M, m), (N, v) $ be objects of $\GrothSet \catFV \cong \Fam{\Vect} $ (or $\GrothSet \catFV ^\op \cong \Fam{\Vect ^\op}$). By Propositions \ref{prop:grothendieck-products-covariant} and \ref{theo:grothendieck-products-contravariant}, we have that
\begin{equation}
	(M, m)\times (N, v) = \left( M\times N, (i,j)\mapsto m(i)\times v(j) \right) 
\end{equation}
gives the product of $(M, m)$ and $(N, v) $ in $\GrothSet \catFV $ (and in $\GrothSet \catFV ^\op $). The terminal object in $\GrothSet \catFV $  (and in $\GrothSet \catFV ^\op $) is given by $\left(\terminal, \initial \right) $.

\subsubsection{Coproduct of families of vector spaces}
Let $\left(W, w _ i\right) _{i\in L} $ be a family of objects of $\GrothSet \catFV $ (or $\GrothSet \catFV ^\op $). We have that \eqref{eq:corproduct-FAM}  gives the coproduct of the family $\left(W, w _ i\right) _{i\in L} $ in $\GrothSet \catFV $ and in $\GrothSet \catFV ^\op $.
\begin{equation}\label{eq:corproduct-FAM}
\left( \coprod\limits _{i\in L} W_i , \langle w_i \rangle _{i\in L} :\coprod\limits _{i\in L} W_i \to \Vect  \right) 
\end{equation}
The initial objects in $\GrothSet \catFV $ and in $\GrothSet \catFV ^\op $ are given by $\left( \emptyset , \initial \right) $.

\subsubsection{Lists and Streams}

Let \eqref{eq:endofunctor-E-FAM} and \eqref{eq:endofunctor-H-FAM} be endofunctors on both $\Fam{\Vect}$ and $\Fam{\Vect ^\op}$. We can compute the initial algebras and terminal coalgebras of \eqref{eq:endofunctor-E-FAM} and \eqref{eq:endofunctor-H-FAM} via colimits and limits of chains~\citep{adamek1979least}. 
We get \eqref{eq:endofunctor-E-FAM-inductive-type} and \eqref{eq:endofunctor-H-FAM-codinductive-type} in both $\Fam{\Vect}$ and $\Fam{\Vect ^\op}$.
\\
\noindent\begin{minipage}{.5\linewidth} 
	\begin{equation}\label{eq:endofunctor-E-FAM}
	E (X,x) = \left( \terminal , \initial  \right) \sqcup  (X,x) \times \left( V , \underline{V}  \right)
	\end{equation} 
\end{minipage}
\begin{minipage}{.5\linewidth} 
	\begin{equation} \label{eq:endofunctor-H-FAM}
	 H (X,x) = (X,x) \times \left( V , \underline{V}  \right)
	\end{equation} 
\end{minipage}
\\
\noindent\begin{minipage}{.5\linewidth} 
	\begin{equation}\label{eq:endofunctor-E-FAM-inductive-type}
	\mu E = \coprod _{n=0}^\infty   \left( V , \underline{V}  \right) ^n,
	\end{equation} 
\end{minipage}
\begin{minipage}{.5\linewidth} 
	\begin{equation} \label{eq:endofunctor-H-FAM-codinductive-type}
	\nu H = \prod _{i=0}^\infty   \left( V , \underline{V}  \right) 
	\end{equation} 
\end{minipage}

Considering the case where $H$ is an endofunctor on $\Fam {\Vect }$, we have that $\hat{\nu H} $ in \eqref{eq:endofunctor-H-FAM-codinductive-type-nuH} 
is the functor constantly equal to the product $\displaystyle\prod _{n=0}^\infty V $. When we consider $H$ on $\Fam{\Vect ^\op}$, $\hat{\nu H} $ is the functor constantly equal to $\displaystyle\coprod _{i=0}^\infty V $.

In the case of the endofunctor $E$, $\hat{\mu E}$ in \eqref{eq:endofunctor-E-FAM-inductive-type-muE} is defined by the constant families $\underline{V^n}: V^n\to \Vect$ in each component $V^n$ of the set $\coprod_{i=0}^{\infty} V^n$. This holds true for both $\Fam{\Vect}$ and $\Fam{\Vect^\op}$.\\
\noindent\begin{minipage}{.5\linewidth} 
	\begin{equation}\label{eq:endofunctor-E-FAM-inductive-type-muE}
		 \mathsf{List} \left( V , \underline{V}  \right) = \mu E = \left( \coprod _{n=0}^\infty V ^n,   \hat{\mu E}  : \coprod _{n=0}^\infty V ^n \to \Vect\right)
	\end{equation} 
\end{minipage}
\begin{minipage}{.5\linewidth} 
	\begin{equation} \label{eq:endofunctor-H-FAM-codinductive-type-nuH}
		 \mathsf{Stream} \left( V , \underline{V}  \right) =  \nu H = \left( \prod _{i=0}^\infty V ,  \hat{\nu H} : \prod _{i=0}^\infty V  \to \Vect\right)  
	\end{equation} 
\end{minipage}

\subsection{Euclidean spaces and coproducts}
We introduce the notion of derivatives as it pertains to our work. Our definition aligns with the conventional understanding of derivatives of functions between manifolds, but with added flexibility to accommodate manifolds of varying dimensions. Readers interested in the basics of differentiable manifolds can refer to \citep{lee2013smooth, tu2011manifolds}.

Let $\Man $ be the category of differentiable manifolds and differentiable maps between them.
An \textit{Euclidean space} is an object of $\Man $ that is isomorphic to some differentiable manifold $\RR^n$.

We denote by $\Euclidean$ the category of Euclidean spaces and differentiable maps between them. In other words, 
$\Euclidean $ is the full and replete subcategory of  $\Man $ containing the
differentiable manifolds  $\RR ^k $ for all $k\in\NN $. 

\begin{definition}[Basic definition of derivatives]\label{def:very-basic-definition-of-derivative}
	Let $f: \RR ^n \to \RR ^m $ be a morphism in $\Euclidean$. We define 
the morphisms \eqref{eq:derivative-in-fam-vect-new} in $\Fam{\Vect}$
	and \eqref{eq:derivative-in-fam-vectop-new} in $\Fam{\Vect^\op}$, where  
	$Df _{x}\coloneqq  f'(x)$ is the usual Fréchet derivative, and $Df _{x}^t\coloneqq  f'(x)^t$ is the transpose of $f'(x)$. 
	\begin{equation}\label{eq:derivative-in-fam-vect-new}
		\underline{\Ds{f}} :=\left( f,  Df\right) :   \left( \RR ^{n},  \underline{ \RR ^{n} }   \right)  \to   \coprod _{k\in K} \left( \RR ^m,  \underline{\RR ^m }  \right)
	\end{equation}	
	\begin{equation}\label{eq:derivative-in-fam-vectop-new}
		\underline{\Dsr{f}} :=\left( f,  Df^t\right) :  \left( \RR ^{n},  \underline{ \RR ^{n} }   \right)  \to   \coprod _{k\in K} \left( \RR ^m,  \underline{\RR ^m }  \right)
	\end{equation}
\end{definition} 	
It follows from the usual properties of derivatives and chain rule that:
\begin{lemma}[Derivative of maps between Euclidean spaces]
\eqref{eq:derivative-in-fam-vect-new} and \eqref{eq:derivative-in-fam-vectop-new} uniquely extend to strictly cartesian functors \eqref{DSEUCLIDEAN} and \eqref{DSREUCLIDEAN}, respectively.
\begin{equation}\label{DSEUCLIDEAN}
	\underline{\Ds{}} : \Euclidean\to\Fam{\Vect }
\end{equation}	
\begin{equation}\label{DSREUCLIDEAN}
	\underline{\Dsr{} } : \Euclidean\to\Fam{\Vect ^\op}
\end{equation}
\end{lemma}

While the definitions provided above are presented in the CHAD-style, they are essentially the same as the ones used to define derivatives between Euclidean spaces, which are commonly taught in calculus courses.

In order to establish a consistent and rigorous framework for proving the correctness of CHAD for inductive data types, we will extend the definition of derivatives by using \textit{cotupling}. More precisely, from a categorical perspective, this extension will rely on the universal property of the free cocompletion under coproducts.

\begin{definition}[Derivative of families]
The universal  property of the free cocompletion  under coproducts $\Fam{\Euclidean}$ of $\Euclidean$
induces unique coproduct-preserving functors 
\begin{equation}\label{DSEUCLIDEAN-FUNCTOR}
	\overline{\Ds{}} : \Fam{\Euclidean}\to\Fam{\Vect }
\end{equation}	
\begin{equation}\label{DSREUCLIDEAN-FUNCTOR}
	\overline{\Dsr{} } : \Fam{\Euclidean}\to\Fam{\Vect ^\op}
\end{equation}
that (genuinely) extend the functors  \eqref{DSEUCLIDEAN} and \eqref{DSREUCLIDEAN}, respectively.
\end{definition} 

Let $\Fam{-}$ be the $2$-functor that takes each category to its free cocompletion under coproducts.
Denoting by $\coprod $ the respective functors that give the coproduct of families, recall that, by the definition above, 
\eqref{DSEUCLIDEAN-FUNCTOR} and \eqref{DSREUCLIDEAN-FUNCTOR} are respectively  given by the composition \eqref{eq:new-DS} and  \eqref{eq:new-DSR}. 
\begin{equation}\label{eq:new-DS}
	\begin{tikzpicture}[x=4cm, y=1cm]
		\node (a) at (0,0) {$\Fam{\Euclidean } $};
		\node (b) at (1, 0) { $\Fam{\Fam{\Vect }} $ };
		\node (c) at (2,0) {$   \Fam{\Vect}   $ };
		\draw[->] (a)--(b) node[midway,above] {$\Fam{\overline{\Ds{} } } $};
		\draw[->] (b)--(c) node[midway,above] {$\coprod $};
	\end{tikzpicture} 
\end{equation}  
\begin{equation}\label{eq:new-DSR} 
	\begin{tikzpicture}[x=4cm, y=1cm]
		\node (a) at (0,0) {$\Fam{\Euclidean } $};
		\node (b) at (1, 0) { $\Fam{\Fam{\Vect ^\op }} $ };
		\node (c) at (2,0) {$   \Fam{\Vect ^\op} $ };
		\draw[->] (a)--(b) node[midway,above] {$\Fam{\overline{\Dsr{}}} $};
		\draw[->] (b)--(c) node[midway,above] {$\coprod $};
	\end{tikzpicture} 
\end{equation}

\subsection{Euclidean families, differentiable morphisms, derivatives and diffeomorphisms} \label{subsect:differentiation-seminarcs-subsection}
We introduce the notion of differentiable morphisms in $\Fam{\Set }$, fundamental to establish the specification and correctness of CHAD. 
To this end, we first define  Euclidean families.
\begin{definition}[$\Euc $: Euclidean families] \label{eq:Euclidea-families}
We inductively define the set $\Euc $ of \textit{Euclidean families} by \ref{def:EUC-1}, \ref{def:EUC-2} and \ref{def:EUC-3}. 
\begin{enumerate}[$\Euc$1)]
	\item For any $k\in\NN$, the singleton family with $\RR ^k $ is a member is an element of $\Euc $.\label{def:EUC-1}
	\item Assuming that $A $ and $B$ are elements of $ \Euc $,  the product $A\times B$ in $\Fam{\Set}$ belongs to $\Euc$.\label{def:EUC-2}
	\item Assuming that $ \left(L _i, A^\ast_i \right) _{i\in L}  $ is a (possibly infinite) family of objects in  $ \Euc $,  the coproduct 
	$$ \left( \coprod _{i\in L } L _i, \langle A^\ast_i\rangle  _{i\in L } \right)  =  \coprod  _{i\in L} \left(L _i, A^\ast_i \right)  $$ 
	 in $\Fam{\Vect } $ also belongs to $\Euc$.\label{def:EUC-3}
\end{enumerate}
\end{definition} 

We denote by
\begin{equation} \label{eq:EuclideU}
	\EuclidU : \Fam{\Euclidean} \to \Fam{\Set}.
\end{equation} 
the forgetful functor obtained by $\EuclidU \coloneqq \Fam{ \underline{\EuclidU } }$ where $\underline{\EuclidU } : \Euclidean \to \Set$ denotes the obvious forgetful functor.

\begin{definition}[Differentiable morphisms and their derivatives]\label{def:derivative-of-a-functioonn}
	A morphism $f: A\to B $ in $\Fam{\Set} $ is differentiable if $A, B\in \Euc $ and  there is a morphism $ \mathfrak{f} $ in $\Fam{\Diff} $ such that
	$\EuclidU\left( \mathfrak{f}  \right) = f $. In this case, we define:
	\begin{equation}
		\Ds{f}\coloneqq \overline{\Ds{\mathfrak{f}} } \qquad\mbox{ and }\qquad \Dsr{f}\coloneqq \overline{\Dsr{\mathfrak{f}} }.
	\end{equation}
We call $f$ a differentiable map, $\Ds{f}$ the \textit{derivative}, and $\Dsr{f} $ the \textit{transpose derivative} of $f$.
\end{definition} 	

\begin{definition}[Diffeomorphism and diffeomorphic Euclidean families]
	We say that a morphism  $f$ of $\Fam{\Set } $ is a \textit{diffeomorphism} if it is an isomorphism in $\Fam{\Vect }$ such that both $f$ and $f^{-1}$ are 
	differentiable. 
	
	We say that two objects $\left( A_l\right) _{l\in L}$ and $\left( B_j\right) _{j\in J}$ of $\Fam{\Vect } $ are \textit{diffeomorphic} 
	if there is a diffeomorphism $\left( A_l\right) _{l\in L}\to \left( B_j\right) _{j\in J}$.
\end{definition}

It should be noted that the chain rule applies. More precisely:
\begin{lemma}[Chain rule]\label{lem:chain-rule}
	If $g$ and $ f$ are composable differentiable morphisms in $\Fam{\Set}$, $g\circ f $ is differentiable. Moreover,
	Equations \eqref{eq:chain-rule-equation-forward} and \eqref{eq:chain-rule-equation-reverse} respectively hold
	in $\Fam{\Vect}$ and $\Fam{\Vect ^\op}$.\\
	\noindent\begin{minipage}{.5\linewidth} 
		\begin{equation} \label{eq:chain-rule-equation-forward}
			\Ds{g}\circ \Ds{f} = \Ds{\left( g\circ f\right) }
		\end{equation} 
	\end{minipage} 
	\begin{minipage}{.5\linewidth} 
		\begin{equation} \label{eq:chain-rule-equation-reverse}
			\Dsr{g}\circ \Dsr{f} = \Dsr{\left( g\circ f\right) }
		\end{equation} 
	\end{minipage}
\end{lemma}

We spell out the definition of the derivative of a function between some particular Euclidean families below.

\begin{remark}[Explicit derivatives]\label{def:derivative-function-differentiable}
By Definition \ref{def:derivative-of-a-functioonn},  \eqref{eq:def-derivative} in $\Fam{\Vect}$ is differentiable if, for each 
$j\in J $, \eqref{eq:def-derivative-restriction} is differentiable in the usual sense; namely, if \eqref{eq:def-derivative} is the underlying function of 
	a map $\RR ^{n_j}\to \RR ^{m_{\underline{f}(j)}} $ in $\Euclidean $.\\
	\noindent\begin{minipage}{.5\linewidth} 
		\begin{equation} \label{eq:def-derivative}
			f = (\underline{f}, f): \coprod _{j\in J} {\RR }^{n_j}\to  \coprod _{k\in K} {\RR }^{m_k}
		\end{equation} 
	\end{minipage} 
	\begin{minipage}{.5\linewidth} 
		\begin{equation} \label{eq:def-derivative-restriction}
			f_{j} : \RR ^{n_j}\to \RR ^{m_{\underline{f}(j)}}
		\end{equation} 
	\end{minipage}
\end{remark}

Lemma \ref{LEM:CANONICAL-EUCLIDEAN-FAMILIES} shows that all differentiable maps can be expressed in the form specified in \eqref{def:derivative-function-differentiable} through the use of canonical diffeomorphisms.
More precisely, we show that every Euclidean family is canonically diffeomorphic to something of the form $\coprod _{j\in L} {\RR }^{l_j}$.

\begin{definition}[Normal form]\label{DEF:CANONICAL-EUCLIDEAN-FAMILIES}
For each Euclidean family $A\in \Euc $, we inductively define a (possibly infinite) family $\NORMAL{A} = \left( n_j \right) _{j\in J}$ of natural numbers,  and a morphism 
\begin{equation} 
\CANsh{A}: A\to \coprod _{j\in J }\RR ^{n_j} 
\end{equation} 
in $\Fam{\Set }$ by \ref{eq:yet-another-induction1}, \ref{eq:yet-another-induction2} and \ref{eq:yet-another-induction3}.
\begin{enumerate}[$\mathcal{N}$a)]
	\item For each $k\in\NN $,  $\NORMAL{\RR ^k} \coloneqq \RR ^k $ and $\CANsh{\RR ^k}\coloneqq \id _{\RR ^k} $.\label{eq:yet-another-induction1}
	\item Assuming that $(A,B)\in \Euc\times \Euc $,   $\NORMAL{A} = \left( n_j \right) _{j\in J}$ and $\NORMAL{B} = \left( m_l \right) _{l\in L}$, we set  $$\NORMAL{A\times B} \coloneqq \left( n_j + m_l \right) _{\left( j,l\right)\in J\times L} .$$ We define $\CANsh{A\times B} $\label{eq:yet-another-induction2}
	by the morphism given by the composition \eqref{eq:new-induction-canonical-form}, where the unlabeled arrow is the canonical isomorphism induced by the universal property of the product and the distributive property of $\Set $.
\begin{equation}\label{eq:new-induction-canonical-form} 
	\begin{tikzpicture}[x=4cm, y=1cm]
		\node (a) at (0,0) {$A\times B $};
		\node (b) at (1, 0) { $ \coprod\limits _{j\in J }\RR ^{n_j}  \times\coprod\limits _{l\in L }\RR ^{m_l}   $ };
		\node (c) at (2,0) {$   \coprod\limits _{(j,l)\in J\times L }\RR ^{n_j + m_l}   $ };
		\draw[->] (a)--(b) node[midway,above] {$ \CANsh{A}\times \CANsh{B}  $};
		\draw[->] (b)--(c) node[midway,above] {};
	\end{tikzpicture} 
\end{equation}  	
	\item Assuming that $\left( L _ j , A_j^\ast \right) _{j\in J} $ is a family of objects in $\Euc $ such that $\NORMAL{\left( L _ j , A_j^\ast \right) } = \left( m_{(j,l)} \right) _{l\in L_j} $, we set \label{eq:yet-another-induction3}
	$$\NORMAL{\coprod\limits _{j\in J } A _j} \coloneqq \left( m_{t} \right) _{t\in \mathsf{I}} , $$ where
	$ \mathsf{I}\coloneqq \bigcup \limits _ {j\in J} \left\{ j \right\}\times L_j .$
	Finally, we define $\CANsh{\coprod\limits _{j\in J } A _j} $ by the composition \eqref{eq:new-induction-canonical-form-coproducts}  where the unlabeled arrow is the canonical isomorphism induced by the universal property of coproducts.
\end{enumerate}
\begin{equation}\label{eq:new-induction-canonical-form-coproducts} 
	\begin{tikzpicture}[x=4cm, y=1cm]
		\node (a) at (0,0) {$\coprod\limits _{j\in J } A _j  $};
		\node (b) at (1, 0) { $ \coprod\limits _{j\in J } \left(\coprod\limits _{l\in L _j }\RR ^{ m_{(j,l)} }  \right)   $ };
		\node (c) at (2,0) {$   \coprod\limits  _{t\in \mathsf{I}} \RR ^{ m_{t} }   $ };
		\draw[->] (a)--(b) node[midway,above] {$ \coprod\limits _{j\in J } \CANsh{ A _j}   $};
		\draw[->] (b)--(c) node[midway,above] {};
	\end{tikzpicture} 
\end{equation}  
\end{definition}	

It is simple to verify by induction that:

\begin{lemma}[Canonical form of Euclidean families]\label{LEM:CANONICAL-EUCLIDEAN-FAMILIES}
	For every object $A\in \Euc $, $\CANsh{A}$ is a diffeomorphism.
\end{lemma}

By utilizing these normal forms, we are able to establish a valuable characterization of differentiable maps (Lemma \ref{lem:differentiability-derivative-characterization}). This characterization is then leveraged in our logical relations argument, which is detailed in Sections \ref{sec:logical-relations-argument} and \ref{sect:correctness-inductive-types}.

\begin{lemma}\label{lem:differentiability-derivative-characterization}
	Let $f: W\to X$ be a morphism in $\Fam{\Set } $ and
	$\left( g,h\right)  $ a morphism in $\Fam{\Vect}\times \Fam{\Vect ^\op} $.
	Assuming that $W\in \Euc $, we have that
\begin{center} 	
	 $f$ is differentiable and $\left( g,h \right) = \left( \Ds{f} , \Dsr{f} \right)  $     \\ if, and  only if, \\ 
	$f\circ \gamma  $ is differentiable, 
 $ g  \circ \Ds{\gamma } = \Ds{\left( f\circ \gamma \right) } $,  and
	 $h  \circ \Dsr{\gamma } = \Dsr{\left( f\circ \gamma \right) } $
\end{center}
for any differentiable map $\gamma : \RR ^n \to W $  in $\Fam{\Set } $ (where $n$ is any natural number).
\end{lemma} 
\begin{proof}
		It should be noted that one direction follows from chain rule; namely, if $f$ is differentiable and  $\left( g,h \right) = \left( \Ds{f} , \Dsr{f} \right)  $, then $f\circ\alpha $ is differentiable, $ g  \circ \Ds{\alpha} = \Ds{\left( f\circ \alpha \right) } $, and $h  \circ \Dsr{\alpha} = \Dsr{\left( f\circ \alpha \right) } $.

		Reciprocally, we assume that $f$ and $\left( g,h \right)$ are such that 	$f\circ \gamma  $ is differentiable, 
		$ g  \circ \Ds{\gamma } = \Ds{\left( f\circ \gamma \right) } $,  and
		$h  \circ \Dsr{\gamma } = \Dsr{\left( f\circ \gamma \right) } $ for any differentiable map $\gamma : \RR ^n \to W $ in $\Fam{\Set } $.
		
Since $W\in \Euc $, we conclude that so is $X$ since, by hypothesis, we can conclude that there is at least a morphism $W\to X $ that is differentiable. 

By Lemma \ref{LEM:CANONICAL-EUCLIDEAN-FAMILIES}, we have canonical diffeomorphism $$\CANsh{W} : W\to \coprod\limits _{j\in J} {\RR }^{n_j}, \qquad \CANsh{X} : X\to \coprod\limits _{l\in L} {\RR }^{m_l} $$ 
		where $ \left(  n_j \right) _ {j\in J} = \NORMAL{W}$ and $ \left(  m_l \right) _ {l\in L} = \NORMAL{X}$ as defined in \ref{DEF:CANONICAL-EUCLIDEAN-FAMILIES}.

For each $j\in J $, we define $\gamma _ j \coloneqq \CANsh{W}\circ  \ic _{\RR^{n_j} } $ where $$\ic _{\RR^{n_j} }  : \RR^{n_j} \to \coprod\limits _{j\in J} {\RR }^{n_j} $$ is the coproduct coprojection in $\Fam{\Set } $.
By hypothesis, for all $j\in J $,  $$f\circ \CANsh{W}\circ  \ic _{\RR^{n_j} } = f\circ \gamma _ j $$ is differentiable and, hence,   $\CANsh{X}\circ f\circ \CANsh{W}\circ  \ic _{\RR^{n_j} } =\CANsh{X}\circ f\circ \gamma _ j $ is differentiable by the chain rule.
This shows that 
\begin{equation}\label{eq:differentiable-function-componentwise} 
\CANsh{X}\circ f\circ \CANsh{W} : \coprod\limits _{j\in J} {\RR }^{n_j} \to X
\end{equation} 
is componentwise differentiable, that is to say,  
\eqref{eq:differentiable-function-componentwise} 
is such that $$\left( \CANsh{X}\circ f\circ \CANsh{W} \right)  _{j\in J} :  {\RR }^{n_j} \to  {\RR }^{m_{\underline{\CANsh{X}\circ f\circ \CANsh{W}} (j) } } $$ is differentiable for all $j\in J$. 
By Remark \ref{def:derivative-function-differentiable},
we conclude that $\CANsh{X}\circ f\circ \CANsh{W}$ is differentiable. Since $\CANsh{X}$ and $\CANsh{W}$ are diffeomorphisms, this proves that $f$ is differentiable by the chain rule.

Analogously, by using the morphisms $\gamma _ j$ defined above, we conclude that  $\left(\sDs{\CANsh{X}}\circ  g\circ \sDs{\CANsh{W}}  , \sDsr{\CANsh{X}}\circ  h\circ \sDsr{\CANsh{W}} \right) = \left( \sDs{\CANsh{X} \circ f \circ \CANsh{W} }  , \sDsr{\CANsh{X} \circ f \circ \CANsh{W} } \right)  $. Therefore, since $\CANsh{W}$ and $\CANsh{X}$ are diffeomorphisms, $\left( g, h\right)  = \left( \Ds{f}  , \Dsr{f  } \right) $ by the chain rule.
\end{proof}

\subsection{Semantic functors}\label{subsect:concrete-semantics-functors}

We establish the concrete denotational semantics of our languages as suitable structure-preserving functors induced by the respective universal properties.

\subsubsection{The concrete denotational model for the source language}\label{subsect:concrete-semantics-functors-target-language}
Recall that $\Fam{\Set}$ is cartesian closed and has $\mu\nu $-polynomials (Proposition \ref{prop:FAMSET-as-CONCRETE-MODEL}).
By the universal property of the source language $\Syn$ established in Corollary \ref{cor:universal-property-of-source-language}, we can define the semantic functor from $\Syn $ to $\Fam{\Set}$:
\begin{corollary}[Concrete semantics of the source language]
	\label{cor:concrete-semantics-source-language}
We fix the concrete semantics of the ground types and primitive operations of $\Syn $ by 
	defining
	\begin{enumerate}[s-a)]
		\item for each $n$-dimensional array $\reals^n\in\Syn$, 
		$\sem{\reals^n}\defeq \RR^n\in\sobjects{\Fam{\Set}}$ in which $\RR^n $ is the singleton family with $\RR ^n $ as unique member,
		\item for each primitive $\op\in \Op_{n_1,\ldots, n_k}^m$,  $\sem{\op}:\RR^{n_1}\times\cdots\times \RR^{n_k}\To \RR^m$ is the map in $\Fam{\Set} $ corresponding to the function that $\op $ intends to implement.
	\end{enumerate}
By Corollary \ref{cor:universal-property-of-source-language}, we  obtain a unique functor 
	$$
	\sem{-}:\Syn\to \Set
	$$
	that extends these definitions to give a concrete denotational semantics for the entire source language
	such that $\sem{-}$  is a strictly bicartesian closed functor that (strictly) preserves $\mu\nu$-polynomials.
\end{corollary}

\subsubsection{The concrete denotational model for the target language}
We establish the concrete denotational semantics of our target language. Recall that  $\catFV$ is a $\Sigma$-bimodel for tuples, function types, sum types, inductive and coinductive types by Corollary~\ref{coro:catFV-is-a-model-for-target-language}. 

We define the functors \eqref{eq:semantics-target-language-forward-mode} and  \eqref{eq:semantics-target-language-reverse-mode} induced by 
the universal property of $(\CSyn,\LSyn)$ established in Corollary~\ref{cor:universal-property-of-target-language}.

Henceforth, we make use of the terminology and notation established in \ref{subsect:differentiation-seminarcs-subsection}
and \ref{def:derivative-function-differentiable}.

\begin{corollary}[Concrete semantics of the target language]
	\label{cor:concrete-semantics-target-language}
Let	$\FVect : \Set ^\op \to \Cat $	be the $\Sigma$-bimodel for inductive, coinductive and function types $\FVect : \Set ^\op \to \Cat $	established in \ref{subsect:FVECT-concrete-model} (see Corollary \ref{coro:catFV-is-a-model-for-target-language}). We establish the following assignment.
	\begin{enumerate}[t-a)]
		\item for each $n$-dimensional array $\reals^n\in\Syn$, \label{SEM-TARGET1}
		$\overline{\semt{\reals^n} } = \overline{\semtt{\reals^n} }\defeq\sem{\reals^n} \RR^n\in\Set$;
		\item \label{SEM-TARGET2} for each $n$-dimensional array $\reals^n\in\Syn$, 
		$$\underline{\semt{\reals^n} } = \underline{\semtt{\reals^n} }\defeq L _{ \reals^n } \in  \catFV \left( \RR^n \right) $$
		in which $L _{ \reals^n }  =\cRR^n : \RR^n \to \Vect $;
		\item \label{SEM-TARGET3} for each primitive $\op\in \Op_{n_1,\ldots, n_k}^m$:
		\begin{enumerate}[t-i)]
			\item  $\overline{\semt{\op } } = \sem{\op}  :\RR^{n_1}\times\cdots\times \RR^{n_k}\To \RR^m$ is the map in $\Set$ corresponding to the function that $\op$ intends to implement;\label{SEM-TARGET1a}
			\item $f_\op=\sem{D\op }\in\catFV (\RR^{n_1}\times\cdots\times \RR^{n_k}) (\cRR^{n_1}\times\cdots\times \cRR^{n_k}, \cRR^m    ) $ is the family of linear transformations that $D\op$ intends to implement;\label{SEM-TARGETb}
			\item $f_\op^t = \sem{\transpose{\left( D\op\right) } } \in \catFV (\RR^{n_1}\times\cdots\times \RR^{n_k})(\cRR^m ,\cRR^{n_1}\times\cdots\times \cRR^{n_k}   )
			$ is the family of linear transformations that $\transpose{D\op}$ intends to implement.\label{SEM-TARGET1c}
		\end{enumerate}
	\end{enumerate}
By Corollary \ref{cor:universal-property-of-target-language}, we obtain canonical functors 
	\begin{eqnarray}
		&\semt{-}: & \Sigma_{\CSyn}\LSyn  \to\GrothSet\catFV \cong \Fam{\Vect } \label{eq:semantics-target-language-forward-mode} \\ 
		&\semtt{-}: & \Sigma_{\CSyn}\LSyn^{op} \to\GrothSet\catFV ^{op} \cong \Fam{\Vect ^\op}\label{eq:semantics-target-language-reverse-mode}
	\end{eqnarray}
	that extend \ref{SEM-TARGET1}, \ref{SEM-TARGET2} and \ref{SEM-TARGET3} to give a concrete denotational semantics for the entire target language of the forward AD and the reverse AD respectively such that $\semt{-}$ and $\semtt{-} $ are bicartesian closed functors that preserve $\mu\nu$-polynomials.
\end{corollary}

\subsection{Semantic assumptions and specification of CHAD}\label{subsect:CHAD-SOUNDNESS-ASSUMPTIONS}
Although our work applies to more general contexts, \textit{we assume that every primitive operation in the 
source language intends to implement a differentiable function}. We claim that, whenever we have an AD correct macro in this setting, this can be applied to further general cases. For the case of dual-numbers AD, we refer to the revised version of \citep{2022arXiv221007724L}  for comments on general contexts involving non-differentiable functions.

More precisely, for any primitive operation $\op \in \Op_{n_1,\ldots,n_k}^m $ of the source language, we assume that
$$\sem{\op } : \prod_{i=1}^k\RR ^{n_i}\to \RR ^m$$
is differentiable. Moreover, we assume that \eqref{eq:semDop} and \eqref{eq:semDoptranspose} hold.\\
\noindent\begin{minipage}{.5\linewidth} 
	\begin{equation}\label{eq:semDop}
\semt{\Dsyn{\op}} = \Ds{\sem{\op}}, 
	\end{equation} 
\end{minipage}
\begin{minipage}{.5\linewidth} 
	\begin{equation} \label{eq:semDoptranspose}
\semtt{\Dsynrev{\op}} = \Dsr{\sem{\op}}.
	\end{equation} 
\end{minipage}\\ \\ \\
It should be noted that \eqref{eq:semDoptranspose} and \eqref{eq:semDop} hold as long as $\semt{\Dsyn{\op}} =\left( \sem{\op}, \sem{D\op } \right) = \left( \sem{\op}, f_\op \right)  $ and $\semtt{\Dsyn{\op}} =\left( \sem{\op}, \sem{\transpose{D\op} } \right) = \left( \sem{\op}, f_\op ^t \right)  $. In other words, \eqref{eq:semDoptranspose} and \eqref{eq:semDop} hold as long as $D\op $ and $\transpose{D\op} $ implement the family of linear transformations corresponding to the respective derivatives of $\sem{\op}$, as explained 
in Sect~\ref{sec:target-language}

\subsubsection{Specification}
We can inductively define what we mean by \textit{data types} in the source language. These are those types constructed out of ground types, tuples, variant types and inductive types. 

We show  in Sect~\ref{sect:correctness-inductive-types} that the semantics for the inductive data types are rather simple, 
as they are Euclidean families, that is to say, elements of $\Euc $. This shows, by Lemma \ref{LEM:CANONICAL-EUCLIDEAN-FAMILIES}, 
that the semantics of any data type is isomorphic (actually, canonically diffeomorphic) to a (possibly  infinite) coproduct of $\coprod \limits _{j\in J} \RR ^{n_j}$.

In Sect~\ref{sect:correctness-inductive-types}, we prove the full correctness theorem of CHAD for data types. More precisely,
 given any  well-typed program $\var_1:\ty[1]\vdash {\trm}:\ty[2] $ in the source language,  where $\ty[1], \ty[2] $ are data types, we have that: 
 \begin{enumerate}[C1)]
\item $\sem{\ty[1]} $ and $\sem{\ty[2]}$ are Euclidean families;
\item $\sem{t}$ is differentiable; 
\item $\semt{\Dsyn{\trm}}=\Ds{\sem{\trm}} $ and $\semtt{\Dsynrev{\trm}}=\Dsr{\sem{\trm}}$. 
 \end{enumerate}

\section{Sconing}\label{sec:subsconing}
Our approach to categorical semantics for logical relations emphasizes principled constructions of concrete categories from elementary ones, guided by the properties we seek to prove in each setting, \textit{e.g.} \citep[Sect.~4]{2022arXiv221008530L}.
In this section, we introduce the basic categorical framework for our open semantic logical relations proof; namely, we study the \textit{scone}, also called \textit{Artin gluing}.

Recall that, given a functor $G:\catC\to\catD$, 
the \textit{scone} of $G$ is the comma category $\catD\downarrow G$ of the identity along $G$.
Explicitly, the scone's objects are triples $(C_0\in\catD  , C_1\in\catC  , f:C_0\to G(C_1) )$ in which  $f$ is a morphism of $\catD$. Its morphisms $(C_0,C_1, f)\to (C_0', C_1', f')$ are pairs
$(h_0 : C_0\to C_0'  , h_1 : C_1\to C_1')$ such that \eqref{eq:morphism-scone-definition} commutes in $\catD $. 
\begin{equation}\label{eq:morphism-scone-definition} 
	\begin{tikzpicture}[x=4cm, y=1cm]
		\node (a) at (0,0) {$C_0$};
		\node (b) at (1, 0) { $C_0'$ };
		\node (c) at (1,-1) {$G(C_1')$};
		\node (d) at (0,-1) {$ G(C_1) $};
		\draw[->] (a)--(b) node[midway,above] {$ h_ 0  $};
		\draw[->] (b)--(c) node[midway,right] {$ f'  $};
		\draw[->] (a)--(d) node[midway,left] {$ f   $};
		\draw[->] (d)--(c) node[midway,below] {$ G(h_1)  $};
	\end{tikzpicture} 
\end{equation}

The scone $\catD\downarrow G$ inherits  
much of the structure of $\catD\times \catC $. For that reason, under suitable conditions, sconing can be seen as a principled way of building a suitable categorical model from a previously given categorical model $\catD\times\catC $, providing an appropriate semantics for our problem. This is, indeed, the fundamental aspect that underlies our logical relations argument in Sect~\ref{sec:logical-relations-argument} and also in \citep{vakar2021chad, 2022arXiv221008530L, 2022arXiv221007724L}.

In this section, we present our comonadic-monadic approach to studying the properties of $\catD\downarrow G$; it consisting of studying $\catD\downarrow G$ via its comonadicity and monadicity over $\catD\times\catC $. This approach allows us to establish conditions under which $\catD\downarrow G$ has $\mu\nu$-polynomials. The key contribution of this section is twofold: (1) our approach provides a systematic and principled way to understand the nice properties of $\catD\downarrow G$ under suitable conditions; and (2) the conditions we establish for the existence of $\mu\nu$-polynomials are particularly useful for building categorical models for logical relations arguments.

Specifically, our approach shows that the forgetful functor 
\begin{equation}\label{eq:forgetful-scone}
	\forgetfulS : \catD \downarrow G\to \catD\times\catC  
\end{equation}
 is comonadic and, in our case, monadic, and that the properties of $\catD\downarrow G$ can be seen as consequences of this fact.

To lay the groundwork for our approach, we begin by recalling Beck's Monadicity Theorem, since Theorem \ref{theo:basic-comonadicity-sconing}  holds a fundamental place in our approach. The original statement of this theorem involves split (co)equalizers; see, for instance, \cite[Theorem3.14]{MR2178101} or \cite[TheoremII.2.1]{zbMATH03362059}  for the enriched case. However, for our purposes, we will make use of a slightly modified version; namely \textit{a left adjoint functor is comonadic if and only if it creates absolute limits}. This version can be found, for instance, in \citep[pag.~550]{MR4266479}.

\begin{therm}\label{theo:basic-comonadicity-sconing}
	If $\catD $ has binary products, then \eqref{eq:forgetful-scone} is comonadic. 
\end{therm}	
\begin{proof}
By the universal property of comma categories, a diagram $D: \catS\to \catD\downarrow G  $ corresponds biunivocally
with triples 
\begin{equation}\label{eq:diagram-in-scone}
\left( D_0 :\catS\to \catD, D_1 :\catS\to \catC , \mathfrak{d} : D_0 \rightarrow GD_1       \right) 
\end{equation}
in which $D_0 , D_1 $ are diagrams and $\mathfrak{d} $ is a natural transformation. In this setting, it is clear that, assuming that 
$\lim D_0 $ exists, if $   \lim D_1 $ exists and is preserved by $G$, we have that
\begin{equation}\label{eq:description-lim-comma-category}
	\left( \lim D_0, \lim D_1,  \lim D_0 \xrightarrow{\mathsf{d} } \lim \left( G\circ D_1\right) \xto{\cong } G\left( \lim D_1 \right)    \right), 
\end{equation}
is the limit of $D$ in $\catD \downarrow G $, in which $\mathsf{d} $ is the morphism induced by the natural transformation $\mathfrak{d}$.

Now, given a diagram $D : \catS\to \catD\downarrow G $ such that $\forgetfulS\circ D = \left( D_ 0 ,  D_1\right) : \catS \to \catD\times \catC   $ has an absolute limit, we get that  $\lim D_0 $ and $  \lim D_1 $
exist and are preserved by any functor. Hence, by the observed above, in this case, the limit of $D$ exists and is given by 
\eqref{eq:description-lim-comma-category}. Thus it is preserved by $\forgetfulS$.
Since \eqref{eq:forgetful-scone} is conservative, this completes the proof that \eqref{eq:forgetful-scone} creates absolute limits.

Finally, since \eqref{eq:forgetful-scone}  has a right adjoint defined by 
$$(Y,X)\mapsto \left( Y\times G(X), X, \pi _2 : Y\times G(X)\to G(X)    \right) ,$$
the proof that \eqref{eq:forgetful-scone}  is comonadic is complete 
by Beck's Monadicity Theorem.
\end{proof}
\begin{remark}
	If $\catC $ has a terminal object and $\catD $ has binary products as above,  \eqref{eq:forgetful-scone} is comonadic and, furthermore, the comonad induced by it is the free comonad
	over the endofunctor on $\catD\times \catC $ defined by $(Y,X)\mapsto \left( G(X), \terminal \right)  $.
\end{remark}

\begin{corollary}\label{coro:basic-comonadicity-monadicity-sconing}
	Assume that $\catC $ has binary coproducts and $\catD $ has binary products. We have that \eqref{eq:forgetful-scone} is comonadic and monadic provided that $G$ has a left adjoint $F$.
\end{corollary}
\begin{proof}
	Firstly, of course, by Theorem \ref{theo:basic-comonadicity-sconing}, we have that \eqref{eq:forgetful-scone} is comonadic.
	Secondly,  by the dual of Theorem \ref{theo:basic-comonadicity-sconing}, we have that the forgetful
	functor $ F \downarrow \catC  \to \catC\times\catD $ is monadic. Hence, since
	\begin{equation*}
		\begin{tikzpicture}[x=4cm, y=1cm]
			\node (a) at (0,0) {$ \catD \downarrow G  $};
			\node (b) at (2, 0) { $\catD\times\catC $ };
			\node (c) at (1,-1) {$    F \downarrow \catC  $ };
			\draw[->] (a)--(b) node[midway,above] {$\forgetfulS $};
			\draw[->] (c)--(b);
			\draw[<->] (a)--(c) node[midway,below left] {$ \cong   $};
		\end{tikzpicture} 
	\end{equation*}  
commutes, we get that $\forgetfulS $ is monadic as well.
\end{proof}

Indeed, in our case, all the properties of the scone we are interested in follow from 
the comonadicity and monadicity of \eqref{eq:forgetful-scone}, that is to say, Corollary \ref{coro:basic-comonadicity-monadicity-sconing}.\footnote{Some of the results presented here hold under slightly more general conditions. But we chose to make the most of our setting, which is general enough for our proof and many others cases of interest.}

\subsection{Bicartesian structure of the scone}\label{subsect:bicartesian-scone}
The bicartesian closed structure of the scone $\catD \downarrow G$ follows from 
the well known result about monadic functors and creation of limits. Namely:

\begin{proposition}\label{prop:creation-limits-monadicity}
	Monadic functors create all limits. Dually, comonadic functors create all colimits. 
\end{proposition}
\begin{proof}
	See, for instance, \citep[Section~1.4]{MR2056584}.
\end{proof}	

As a corollary, then, we have the following explicit constructions.

\begin{corollary}
	Assuming that $\catC $ and $\catD $ have finite products and finite coproducts,
	if $G : \catC\to\catD $ has a left adjoint, then 
$	\forgetfulS : \catD \downarrow G\to \catD\times\catC $ creates limits and colimits.
	In particular, $\catD \downarrow G $ is bicartesian and, in the case $\catD\times\catC$ is a distributive category, so is $\catD \downarrow G$.
\end{corollary}
\begin{proof}
Given a diagram $D : \catS\to \catD\downarrow G $, we have that it is uniquely determined by a triple $ \left( D_0 :\catS\to \catD, D_1 :\catS\to \catC , \mathfrak{d} : D_0 \rightarrow GD_1       \right) $
like in \eqref{eq:diagram-in-scone}. In this case, we have that:	
\begin{enumerate}
	\item In the proof of Theorem \ref{theo:basic-comonadicity-sconing}, we implicitly addressed the problem of creation of limits that are preserved by $G$. Since $G$ has a left adjoint, it preserves all the limits and, hence, all the limits are created like \eqref{eq:description-lim-comma-category2}. 
	
	More precisely, assuming that $\forgetfulS\circ D = \left( D_ 0 ,  D_1\right) : \catS \to \catD\times \catC   $ has a limit, we get that both $\lim D_0$  and $ \lim D_1$
	exist, since the projections $\catD\times \catC\to\catD $ and $\catD\times\catC\to\catC $ have left adjoints (because $\catC $ and $\catD $ have initial objects). 
	
	Since $G$ has a left adjoint, it preserves  the limit of $D_1$.  Hence, the limit of $D$ is given by 
\begin{equation}\label{eq:description-lim-comma-category2}
	\left( \lim D_0, \lim D_1,  \lim D_0, \xrightarrow{\mathsf{d} } \lim \left( G\circ D_1\right) \xto{\cong } G\left( \lim D_1 \right)    \right), 
\end{equation}
like in \eqref{eq:description-lim-comma-category}, in which $ \mathsf{d} $ is the morphism induced by $\mathfrak{d}$ and $\lim \left( G\circ D_1\right) \cong G\left( \lim D_1 \right) $ comes from the fact that $G$ preserves limits.
\item Assuming that $\forgetfulS\circ D = \left( D_ 0 ,  D_1\right) : \catS \to \catD\times \catC   $ has a limit, we get that both $\colim D_0$  and $ \colim D_1$ exist. In this case, the colimit of $D$ is given by
\begin{equation}\label{eq:description-colim-comma-category}
	\left( \colim D_0, \colim D_1,  \colim D_0 \xrightarrow{\mathsf{d} } \colim \left( G\circ D_1\right) \to G\left( \colim D_1 \right)    \right), 
\end{equation}
in which $\colim \left( G\circ D_1\right) \to G\left( \colim D_1 \right)  $ is the induced comparison.
\end{enumerate} 	
\end{proof}

\begin{remark}
	It will be particularly important for our correctness proof in Section \ref{sect:correctness-inductive-types}  that $ \catD \downarrow G $  has infinite coproducts
	whenever $\catC $ and $\catD $ have finite products and infinite coproducts. This is a consequence of the fact stated above.
\end{remark}

\subsection{Monadic-comonadic functors and the cartesian closedness of the scone}\label{subsect:closed-structure-comonadic-monadic}
Under the conditions of our proof, the scone $\catD \downarrow G$ is cartesian closed. 
In our case, we can see as a consequence of the well known result below.

\begin{proposition}\label{prop:closed-structure-comonadic-monadic}
	If a category is monadic-comonadic over a finitely complete cartesian closed category, then it is 
	finitely complete cartesian closed as well. 
	
	More precisely, if $\catD $ is finitely complete and $G:\catC\to\catD $ is monadic and comonadic, then $G$ reflects exponentiable objects.
\end{proposition}	
\begin{proof}
	See, for instance, a slightly more general version in \citep[Theorem~1.8.2]{2018arXiv180201767L}.
	Indeed, assuming that $G:\catC\to\catD $ is monadic and comonadic and that $\catD $ is finitely complete, 
	we get that $\catC $ is finitely complete as well and, moreover, $G$ preserves them (since monadic functors create limits).
	
Denoting the right adjoint of $G$ by $H$, given an object $W\in \catC $,  we have an isomorphism 
		\begin{equation}\label{eq:commutative-diagram-closed-structure-comonadicity}
		\begin{tikzpicture}[x=4cm, y=0.5cm]
			\node (a) at (0,0) {$ \catC  $};
			\node (b) at (2, 0) { $\catC $ };
			\node (c) at (2,-2) {$\catD $ };
			\node (d) at (0,-2) {$\catD $ };
			\node (e) at (1,-1) {$\cong $ };
			\draw[->] (a)--(b) node[midway,above] {$\left( W\times -\right)  $};
			\draw[->] (d)--(c) node[midway,below] {$ \left( G (W)\times - \right)  $};
			\draw[->] (a)--(d) node[midway,left] {$ G  $};
			\draw[->] (b)--(c) node[midway,right] {$ G  $};
		\end{tikzpicture} 
	\end{equation} 
If $G(W) $ is exponentiable, we know that 	$ \left( G(W)\times G(-)\right)  \dashv H \left( G(W)\Rightarrow -\right)   $. Since 
  $\catC$ has equalizers and $G$ is comonadic, we get that $\left( W\times - \right) $ has a right adjoint by Dubuc's adjoint triangle theorem.\footnote{The original result on adjoint triangles was proven in \citep{MR0233864}. Further comments and generalizations are given in \citep{zbMATH06970806}, while a precise statement for our case is given in \citep[Corollary~1.2]{MR3491845}.} That is to say, $W$ is exponentiable.
\end{proof}

Explicitly, we get:

\begin{corollary}\label{coro:closed-structure-comonadic-monadic-special-case}
Let  $\catC$ and $\catD$ be finitely complete cartesian closed categories. If $G: \catC\to\catD $ has a left adjoint,
we get that $\catD \downarrow G$ is finitely complete cartesian closed. More precisely, the exponential in $\catD \downarrow G$ 
is given by \eqref{eq:exponential-scone} where we write $f\Rightarrow f'$ for the Pullback \eqref{eq:commutative-diagram-closed-structure-pullback}. 
\begin{equation} \label{eq:exponential-scone}
(C_0 , C_1, f )\Rightarrow (D_0, D_1, f') \, = \, (P, C_1\Rightarrow D_1, f\Rightarrow f')
\end{equation} 
\begin{equation}\label{eq:commutative-diagram-closed-structure-pullback}
	\begin{tikzpicture}[x=4cm, y=0.8cm]
		\node (a) at (0,0) {$ P $};
		\node (b) at (1.8, 0) { $ C_0\Rightarrow D_0 $ };
		\node (c) at (1.8,-2) {$C_0\Rightarrow G(D_1) $ };
		\node (d) at (0,-2) {$G(C_1\Rightarrow D_1 )  $ };
		\node (e) at (0.8,-2) {$ G(C_1)\Rightarrow G( D_1 )  $ };
		\draw[->] (a)--(b);
		\draw[->] (d)--(e);
		\draw[->] (e)--(c) node[midway,below] {$ f\Rightarrow G( D_1 )  $};
		\draw[->] (a)--(d) node[midway,left] {$ f\Rightarrow f'  $};
		\draw[->] (b)--(c) node[midway,right] {$ C_0\Rightarrow f'  $};
	\end{tikzpicture} 
\end{equation} 
\end{corollary}

\subsection{Monadic functors create terminal coalgebras of compatible endofunctors}

Recall the definition of preservation, reflection and creation of initial algebras and terminal coalgebras, see Definitions \ref{def:preservation-of-initial-algebras-of-a-specific-endofunctor} and \ref{def:preservation-of-terminal-coalgebras-specific-endofunctor}.
We prove and establish the result that says that monadic functors create initial algebras, while, dually, comonadic functors
create terminal coalgebras. 

We first establish the fact that left adjoint functors preserve initial algebras and, dually, right adjoint functors preserve terminal coalgebras. In order to do so, we start by observing that:

\begin{lemma}\label{lem:induced-adjunctions-between-coalgebras-dual}
	Let
	\begin{equation*}
		\begin{tikzcd}
			\catC \arrow[bend right=15, rrrr, swap, "{G}"] &&\bot (\varepsilon , \eta ) && \catD \arrow[bend right=15, llll, swap, "{F}"]                            
		\end{tikzcd}
	\end{equation*}	
	be an adjunction. Assume that $ \gamma : E\circ F \cong F\circ E' $ is a natural isomorphism in which $E$ and $E ' $ are endofunctors. 
 In this case, we have an induced adjunction
	\begin{equation}
		\begin{tikzcd}
			E\AAlg \arrow[bend right=15, rrrr, swap, "{\underline{\hat{G}} _ \gamma  }"] &&\bot ( \underline{\hat{\varepsilon } }, \underline{\hat{\eta }} ) && E'\AAlg \arrow[bend right=15, llll, swap, "{\check{F}_ \gamma  }"]                            
		\end{tikzcd}
	\end{equation}	
	in which $\check{F} _\gamma  $ is defined as in Definition~\ref{def:preservation-of-initial-algebras-of-a-specific-endofunctor} and $\underline{\hat{G}} _ \gamma  $
	is defined as follows:
	\begin{eqnarray*}
		\underline{\hat{G}} _ \gamma   : &E\AAlg                & \to E'\AAlg  \\
		&\left( Y, \xi  \right) & \mapsto  \left( G(Y) , G\left( \xi \right)\circ   GE \left( \varepsilon _{Y}\right)\circ G\left( \gamma ^{-1} _ {G(Y) } \right)   \circ   \eta _{E'G (Y) } \right)\\
		& f & \mapsto G(f).   
	\end{eqnarray*}	
\end{lemma}
\begin{proof}
	In fact, the counit and unit, $\underline{\hat{\varepsilon }} , \underline{\hat{\eta }} $, are defined pointwise by the original counit and unit. That is to say,
	$  \underline{\hat{\varepsilon }} _{(Y, \xi )} = \varepsilon _ Y$  and $ \underline{\hat{\eta }} _{(W, \zeta )} = \eta _ W $.
\end{proof}

\begin{remark}[Doctrinal adjunction]
	The right adjoint $\underline{\hat{G}} _\gamma $ does not come out of the blue. The association $\left( F, \gamma\right) \mapsto \check{F}\gamma $ in Lemma \ref{lem:comparsion-of-algebras-basic-lemma} is part of a $2$-functor, with the domain being the $2$-category of endomorphisms in $\Cat $, lax natural transformations and modifications, and the codomain being $\Cat$. By the doctrinal adjunction,\footnote{For the original statement, please refer to \citep{MR0360749}. For the general case of lax algebras, see, for instance,  \citep[Corollary~1.4.15]{2018arXiv180201767L}.} we know that whenever $\left( F, \gamma \right) $ is pseudonatural (i.e., $\gamma$ is invertible) and $F$ has a right adjoint in $\Cat$, the pair $\left( F, \gamma\right) $ has a right adjoint $\left( G, \left( GE \varepsilon\right)\cdot \left( G \gamma ^{-1} _ G\right)\cdot \left( \eta {E'G } \right) \right) $ in the $2$-category of endofunctors. Therefore, since $2$-functors preserve adjunctions, we obtain that $\check{F}\gamma $ has a right adjoint given by $\check{G}_{\left( GE \varepsilon\right)\cdot \left( G \gamma ^{-1} _ G\right)\cdot \left( \eta _{E'G } \right)} $, denoted by $\underline{\hat{G}} _\gamma $, whenever $\gamma $ is invertible and $F$ has a right adjoint. 
\end{remark}

The dual of Lemma~\ref{lem:induced-adjunctions-between-coalgebras-dual} is given by:

\begin{lemma}\label{lem:induced-adjunctions-between-coalgebras}
Let
\begin{equation*}
	\begin{tikzcd}
		\catC \arrow[bend right=15, rrrr, swap, "{G}"] &&\bot (\varepsilon , \eta ) && \catD \arrow[bend right=15, llll, swap, "{F}"]                            
	\end{tikzcd}
\end{equation*}	
be an adjunction. Assume that $ \beta : G\circ E\cong E '\circ G $
is a natural isomorphism in which $E$ and $E'$ are endofunctors. In this case, we have an induced adjunction
\begin{equation}
	\begin{tikzcd}
		E\CCoAlg \arrow[bend right=15, rrrr, swap, "{\tilde{G} ^\beta  }"] &&\bot (\hat{\varepsilon } , \hat{\eta } ) && E'\CCoAlg \arrow[bend right=15, llll, swap, "{\hat{F}^\beta }"]                            
	\end{tikzcd}
\end{equation}	
in which $\tilde{G} ^\beta $ is defined as in \ref{def:preservation-of-terminal-coalgebras-specific-endofunctor} and $\hat{F} ^\beta $
is defined as follows:
\begin{eqnarray*}
	\hat{F} : &E'\CCoAlg                & \to E\CCoAlg  \\
	&\left( W, \zeta  \right) & \mapsto  \left( W, \varepsilon _{EF(W)}\circ F(\beta _{F(W) } ^{-1} )   \circ FE'\left( \eta _ W  \right)\circ  F\left( \zeta \right)   \right)\\
	& g & \mapsto F(g).   
\end{eqnarray*}	
\end{lemma}

As an immediate consequence, we have that:

\begin{therm}\label{theo:right-adjoint-preserves-terminal-coalgebras}
	Right adjoint functors preserve terminal coalgebras. Dually, left adjoints preserve initial algebras.
\end{therm}
\begin{proof}
	Let $G:\catC\to\catD $ be a functor and $ \beta : G\circ E\cong E '\circ G $ a natural isomorphism in which $E, E'$ are endofunctors. 
	If  $F\dashv G $, we get that $\tilde{G} ^\beta  : E\CCoAlg\to E'\CCoAlg $ (as defined in \ref{def:preservation-of-terminal-coalgebras-specific-endofunctor}) has a left adjoint by Lemma \ref{lem:induced-adjunctions-between-coalgebras}. Therefore $\tilde{G} ^\beta $ preserves limits and, in particular, terminal objects. This completes the proof that $G$ preserves terminal coalgebras (see Definition \ref{def:preservation-of-terminal-coalgebras-specific-endofunctor}).
\end{proof}

Finally, we can state the result about monadic functors; namely:
\begin{therm}\label{theo:creation-of-initial-algebras}
	Monadic functors create terminal coalgebras. Dually, comonadic functors create initial algebras.
\end{therm}

\begin{proof}
Let $ G: \catC\to\catD  $ be a monadic functor.
Assume that $\beta : G\circ E \cong E' \circ G  $ is a natural isomorphism in which $E, E '$ are endofunctors.

We have that $\tilde{G} ^\beta  : E\CCoAlg\to E'\CCoAlg $ (as defined in \ref{def:preservation-of-terminal-coalgebras-specific-endofunctor})  has a left adjoint by Lemma \ref{lem:induced-adjunctions-between-coalgebras}. Moreover, 
we have the commutative diagram 
	\begin{equation}\label{eq:commutative-diagram-preservation-terminal-coalgebras}
		\begin{tikzpicture}[x=4cm, y=0.5cm]
			\node (a) at (0,0) {$ E\CCoAlg $};
			\node (b) at (2, 0) { $E'\CCoAlg $ };
			\node (c) at (2,-2) {$\catD $ };
			\node (d) at (0,-2) {$\catC $ };
			\draw[->] (a)--(b) node[midway,above] {$\tilde{G} ^\beta $};
			\draw[->] (d)--(c) node[midway,below] {$ G $};
			\draw[->] (a)--(d);
			\draw[->] (b)--(c);
		\end{tikzpicture} 
	\end{equation} 
	in which the vertical arrows are the forgetful functors.
	
	Since we know that all the functors in \eqref{eq:commutative-diagram-preservation-terminal-coalgebras} but $\tilde{G} ^\beta  $ create absolute colimits, we conclude that $\tilde{G} ^\beta $ creates absolute colimits as well. Therefore $\tilde{G} ^\beta $ is monadic and, thus, it creates all limits. In particular, $\tilde{G} ^\beta  $ creates terminal objects. This completes the proof that $G$ creates terminal coalgebras (see Definition \ref{def:preservation-of-terminal-coalgebras-specific-endofunctor}).
\end{proof}

\subsection{Monadic-comonadic functors create $\mu\nu$-polynomials}
We establish that monadic-comonadic functors create $\mu\nu$-polynomials below, a crucial result for our approach to the study of $\mu\nu$-polynomials in the scone.

\begin{corollary}\label{coro:creating-munu-polynomials}
	Monadic-comonadic functors create $\mu\nu$-polynomials. More precisely, if $G: \catA\to\catB $ is monadic-comonadic and $\catB $ has $\mu\nu $-polynomials, then
\begin{enumerate}	
	\item $G$ creates products and coproducts;
	\item $\catA $ has $\mu\nu $-polynomials;
	\item for each $\mu\nu$-polynomial endofunctor $E$ on $\catA $, there is a $\mu\nu $-polynomial endofunctor $\overline{\underline{E}} $ on $\catB $
	such that $G\circ E \cong \overline{\underline{E}}\circ G $ (and $G$ creates the initial algebra and the terminal coalgebra of $E$).
\end{enumerate}
\end{corollary}	
\begin{proof}
	Let $ G : \catA\to\catB $ be a monadic-comonadic functor in which $\catB $ has $\mu\nu$-polynomials. We inductively define the set $\productclosure{G}$ as follows:
\begin{enumerate}[$\productclosure{G}$1)]
	\item the identity functor $\terminal \to \terminal $ belongs to $\productclosure{G}$;
	\item $G : \catA\to\catB $ belongs to $\productclosure{G}$;
	\item if $G ' : \catA '\to\catB '$  and $G '' : \catA ''\to\catB ''  $ belong $\productclosure{G}$, then so does the product $G '\times G '' : \catA '\times \catA ''\to\catB '\times \catB ''  $.
\end{enumerate}

We have bijections \eqref{eq:dom-definition}  and \eqref{eq:codom-definition}  inductively defined by \eqref{eq:codom-definition1}, \eqref{eq:codom-definition2} and \eqref{eq:codom-definition3}.\\
\noindent\begin{minipage}{.5\linewidth} 
	\begin{equation} \label{eq:dom-definition}
\mathrm{dom} : \productclosure{G}\to \objects \left( \mnPoly _ \catA \right) 
	\end{equation} 
\end{minipage} 
\begin{minipage}{.5\linewidth} 
	\begin{equation} \label{eq:codom-definition}
\mathrm{codom} : \productclosure{G}\to \objects \left( \mnPoly _ \catB \right) 
	\end{equation} 
\end{minipage}\\
\begin{enumerate}[bi)]
\item  \label{eq:codom-definition1} $ \mathrm{codom}\left( \terminal \to\terminal \right) = \mathrm{dom}\left( \terminal \to\terminal \right) =\terminal $; 
\item \label{eq:codom-definition2} $\mathrm{dom}\left( G \right) = \catA $ and $\mathrm{codom}\left( G \right) = \catB $; 
\item \label{eq:codom-definition3} $ \mathrm{dom}\left( G'\times G '' \right) =    \mathrm{dom}\left( G'\right) \times \mathrm{dom}\left( G '' \right) $ and $ \mathrm{codom}\left( G'\times G '' \right) =    \mathrm{codom}\left( G'\right) \times \mathrm{codom}\left( G '' \right) $.
\end{enumerate} 
In other words, the function  $\mathrm{dom} : \productclosure{G}\to \objects \left( \mnPoly _ \catA \right)  $ and
$\mathrm{codom} : \productclosure{G}\to \objects \left( \mnPoly _ \catB \right)  $ 
 give, respectively, the domain and codomain of each functor in $\productclosure{G}$.

Since $G$ creates initial algebras and terminal coalgebras, it is enough to show that, for any $\mu\nu$-polynomial
$H : \catC\to\catD $ in $\mnPoly _ \catA $, there is a morphism $ \underline{\overline{H}} $ of $\mnPoly _ \catB $ such that there is  an isomorphism
\begin{equation}\label{eq:isomorphism-for-each-polynomial}
	\begin{tikzpicture}[x=3.5cm, y=0.6cm]
		\node (a) at (0,0) {$ \catC  $};
		\node (b) at (2, 0) {$ \catD  $ };
		\node (c) at (2,-2) {$ \underline{\overline{\catD }} $ };
		\node (d) at (0,-2) {$\underline{\overline{\catC }}$ };
		\node (e) at (1,-1) {$\xleftrightarrow{\cong } $ };
		\draw[->] (a)--(b) node[midway,above] {$ H $};
		\draw[->] (d)--(c) node[midway,below] {$ \underline{\overline{H}} $};
		\draw[->] (a)--(d) node[midway,left] {$ \mathrm{dom} ^{-1} \left( \catC \right) $};
		\draw[->] (b)--(c) node[midway,right] {$ \mathrm{dom} ^{-1} \left( \catD \right) $};
	\end{tikzpicture} 
\end{equation} 
where $\underline{\overline{\catD }} : = \mathrm{codom}\circ  \mathrm{dom} ^{-1} \left( \catD \right)$  and 
$\underline{\overline{\catD }} : = \mathrm{codom}\circ  \mathrm{dom} ^{-1} \left( \catC \right)$.  

We start by proving that the objects of $\mnPoly _ \catA $ together with the functors that satisfy the property above 
do form a subcategory of $\Cat $. Indeed, observe that the identities do satisfy the condition above, since it is always true that
$$ \id_{\underline{\overline{\catC }} } \circ  \mathrm{dom} ^{-1} \left( \catC \right) =    \mathrm{dom} ^{-1} \left( \catC \right) \circ \id _{\catC } 
$$
for any given object $\catC $ of $\mnPoly _ \catA $. Moreover, given morphisms $J : \catD '' \to\catD ''' $ and $E : \catD ' \to\catD '' $
of  $\mnPoly _ \catA $ such that we have natural isomorphisms 
\begin{eqnarray*}
	\gamma : & \overline{\underline{ E } } \circ \mathrm{dom} ^{-1} \left(  \catD '  \right) & \cong \mathrm{dom}^{-1} \left(  \catD ''\right) \circ E\\
	\gamma ' : & \overline{\underline{ J } } \circ \mathrm{dom} ^{-1} \left(  \catD ''  \right) & \cong  \mathrm{dom} ^{-1} \left(  \catD '''  \right) \circ J    
\end{eqnarray*}
in which $ \overline{\underline{ J } }  $ and $\overline{\underline{ E } } $ are morphisms of $\mnPoly _ \catB $, we have that
\begin{equation}
	\begin{tikzpicture}[x=2.5cm, y=1cm]
		\node (a) at (0,0) {$ \catD '  $};
		\node (b) at (2, 0) {$ \catD '' $ };
		\node (bb) at (4, 0) {$ \catD ''' $ };
		\node (c) at (2,-2) {$ \underline{\overline{\catD '' }} $ };
		\node (d) at (0,-2) {$\underline{\overline{\catD ' }}$ };
		\node (dd) at (4,-2) {$\underline{\overline{\catD ''' }}$ };
		\node (e) at (1,-1) {$\xleftrightarrow{\gamma } $ };
		\node (e) at (3.2,-1) {$\xleftrightarrow{\gamma ' } $ };
		\draw[->] (a)--(b) node[midway,above] {$ E $};
		\draw[->] (d)--(c) node[midway,below] {$ \underline{\overline{E}} $};
		\draw[->] (a)--(d) node[midway,left] {$ \mathrm{dom} ^{-1} \left( \catD ' \right) $};
		\draw[->] (b)--(c) node[midway,right] {$ \mathrm{dom} ^{-1} \left( \catD ''\right) $};
		\draw[->] (b)--(bb) node[midway,above] {$ J $};
		\draw[->] (c)--(dd) node[midway,below] {$ \underline{\overline{J}} $};
		\draw[->] (bb)--(dd) node[midway,right] {$ \mathrm{dom} ^{-1} \left( \catD ''' \right) $};
	\end{tikzpicture} 
\end{equation} 
is a natural isomorphism and $\overline{\underline{ J } } \circ \overline{\underline{ E } }   $ is a morphism in $\mnPoly _ \catB $.

Finally, we complete the proof that all the morphisms of $\mnPoly _ \catA $ satisfy the property above by proving by induction over the Definition \ref{def:basic-definition-munupolynomials} of $\mnPoly _ \catA $.
\begin{enumerate}[M1)]
	\item for any object $\catC$ of $\mnPoly _ \catA $, the unique functor 
	$\catC\to \terminal $ is such that 
\begin{equation}
	\begin{tikzpicture}[x=1cm, y=0.7cm]
		\node (a) at (0,0) {$ \catC  $};
		\node (b) at (2, 0) {$ \terminal  $ };
		\node (c) at (2,-2) {$ \terminal $ };
		\node (d) at (0,-2) {$\underline{\overline{\catC }} $ };
		\draw[->] (a)--(b);
		\draw[->] (d)--(c);
		\draw[->] (a)--(d) node[midway,left] {$ \mathrm{dom} ^{-1} \left( \catC \right) $};
		\draw[->] (b)--(c) node[midway,right] {$ \mathrm{dom} ^{-1} \left( \terminal \right) $};
	\end{tikzpicture} 
\end{equation} 	
	commutes and, of course, $\underline{\overline{\catC }}\to \terminal $ is a morphism in $\mnPoly _ \catB $; 
	\item for any object $\catD $ of $\mnPoly _ \catA $,  given a functor $W : \terminal \to \catD $ (which belongs to $\mnPoly _ \catA $), we have that $\mathrm{dom} ^{-1} \left( \catD  \right) \circ W  $ is a morphism of $\mnPoly _ \catB $ such that
\begin{equation}
	\begin{tikzpicture}[x=4cm, y=0.7cm]
		\node (a) at (0,0) {$ \terminal  $};
		\node (b) at (2, 0) {$ \catD  $ };
		\node (c) at (2,-2) {$ \underline{\overline{\catD }}  $ };
		\node (d) at (0,-2) {$\terminal  $ };
		\draw[->] (a)--(b) node[midway,above] {$ W $};
		\draw[->] (d)--(c) node[midway,below] {$ \mathrm{dom} ^{-1} \left( \catD  \right) \circ W    $};
		\draw[->] (a)--(d) node[midway,left] {$ \mathrm{dom} ^{-1} \left( \terminal \right)    $};
		\draw[->] (b)--(c) node[midway,right] {$ \mathrm{dom} ^{-1} \left( \catD \right)  $};
	\end{tikzpicture} 
\end{equation} 	
	commutes;
	\item  consider the binary product $\times : \catA \times\catA   \to \catA $ (which exists, since $G$ is monadic). We have that $\times : \catB \times\catB   \to \catB $ (which is a morphism of $\mnPoly _ \catB $) is such that we have an isomorphism
\begin{equation}
	\begin{tikzpicture}[x=2cm, y=0.7cm]
		\node (a) at (0,0) {$ \catA\times \catA   $};
		\node (b) at (2, 0) {$ \catA  $ };
		\node (c) at (2,-2) {$ \catB  $ };
		\node (d) at (0,-2) {$\catB\times\catB   $ };
		\node (z) at (1,-1) {$\xleftrightarrow{\cong } $ };
		\draw[->] (a)--(b) node[midway,above] {$ \times  $};
		\draw[->] (d)--(c) node[midway,below] {$ \times    $};
		\draw[->] (a)--(d) node[midway,left] {$ \mathrm{dom} ^{-1} \left( \catA\times \catA   \right)    $};
		\draw[->] (b)--(c) node[midway,right] {$ \mathrm{dom} ^{-1} \left( \catA \right)  $};
	\end{tikzpicture} 
\end{equation} 	
since $G: \catA\to\catB $ preserves products and
\begin{center}
$\mathrm{dom} ^{-1}\left( \catA \right) = G  $, \qquad $\mathrm{dom} ^{-1}\left( \catA\times \catA  \right) = G\times G $;
\end{center} 
	\item  consider the binary coproduct $\sqcup : \catA \times\catA   \to \catA $ (which exists, since $G$ is comonadic). We have that $\sqcup : \catB \times\catB   \to \catB $ (which is a morphism of $\mnPoly _ \catB $) is such that we have an isomorphism
\begin{equation}
	\begin{tikzpicture}[x=2cm, y=0.7cm]
		\node (a) at (0,0) {$ \catA\times \catA   $};
		\node (b) at (2, 0) {$ \catA  $ };
		\node (c) at (2,-2) {$ \catB  $ };
		\node (d) at (0,-2) {$\catB\times\catB   $ };
		\node (z) at (1,-1) {$\xleftrightarrow{\cong } $ };
		\draw[->] (a)--(b) node[midway,above] {$ \sqcup  $};
		\draw[->] (d)--(c) node[midway,below] {$ \sqcup    $};
		\draw[->] (a)--(d) node[midway,left] {$ G\times G = \mathrm{dom} ^{-1} \left( \catA\times \catA   \right)    $};
		\draw[->] (b)--(c) node[midway,right] {$ G = \mathrm{dom} ^{-1} \left( \catA \right)  $};
	\end{tikzpicture} 
\end{equation} 	
since $G: \catA\to\catB $ preserves coproducts.
	\item for any pair of objects  $\left( \catC, \catD  \right) \in \mnPoly _ \catA\times \mnPoly _ \catA $,   we have, of course, that
\begin{equation}
	\begin{tikzpicture}[x=1cm, y=0.7cm]
		\node (a) at (0,0) {$ \catC\times \catD   $};
		\node (b) at (2, 0) {$ \catC  $ };
		\node (c) at (2,-2) {$ \underline{\overline{\catC }} $ };
		\node (d) at (0,-2) {$\underline{\overline{\catC }}\times \underline{\overline{\catD }}  $ };
		\draw[->] (a)--(b) node[midway,above] {$ \pi _1  $};
		\draw[->] (d)--(c) node[midway,below] {$ \pi _1    $};
		\draw[->] (a)--(d) node[midway,left] {$ \mathrm{dom} ^{-1} \left( \catC\times \catD   \right)    $};
		\draw[->] (b)--(c) node[midway,right] {$  \mathrm{dom} ^{-1} \left( \catC \right)  $};
		\node (aa) at (7,0) {$ \catC\times \catD   $};
		\node (bb) at (9, 0) {$ \catD  $ };
		\node (cc) at (9,-2) {$ \underline{\overline{\catD }} $ };
		\node (dd) at (7,-2) {$\underline{\overline{\catC }}\times \underline{\overline{\catD }}  $ };
		\draw[->] (aa)--(bb) node[midway,above] {$ \pi _2  $};
		\draw[->] (dd)--(cc) node[midway,below] {$ \pi _2    $};
		\draw[->] (aa)--(dd) node[midway,left] {$ \mathrm{dom} ^{-1} \left( \catC\times \catD   \right)    $};
		\draw[->] (bb)--(cc) node[midway,right] {$  \mathrm{dom} ^{-1} \left( \catD \right)  $};
	\end{tikzpicture} 
\end{equation} 	
commute and $\pi _1 : \underline{\overline{\catC }}\times \underline{\overline{\catD }}\to \underline{\overline{\catC }} $ and
$\pi _2  : \underline{\overline{\catC }}\times \underline{\overline{\catD }}\to \underline{\overline{\catD }} $ are morphisms
in $\mnPoly _ \catB $.
	\item given objects $ \catD ', \catD '' ,  \catD '''$  of $\mnPoly _\catA  $, if $E: \catD '   \to \catD ''    $ and $J : \catD '  \to \catD '''    $ are morphisms of $\mnPoly _ \catA $ such that we have natural isomorphisms 
	\begin{eqnarray*}
\gamma : & \overline{\underline{ E } } \circ \mathrm{dom} ^{-1} \left(  \catD '  \right) & \cong \mathrm{dom}^{-1} \left(  \catD ''\right) \circ E\\
\gamma ' : & \overline{\underline{ J } } \circ \mathrm{dom} ^{-1} \left(  \catD '  \right) & \cong  \mathrm{dom} ^{-1} \left(  \catD '''  \right) \circ J    
	\end{eqnarray*}
in which $ \overline{\underline{ J } }  $ and $\overline{\underline{ E } } $ are morphisms of $\mnPoly _ \catB $, then
$\left( \overline{\underline{ E } }, \overline{\underline{ J } }\right)  $ is a morphism in $\mnPoly _ \catB $ and
$\left( \gamma , \gamma '\right) $ defines an isomorphism 
\begin{equation}
	\left( \overline{\underline{ E } }, \overline{\underline{ J } }\right)  \circ \mathrm{dom} ^{-1} \left(  \catD '  \right) \cong \mathrm{dom} ^{-1} \left(  \catD ''\times  \catD ''' \right) \circ (E, J) . 
\end{equation}
	\item   if $\catC$ is an object of $\mnPoly _ \catA $ and  $H: \catC\times \catA   \to\catA  $ is a morphism of $\mnPoly _ \catA $ 
	such that there is an isomorphism 
	$$ \gamma :  \overline{\underline{ H } } \circ \mathrm{dom} ^{-1} \left(  \catC\times \catA   \right) \cong \mathrm{dom} ^{-1} \left(  \catA \right) \circ H $$
	in which $\overline{\underline{ H } } $ is a morphism of $\mnPoly _ \catB $, then, since $G$ creates initial algebras and terminal coalgebras, we get that
	there are natural transformations
\begin{eqnarray*}
& \mu\overline{\underline{ H } } \circ \mathrm{dom} ^{-1} \left(   \catA   \right) & \cong \mathrm{dom} ^{-1} \left(  \catA \right) \circ \mu H\\
 & \nu\overline{\underline{ H } } \circ \mathrm{dom} ^{-1} \left(   \catA   \right) & \cong \mathrm{dom} ^{-1} \left(  \catA \right) \circ \nu H	
\end{eqnarray*} 		 
and, of course, $\mu\overline{\underline{ H } } $ and $\nu\overline{\underline{ H } } $ are morphisms of $\mnPoly _ \catB $.
\end{enumerate}
\end{proof}

\subsection{$\mu\nu$-polynomials in product categories}
Before applying the results above to study the $\mu\nu$-polynomials in suitable scones $\catD\downarrow G $,
we need to study the $\mu\nu$-polynomials in product categories $\catC\times\catD $. We start by showing that:

\begin{lemma}\label{lem:product-initial-algebras}
Let $\left( E _ i: \catC _ {i} \to \catC _{i} \right) _ {i\in L}  $  be a (possibly infinite) family of endofunctors such that $E_i$ has initial algebra $\left( \mu E_i, \ind_{E_i} \right)$ and terminal coalgebra $\left( \nu E_i, \coind_{E_i} \right)$. The functor defined by the product  
	\begin{equation}
		\prod\limits _{i\in L} E _ i: \prod\limits _{i\in L} \catC _ i \to \prod\limits _{i\in L} \catC _ i
	\end{equation}
has initial algebra given by  $\left( \mu E_i, \ind_{E_i} \right) _{i\in L} $ and terminal coalgebra given by 
$\left( \nu E_i, \coind_{E_i} \right) _{i\in L} $.

As a consequence, if  $\left( H _ i: \catA_ i\times \catC _ {i} \to \catC _{i} \right) _ {i\in L}  $ is a a (possibly infinite) family of functors with parameterized initial algebras and terminal coalgebras, then $\prod\limits _{i\in L} H _ i$
has parameterized initial algebra given by $\prod\limits _{i\in L} \mu H _ i : \prod\limits _{i\in L }\catA_ i\to \prod\limits _{i\in L }\catC _i $ and parameterized terminal coalgebra 
given by $\prod\limits _{i\in L} \nu H _ i : \prod\limits _{i\in L }\catA_ i\to \prod\limits _{i\in L }\catC _i$.
\end{lemma} 	
\begin{proof}
Given an $\left( \prod\limits _{i\in L} E _ i \right)$-algebra $\left( Y_i, \xi _i \right) _{i\in L} $, we have that 
$\left( Y _i , \xi _i \right)$  is an $E_i$-algebra for every $i\in L$. Therefore, by the universal property of $\left( \mu E_i ,  \ind _{E_i}\right)  $ for each $i$, we conclude that 
\begin{equation} 
\fold_{\prod\limits _{i\in L} E _ i } \left( Y_i, \xi _i \right) _{i\in L} \coloneqq \left( \fold_{E_i} \left( Y _i, \xi _i \right)\right) _{i\in L} 
\end{equation} 
is the unique morphism in $\prod\limits _{i\in L} \catC  _ i  $ such that  
\begin{equation}\label{eq:fold-product-of-categories} 
	\begin{tikzpicture}[x=9.5cm, y=3cm]
		\node (a) at (0,0) {$	\prod\limits _{i\in L} E _ i\left( \mu E_i\right) _{i\in L} =  \left( E_i\left( \mu E_i\right) \right) _{i\in L} $};
		\node (d) at (0,-1) {$ \left( \mu E_i \right) _{i\in L}  $};
		\node (b) at (1, 0) { $\prod\limits _{i\in L} E _ i \left( Y_i\right) _{i\in L} $ };
		\node (c) at (1,-1) {$\left( Y_i\right) _{i\in L} $};
		\draw[->] (a)--(b) node[midway,above] {$ \prod\limits _{i\in L} E _ i \left( \fold_{E_i} \left( Y _i, \xi _i \right)\right) _{i\in L}  $};
		\draw[->] (a)--(b) node[midway,below] {$ \prod\limits _{i\in L} E _ i \left( \fold_{\prod\limits _{i\in L} E _ i} \left( Y_i, \xi _i \right) _{i\in L}  \right)   $};
		\draw[->] (b)--(c) node[midway,right] {$  \left( \xi _i\right) _{i\in L}   $};
		\draw[->] (a)--(d) node[midway,left] {$ \left( \ind _{E_i} \right) _{i\in L}   $};
		\draw[->] (d)--(c) node[midway,above] {$ \fold_{\prod\limits _{i\in L} E _ i} \left( Y_i, \xi _i \right) _{i\in L}  $};
		\draw[->] (d)--(c) node[midway,below] {$ \left( \fold_{E_i} \left( Y _i, \xi _i \right)\right) _{i\in L}   $};
	\end{tikzpicture} 
\end{equation}
holds. 
   This proves that $\left( \mu E_i, \ind_{E_i} \right) _{i\in L}   $ is the initial $\left( \prod\limits _{i\in L} E _ i \right) $-algebra. Dually, $\left( \nu E_i, \coind_{E_i} \right) _{i\in L}  $ is the terminal $\left( \prod\limits _{i\in L} E _ i \right) $-coalgebra.
\end{proof}

We prove below that the binary product of categories with $\mu\nu$-polynomials has $\mu\nu$-polynomials.
We start by:
\begin{definition}[$\deckk{\catA}$]
	Let $\left( \catC _i \right) _ {i\in L}$ be a (possibly infinite) family of categories. We  establish a family  	\begin{equation}\label{eq:triple-for-the-isomorphism-product-of-categories-munu-polynomials}
		\left( \deckk\catA^i : \catA \to \obdeck{i}{\catA }    \right) _{\left( \catA , i \right) \in \left( \objects\left(\mnPoly _ {\prod\limits _{i\in L} \catC _ i}\right)\times L\right)  }
	\end{equation}
 of functors, where $\obdeck{i}{\catA } \in \objects\left( \mnPoly _{\catC_i }\right) $,  by induction on 
 the objects of $\mnPoly _ {\prod\limits _{i\in L} \catC _ i}$:\\
\noindent\begin{minipage}{.5\linewidth} 
	\begin{equation}
		\deckk\terminal ^i \coloneqq  \id _{\terminal}  ; 
	\end{equation}
\end{minipage} 
\begin{minipage}{.5\linewidth} 
	\begin{equation} 
		\deckk{\prod\limits _{i\in L} \catC _ i} ^i \coloneqq  \projection _{\catC _i }: \prod\limits _{i\in L} \catC _ i\to \catC _i ; 
	\end{equation}  
\end{minipage} 
\begin{equation}
	\deckk{\catA\times\catA '} ^i \coloneqq  \deckk{\catA} ^i\times \deckk{\catA '} ^i  ,\mbox{ if }\left( \catA , \catA '\right) \in\objects\left(\mnPoly _ {\prod\limits _{i\in L} \catC _ i} \right) ^2.
\end{equation}
Finally, for each $\catA\in \objects\left(\mnPoly _ {\prod\limits _{i\in L} \catC _ i}\right) $, we define the isomorphism of categories
\begin{equation} 
	\deckk{\catA}  \coloneqq  \left( \deckk{\catA}^0 , \deckk{\catA}^1\right).
\end{equation}
\end{definition} 	

\begin{lemma}\label{lem:support-munupolynomials-products}
	Let $\left( \catC _i \right) _ {i\in L}$ be a (possibly infinite) family of categories with $\mu\nu$-polynomials. For each pair $\left(  \catA, \catA '\right)\in\objects\left( \mnPoly _ {\prod\limits _{i\in L} \catC _ i} \right) ^2  $ and any functor $H: \catA \to \catA '$ in $\mnPoly _ {\prod\limits _{i\in L} \catC _ i}$, we have that 
$ \deckk{\catA'}\circ H\circ \deckk{\catA}^{-1} = \prod\limits _{i\in L} H _ i $ for some 
morphism $\left( H _i \right) _ {i\in L}$ in $\prod\limits _{i\in L} \mnPoly _ {\catC _ i}$.
\end{lemma} 	
\begin{proof}
	It is clear the property above is closed under composition, and the identity on $\prod\limits _{i\in L} \catC _ i $ satisfies the property.
	Moreover, for the base case (see Definition \ref{def:basic-definition-munupolynomials}), it is clear that the functors in $\mnPoly _ {\prod\limits _{i\in L} \catC _ i}$ defined by the base cases \ref{eq:munu-M1} and \ref{eq:munu-M2} satisfy the statement above.
	Moreover, since the binary products and coproducts in $\prod\limits _{i\in L} \catC _ i$  are defined pointwise, it is also true that
	\ref{eq:munu-M3} and \ref{eq:munu-M4} satisfy the statement above. Finally, it is also clear that the statement above holds for \ref{eq:munu-M5} and \ref{eq:munu-M6}.
	
	We assume, by induction, that $H: \catA \times \prod\limits _{i\in L} \catC _ i \to \prod\limits _{i\in L} \catC _ i $ is a morphism of $\mnPoly _ {\prod\limits _{i\in L} \catC _ i}$
	such that $H\circ \deckk{\catA\times\prod\limits _{i\in L} \catC _ i}^{-1} = \prod\limits _{i\in L} H _ i$ for some 
	morphism $\left( H _i \right) _ {i\in L}$ in $\prod\limits _{i\in L} \mnPoly _ {\catC _ i}$.
	
	Since $\catC _i $ has $\mu\nu$-polynomials for all $i\in L$, we have that $H_i $ has parameterized initial algebras and parameterized terminal coalgebras for all $i\in L$. Therefore, $\prod\limits _{i\in L} H _ i$ has parameterized initial algebra $\prod\limits _{i\in L} \mu H _ i $ and parameterized terminal coalgebra $\prod\limits _{i\in L} \nu H _ i  $ by Lemma \ref{lem:product-initial-algebras}.
	Hence $\mu H = \left( \prod\limits _{i\in L} \mu H _ i \right) \circ \deckk{\catA} $ and $\nu H = \left( \prod\limits _{i\in L} \nu H _ i  \right) \circ \deckk{\catA} $ where $\left( \mu H _i \right) _ {i\in L}$
	and $\left( \nu H _i \right) _ {i\in L}$ are morphisms in $\prod\limits _{i\in L} \mnPoly _ {\catC _ i}$. This completes the proof.	
\end{proof}

\begin{therm}\label{theo:munupolynomials-of-product-of-categories}
	Let $\left( \catC _i \right) _ {i\in L}$ be a (possibly infinite) family of categories with $\mu\nu$-polynomials.  The category $\prod\limits _{i\in L} \catC _ i $ has $\mu\nu$-polynomials.
\end{therm} 
\begin{proof}

For each endofunctor $E: \prod\limits _{i\in L} \catC _ i \to \prod\limits _{i\in L} \catC _ i $ in $\mnPoly _ {\catC\times \catD}$,
we conclude that $E= \prod\limits _{i\in L} E_i $ for some morphism $\left( E_i : \catC _i \to \catC _i \right) _{i\in L}$ of $\mnPoly _ {\prod\limits _{i\in L} \catC _ i } $ by Lemma \ref{lem:support-munupolynomials-products}. Therefore, by Lemma \ref{lem:product-initial-algebras}, 
$E$ has initial algebra and terminal coalgebra, since the functors of the family $\left( E_i : \catC _i \to \catC _i \right) _{i\in L}$ do.
\end{proof}

\subsection{Suitable scones have $\mu\nu$-polynomials} 

Finally, we establish the existence of $\mu\nu$-polynomials in the scone, and the preservation of the initial algebras and terminal coalgebras by the forgetful functor. 

\begin{corollary}\label{cor:sconing-munupolynomials}
	Let $\catC $ and $\catD $ be categories with $\mu\nu $-polynomials.
	If $G: \catC\to \catD $ has a left adjoint, then $\catD\downarrow G $
	has $\mu\nu $-polynomials and
	\begin{equation}
		\forgetfulS : \catD\downarrow G \to \catD\times \catC 
	\end{equation}
(strictly) preserves (in fact, creates) $\mu\nu$-polynomials.
\end{corollary}	
\begin{proof}
	By Corollary \ref{coro:basic-comonadicity-monadicity-sconing}, we have that
	$\forgetfulS $ is monadic and comonadic. Hence it creates $\mu\nu $-polynomials and we get the conclusion of the result provided that  
	$\catD\times \catC  $ has $\mu\nu $-polynomials.
	
	Indeed, by Theorem \ref{theo:munupolynomials-of-product-of-categories}, $\catD\times \catC  $  has $\mu\nu $-polynomials provided that $\catD$ and $\catC$ have $\mu\nu $-polynomials.
\end{proof}

\subsection{The projection $\catD\downarrow G\to \catC $}
Let $\catC $ and $\catD $ be bicartesian closed categories with finite limits.
Recall that $\pi _{\catC } : \catD\times \catC \to \catC $ has left and right adjoints, respectively given by  $W\mapsto \left( W, \initial   \right)  $ and $W\mapsto \left( W, \terminal \right) $.
Therefore, assuming that $G: \catC\to\catD $ has a left adjoint,
we get that
\begin{equation}\label{eq:projection-semantics-transformation}
\catD\downarrow G \xrightarrow{\forgetfulS } \catD\times \catC \xrightarrow{\pi_2 } \catC  
\end{equation}
has a left adjoint and a right adjoint. Therefore it preserves limits, colimits, initial algebras and terminal coalgebras. Finally, \eqref{eq:projection-semantics-transformation} preserves the closed structure by Corollary \eqref{coro:closed-structure-comonadic-monadic-special-case}.

\begin{corollary}\label{coro:main-sconing}	
Let $\catC $ and $\catD $ be finitely complete bicartesian closed categories that have $\mu\nu $-polynomials.
If $G : \catC \to\catD $ has a left adjoint, the category $\catD\downarrow G $ is a finitely complete bicartesian closed category with $\mu\nu $-polynomials, and \eqref{eq:projection-semantics-transformation}  is a (strictly) bicartesian closed functor that (strictly) preserves $\mu\nu $-polynomials.

Furthermore, if, additionally, $\catC$ and $\catD $ have infinite coproducts, so does $\catD\downarrow G$ and \eqref{eq:projection-semantics-transformation} (strictly) preserves them.
\end{corollary}

\section{Correctness of CHAD for tuples and variant types, by logical relations}\label{sec:logical-relations-argument}
Henceforth we assume the hypothesis established in \ref{subsect:CHAD-SOUNDNESS-ASSUMPTIONS}, and  rely on the concrete semantics and notation established in Section \ref{sec:concrete-semantics-specification}. 

In this section, we present the basic correctness theorem for tuples and variant types, which serves as a crucial step towards establishing the full correctness theorem for data types. 
More precisely, we prove:
\begin{therm}[Correctness of CHAD for tuples and variant tuples]\label{theorem:correctness-theorem-for-data-types}
	For any well-typed program $$ \var[1] : \ty[1] \tinf \trm[1] : \ty[2], $$
	where $\ty[1], \ty[2] $ are data types that do not involve inductive types,  
	we have that $\sem{\trm[1]}$ is differentiable. Moreover, \eqref{eq:derivateive-forward-modee} and \eqref{eq:derivateive-reverse-modee} hold. \\	
	\noindent\begin{minipage}{.5\linewidth} 
\begin{equation}\label{eq:derivateive-forward-modee}
	\semt{\Dsyn{\trm[1]} } = \Ds{\sem{\trm[1]} } 
\end{equation}
	\end{minipage} 
	\begin{minipage}{.5\linewidth} 
\begin{equation}\label{eq:derivateive-reverse-modee}
	\semtt{\Dsynrev{\trm[1]} } = \Dsr{\sem{\trm[1]} } 
\end{equation}
	\end{minipage}
\end{therm}
It should be noted that: (1) we prove our result only assuming that the semantics of the primitive operations are differentiable
instead of requiring them to be smooth;\footnote{We even claim that the result is useful when the semantics of the primitive operations is not differentiable 
everywhere in the domain; see the revised version of \citep{2022arXiv221007724L}.} (2) $\trm[1]$ above might, in particular, have subprograms that use higher-order functions and (co)inductive types. 

The argument we present below is a categorical version of a semantic open logical relations proof; see, for instance, \citep{openlogicalrelationsBarthe, hsv-fossacs2020, vakar2020reverse, vakar2021chad}.
We follow the perspective described in \citep[Section~4]{2022arXiv221008530L} and in \citep{2022arXiv221007724L}.

The precise statement \ref{theorem:correctness-theorem-for-data-types} is presented in Theorem \ref{theorem:correctness-theorem-for-data-typess}. 
\subsection{The scone for the correctness proof}\label{subsect:scones-in-the-proof}
We first establish the appropriate scone for our proof (see Section~\ref{sec:subsconing}).

By Proposition \ref{prop:FAMSET-as-CONCRETE-MODEL}, Corollary \ref{coro:catFV-is-a-model-for-target-language} and Corollary \ref{coro:catFV-is-complete-cocomplete},  
we conclude, in particular, that $\Fam{\Set } $, $\Fam{\Vect}$ and $\Fam{\Vect ^\op}$ are finitely complete cartesian closed categories with $\mu\nu$-polynomials and infinite coproducts. Therefore, we conclude 
that $\Fam{\Set } \times\Fam{\Vect}\times\Fam{\Vect ^\op } $ is a finitely complete cartesian closed category with $\mu\nu $-polynomials and infinite coproducts:
see Theorem \ref{theo:munupolynomials-of-product-of-categories} for the result on $\mu\nu$-polynomials.

We consider the scone along \eqref{eq:functor-SCONE}, which is representable by the coproduct $\coprod\limits _{k\in\NN }  \left( \RR ^k  , \left( \RR ^k , \cRR ^k \right) , \left( \RR ^k , \cRR ^k \right) \right)  $ 
in $\Fam{\Set} \times\Fam{\Vect}\times\Fam{\Vect ^\op } $.
\begin{eqnarray}
	&\forG   &:  \Fam{\Set }\times \Fam{\Vect }  \to \Set \label{eq:functor-SCONE}\\
	&\forG  & \coloneqq \prod _{k\in \NN } \left( \Fam{\Set }\times \Fam{\Vect  }\times \Fam{\Vect ^\op } \left( \left( \RR ^k  , \left( \RR ^k , \cRR ^k \right), \left( \RR ^k , \cRR ^k \right)\right)  , -\right)  \right)\nonumber
\end{eqnarray}
Moreover, \eqref{eq:functor-SCON-copower} given by the copower in $\Fam{\Set}\times \Fam{\Vect }\times \Fam{\Vect ^\op}$ defines the left adjoint $\forF\dashv\forG $.
As a consequence, we get Theorem~\ref{theo:fundamental-scone-result-proof} by Corollary \ref{coro:main-sconing}.
\begin{eqnarray}\label{eq:functor-SCON-copower}
	\forF : &\Set & \to \Fam{\Set}\times \Fam{\Vect }\times \Fam{\Vect ^\op} \\
	&W&\mapsto  W\otimes \coprod_{k\in\NN }  \left( \RR ^k  , \left( \RR ^k , \cRR ^k \right) , \left( \RR ^k , \cRR ^k \right) \right)  \cong \coprod _{x\in W} \left( \coprod_{k\in\NN }  \left( \RR ^k  , \left( \RR ^k , \cRR ^k \right) , \left( \RR ^k , \cRR ^k \right)\right) \right) \nonumber
\end{eqnarray}

\begin{therm}\label{theo:fundamental-scone-result-proof}
	$\Set \downarrow \forG $ is a finitely complete cartesian closed categories with $\mu\nu $-polynomials and infinite coproducts. Moreover, \eqref{eq:projection-semantics-transformation-logical-relations}  is a strictly bicartesian closed functor that preserves $\mu\nu $-polynomials and (infinite) coproducts.
\begin{equation}
	\Set  \downarrow \forG \rightarrow  \Set \times \Fam{\Set}\times \Fam{\Vect }\times \Fam{\Vect ^\op}  \rightarrow \Fam{\Set}\times \Fam{\Vect } \times \Fam{\Vect ^\op}   \label{eq:projection-semantics-transformation-logical-relations} 
\end{equation}	 
\end{therm}

\begin{definition}[$\Fscone$]
	For short, we henceforth denote by \eqref{eq:pifor}, where $\Fscone \ceq \Set  \downarrow \forG $, the forgetful functor \eqref{eq:projection-semantics-transformation-logical-relations}.\\
\begin{equation} \label{eq:pifor}
			\forpi : \Fscone  \to  \Fam{\Set} \times\Fam{\Vect } \times \Fam{\Vect ^\op} 
\end{equation} 
\end{definition}

\subsection{The logical relations}
Guided by the characterization of differentiable morphisms and their derivatives (Lemma \ref{lem:differentiability-derivative-characterization}), 
we now define the objects in $\Fscone$ that will provide us with the appropriate predicates for our logical relations argument.

It should be noted that, for any object $\left( Y, (W,w), (Z,z)\right)$ in $ \Fam{\Set}\times \Fam{\Vect }\times \Fam{\Vect ^\op }  $, the elements of $\revG{\left( Y, (W,w), (Z,z)\right)}$
are families $\left( f _k , g_k, h_k \right) _{k\in\NN } $ where, for each $k\in\NN$,  $f_k : \RR ^k\to Y $ is a morphism in $\Fam{\Set}$, $g_k :  \left( \RR ^k , \cRR ^k \right)\to (W,w)$ is a morphism in $\Fam{\Vect } $ and 
$h_k :  \left( \RR ^k , \cRR ^k \right)\to (Z,z)$ is a morphism in $\Fam{\Vect ^\op } $.
\begin{definition}[$\forsem{\reals ^n }$]\label{def:logical-relation-very-basic-definition}
	For each $n$-dimensional array $\reals^n\in\Syn$, we define the subset \eqref{eq:basic-definition-LR-subset} of $\revR\coloneqq \forG{ \left( \RR ^n , \left( \RR ^n , \cRR ^n \right), \left( \RR ^n , \cRR ^n \right)\right) }$.
	\begin{equation}\label{eq:basic-definition-LR-subset}
		\underline{\forsem{\reals ^n }} \coloneqq \left\{ \left( f_k, g_k, h_k  \right) _{k\in \NN }\in \revR :\forall k\in\NN , \, f_k \,\mbox{is differentiable,}\, g_k  = \Ds{f_k}, \, h_k = \Dsr{f_k} \right\} 
	\end{equation}
	Denoting the subset inclusion by   $$\incLRR :  \underline{\forsem{\reals ^n }} \to\forG{ \left( \RR ^n , \left( \RR ^n , \cRR ^n \right), \left( \RR ^n , \cRR ^n \right)\right) } ,$$ we define the object \eqref{eq:realsn-logicalrelations}  of $\Fscone$.
	\begin{equation}\label{eq:realsn-logicalrelations} 
		\forsem{\reals ^n }\coloneqq \left( \underline{\forsem{\reals ^n }} , \left( \RR ^n , \left( \RR ^n , \cRR ^n \right), \left( \RR ^n , \cRR ^n \right)\right) , \incLRR \right). 
	\end{equation}

\end{definition} 	

Recall that we denote by $\Euc$ the set of Euclidean families defined in \ref{eq:Euclidea-families}.
Theorem \ref{theo-basic-result-for-LR} relies on the canonical diffeomorphisms given in Definition \ref{DEF:CANONICAL-EUCLIDEAN-FAMILIES}. 

\begin{therm}\label{theo-basic-result-for-LR}
Let  $\left( f ,g,h \right)$  be a morphism in $\Fam{\Set } \times \Fam{\Vect }\times \Fam{\Vect ^\op} $.
Assuming that $f :A\to B $ is such that  $A$ and $B$ are Euclidean families,  we have that  \ref{fundamentallemma:condition2)LR}  implies \ref{fundamentallemma:condition1)differentiability-of-f}.
\begin{enumerate}[i)] 
\item \label{fundamentallemma:condition2)LR} There is a morphism
		\begin{equation}\label{eq:proof-scone-logical-relations}
			\alpha : \coprod _{j\in J}\left( \prod _{i=1 }^{n_j} {\forsem{\reals ^{q_{(j,i)} } } } \right) \to  \coprod _{l\in L}\left( \prod _{t=1 }^{m_l} {\forsem{\reals ^{s_{(l,t)} } } } \right)
		\end{equation}
			in $\Fscone$, where $\left( n_{j} \right) _{j\in J } $, $\left( m_l \right) _{l\in L } $, 
			$\left( \left( q_{(j,i)} \right)_{i\in \left\{ 1, \ldots, n_j\right\} } \right)  _{j\in J } $ and $\left( \left( s_{(l,t)} \right)_{t\in \left\{ 1, \ldots, m_l\right\} }  \right) _{l\in L } $  are (possibly infinite) families of natural numbers, such that
			\begin{equation} 
			\forpi\left( \alpha \right) = \left( \CANsh{B}\circ  f \circ \CANsh{A} ^{-1}, \sDs{\CANsh{B}} \circ g \circ \sDs{\CANsh{A}} ^{-1} ,\sDsr{\CANsh{B}} \circ  h \circ \sDsr{\CANsh{A}} ^{-1}\right).
			\end{equation} 
\item\label{fundamentallemma:condition1)differentiability-of-f} The morphism
$f$ is differentiable, $\Ds{f} = g $ and $\Dsr{f} = h $.	
\end{enumerate} 
\end{therm} 	
\begin{proof}
We start by establishing the objects $\okS _ 0 $ and $\okS _1 $ of $\Fscone $ together with the canonical isomorphisms \eqref{eq:sh0-LR-FSCONE} 	and \eqref{eq:sh1-LR-FSCONE}. 

Let $q_j \coloneqq \sum\limits _{i=1}^{n_j} q_{(i,l)} $ and $s_l \coloneqq \sum\limits _{t=1}^{m_l} s_{(l,t)} $. We define the objects $\okCO _0$ and $\okCO _1 $  of $\Set\times\Fam{\Vect}\times\Fam{\Vect ^\op}$ by \eqref{eq:coproduct-scone-LR1}
and \eqref{eq:coproduct-scone-LR2}: the construction of infinite coproducts in $\Fscone$ follows from \ref{subsect:bicartesian-scone}.
\begin{eqnarray} 
	\okCO _0 &\coloneqq & \coprod _{j\in J} \left(\RR ^{q_j} , \left(  \RR ^{q_j}, \cRR ^{q_j} \right) ,  \left(  \RR ^{q_j}, \cRR ^{q_j} \right)  \right) \nonumber \\ &= & \left(\coprod _{j\in J} \RR ^{q_j} , \left(  \coprod _{j\in J} \RR ^{q_j}, \langle \cRR ^{q_j} \rangle _{j\in J}  \right) ,  \left(  \coprod _{j\in J} \RR ^{q_j}, \langle \cRR ^{q_j} \rangle _{j\in J}  \right)  \right)\label{eq:coproduct-scone-LR1} 
\end{eqnarray} 
\begin{eqnarray} 	
	\okCO _1 &\coloneqq &   \coprod _{l\in L} \left(\RR ^{s_l} , \left(  \RR ^{s_l}, \cRR ^{s_l} \right) ,  \left(  \RR ^{s_l},  \cRR ^{s_l}  \right)  \right) \nonumber \\ & = & \left( \coprod _{l\in L}\RR ^{s_l} , \left(  \coprod _{l\in L}\RR ^{s_l}, \langle \cRR ^{s_l} \rangle _{l\in L}  \right) ,  \left(  \coprod _{l\in L}\RR ^{s_l}, \langle \cRR ^{s_l} \rangle _{l\in L}  \right)  \right)\label{eq:coproduct-scone-LR2}
\end{eqnarray} 
We consider the subsets $\underline{\okS}_0 \subset \sforG{\okCO _ 0} $ and $\underline{\okS}_1 \subset \sforG{\okCO _ 1} $ defined by \eqref{eq:subset-0}  and \eqref{eq:subset-1}.
Denoting by $\incLRR$ the appropriate subset inclusions, we define the objects ${\okS}_0\coloneqq \left( \underline{\okS}_0 , \okCO _0 , \incLRR \right) $ and ${\okS}_1\coloneqq \left( \underline{\okS}_1 , \okCO _1 , \incLRR \right) $ of $\Fscone $.
\begin{equation}\label{eq:subset-0} 
	\underline{\okS}_0 \coloneqq  \left\{ \left( f_k, g_k, h_k  \right) _{k\in \NN }\in \forG\left(\okCO _0\right)  :\forall k\in\NN , \, f_k \,\mbox{is differentiable,}\, g_k  = \Ds{f_k}, \, h_k = \Dsr{f_k} \right\}
\end{equation}
\begin{equation}\label{eq:subset-1}  
	\underline{\okS}_1 \coloneqq  \left\{ \left( f_k, g_k, h_k  \right) _{k\in \NN }\in \forG\left(\okCO _1\right)  :\forall k\in\NN , \, f_k \,\mbox{is differentiable,}\, g_k  = \Ds{f_k}, \, h_k = \Dsr{f_k} \right\} 
\end{equation}
By the results of \ref{subsect:bicartesian-scone}, the chain rule (Lemma \ref{lem:chain-rule})  and Definition \ref{def:logical-relation-very-basic-definition}, since the canonical isomorphisms \eqref{eq:sh0}
and \eqref{eq:sh1} are diffeomorphisms, there are (invertible) functions $\underline{\sh{0}}$ and $\underline{\sh{1}} $, respectively induced by the compositions with $\left( \sh{0}, \sDs{\sh{0}},\sDsr{\sh{0}}\right)   $ and $	\left( \underline{\sh{1}}, \left( \sh{1}, \sDs{\sh{1}},\sDsr{\sh{1}}\right)  \right) $,  such that \eqref{eq:sh0-LR-FSCONE} and \eqref{eq:sh1-LR-FSCONE} define isomorphisms in $\Fscone $.
\begin{equation}\label{eq:sh0}
	\sh{0}:\coprod _{j\in J} \left( \prod _{i=1 }^{n_j} {\RR ^{q_{(j,i)}  } } \right) \xto{\cong} \coprod _{j\in J} {\RR ^{q_{j}  } } 
\end{equation} 
\begin{equation}\label{eq:sh1}
	\sh{1}: \coprod _{l\in L} \left( \prod _{t=1 }^{m_l} {\RR ^{s_{(l,t)}  } } \right) \xto{\cong}  \coprod _{l\in L} {\RR ^{s_{l}  } } 
\end{equation} 
\begin{equation}\label{eq:sh0-LR-FSCONE}
	\tilde{\sh{0}} \coloneqq  \left( \underline{\sh{0}}, \left( \sh{0}, \sDs{\sh{0}},\sDsr{\sh{0}}\right)  \right) : \coprod _{j\in J}\left( \prod _{i=1 }^{n_j} {\forsem{\reals ^{q_{(j,i)} } } } \right) \xto{\cong} {\okS}_0
\end{equation} 
\begin{equation}\label{eq:sh1-LR-FSCONE}
	\tilde{\sh{1}}  \coloneqq  \left( \underline{\sh{1}}, \left( \sh{1}, \sDs{\sh{1}},\sDsr{\sh{1}}\right)  \right)  : \coprod _{l\in L}\left( \prod _{t=1 }^{m_l} {\forsem{\reals ^{s_{(l,t)} } } } \right) \xto{\cong}  {\okS}_1
\end{equation} 		
\\ \\
Proof of \ref{fundamentallemma:condition2)LR} $\Rightarrow $ \ref{fundamentallemma:condition1)differentiability-of-f}.\\
By \ref{fundamentallemma:condition2)LR} and chain rule, denoting $\mathfrak{f} = \CANsh{B}\circ  f \circ \CANsh{A} ^{-1} $,  $ \mathfrak{g} = \sDs{\CANsh{B}} \circ g \circ \sDs{\CANsh{A}} ^{-1}$
and $\mathfrak{h} = \sDsr{\CANsh{B}} \circ  h \circ \sDsr{\CANsh{A}} ^{-1} $, we conclude that there is a morphism  $\alpha $ in $\Fscone$ such that
\begin{eqnarray*}
&\forpi{\left( \tilde{\sh{1}} \circ	\alpha\circ \tilde{\sh{0}}  ^{-1} \right)  }& \\
& =& \\
&\left(   \sh{1} \circ  \mathfrak{f}  \circ \sh{0}  ^{-1}, \sDs{\sh{1}} \circ \mathfrak{g}  \circ \sDs{\sh{0}}^{-1}, \sDsr{\sh{1}} \circ \mathfrak{h}  \circ \sDsr{\sh{0}}^{-1}\right).& 
\end{eqnarray*} 
\normalsize
This implies, by the definitions of $\okS _0 $ and $\okS_1$, that, for any family $\left( \gamma _k: \RR ^k \to \coprod\limits _{j\in J} \RR ^{q_j} \right) _{k\in \NN} $ of differentiable functions, we have that, for all $k\in \NN $: 
\begin{enumerate}[(I)]
	\item  $\sh{1}\circ \mathfrak{f} \circ \sh{0}^{-1} \circ \gamma _k  $ is differentiable,
	\item $\sDs{\sh{1}} \circ \mathfrak{g} \circ \sDs{\sh{0}}^{-1}  = \sDs{\sh{1}\circ \mathfrak{f} \circ \sh{0}^{-1}\circ \gamma _k}$, and
	\item $\sDsr{\sh{1}} \circ \mathfrak{h}  \circ \sDsr{\sh{0}}^{-1} = \sDsr{\sh{1}\circ \mathfrak{f} \circ \sh{0}^{-1}\circ \gamma _ k}$. 
\end{enumerate} 
By Lemma \ref{lem:differentiability-derivative-characterization}, this implies that: 
\begin{enumerate}[(A)]
	\item  $\sh{1}\circ \mathfrak{f} \circ \sh{0}^{-1}   $ is differentiable,
	\item $\sDs{\sh{1}} \circ \mathfrak{g} \circ \sDs{\sh{0}}^{-1}  = \sDs{\sh{1}\circ \mathfrak{f} \circ \sh{0}^{-1}}$, and
	\item $\sDsr{\sh{1}} \circ \mathfrak{h} \circ \sDsr{\sh{0}}^{-1} = \sDsr{\sh{1}\circ \mathfrak{f} \circ \sh{0}^{-1}}$. 
\end{enumerate} 
By the chain rule (Lemma \ref{lem:chain-rule}) and the fact that $\sh{1}$, $\sh{0}$, $\CANsh{A}$ and $\CANsh{B}$ are diffeomorphisms, this implies that  $f    $ is differentiable, $  g   = \Ds{ f }$ and $  h   = \Dsr{ f }$. This completes the proof.
\end{proof}

\subsection{Logical relations as a functor}
For each primitive operation $\op\in \Op_{n_1,\ldots, n_k}^m$ of the source language, recall that $$\sem{\op}: \RR ^{n_1}\times \cdots \times \RR ^{n_k}\to \RR ^m $$
is differentiable, $\semt{\Dsyn{\op}} = \Ds{\sem{\op}}$, and $ \semtt{\Dsynrev{\op}} = \Dsr{\sem{\op}}$ (see \ref{subsect:CHAD-SOUNDNESS-ASSUMPTIONS}).
Therefore, for each primitive operation $\op\in \Op_{n_1,\ldots, n_l}^m$,  we conclude, by the chain rule (Lemma \ref{lem:chain-rule}), that we can define the morphism
\begin{equation}
	\forsem{\op }  \coloneqq \left( \underline{\sem{\op }} , \left( \sem{\op} , \semt{\Dsyn{\op } }, \semtt{\Dsynrev{\op}}   \right) \right) : \prod _{i=1}^l\forsem{\reals ^{n_i} }\to \forsem{\reals ^m }
\end{equation}
in $\Fscone $. 

Since $\Fscone $ is bicartesian closed and has $\mu\nu$-polynomials, by the universal property of the category $\Syn $  established in Corollary \ref{cor:universal-property-of-source-language}, we conclude:
\begin{lemma}
	There is a unique strictly bicartesian closed functor
	\begin{equation}
		\forsem{-} : \Syn \to \Fscone 
	\end{equation}	
	that strictly preserves $\mu\nu$-polynomials such that $\forsem{-}$ extends the consistent assignment given by \eqref{eq:association-extending-logical-relations}.
\begin{equation}\label{eq:association-extending-logical-relations}
 \reals ^n \mapsto \forsem{\reals ^n}, \qquad \qquad \op \mapsto \forsem{\op }.
\end{equation}		
\end{lemma}	

Let us recall that we defined the forward-mode and reverse-mode CHAD corresponding functors in Corollary \ref{cor:CHAD-definition}, which we denote by $\Dsyn{-}$ and $\Dsynrev{-}$, respectively.
By the universal property of $\Syn $ and the hypothesis established in \ref{subsect:CHAD-SOUNDNESS-ASSUMPTIONS}, we can further conclude that:
\begin{therm}[Correctness commutative diagram]\label{theo:correctness-commutative-diagram}
Diagram \eqref{eq:diagram-correctness-proof} commutes.
\begin{equation}\label{eq:diagram-correctness-proof}
	\begin{tikzpicture}[x=6cm, y=1.7cm]
		\node[scale=0.9] (a) at (0,0) {$\Syn $};
		\node[scale=0.9] (b) at (1,0) {$\Syn\times \Sigma_{\CSyn}\LSyn\times \Sigma_{\CSyn}\LSyn ^\op  $ };
		\node[scale=0.9] (c) at (0,-1) {$\Fscone $ };
		\node[scale=0.9] (d) at (1,-1) {$\Fam{\Set}\times \Fam{\Vect}\times \Fam{\Vect ^\op } $ };
		\draw[->] (a)--(b) node[midway,above] {$\left( \id , \Dsyn{-}, \Dsynrev{-}  \right) $ };
		\draw[->] (c)--(d) node[midway,below] {$\forpi$ };
		\draw[->] (a)--(c) node[midway,left] {$\forsem{-} $};
		\draw[->] (b)--(d) node[midway,right] {$\sem{-}\times \semt{-}\times \semtt{-} $};
	\end{tikzpicture}
\end{equation}
\end{therm}	
\begin{proof}
For each primitive type $\reals ^n $ and each primitive operation $\op$, we have that Equations \eqref{eq:reals-commutative-diagram1} and \eqref{eq:reals-commutative-diagram2}   hold by CHAD's soundness for primitives (by which we mean the assumptions of \ref{subsect:CHAD-SOUNDNESS-ASSUMPTIONS}). 
\begin{equation} \label{eq:reals-commutative-diagram1}
	\forpi{\left( \forsem{\reals ^n}\right) } = \left( \RR ^n , \left( \RR ^n, \cRR ^n \right) , \left( \RR ^n, \cRR ^n \right)\right)   = \left( \sem{\reals ^n}, \semt{\Dsyn{\reals ^n}}, \semtt{\sDsr{\reals ^n}}   \right)  
\end{equation} 
\begin{equation} \label{eq:reals-commutative-diagram2}
	\forpi{\left( \forsem{\op }\right) } = \left( \sem{\op } , \sDs{\sem{\op }  } , \sDsr{\sem{\op }  } \right)   = \left( \sem{\op }, \semt{\Dsyn{\op}}, \semtt{\Dsynrev{\op}}   \right)  
\end{equation}

Since  $\left( \sem{-}\times \semt{-}\times \semtt{-}  \right) \circ \left( \id , \Dsyn{-}, \Dsynrev{-}  \right) $ and $\forpi\circ \forsem{-}$ are (compositions) of strictly $\mu\nu $-polynomial preserving  bicartesian closed functors satisfying \eqref{eq:reals-commutative-diagram1} and \eqref{eq:reals-commutative-diagram2} for any ground type $\reals ^n$ and any primitive operation $\op$, we conclude that $\left( \sem{-}\times \semt{-}\times \semtt{-}  \right) \circ \left( \id , \Dsyn{-}, \Dsynrev{-}  \right) = \forpi\circ \forsem{-}$
by the universal property of $\Syn $ established in Corollary \ref{cor:universal-property-of-source-language}.
\end{proof}

\subsection{Correctness result}
We are now ready to establish the fundamental correctness result for both forward-mode and reverse-mode CHAD. Specifically, we prove that these techniques yield the correct derivatives for any well-typed program of the form $\var[1] : \ty[1] \tinf \trm[1] : \ty[2]$, where $\ty[1]$ and $\ty[2]$ are types constructed from sum and product types. 
\begin{therm}[Correctness of CHAD for tuples and variant tuples]\label{theorem:correctness-theorem-for-data-typess}
	Let	$\left( n_{j} \right) _{j\in J } $, $\left( m_l \right) _{l\in L } $,   
	$\left( \left( q_{(j,i)} \right) _{i\in \left\{ 1, \ldots, n_j\right\} }  \right)  _{j\in J }  $ and $\left( \left( s_{(l,t)} \right)_{t\in \left\{ 1, \ldots, m_l\right\} } \right)_{l\in L }   $  	be finite families of natural numbers. 
	
	For any well-typed program $ \var[1] : \ty[1] \tinf \trm[1] : \ty[2]$,
	where 
\begin{equation}
	\ty[1] = \coprod _{j\in J} \left( \prod _{i=1 }^{n_j} {\reals ^{q_{(j,i)}  } } \right) \qquad\mbox{and} \qquad \ty[2] = \coprod _{l\in L} \left( \prod _{t=1 }^{m_l} {\reals ^{s_{(l,t)}  } } \right),
\end{equation}	 
	we have that $\sem{\trm[1]}$ is differentiable. Moreover, \eqref{eq:derivateive-forward-modeee} and \eqref{eq:derivateive-reverse-modeee} hold. \\	
	\noindent\begin{minipage}{.5\linewidth} 
		\begin{equation}\label{eq:derivateive-forward-modeee}
			\semt{\Dsyn{\trm[1]} } = \Ds{\sem{\trm[1]} } 
		\end{equation}
	\end{minipage} 
	\begin{minipage}{.5\linewidth} 
		\begin{equation}\label{eq:derivateive-reverse-modeee}
			\semtt{\Dsynrev{\trm[1]} } = \Dsr{\sem{\trm[1]} } 
		\end{equation}
	\end{minipage}
\end{therm}
\begin{proof}
	Let $ t: 	\coprod\limits _{j\in J} \left( \prod\limits _{i=1 }^{n_j} {\reals ^{q_{(j,i)}  } } \right)\to  \coprod\limits _{l\in L} \left( \prod\limits _{t=1 }^{m_l} {\reals ^{s_{(l,t)}  } } \right) $ be a morphism in $\Syn $.
	By the commutativity of Diagram \ref{eq:diagram-correctness-proof}, the morphism $\left( \sem{t}, \semt{\Dsyn{t} }, \semtt{\Dsynrev{t} }  \right) $ in $\Fam{\Set}\times\Fam{\Vect}\times\Fam{\Vect ^\op} $
	satisfies $\forpi{\left( \forsem{t}\right)  }=\left( \sem{t}, \semt{\Dsyn{t} }, \semtt{\Dsynrev{t} }  \right) $.

	By Theorem \ref{theo-basic-result-for-LR},  we conclude that $\sem{t}$ is differentiable and $\left( \semt{\Dsyn{t} }, \semtt{\Dsynrev{t} }  \right) =\left( \Ds{\sem{t} } , \Dsr{\sem{t} }  \right)   $.
\end{proof}

\section{Inductive data types: $\mu$-polynomials}\label{sect:correctness-inductive-types}
We establish the correctness of CHAD for any well-typed program of the form
$$ \var[1] : \ty[1] \tinf \trm[1] : \ty[2], $$
where $\ty[1]$ and $\ty[2]$ are data types in our source language in \ref{subsect:aim-1-section-inductive-data-types}. 

It should be noted that our source language supports inductive types, which enable us to represent lists, trees, or other more complex inductive types.
In order to emphasize this fact, we refer to our data types as \textit{inductive data types}.

We start by clarifying the categorical semantics of inductive data types, which are referred to as $\mu$-polynomials and defined in \ref{subsect:definition-mu-polynomials}. We demonstrate in \ref{subsect:normal-form-inductive-data-types} how $\mu$-polynomials can be created from coproducts and finite products in concrete models that feature infinite coproducts. As a result, we can deduce that whenever $\ty[1]$ is an inductive data type, $\sem{\ty[1]}$ is an Euclidean family. This allows us to establish the specification and correctness of forward and reverse-mode CHAD for general inductive data types in \ref{subsect:aim-1-section-inductive-data-types}.

The definitions and results presented below heavily rely on the terminology, notation, and results established in Subsection \ref{SUBSECT:INDUCTIVE-COINDUCTIVE-TYPES} and Section \ref{sect:structure-preserving-functors}.

\subsection{$\mu$-polynomials}\label{subsect:definition-mu-polynomials}
In our source language, data types are constructed using tupling, cotupling, and the $\mu$-fixpoint operator. From a categorical semantic viewpoint, this implies that we want to examine objects that arise from products, coproducts, and initial algebras. Specifically, we consider the $\mu$-polynomials as defined below.
\begin{definition}[$\mu$-polynomials]
	The set $\MPoly $ of $\mu $-polynomial functors in $\Syn $ is the smallest set satisfying \ref{def:mu-polynomials(-1)},  \ref{def:mu-polynomials0}, \ref{def:mu-polynomials1}, \ref{def:mu-polynomials2}, and \ref{def:mu-polynomials3}. 
\begin{enumerate}[$\MPoly$1)]
	\item For every $k\in \NN $, every projection $\projection_t : \Syn ^k \to \Syn $ is an element $MPoly $. \label{def:mu-polynomials(-1)}
	\item For any $k\in\NN $, the constant functors $$\terminal : \Syn ^k \to \Syn , W\mapsto \terminal\quad \mbox{and}\quad \initial  : \Syn ^k \to \Syn , W\mapsto \initial $$
	belong to $\MPoly $.\label{def:mu-polynomials0}
	\item For any $k\in\NN $ and any primitive type $\reals ^n \in \sobjects{\Syn}$, the functor 
	$$ H_{\reals ^n }: \Syn ^k \to \Syn $$
	constantly equal to $\reals ^n$ belongs to $\MPoly$. \label{def:mu-polynomials1}
	\item If $H : \Syn ^k \to \Syn  $ and $J: \Syn ^k \to \Syn $
	are functors in $\MPoly $, then \eqref{eq:product-mupolynomials} and \eqref{eq:coproduct-mupolynomials}            belong to $\MPoly$.\label{def:mu-polynomials2}
\begin{equation} \label{eq:product-mupolynomials}
 \times \circ \left( H, J\right) : \Syn ^k \to \Syn , W\mapsto H(W)\times J(W) 
\end{equation} 
\begin{equation}  \label{eq:coproduct-mupolynomials}
\sqcup \circ \left( H, J\right) : \Syn ^k \to \Syn, W\mapsto H(W)\sqcup J(W) 
\end{equation} 
	\item If $k\in \NN - \left\{ 0 \right\} $ and $H : \Syn ^k\to \Syn $ belongs to $\MPoly $, then the parameterized initial algebra (initial algebra) $\mu H : \Syn ^{k-1}\to \Syn $ ($\mu H $) belongs to $\MPoly $.	\label{def:mu-polynomials3}
\end{enumerate}	 
An \textit{inductive data type} is a type $\ty[1]$ in our source language that corresponds to a initial algebra of a $\mu$-polynomial functor $E: \Syn \to \Syn $. 
\end{definition} 	

\subsection{$\mu$-polynomials in concrete models: a normal form}\label{subsect:normal-form-inductive-data-types}
Similarly to Euclidean families, in concrete models of our source language,  we can reduce every $\mu$-polynomial functor to a canonically isomorphic normal form.
More precisely, we have Theorem \ref{THEO:NORMAL-FORM-INDUCTIVE-TYPES}.

Let $G: \Syn \to \catD $ be a strictly cartesian closed functor that strictly preserves $\mu\nu$-polynomials. Given functors 
$H : \Syn ^k \to \Syn$ and $J : \catD^n \to \catD $, we say that $J $ is $(H,G)$-compatible if \eqref{DIAG:G-COMPATIBILITY} commutes.
\begin{equation} \label{DIAG:G-COMPATIBILITY}
	\begin{tikzcd}
		\Syn ^k \arrow[rrr, "{G^k}"] \arrow[swap,dd, "{ H }"] &&& 
		\catD ^k \arrow[dd, "{J}"] \\
		&&&\\
		\Syn \arrow[swap, rrr, "{G}"] &&& \mu \catD                               
	\end{tikzcd}
\end{equation}

In the result below, we denote $\NNN{n}\coloneqq \left\{ 1, \ldots , n\right\}$, for each $n\in\NN$. 
\begin{therm}\label{THEO:NORMAL-FORM-INDUCTIVE-TYPES}
	Let $\catD $ be a cartesian closed category with $\mu\nu $-polynomials and infinite coproducts. We assume that $G: \Syn \to \catD $ 
	is strictly cartesian closed functor that strictly preserves $\mu\nu$-polynomials. 
	
	If $H :\Syn ^n \to\Syn$ is a functor in $\MPoly $, then
	there is a quadruple $\left( J, \MUNORMAG{H}, \mathtt{m}, \MUCANG\right) $, where $F : \catD ^n\to \catD $ is an $(H,G)$-compatible functor,  $\mathtt{m} = \displaystyle\left( \mind{j, \mathtt{T}} \right) _{\left( j, \mathtt{T}\right)\in \left( \NNN{n}\cup \left\{0\right\} \right)\times {\TreeIndex}} $ is a countable family of natural numbers and 
\begin{equation} \label{eq:canonical-isomorphism-coproduct-inductive}	
\MUCANG _ {(Y_i)_{i\in\NNN{n}}} : F (Y_i)_{i\in\NNN{n}} \cong \coprod _{\mathtt{T}\in {\TreeIndex} }\left(  \MUNORMAG_{\mathtt{T}}^{\mind{0,\mathtt{T}}}\times \prod _{j=1}^{n} Y_j ^{\mind{j,\mathtt{T}} }\right) 
\end{equation} 
	is a natural isomorphism, where, for each $\mathtt{T}\in {\TreeIndex} $, 
	\begin{equation} \label{eq:products-form}
		\MUNORMAG_{\mathtt{T}} = \prod _{l\in {L_{\mathtt{T}}}} G\left(\reals ^{z_{(l,\mathtt{T} )} } \right)
	\end{equation} 	
for some finite family $\left( z_{(l,\mathtt{T} )} \right) _{l\in {L_{\mathtt{T}}} } $ of natural numbers.
\end{therm} 
\begin{proof}
	The result follows from induction over the definition of $\MPoly $. The only non-trivial part of the proof  is related to \ref{def:mu-polynomials3}, that is to say, the stability of  $\MPoly $
	under the parameterized initial algebras, which we sketch below. 
	
	Let $\tilde{H} :\Syn ^{n+1} \to\Syn$ be a member of $\MPoly $. We assume, by induction, that $\tilde{F}: \catD ^ {n+1}\to \catD $ satisfies the above. That is to say, 
	it is an $(H,G)$-compatible functor and we have a natural isomorphism
$$\tilde{F} (Y_i)_{i\in\NNN{n+1}} \cong \coprod _{r\in {\KI } }\left(  {\tilde{\MUNORMAG}_{r} }^{\sind{0,r}}\times \prod _{i=1}^{n+1} Y_i ^{\sind{i,r} }\right) . $$
where $\tilde{\MUNORMAG_{r}}$ is equal to some finite product 
$$ \prod _{l\in {L_{r} }} G\left(\reals ^{z_{(l, r )} } \right) .$$

It is clear that $\tilde{F} $ preserves colimits of $\omega $-chains. 
Hence, given $W = (W_i)_{i\in\NNN{n}}$,  $\mu\tilde{F} ^W = \mu\tilde{F}\left( W \right)   $
exists and is given by the colimit of the $\omega $-chain
\begin{equation}\label{eq:colimit-omeage-chain} 
\initial\rightarrow  \tilde{F} ^W \left( \initial \right) \rightarrow \left( \tilde{F} ^W\right) ^2 \left( \initial \right)\rightarrow \cdots  
\end{equation} 
 provided that it exists. 
 
 We claim that the colimit  \eqref{eq:colimit-omeage-chain}  indeed exists. More precisely, the colimit is given by the coproduct
 $$ \coprod _{q=0}^\infty S_v (W) $$
where $\left( S_v(W)\right) _{v\in\NN} $ is defined inductively by \ref{eq:inductive-induction-mupolynomials1} and  \ref{eq:inductive-induction-mupolynomials2}.

\begin{enumerate}[S1)]
\item Denoting by $\overline{K}_0\coloneqq \left\{r\in\KI\mbox{ such that }  \sind{n+1,r} = 0 \right\} $,
$$ S_0 (W)\coloneqq \coprod _{r\in {\overline{K}_0}} \left( \tilde{\MUNORMAG}_{r} \times \prod _{i=1}^{n} W_i ^{\sind{i,r} }\right) .  $$\label{eq:inductive-induction-mupolynomials1}
\item Denoting by $\overline{K}_{a}\coloneqq \left\{r\in\KI\mbox{ such that }  \sind{n+1,r} = a \right\} $,\label{eq:inductive-induction-mupolynomials2}
$$ S_{v+1}(W)\coloneqq \coprod _{a=1}^\infty\coprod _{r\in {\overline{K}_{a}}} \left(\left( S_v(W) \right) ^a\times \tilde{\MUNORMAG}_{r} \times \prod _{i=1}^{n} W_i ^{\sind{i,r} }\right). $$
\end{enumerate}
By the infinitely distributive property and the universal property of the coproduct and product, we conclude that there is a canonical isomorphism between
 $$\mu\tilde{F}\left( W \right) = \coprod\limits _{q=0}^\infty S_v (W) $$
and something of the form 
$\coprod\limits _{\mathtt{T}\in {\TreeIndex} }\left(  \MUNORMAG_{\mathtt{T}}^{\mind{0,\mathtt{T}}}\times \prod\limits _{j=1}^{n} Y_j ^{\mind{j,\mathtt{T}} }\right), $
as described in \eqref{eq:canonical-isomorphism-coproduct-inductive}.	

Since $G$ preserves $\mu\nu $-polynomials, we conclude that $\mu\tilde{F}$ is a $\left( \mu\tilde{H}, G\right)$-compatible satisfying the required conditions.
\end{proof}	

As consequence, we get:
\begin{corollary}\label{coro:normal-form-inductive-data-types}
		Let $\catD $ be a cartesian closed category with $\mu\nu $-polynomials and infinite coproducts. We assume that $G: \Syn \to \catD $ 
	is strictly cartesian closed functor that strictly preserves $\mu\nu$-polynomials. If $E : \Syn \to \Syn $ is an endofunctor 
	in $\MPoly $, then there is a canonical isomorphism 
\begin{equation} 
\MUNORMAG :  G\left( \mu E \right)\cong  \coprod _{l\in L}\left( \prod _{t=1 }^{m_l} {G\left( \reals ^{s_{(l,t)}} \right)     } \right) ,
\end{equation} 
where $\left( m_l \right) _{l\in L } $ and $\left( \left( s_{(l,t)} \right)_{t\in \left\{ 1, \ldots, m_l\right\} } \right)_{l\in L }$ 
are (possibly infinite) families of natural numbers.
\end{corollary}

\subsection{Correctness of CHAD for inductive data types, by logical relations}\label{subsect:aim-1-section-inductive-data-types}
Since the canonical isomorphisms $\MUNORMAG$ given in Corollary \ref{coro:normal-form-inductive-data-types} 
are indeed canonical in the sense that they are given by the composition of isomorphisms coming from the distributively property and universal property of (co)products,
we have that:
\begin{lemma}\label{lemma:NORMAL-INDUCTIVE-DATA-TYPE}
	Let $\ty[1]$ be an inductive data type as defined in \ref{subsect:definition-mu-polynomials}. It follows that there is a canonical isomorphism 
\begin{equation} \label{eq:canonical-isomorphism-diffeomorphism-inductive-data-type}
	\MUNORMAG _{\ty[1]} :  \forsem{ \ty[1] }\cong  \coprod _{l\in L}\left( \prod _{t=1 }^{m_l} { \forsem{ \reals   ^{s_{(l,t)} }  }} \right) ,
\end{equation} 
such that:
\begin{enumerate}[$\mathcal{C}$1)] 
 \item $\left( m_l \right) _{l\in L } $ and $\left( \left( s_{(l,t)} \right)_{t\in \left\{ 1, \ldots, m_l\right\} } \right)_{l\in L }$ 
are (possibly infinite) families of natural numbers;
\item $\underline{\MUNORMAG _{\ty[1]} }$ is a diffeomorphism;
\item $\MUNORMAG _{\ty[1]} = \left( \underline{\MUNORMAG _{\ty[1]} }, \sDs{\underline{\MUNORMAG _{\ty[1]} }}, \sDsr{\underline{\MUNORMAG _{\ty[1]} }}   \right)   $.
\end{enumerate} 
\end{lemma} 	
By making use of the canonical isomorphisms \eqref{eq:canonical-isomorphism-diffeomorphism-inductive-data-type}, we can prove our correctness theorem; namely:
\begin{therm}[Correctness of CHAD for tuples and variant tuples]\label{theorem:correctness-theorem-for-inductive-data-types}
	For any well-typed program $ \var[1] : \ty[1] \tinf \trm[1] : \ty[2]$,
	where  $\ty[1], \ty[2]$ are inductive data types, we have that $\sem{\trm[1]}$ is differentiable. Moreover, \eqref{eq:derivateive-forward-modee-inductive-data-types} and \eqref{eq:derivateive-reverse-mode-inductive-data-types} hold. \\	
	\noindent\begin{minipage}{.5\linewidth} 
		\begin{equation}\label{eq:derivateive-forward-modee-inductive-data-types}
			\semt{\Dsyn{\trm[1]} } = \Ds{\sem{\trm[1]} } 
		\end{equation}
	\end{minipage} 
	\begin{minipage}{.5\linewidth} 
		\begin{equation}\label{eq:derivateive-reverse-mode-inductive-data-types}
			\semtt{\Dsynrev{\trm[1]} } = \Dsr{\sem{\trm[1]} } 
		\end{equation}
	\end{minipage}
\end{therm}
\begin{proof}
	Let $ t: \ty[1]\to \ty[2]$ be a morphism in $\Syn $.
	By the commutativity of Diagram \ref{eq:diagram-correctness-proof} and Lemma \ref{lemma:NORMAL-INDUCTIVE-DATA-TYPE}, the morphism $\left( \sem{t}, \semt{\Dsyn{t} }, \semtt{\Dsynrev{t} }  \right) $ in $\Fam{\Set}\times\Fam{\Vect}\times\Fam{\Vect ^\op} $
	is such that $\forpi{\left( \MUNORMAG _{\ty[2]}\circ  \forsem{t}  \circ \MUNORMAG _{\ty[1]}^{-1} \right)  } $ is equal to 
	$$ \left( \underline{\MUNORMAG _{\ty[2]}} \circ  \sem{t} \circ 	\underline{\MUNORMAG _{\ty[1]}} ^{-1}, \sDs{\underline{\MUNORMAG _{\ty[2]}}} \circ \semt{\Dsyn{t} } \circ \sDs{\underline{\MUNORMAG _{\ty[1]}}} ^{-1},
	\sDsr{\underline{\MUNORMAG _{\ty[2]}}} \circ  \semtt{\Dsynrev{t} } \circ \sDsr{\underline{\MUNORMAG _{\ty[1]}}} ^{-1}  \right) .$$

	By Theorem \ref{theo-basic-result-for-LR},  we conclude that
	\begin{enumerate}[$\mathfrak{C}$1)]
		\item  $\underline{\MUNORMAG _{\ty[2]}} \circ  \sem{t} \circ 	\underline{\MUNORMAG _{\ty[1]}} ^{-1}$ is differentiable;
		\item  $\left( \sDs{\underline{\MUNORMAG _{\ty[2]}}} \circ \semt{\Dsyn{t} } \circ \sDs{\underline{\MUNORMAG _{\ty[1]}}} ^{-1}, 	\sDsr{\underline{\MUNORMAG _{\ty[2]}}} \circ  \semtt{\Dsynrev{t} } \circ \sDsr{\underline{\MUNORMAG _{\ty[1]}}} ^{-1}  \right) =\left( \Ds{\sem{t} } , \Dsr{\sem{t} }  \right)   $.
	\end{enumerate} 
By the chain rule, since $\underline{\MUNORMAG _{\ty[2]}} $ and $\underline{\MUNORMAG _{\ty[1]}} $ are diffeomorphisms, we conclude that $\sem{t}$ is differentiable and $\left( \semt{\Dsyn{t} }, \semtt{\Dsynrev{t} }  \right) =\left( \Ds{\sem{t} } , \Dsr{\sem{t} }  \right)   $.
\end{proof}

\section{Examples of reverse-mode CHAD}\label{sect:EXAMPLE}
We provide examples of reverse-mode CHAD computation of derivatives, with a focus on computing derivatives of functions involving inductive types. In particular, we consider the simplest example of an inductive type: the type of non-empty lists of real numbers, denoted by $\nonList{\reals}$.

We present three examples. The function $\summ :\nonList{\reals}\to\reals  $ that computes the sum of elements of a list; $\productt :\nonList{\reals}\to\reals $ that gives the product of elements of a list; and the polynomial evaluator $\evp :\nonList{\reals}\to\reals  $. The semantics of these functions are roughly described below: 
\begin{eqnarray} 
	\sem{\summ} :  \sem{\nonList{\reals}}\to\sem{\reals}, & \qquad & \left[ a_0, \ldots , a_n\right]\mapsto a_0+a_1+\cdots + a_n \\
	\sem{\productt} : \sem{\nonList{\reals}} \to \sem{\reals}, &\qquad&  \left[ a_0, \ldots , a_n\right]\mapsto a_0a_1\cdots a_n\\
	\sem{\evp} : \sem{\nonList{\reals}} \to \sem{\reals}, & \qquad & \left[ a_0, \ldots , a_n, v\right]\mapsto  a_0  + a_1 v + \cdots + a_nv^n
\end{eqnarray} 	

The examples presented below heavily rely on the terminology, notation, and results established in Subsection \ref{SUBSECT:INDUCTIVE-COINDUCTIVE-TYPES}, Section \ref{sect:structure-preserving-functors} and Section \ref{sect:correctness-inductive-types}.

\subsection{The derivative of $0$}
In order to express the polynomial evaluator, we assume that we have a morphism $0: \reals\to \reals $
whose semantics correspond to the function $0 : \RR \to \RR $ constantly equal to $0\in\RR $. 

The morphism $0: \reals\to \reals $ can be either a primitive operation, or a function obtained by composing
$$ \reals \to \terminal \xto{0} \reals ,$$
where the constant $0:\terminal\to\reals $ would be taken to be the primitive operation.
Either way, by our semantic assumptions of \ref{subsect:CHAD-SOUNDNESS-ASSUMPTIONS}, we
get that 
\begin{equation}
		\semtt{\Dsynrev{0}} : \left( \RR , \cRR \right) \to \left(\RR , \cRR \right) 
\end{equation}
is the morphism in $\Fam{\Vect ^\op} $ defined by the pair $\left( 0, 0'\right) $ where, for each $a\in \RR $, 
$0'_a : \RR \to \RR $ is the linear transformation constantly equal to $0$.

\subsection{The derivatives of $(+)$ and $(\cdot)$}
We assume that 
$$(\cdot) : \reals \times \reals \to \reals\qquad\mbox{ and }\qquad(+) : \reals\times \reals \to \reals $$
are primitive operations in the source-language whose semantics are given, respectively, by the addition $\plus : \RR \times \RR \to \RR $ and multiplication $\multI : \RR \times \RR \to \RR$.
Since  $\left( +\right)$ and $\left( \cdot\right)$ are primitive operations in the source-language, $\Dsynrev{+}$ and $\Dsynrev{\cdot}$ are set by definition. 

By our semantic assumptions as per \ref{subsect:CHAD-SOUNDNESS-ASSUMPTIONS}, we have:
\begin{enumerate}
	\item[$(+)$] the morphism $\semtt{\Dsynrev{+}} : \left( \RR , \cRR \right) \to \left( \RR\times \RR  , \cRR \times \cRR \right) $ of $\Fam{\Vect ^\op }$ is defined by  \label{eq:derivative-of-+}
	\begin{equation} 
		\semtt{\Dsynrev{+}} = \left( \plus , \plus '  \right) 
	\end{equation} 
	where $\plus (a,b) = a+b $ and, for each $(a,b)\in \RR\times \RR$, 
	$\plus _{(a,b)} ' : \RR \to \RR \times \RR $ is defined by $x\mapsto  \left( x, x\right)  $. 
	\item[$(\cdot)$] the morphism $\semtt{\Dsynrev{\cdot}} : \left( \RR , \cRR \right) \to \left( \RR\times \RR  , \cRR \times \cRR \right) $ of $\Fam{\Vect ^\op }$ is defined by  \label{eq:derivative-of-.}
	\begin{equation} 
		\semtt{\Dsynrev{\cdot}} = \left( \multI , \multI  '  \right) 
	\end{equation} 
	where $\multI (a,b) = ab $ and, for each $(a,b)\in \RR\times \RR$, 	$\multI _{(a,b)} ' : \RR \to \RR \times \RR $ is defined by $x\mapsto  \left( bx, ax\right)  $. 
\end{enumerate}

\subsection{Type of non-empty lists of real numbers in $\Syn $ and $\Fam{\Vect ^\op}$} \label{subsect:derivative-of-fold-fam-vect}
As our examples mainly concern the type $\nonList{\reals}$ of non-empty lists of real numbers in $\Syn$, let us first recall its categorical semantics, and discuss its image under the reverse-mode CHAD $\Dsynrev{-}$.

The $\nonList{\reals}\coloneqq \mu \liste $ where the endofunctor $\liste$ is defined by
\begin{eqnarray}
	\liste : &\Syn& \to \Syn \label{eq:endofunctor-list-Syn}\\
	&W& \mapsto  \reals \sqcup W\times\reals. \nonumber
\end{eqnarray} 
Denoting by $\listee : \Fam{\Vect ^\op}\to \Fam{\Vect ^\op} $  the endofunctor defined by 
\begin{equation} 
\listee (W,w) = \left( \RR , \cRR \right) \sqcup (W,w) \times \left( \RR , \cRR \right), 
\end{equation}
we conclude that 
\begin{eqnarray}
	\semtt{\Dsynrev{\nonList{\reals}}}& = &\mu \listee \\
	& = & \left(\coprod _{j\in \NN -\left\{0\right\}} \RR ^j  , \langle \cRR ^j \rangle _{j\in \NN -\left\{0\right\}} : \coprod _{j\in \NN -\left\{0\right\}} \RR ^j   \to \Vect \right) \nonumber
\end{eqnarray}
by the structure-preserving property of $\Dsynrev{-}$.

Let $\langle \left( \zeta, \zeta '\right)    , \left( \beta, \beta '\right) \rangle   :  \left( \RR , \cRR \right) \sqcup (W,w) \times \left( \RR , \cRR \right)\to \left( W,w\right) $ 
be the morphism
in $\Fam{\Vect ^\op }$ induced by given morphisms 
\begin{eqnarray*}
	\left( \zeta, \zeta '\right) : & \left( \RR , \cRR \right) & \to \left( W,w\right)\\
	\left( \beta, \beta '\right) : & (W,w) \times \left( \RR , \cRR \right) & \to \left( W,w\right) 
\end{eqnarray*}
in $\Fam{\Vect ^\op} $.
Denoting
\begin{equation}
\left( \xi, \xi '   \right) \coloneqq	\fold _{\listee}\left( \left( W, w \right) , \langle \left( \zeta, \zeta '\right)    , \left( \beta , \beta '\right) \rangle  \right)   : \mu \listee \to \left( W, w\right), 
\end{equation}
we have the following:
\begin{enumerate} 
\item[$\xi$)]  $\xi  : \coprod\limits _{j\in \NN -\left\{0\right\}} \RR ^j \to W $ is induced by the family
\begin{equation}
 \xi =	\langle \xi _j : \RR^j\to W  \rangle _{j\in \NN -\left\{0\right\}}  
\end{equation}
defined by $\xi _1 =  \zeta : \RR  \to W   $ and $\xi _{j+1} =  \beta \circ \left( \xi _j \times \id_{\RR } \right)  $;
\item[$\xi _r '$)] for each $r \in \RR  \subset  \coprod\limits _{j\in \NN -\left\{0\right\}} \RR ^j $, the component $$\xi '_r : w\circ \xi (r) \to \RR $$ is given by $\zeta '_r : w\circ \zeta (r) \to \RR $.
\item[$\xi '$)] for each $p = \left(  p_\ast, p_0 \right)\in \RR ^k\times \RR   =  \RR ^{k+1} \subset  \coprod\limits _{j\in \NN -\left\{0\right\}} \RR ^j $, 
\begin{equation}
	\xi ' _p =   \left( \xi ' _{p_\ast} \times \id_{\RR } \right)  \circ \beta '_{\left( \xi \left( p_\ast\right) , p_0  \right)} . \label{eq:derivative-of-fold-Fam}
\end{equation} 
\end{enumerate}

\subsection{Reverse-mode CHAD derivative of sum}
The function $\summ : \nonList{\reals}\to\reals $ computes the sum of the elements of a non-empty list of real numbers.
We can express $\summ $ in $\Syn$ by:
\begin{equation}
	\summ  \coloneqq   \fold _{\liste}\left( \reals , \langle \id _{\reals} , \left( + \right) \rangle : \reals \sqcup \reals\times \reals \to \reals  \right) : \mu\liste = \nonList{\reals} \to \reals  .
\end{equation}
By the structure-preserving property of CHAD, we conclude that:
\begin{equation} 
	\semtt{\Dsynrev{\summ }} =\fold _{\listee} \left( \left( \RR ,\cRR \right) , \langle \id _{\left( \RR ,\cRR \right) } , \left( \plus , \plus '\right) \rangle : \left( \RR ,\cRR \right) \sqcup \left( \RR \times \RR ,\cRR ^2 \right) \to \left( \RR ,\cRR \right)  \right).
\end{equation} 	
Therefore, by \ref{eq:derivative-of-+} and \ref{subsect:derivative-of-fold-fam-vect}, we conclude that, denoting $\semtt{\Dsynrev{\summ }} = \left( \sem{\summ }, \sem{\summ } ' \right) $, we have the following:
\begin{enumerate}[A)] 
	\item the function 
\begin{equation}
	\sem{\summ } : \coprod _{j\in \NN -\left\{0\right\}} \RR ^j   \to \RR
\end{equation}  
is induced by the family $\langle\sem{\summ }_j : \RR ^j \to \RR\rangle _{j\in \NN -\left\{0\right\}} $ defined by $$\sem{\summ }_j\left(w_1, \ldots , w_ j\right) = \sum _{i=1}^j w_i ;$$
	\item for each $p\in \RR ^k \subset \coprod\limits _{j\in \NN -\left\{0\right\}}  \RR ^j$, we have that
\begin{equation}
	\sem{\summ }'_p :  \RR    \to \RR ^k
\end{equation} 	
is defined by $x\mapsto\left( x, \ldots , x \right)$.
\end{enumerate}

\subsection{Reverse-mode CHAD derivative of product}
The function $\productt : \nonList{\reals}\to\reals $ computes the product of the elements of a non-empty list of real numbers.
We can express $\productt $ in $\Syn$ by:
\begin{equation}
	\productt  \coloneqq   \fold _{\liste}\left( \reals , \langle \id _{\reals} , \left( \cdot \right) \rangle : \reals \sqcup \reals\times \reals \to \reals  \right) : \mu\liste = \nonList{\reals} \to \reals .
\end{equation}
By the structure-preserving property of CHAD, we have that:
\begin{equation} 
	\semtt{\Dsynrev{\summ }} =\fold _{\listee} \left( \left( \RR ,\cRR \right) , \langle \id _{\left( \RR ,\cRR \right) } , \left( \multI , \multI '\right) \rangle : \left( \RR ,\cRR \right) \sqcup \left( \RR \times \RR ,\cRR ^2 \right) \to \left( \RR ,\cRR \right)  \right).
\end{equation} 	
Therefore, by \ref{eq:derivative-of-.} and \ref{subsect:derivative-of-fold-fam-vect}, we conclude that, denoting $\semtt{\Dsynrev{\productt }} = \left( \sem{\productt }, \sem{\productt } ' \right) $, we have the following:
\begin{enumerate}[I)] 
	\item the function 
	\begin{equation}
		\sem{\productt} : \coprod _{j\in \NN -\left\{0\right\}} \RR ^j   \to \RR
	\end{equation}  
	is induced by the family $\langle\sem{\productt}_j : \RR ^j \to \RR\rangle _{j\in \NN -\left\{0\right\}} $ defined by $$\sem{\productt}_j\left(w_1, \ldots , w_ j\right) = \prod _{i=1}^j w_i;   $$
	\item for each $$p=\left( p_1, \ldots , p_k\right) \in \RR ^k \subset \coprod _{j\in \NN -\left\{0\right\}} \RR^j ,$$ we have that
	\begin{equation}
		\sem{\productt}'_p :  \RR    \to \RR ^k
	\end{equation} 	
	is defined by $x\mapsto\left( \hat{p_1} x,  \hat{p_2} x, \ldots , \hat{p_k} x \right)$, where
	$$ \hat{p_t} = \prod _{i\in\left\{1, \ldots , k \right\} -\left\{t\right\} } p_i . $$
\end{enumerate}

\subsection{Reverse-mode CHAD derivative of $(+)\circ  ( \id _{\reals}\times (\cdot) )$}
In order to compute the derivative of the polynomial evaluator as expressed in \eqref{eq:polynomial-evaluator-as-a-composition}, we need to compute the derivative of the function
\begin{equation}
	(+)\circ  ( \id _{\reals}\times (\cdot) ) : \reals \times \reals \times \reals \to \reals 
\end{equation}
whose semantics is defined by $(a,b,c)\mapsto a + bc$.

We use the structure-preserving property of $\Dsynrev{-}$ to compute $\semtt{\Dsynrev{\id _{\reals}\times (\cdot)}}$. This gives us:
\begin{eqnarray*}
	\semtt{\Dsynrev{ \id _{\reals}\times (\cdot)     }} &=& \semtt{\Dsynrev{\id _{\reals}}}\times \semtt{\Dsynrev{\cdot}}\\
	& = &       \id _{\left( \RR , \cRR \right) }\times \left( \multI , \multI  '  \right)  \\
	& = &  \left( \overline{\multI}, \overline{\multI} '\right) 
\end{eqnarray*}
in $\Fam{\Vect ^\op}$, where $\overline{\multI} : \RR\times \RR \times \RR \to \RR\times\RR $ is defined by $\left( a,b,c\right)\mapsto \left( a, bc\right)  $
and, for each $\left( a,b,c\right)\in \RR\times \RR\times \RR $,  
\begin{equation}
	\overline{\multI}'_{(a,b,c)} : \RR\times\RR  \to \RR\times \RR\times \RR, \qquad
	 \left( w, x\right) \mapsto  \left( w, cx, bx\right). 
\end{equation}
We conclude, then, that 
\begin{eqnarray*}
	\semtt{\Dsynrev{(+)\circ \left( \id _{\reals}\times (\cdot)  \right)   }} & = & \semtt{\Dsynrev{+}}\circ \semtt{\Dsynrev{\left( \id _{\reals}\times (\cdot)  \right)   }}\\
	& = & \semtt{\Dsynrev{+ }}\circ \left( \semtt{\Dsynrev{\id _{\reals}}}\times \semtt{\Dsynrev{\cdot}}     \right) \\
	& = &  \left( \plus , \plus '  \right) \circ \left(  \id _{\left( \RR , \cRR \right) }\times \left( \multI , \multI  '  \right)   \right) 
\end{eqnarray*}
 is equal to the morphism
\begin{equation}
\left( \overline{\plus}, \overline{\plus}'   \right) \coloneq	\left( \plus \circ \left( \id _{\RR}\times \multI \right), \left( \plus \circ \left( \id _{\RR}\times \multI \right) \right) '    \right) : \left( \RR \times \RR \times \RR ,   \cRR ^3 \right) \to  \left( \RR , \cRR \right) 
\end{equation}
where, for each $\left( a, b, c \right)\in\RR\times\RR\times\RR$,  
\begin{eqnarray*}
	\overline{\plus}'_{(a,b,c)} = \left( \plus \circ \left( \id _{\RR}\times \multI \right) \right) ' _{(a,b,c)} : & \RR & \to \RR\times \RR \times \RR\\
	& x & \mapsto \left( x, cx, bx\right) . 
\end{eqnarray*}

\subsection{Reverse-mode CHAD derivative of polynomial evaluator}
For convenience, we represent a pair $\left( p(x), v\right) $, where 
 \begin{equation} 
p(x) =	a_0  + \cdots + a_nx^n  
\end{equation} 
is a polynomial and $v\in \RR$, by a non-empty list $\left[ a_0 , \ldots , a_n, v\right]$. With this notation, the polynomial evaluator 
$$ \evp : \nonList{\reals}\to \reals $$
can be expressed as the composition 
\begin{equation}\label{eq:polynomial-evaluator-as-a-composition}
	\mu\liste = \nonList{\reals} \xto{\fold _ \liste \left( \reals \times \reals , \langle \left(  0, \id _\reals \right),   \left(   (+)\circ  ( \id _{\reals}\times (\cdot) ), \projection _{3}\right) \rangle \right)  }\reals\times\reals \xto{\projection _1} \reals .
\end{equation}
It should be noted that $\sem{\langle \left(  0, \id _\reals \right),   \left(   (+)\circ  ( \id _{\reals}\times (\cdot) ), \projection _{3}\right) \rangle  }$ is the morphism
$$ \RR\sqcup \left( \RR\times\RR\times \RR\right) \to \RR\times \RR  $$
in $\Fam{\Set}$ induced by the morphism $\RR \ni r\mapsto \left( 0, r\right)  $ and $ \RR\times\RR\times \RR  \ni \left( a,b,c\right) \mapsto \left( c, a + bc\right) $ and, hence, indeed, 
$$\sem{\evp}\left( a_0 , \ldots , a_k , v\right) = \left( a_0  + \cdots + a_kv^k , v\right)  $$ for each $\left( a_0 , \ldots , a_k , v\right)\in \RR ^k \subset \coprod _{j\in \NN -\left\{0\right\}} \RR ^j $.

By the structure-preserving property of $\Dsynrev{-}$, we conclude that $\semtt{\Dsynrev{\evp}}$ is given by the composition 
\begin{equation}\label{eq:derivative-polynomial-evaluator-as-a-composition}
	\mu\listee \xto{\fold _ \listee \left( \RR \times \RR , \langle \left( \left( 0 , 0'\right) ,   \left( \id _\RR , \id _\RR '\right)    \right),   \left(   \left( \overline{\plus} , \overline{\plus} '  \right) , \left( \projection _{3} , \projection _{3}'\right)  \right)  \rangle \right)  }\left( \RR ^2, \cRR ^2 \right)   \xto{\left( \projection _{1} , \projection _{1}'\right) } \left( \RR , \cRR \right)  ,
\end{equation}
where $\left( \projection_{3}, \projection_{3}'\right) $ and $\left( \projection_{1}, \projection_{1}'\right) $ denote the respective projections in $\Fam{\Vect ^\op }$.
By \ref{subsect:derivative-of-fold-fam-vect}, denoting  
\begin{equation}
\left( \mathsf{g} , \mathsf{g}'\right) \coloneqq{\fold _ \listee \left( \RR \times \RR , \langle \left( \left( 0 , 0'\right) ,   \left( \id _\RR , \id _\RR '\right)    \right),   \left(   \left( \overline{\plus} , \overline{\plus} '  \right) , \left( \projection _{3} , \projection _{3}'\right)  \right)  \rangle \right)  },
\end{equation}
we have the following. 
\begin{enumerate}[a)]
\item The function 
\begin{equation}
\mathsf{g} : \coprod _{j\in \NN -\left\{0\right\}} \RR ^j  \to  \RR \times \RR
\end{equation} 
takes each $$\left( a_0, \ldots , a_k, v \right)\in \RR ^{k+1}\subset \coprod \limits_{j\in \NN -\left\{0\right\}} \RR ^j $$ to $\left( a_0 +a_1v + \cdots + a_k v^k, v \right)\in \RR\times \RR $.
\item For each $\left( a_0, \ldots , a_k, v \right)\in \RR ^{k+1}\subset \coprod _{j\in \NN -\left\{0\right\}} \RR ^j $, 
\begin{equation}
	\mathsf{g}' _{\left( a_0, \ldots , a_k, v \right)}: \RR \times \RR  \to  \RR ^{k+1}
\end{equation}  
is defined by $\left( x,y \right)\mapsto \left( x, vx, v^2x, \ldots , v^k x, \left( a_1 + 2\cdot a_2 v + 3\cdot a_3v^2 + \cdots + k a_k v^{k-1}  \right) x +y\right) $.
\end{enumerate}
Therefore $\left( \sem{\evp} , \semtt{\evp} '\right) \coloneq\semtt{\Dsynrev{\evp}} =  \left( \projection _{1} , \projection _{1}'\right) \circ  \left( \mathsf{g} , \mathsf{g}'\right)$ is such that, for each 
$$\left( a_0, \ldots , a_k, v \right)\in \RR ^{k+1}\subset \coprod _{j\in \NN -\left\{0\right\}} \RR ^j, $$
$\semtt{\evp} '_{\left( a_0, \ldots , a_k, v \right)} : \RR  \to  \RR ^{k+1} $ is defined by 
$$x\mapsto \left( x, vx, v^2x, \ldots , v^k x, \left( a_1 + 2\cdot a_2 v + 3\cdot a_3v^2 + \cdots + k a_k v^{k-1}  \right) x \right) .$$

\section{Practical considerations}\label{sec:practical-considerations}
Despite the theoretical approach this paper has taken, our motivations for this line 
of research are very applied: we want to achieve efficient and correct reverse AD on expressive 
programming languages.
We believe this paper lays some of the necessary theoretical groundwork to achieve that goal. 
We are planning to address the practical considerations around achieving efficient
implementations of CHAD in detail in a dedicated applied follow-up paper.
However, we still sketch some of these considerations in this section to convey that the methods 
described in this paper are not merely of theoretical interest.

\subsection{Addressing expression blow-up and sharing common subcomputations}
We can observe that our source-code transformations of Appendix \ref{sec:inefficient-ad-transformation}
can result in code-blowup due to the interdependence of the transformations $\Dsyn[\vGamma]{-}_1$ and 
$\Dsyn[\vGamma]{-}_2$ (and $\Dsynrev[\vGamma]{-}_1$ and $\Dsynrev[\vGamma]{-}_2$, respectively) on programs.
This is why, in \S\ref{sec:AD-transformations}, we have instead defined
a single code transformation on programs $\Dsyn[\vGamma]{-}$ for forward mode and $\Dsynrev[\vGamma]{-}$ for reverse mode that simultaneously computes the primals and 
(co)tangents and shares any subcomputations they have in common.
These more efficient CHAD transformations are still representations of the canonical CHAD functors $\Dsyn{-}:\Syn\to \Sigma_\CSyn\LSyn$ and 
$\Dsynrev{-}:\Syn\to \Sigma_\CSyn \LSyn^{op}$
in the sense that $\Dsyn[\vGamma]{\trm}\bepeq \tPair{\Dsyn[\vGamma]{\trm}_1}{\lfun{\lvar}\Dsyn[\vGamma]{\trm}_2}$ and $\Dsynrev[\vGamma]{\trm}\bepeq \tPair{\Dsynrev[\vGamma]{\trm}_1}{\lfun{\lvar}\Dsynrev[\vGamma]{\trm}_2}$ and hence are 
equivalent to the infficient CHAD transformations from the point of view of denotational semantics and correctness.

We can observe that the efficient CHAD code transformations $\Dsyn[\vGamma]{-}$  and $\Dsynrev[\vGamma]{-}$ have the 
property that the transformation $\Dsyn[\vGamma]{C[\trm_1,\ldots,\trm_n]}$ (resp. $\Dsynrev[\vGamma]{C[\trm_1,\ldots,\trm_n]}$) of a term former $C[\trm_1,\ldots,\trm_n]$ that takes $n$ arguments $\trm_1$, \ldots, $\trm_n$ (e.g., the pair constructor $C[\trm_1,\trm_2]=\tPair{\trm_1}{\trm_2}$, which takes two arguments $\trm_1$ and $\trm_2$) is a piece of code that uses the CHAD transformation $\Dsyn[\vGamma]{\trm_i}$ (resp. $\Dsynrev[\vGamma]{\trm_i}$) of each subterm $\trm_i$ exactly once.
This has as a consequence the following important compile-time complexity result that is a necessary condition if this AD technique is to scale up to large code-bases.
\begin{corollary}[No code blow-up]
\label{cor:no-code-blow-up}
The size of the code of the CHAD transformed programs $\Dsyn[\vGamma]{\trm}$ and $\Dsynrev[\vGamma]{\trm}$  grows linearly with the size of the original source program $\trm$.
\end{corollary}
While we have taken care to avoid recomputation as much as possible in defining these code transformations by sharing results of subcomputations through $\mathbf{let}$-bindings, the run-time complexity of the generated code remains to be studied.

\subsection{Removing dependent types from the target language}
In this paper, we have chosen to work with a dependently typed target language,
as this allows our AD transformations to correspond as closely as possible to the 
conventional mathematics of differential geometry, in which spaces of tangent and cotangent vectors form (non-trivial) bundles over the space of primals.
For example, the dimension of the space of (co)tangent vectors to a sum
$\RR^n\sqcup \RR^m $ is either $n$ or $m$, depending on whether the base point (primal) is chosen in the left or right component. 
An added advantage of this dependently typed approach is that it leads to a cleaner categorical story in which all 
$\eta$-laws are preserved by the AD transformations and standard categorical logical relations techniques can be used in the correctness proof.

That said, while the dependent types we presented give extra type safety that simplify mathematical foundations and the correctness argument underlying our AD techniques, nothing breaks if we keep the transformation on programs the same and simply coarse grain the types by removing any type dependency. 
This may be desirable in practical implementations of the algorithms as most practical programming languages have either no or only limited support for type dependency.

To be precise, we can perform the following coarse-graining transformation $\deptrans{(-)}$ on the types of the target language, which removes all type dependency:
\[
\begin{array}{lll}
\deptrans{\ltvar} & \defeq \ltvar \\
\deptrans{\creals^n} &\defeq& \creals^n\\
\deptrans{\lUnit} & \defeq & \lUnit\\
\deptrans{(\cty\t* \cty[2])} & \defeq & \deptrans{\cty}\t* \deptrans{\cty[2]}\\
\deptrans{(\Pi \var:\ty.\cty[2])} & \defeq & \Pi \var:\deptrans{\ty}.\deptrans{\cty[2]}\\
\deptrans{(\Sigma \var:\ty.\cty[2])} & \defeq &  \Sigma\var:\deptrans{\ty}.\deptrans{\cty[2]}\\
\end{array}\qquad\qquad\!\!
\begin{array}{lll}
\deptrans{(\vMatch{\trm}{\ell_1\var_1\To\cty_1\mid\cdots\ell_n\var_n\To\cty_n})} &\defeq  & \deptrans{\cty_1}\ovplus \cdots\ovplus \deptrans{\cty_n}\\
\deptrans{(\llfp{\ltvar}\cty)} & \defeq &  \llfp{\ltvar}{\deptrans{\cty}}\\
\deptrans{(\lgfp{\ltvar}\cty)} & \defeq &  \llfp{\ltvar}{\deptrans{\cty}}\\
\deptrans{(\cty\multimap \cty[2])} & \defeq & \deptrans{\cty}\multimap \deptrans{\cty[2]}\\
\deptrans{(\Pi \var:\ty.\ty[2])} & \defeq & \Pi\var:\deptrans{\ty}.\deptrans{\ty[2]}\\
\deptrans{(\Sigma \var:\ty.\ty[2])} & \defeq &  \Sigma\var:\deptrans{\ty}.\deptrans{\ty[2]}.
\end{array}    
\]
In fact, seeing that  $(\vMatch{\trm}{\ell_1\var_1\To\cty_1\mid\cdots\ell_n\var_n\To\cty_n})$-types 
were the only source of type dependency in our language while these are translated to non-dependent types, all $\Pi$- and 
$\Sigma$-types are simply translated to powers, copowers, function types and product types:
\[
    \begin{array}{lll}
\deptrans{(\Pi \var:\ty.\cty[2])} & = & \deptrans{\ty}\To\deptrans{\cty[2]}\\
\deptrans{(\Sigma \var:\ty.\cty[2])} & = &  !\deptrans{\ty}\otimes\deptrans{\cty[2]}
    \end{array}
    \qquad\qquad 
    \begin{array}{lll}
        \deptrans{(\Pi \var:\ty.\ty[2])} & = & \deptrans{\ty}\To\deptrans{\ty[2]}\\
        \deptrans{(\Sigma \var:\ty.\ty[2])} & = &  \deptrans{\ty}\t*\deptrans{\ty[2]}.
            \end{array}  
\] 
Our translation $\deptrans{(-)}$ is the identity on programs.

The types $\cty_1\ovplus \cdots\ovplus \cty_n$ require some elaboration.
We give this in the next section where we explain how to implement all required linear types and their terms in a standard functional programming language.

\subsection{Removing linear types from the target language}
\subsubsection{Basics}
As discussed in detail in \citep{vakar2021chad,vakar2020reverse} and demonstrated in the Haskell implementation available at \url{https://github.com/VMatthijs/CHAD}, the types $\creals^n$, $\lUnit$, 
$\cty\t* \cty[2]$, $\ty\To\cty[2]$, $!\ty\otimes \cty[2]$ and
$\cty\multimap\cty[2]$ (and, obviously, the ordinary Cartesian function and product types $\ty\To\ty[2]$ and $\ty\t*\ty[2]$) together with their terms can all be implemented in a standard functional language.
The core idea is to implement $\cty$ as the type $\lintrans{\cty}$:
\[
\begin{array}{lll}
\lintrans{\creals^n } & \defeq & \reals^n\\
\lintrans{\lUnit} & \defeq & \Unit\\
\lintrans{(\cty\t* \cty[2])} & \defeq & \lintrans{\cty}\t*\lintrans{\cty[2]}\\
\end{array}\qquad \qquad
\begin{array}{lll}
\lintrans{(\ty\To\cty[2])} & \defeq & \lintrans{\ty}\To\lintrans{\cty[2]}\\
\lintrans{(!\ty\otimes\cty[2])} & \defeq & [(\lintrans{\ty},\lintrans{\cty[2]})]\\
\lintrans{(\cty\multimap\cty[2])} & \defeq & \lintrans{\cty}\To\lintrans{\cty[2]}.
\end{array}
\]
Crucially, we implement the copowers as abstract types that can under the 
hood be lists of pairs $[(\lintrans{\ty},\lintrans{\cty[2]})]$ 
and we implement the linear function types as abstract types that 
can under the hood be plain functions $\lintrans{\cty}\To\lintrans{\cty[2]}$.
As discussed in \citep{vakar2021chad,vakar2020reverse} and 
shown in the Haskell implementation, this translation extends to 
programs and leads to a correct implementation of CHAD on 
a simply typed $\lambda$-calculus.

We explain here how to extend this translation to implement 
the extra linear types $\cty_1\ovplus \cdots\ovplus\cty_n$,
$\llfp{\ltvar}{\cty}$ and $\lgfp{\ltvar}{\cty}$
required to perform AD on source languages that additionally use
sum types, inductive types and coinductive types.

\subsubsection{Linear sum types $\cty_1\ovplus \cdots\ovplus\cty_n$}
We briefly outline three possible implementations $\lintrans{(\cty_1\ovplus \cdots\ovplus\cty_n)}$ of the linear sum types\linebreak $\cty_1\ovplus \cdots\ovplus\cty_n$:
\begin{enumerate}
\item as a finite (bi)product $\lintrans{\cty_1}\t* \cdots\t* \lintrans{\cty_n}$;
\item as a finite lifted sum $\set{Zero\mid Opt_1\,\lintrans{\cty_1}\mid\cdots\mid Opt_n\,\lintrans{\cty_n}}$;
\item as a finite sum $\set{Opt_1\,\lintrans{\cty_1}\mid\cdots\mid Opt_n\,\lintrans{\cty_n}}$.
\end{enumerate}
Approach 1 has the advantage that we can keep the implementation total.
As demonstrated in Appendix \ref{sec:simply-typed-coprod},
this allows us the easily extend the logical relations argument for the correctness of the applied implementation of \citep{vakar2021chad,vakar2020reverse} (in actual Haskell, available at \url{https://github.com/VMatthijs/CHAD}).
Categorically, what is going on is that, for a locally indexed category 
$\catL:\catC^{op}\to\Cat$ with indexed finite biproducts and $\multimap$-types, $(X_1\sqcup \cdots \sqcup X_n, A_1\times \cdots\times A_n)$ is a weak coproduct of $(X_1,A_1)$, ...,
$(X_n, A_n)$ in both $\Sigma_\catC\catL$ and $\Sigma_\catC \catL^{op}$:
i.e. a coproduct for which the $\eta$-law may fail.
The logical relations proof of Appendix \ref{sec:simply-typed-coprod} lifts these weak coproducts to the subscone, demonstrating that this implementation of CHAD for 
coproducts indeed computes semantically correct derivatives.

This approach was first implemented in the Haskell implementation of CHAD.
However, a major downside of approach 1 is its inefficiency: it represents (co)tangents to a coproducts as tuples of (co)tangents to the component spaces, all but one of which are known to be zero.
This motivates approaches 2 and 3.

Approach 2 exploits this knowledge that all but one component of the (co)tangent
space are zero by only storing the single non-zero component, corresponding to the connected component the current primal is in.
To see the correctness of this approach, we can add an extra error element $\bot$ to all our linear types $\Dsyn{\ty}_2$ and $\Dsynrev{\ty}_2$, for which $\bot+x=\bot$, and do a manual (total) logical relations proof.
We can then note that we can also leave out the error element of the data type and throw actual errors at runtime.

We pay for this more efficient representation in two ways:
\begin{itemize}
\item addition on the (co)tangent space is defined by 
\[
Zero + x = x \qquad x + Zero = x \qquad Opt_i(\trm) + Opt_i(\trm[2])= Opt_i(\trm+\trm[2])
\]
and hence is a partial operation that throws an error if we try to add $Opt_i (\trm)+ Opt_j(\trm[2])$ for $i\neq j$;
\item we need to add a new zero element $Zero$ rather than simply reusing the 
zeros $Opt_i (\zero)$ that are present in each of the components, which should be equivalent for all practical purposes.
\end{itemize}
The first issue is not a problem at all in practice, as the more precise dependent types we have erased guarantee that CHAD only ever adds (co)tangents in the same component, meaning that the error can never be trigerred in practice.
However, it requires us to do a manual logical relations proof of correctness.
This is the approach that is currently implemented in the reference Haskell implementation of CHAD.
The second issue is a minor inefficiency that can become more serious if (co)inductive types are built using this representation of coproducts.
This motivates approach 3.

Approach 3 addresses the second issue with approach 2 by removing the unnecessary extra element $Zero$ of the (co)tangent spaces.
To achieve this, however, the zeros $\zero$ at each type $\Dsyn{\ty}_2$ of tangent and $\Dsynrev{\ty}_2$ of cotangents need to be made functions $\zero:\Dsyn{\ty}_1\to\Dsyn{\ty}_2$ and  $\zero:\Dsynrev{\ty}_1\to\Dsynrev{\ty}_2$, rather than mere constant zeros.
Whenever the a zero is used by CHAD, it is called on the corresponding primal value that specifies in which component we want the zero to land.
While a mathematical formalization of this approach remains future work, we have shown this approach to work well in practice in an experimental Haskell implementation of CHAD.
As we plan to detail in an applied follow-up paper, this approach also gives an efficient way of applying CHAD to dynamically sized arrays.

\subsubsection{Linear inductive and coinductive types $\llfp{\ltvar}{\cty}$ and $\lgfp{\ltvar}{\cty}$}
As we have seen, linear coinductive types arise in reverse CHAD of inductive types 
as well as in forward CHAD of coinductive types.
Similarly, linear inductive types arise in reverse CHAD of coinductive types 
as well as in forward CHAD of inductive types.
It remains to be investigated how these can be best implemented.
However, as was the case for the implementation of copowers and linear sum types, we are hopeful that the concrete denotational semantics can guide us

Observe that all polynomials $F:\Vect\to \Vect$ are of the form $W\mapsto L(A)+W^n$, where $L\dashv U:\Set\to \Vect$ is the usual free-forgetful adjunction.
Therefore, $U\circ F=H\circ U$ for the polynomial $H:\Set\to\Set$ defined by $S\mapsto U(L(A))\times S^n$.
As the forgetful functor $F:\Vect\to \Set$ is monadic, it creates terminal coalgebras, hence hence $U(\nu F)=\nu H$.
This suggests that we might be able to implement $\lintrans{(\lgfp{\ltvar}{\cty})}$ as the plain coinductive type $\gfp{\tvar}{\lintrans{\cty}}$, where $\lintrans{\ltvar}\defeq \tvar$.

Similarly, we have that $F\circ L= L\circ E$ for the polynomial $E:\Set\to \Set $ defined by $E(X)=A \sqcup \bigsqcup_n X$.
Therefore, we have that $\mu F=L(\mu E)=(\mu E) \To \RR$.
This suggests that the implementation of linear inductive types might be achieved by "delinearizing" a polynomial $F$ to $E$, taking the initial algebra of $E$ and taking the function type to $\RR$.

We are hopeful that this theory will lead to a practical implementation, but the details remain to be verified.

\section{Related work}
\label{sec:related-work}
Automatic differentiation has long been studied by the scientific computing community.
In fact, its study goes back many decades with forward mode AD
being introduced by \citep{wengert1964simple}
and variants of reverse mode AD seemingly being reinvented several times, 
for example, by 
\citep{linnainmaa1970representation,speelpenning1980compiling}.
For brief reviews of this complex history and the basic ideas behind AD, 
we refer the reader to \citep{baydin2018automatic}.
For a more comprehensive account of the traditional work on AD, see the standard reference
text \citep{griewank2008evaluating}.

In this section, we focus, instead, on the more recent work that 
has proliferated since the programming languages community started 
seriously studying AD.
Their objectives are more closely aligned with those of the present paper.

\citep{pearlmutter2008reverse} is one of the early programming languages papers trying to extend the scope of 
AD from the traditional setting of first-order imperative languages to more expressive 
programming languages.
Specifically, this applied paper proposes a method to use reverse mode AD on an 
untyped higher-order functional language, through the use of an intricate source 
code transformation that employs ideas similar to defunctionalization.
It focuses on implementation rather than correctness or intended semantics.
\citep{alvarez2021functorial} recently simplified this code transformation and formalized its correctness.

Prompted by \citep{plotkin2018some}, there has, more recently, been a push in the 
programming language community to learn from \citep{pearlmutter2008reverse} and 
arrive at a definition of (reverse) AD as a source code 
transformation on expressive languages that should ideally be simple, semantically motivated and correct, compositional and efficient.

Among this work, \citep{wang2018demystifying} specifies and implements much simpler reverse AD transformation on a higher-order functional language with sum types.
The price they have to pay is that the transformation relies on the use of 
delimited continuations in the target language.

Various more theoretical works give formalizations and correctness proofs of reverse AD on expressive languages through the use of custom operational semantics.
\citep{abadi-plotkin2020} gives such an analysis for a first-order 
functional language with recursion, using an operational semantics that mirrors 
the runtime tracing techniques used in practice.
\citep{mak-ong2020} instead works with a total higher-order language that is a variant 
of the differential $\lambda$-calculus.
Using slightly different operational techniques, coming from linear logic, \citep{brunel2019backpropagation,mazza2021automatic} give an analysis of reverse AD on a simply typed $\lambda$-calculus and PCF.
Notably, \citep{brunel2019backpropagation} shows that their algorithm has 
the right complexity if one assumes a specific operational semantics for their linear $\lambda$-calculus with what 
they call a ``linear factoring rule''.
Very recently, \citep{krawiec2021} applied the idea of reverse AD through tracing to 
a higher-order functional language with variant types.
They implement the custom operational semantics as an evaluator and give a denotational correctness proof (using logical relations techniques similar to those of \citep{bcdg-open-logical-relations,hsv-fossacs2020}) as well as an asymptotic complexity proof 
about the full code transformation plus evaluator.

\citep{elliott2018simple} takes a different approach that is much closer to the 
present paper by working with a target language that is a plain functional language and does not depend on a custom operational semantics or an evaluator for traces.
Although this approach also naturally has linear types, it is a fundamentally 
different algorithm from that of \citep{brunel2019backpropagation,mazza2021automatic}:
for example, the linear types can be coarse-grained to plain simply typed code (e.g., Haskell) with the right computational complexity, even under the standard operational semantics of functional languages.
This is the approach that we have been referring to as CHAD.
Elliott's CHAD transformation, however, is restricted to a first-order functional language with tuples.
\citep{vytiniotis2019differentiable, vakar2020reverse} both present (the same) extensions of CHAD to apply to a higher-order functional source 
language, while still working with a functional target language. 
While \citep{vytiniotis2019differentiable} relates CHAD to the approach of 
\citep{pearlmutter2008reverse,alvarez2021functorial},
\citep{vakar2020reverse} and its extended version \citep{vakar2021chad}
give a (denotational) semantic foundation and correctness proof for CHAD, using 
a combination of logical relations techniques that \citep{bcdg-open-logical-relations,hsv-fossacs2020,huot2021higher} had previously used 
to prove correct (higher-order) forward mode AD together with the observation 
that AD can be understood through the framework of lenses or Grothendieck fibrations, which had 
previously been made by \citep{fong2019backprop,rev-deriv-cat2020}.
The present paper extends CHAD to further apply to source languages
with variant types and (co)inductive types.
To our knowledge, it is the first paper to consider reverse AD on languages with 
such expressive type systems.

\subsubsection*{Acknowledgments}                        
This project has received funding from the European Union’s Horizon 2020 research and innovation
programme under the Marie Skłodowska-Curie grant agreement No. 895827 and from the Nederlandse Organisatie voor Wetenschappelijk Onderzoek under NWO Veni grant number VI.Veni.202.124. 
This research was also supported through the programme ``Oberwolfach Leibniz Fellows'' by the Mathematisches Forschungsinstitut Oberwolfach in 2022, and partially supported by the CMUC, Centre for Mathematics of the University of Coimbra - UIDB/00324/2020, funded by the Portuguese Government through FCT/MCTES.

We thank Tom Smeding, Gordon Plotkin, Wouter Swierstra, Gabriele Keller, Ohad Kammar, Dimitrios Vytiniotis,  Patricia Johann, Michelle Pagani,     Michael Betancourt, Bob Carpenter, Sam Staton, Mathieu Huot, Curtis Chin Jen Sem and Amir Shaikhha for helpful discussions about topics related to the present work.

\clearpage
\bibliographystyle{msclike}
\bibliography{bibliography}

\clearpage

\appendix

\section{Pseudo-preterminal objects in Cat}\label{sec:appendix-pseudoterminal-in-Cat}
The appropriate $2$-dimensional analogous to preterminal objects are the pseudo-preterminal ones.
Namely, in the case of $\Cat $:
\begin{definition} 
An object $W$ in $\Cat $ is \textit{pseudo-preterminal} if the category of functors $\Cat \left[ X, W\right] $ is a groupoid for any object $X $ in $\Cat$.
\end{definition}  

Lemma \ref{lem:pseudopreterminal-Cat} establishes that the initial and terminal categories are, up to equivalence, the only pseudo-preterminal objects of $\Cat $.

\begin{lemma}[Pseudo-preterminal objects in $\Cat $]\label{lem:pseudopreterminal-Cat}
	Let $W$ be an object of $\Cat $. Assuming that $W$ is not the initial object of $\Cat $, the following statements are equivalent.
	\begin{enumerate}[i]
		\item The unique functor $W\to \terminal $ is an equivalence.\label{lem:pseudoterminal-wterminalequivalence}
		\item The projection $\pi _ W : W\times W \to W $ is an equivalence.\label{lem:pseudoterminal-projectionequivalence}
		\item The identity $\id _W : W\to W $ is naturally isomorphic to a constant functor $c: W\to W $.\label{lem:pseudoterminal-widentityisomorphictoconstant}
		\item If $f, g: X\to W $ are functors, then there is a natural isomorphism $f\cong g $ (that is to say, $W $ is pseudo-preterminal).\label{lem:pseudoterminal-pseudoterminal}
	\end{enumerate}
\end{lemma}	
\begin{proof}
	Assuming \eqref{lem:pseudoterminal-wterminalequivalence}, denoting by $t:W\to \terminal $
	the unique functor, we have that $\pi _W $ is the composition $W\times W \xto{\id_W \times t} W\times\terminal \cong W  $. Hence, since $\id _W $ and $t $ are equivalences, we conclude that $\pi _W $ is an equivalence.
	This proves that \eqref{lem:pseudoterminal-wterminalequivalence} $\Rightarrow $ \eqref{lem:pseudoterminal-projectionequivalence}.
	
	Given any constant functor 
	$c: W\to W $, we have that $\left( \id _W , c\right) : W\to W\times W $ and the diagonal functor $\left( \id _W , \id _W \right) : W\to W\times W $  are such that
	$\pi _W \circ \left( \id _W , c\right) = \id _ W $ and $\pi _W \circ \left( \id _W , \id _ W\right) = \id _ W $. Hence, assuming \eqref{lem:pseudoterminal-projectionequivalence}, we have that 
	$ \left( \id _W , c\right) $ and $ \left( \id _W , \id _ W\right) $  are inverse equivalences of $ \pi _W $. Thus we have a natural isomorphism 
	$ \left( \id _W , c\right) \cong \left( \id _W , \id _ W\right) $
	which implies that 
	$$c\cong \pi _2 \circ \left( \id _W , c\right) \cong \pi _2 \circ \left( \id _W , \id _W \right) \cong \id _ W. $$
	This proves that  \eqref{lem:pseudoterminal-projectionequivalence} $ \Rightarrow $ \eqref{lem:pseudoterminal-widentityisomorphictoconstant}.
	
	Assuming \eqref{lem:pseudoterminal-widentityisomorphictoconstant}, if $ f, g : X\to W $ are functors, we have the natural isomorphisms 
	$$ f = \id _W \circ f \cong c \circ f = c\circ g \cong \id _W \circ g = g .$$
	This shows that \eqref{lem:pseudoterminal-widentityisomorphictoconstant} $ \Rightarrow $ \eqref{lem:pseudoterminal-pseudoterminal}.
	
	Finally, assuming \eqref{lem:pseudoterminal-pseudoterminal}, we have that,
	given any functor $c : \terminal \to W $, the composition
	$ W\to\terminal \xto{c} W $
	is naturally isomorphic to the identity. Hence $W\to \terminal $ is an equivalence.
	This shows that \eqref{lem:pseudoterminal-pseudoterminal} $ \Rightarrow $ \eqref{lem:pseudoterminal-wterminalequivalence}.
\end{proof}

\begin{remark}
	The equivalence \eqref{lem:pseudoterminal-projectionequivalence} $\Leftrightarrow $ \eqref{lem:pseudoterminal-pseudoterminal}  holds for the general context of any $2$-category.
	The other equivalences mean that $\terminal $ and $\initial $ are, up to equivalence, the unique pseudo-preterminal objects of $\Cat$. The reader might compare the result, for instance, with the characterization of contractible spaces in basic homotopy theory.
\end{remark} 	

\section{CHAD transformation without sharing between primal and (co)tangents}
\label{sec:inefficient-ad-transformation}
In this section, we list the CHAD program transformations $\Dsyn{\Gamma}_1\vdash \Dsyn[\vGamma]{\trm}_1:\Dsyn{\ty}$, $\Dsyn{\Gamma}_1;\lvar:\Dsyn{\Gamma}_2\vdash\Dsyn[\vGamma]{\trm}_2:\subst{\Dsyn{\ty}_2}{\sfor{p}{\Dsyn[\vGamma]{\trm}_1}}$, $\Dsynrev{\Gamma}_1\vdash \Dsynrev[\vGamma]{\trm}_1:\Dsynrev{\ty}$ and $\Dsynrev{\Gamma}_1;\lvar:\subst{\Dsynrev{\ty}_2}{\sfor{p}{\Dsyn[\vGamma]{\trm}_1}}\vdash\Dsynrev[\vGamma]{\trm}_2:\Dsynrev{\Gamma}_2$ of a program $\Gamma\vdash\trm:\ty$ that keep the primals and (co)tangents separate
without sharing computation.
We advise against implementing these, due to
\begin{enumerate}
    \item the code explosion they can result in, leading to a potentially large code size and compilation times;
    \item the lack of sharing of computation they can result in, leading to poor runtime performance.
\end{enumerate}

\sqsubsection{Forward-mode AD}
\begin{align*}
    &\Dsyn[\vGamma]{\op(\trm_1,\ldots,\trm_k)}_1\defeq \letin{\var_1}{\Dsyn[\vGamma]{\trm_1}_1}{\cdots\letin{\var_k}{\Dsyn[\vGamma]{\trm_k}_1}{\op(\var_1,\ldots,\var_k)}}\\&
    \Dsyn[\vGamma]{\var}_1 \defeq \var\\&
    \Dsyn[\vGamma]{\letin{\var}{\trm}{\trm[2]}}_1 \defeq \letin{\var}{\Dsyn[\vGamma]{\trm}_1}{\Dsyn[\vGamma,\var]{\trm[2]}_1} \\&
    \Dsyn[\vGamma]{\tUnit}_1  \defeq \tUnit\\&
    \Dsyn[\vGamma]{\tPair{\trm}{\trm[2]}}_1 \defeq 
    \tPair{\Dsyn[\vGamma]{\trm}_1}{\Dsyn[\vGamma]{\trm[2]}_1}\\&
    \Dsyn[\vGamma]{\tFst(\trm)}_1  \defeq \tFst(\Dsyn[\vGamma]{\trm}_1)\\&
    \Dsyn[\vGamma]{\tSnd(\trm)}_1  \defeq \tSnd(\Dsyn[\vGamma]{\trm}_1)\\&
    \Dsyn[\vGamma]{\fun{\var}{\trm}}_1   \defeq
    \fun{\var}{\tPair{\Dsyn[\vGamma,\var]{\trm}_1}{\lfun{\lvar}{\letin{\lvar}{\tPair{\zero}{\lvar}}{\Dsyn[\vGamma,\var]{\trm}_2}}}}\\&
    \Dsyn[\vGamma]{\trm\,\trm[2]}_1  \defeq \tFst(\Dsyn[\vGamma]{\trm}_1\,\Dsyn[\vGamma]{\trm[2]}_1)\\&
    \Dsyn[\vGamma]{\Cns\trm}_1   \defeq \Cns(\Dsyn[\vGamma]{\trm}_1)\\&
    \Dsyn[\vGamma]{\vMatch{\trm}{\Cns_1\var_1\To\trm[2]_1\mid\cdots\mid \Cns_n \var_n\To\trm[2]_n}}_1   \defeq\\& \qquad \vMatch{\Dsyn[\vGamma]{\trm}_1}{\Cns_1\var_1\To\Dsyn[\vGamma,\var_1]{\trm[2]_1}_1\mid\cdots\mid \Cns_n \var_n\To\Dsyn[\vGamma,\var_n]{\trm[2]_n}_1}\\&
    \Dsyn[\vGamma]{\tRoll\trm}_1 \defeq \tRoll\Dsyn[\vGamma]{\trm}_1\\&
    \Dsyn[\vGamma]{\tFold\trm\var{\trm[2]}}_1 \defeq \tFold{\Dsyn[\vGamma]{\trm}_1}{\var}{\Dsyn[\var]{\trm[2]}_1}\\&
    \Dsyn[\vGamma]{\tGen\trm\var{\trm[2]}}_1 \defeq \tGen{\Dsyn[\vGamma]{\trm}_1}{\var}{\Dsyn[\var]{\trm[2]}_1}\\&
    \Dsyn[\vGamma]{\tUnroll\trm}_1 \defeq \tUnroll\Dsyn[\vGamma]{\trm}_1\\&
    \\&
    \Dsyn[\vGamma]{\op(\trm_1,\ldots,\trm_k)}_2\defeq \letin{\var_1}{\Dsyn[\vGamma]{\trm_1}_1}{\cdots\letin{\var_k}{\Dsyn[\vGamma]{\trm_k}_1}{\\
     & \phantom{\Dsyn[\vGamma]{\op(\trm_1,\ldots,\trm_k)}_2\defeq}D\op(\var_1,\ldots,\var_k;\tTriple{\lapp{\Dsyn[\vGamma]{\trm_1}_2}{\lvar}}{\ldots}{\lapp{\Dsyn[\vGamma]{\trm_k}}_2{\lvar}})}}\\&
    \Dsyn[\vGamma]{\var}_2 \defeq \tProj{\idx{\var}{\vGamma}}(\lvar)\\&
    \Dsyn[\vGamma]{\letin{\var}{\trm}{\trm[2]}}_2 \defeq \letin{\var}{\Dsyn[\vGamma]{\trm}_1}{\letin{\lvar}{\tPair{\lvar}{\Dsyn[\vGamma]{\trm}_2}}{\Dsyn[\vGamma,\var]{\trm[2]}_2}} \\&
    \Dsyn[\vGamma]{\tUnit}_2  \defeq \tUnit\\&
    \Dsyn[\vGamma]{\tPair{\trm}{\trm[2]}}_2 \defeq 
    \tPair{\Dsyn[\vGamma]{\trm}_2}{\Dsyn[\vGamma]{\trm[2]}_2}\\&
    \Dsyn[\vGamma]{\tFst(\trm)}_2  \defeq \tFst(\Dsyn[\vGamma]{\trm}_2)\\&
    \Dsyn[\vGamma]{\tSnd(\trm)}_2  \defeq \tSnd(\Dsyn[\vGamma]{\trm}_2)\\&
    \Dsyn[\vGamma]{\fun{\var}{\trm}}_2   \defeq
    \fun{\var}{\letin{\lvar}{\tPair{\lvar}{\zero}}{\Dsyn[\vGamma,\var]{\trm}_2}}\\&
    \Dsyn[\vGamma]{\trm\,\trm[2]}_2  \defeq 
    \letin{\var[2]}{\Dsyn[\vGamma]{\trm[2]}_1}{
    \Dsyn[\vGamma]{\trm}_2\, \var[2]+\lapp{(\tSnd(\Dsyn[\vGamma]{\trm}_1\,\var[2]))}{\Dsyn[\vGamma]{\trm[2]}_2}}\\&
    \Dsyn[\vGamma]{\Cns\trm}_2   \defeq \Dsyn[\vGamma]{\trm}_2\\&
    \Dsyn[\vGamma]{\vMatch{\trm}{\Cns_1\var_1\To\trm[2]_1\mid\cdots\mid \Cns_n \var_n\To\trm[2]_n}}_2   \defeq\\& \qquad 
    \letin{\lvar}{\tPair{\lvar}{\Dsyn[\vGamma]{\trm}_2}}{\vMatch{\Dsyn[\vGamma]{\trm}_1}{\Cns_1\var_1\To\Dsyn[\vGamma,\var_1]{\trm[2]_1}_2\mid \cdots 
    \mid \Cns_n\var_n\To\Dsyn[\vGamma,\var_n]{\trm[2]_n}_2}}\\&
    \Dsyn[\vGamma]{\tRoll\trm}_2 \defeq \tRoll\Dsyn[\vGamma]{\trm}_2\\&
    \Dsyn[\vGamma]{\tFold\trm\var{\trm[2]}}_2 \defeq  
     \tFold{\Dsyn[\vGamma]{\trm}_2}{\lvar}{
    \\&\qquad\letin{\var}{\tFold{\Dsyn[\vGamma]{\trm}_1}{\var}{\subst{\Dsyn{\ty}_1}{\sfor{\tvar}{\var\vdash \Dsyn[\var]{\trm[2]}_1}}}}{\Dsyn[\var]{\trm[2]}_2}
    }\\&
    \Dsyn[\vGamma]{\tGen{\trm}{\var}{\trm[2]}}_2 \defeq \tGen{\Dsyn[\vGamma]{\trm}_2}{\lvar}
    {\letin{\var}{\Dsyn[\vGamma]{\trm}_1}{\Dsyn[\var]{\trm[2]}_2} }
    \\&
    \Dsyn[\vGamma]{\tUnroll\trm}_2 \defeq \tUnroll \Dsyn[\vGamma]{\trm}_2
    \end{align*}

\sqsubsection{Reverse-mode AD}
\begin{align*}
    &\Dsynrev[\vGamma]{\op(\trm_1,\ldots,\trm_k)}_1 \defeq \letin{\var_1}{\Dsynrev[\vGamma]{\trm_1}}{\cdots\letin{\var_k}{\Dsynrev[\vGamma]{\trm_k}}{\op(\var_1)}}\\&
    \Dsynrev[\vGamma]{\var}_1\defeq\var  \\&  
    \Dsynrev[\vGamma]{\letin{\var}{\trm}{\trm[2]}}_1  
    \defeq \letin{\var}{\Dsynrev[\vGamma]{\trm}_1}{\Dsynrev[\vGamma,\var]{\trm[2]}_1}  \\&  
    \Dsynrev[\vGamma]{\tUnit}_1 \defeq \tUnit \\&  
    \Dsynrev[\vGamma]{\tPair{\trm}{\trm[2]}}_1 \defeq \tPair{\Dsynrev[\vGamma]{\trm}_1}{\Dsynrev[\vGamma]{\trm[2]}_1} \\&  
    \Dsynrev[\vGamma]{\tFst(\trm)}_1\defeq
    \tFst(\Dsynrev[\vGamma]{\trm}_1) \\&  
    \Dsynrev[\vGamma]{\tSnd(\trm)}_1\defeq
    \tSnd(\Dsynrev[\vGamma]{\trm}_1) \\&  
    \Dsynrev[\vGamma]{\fun{\var}{\trm}}_1 \defeq 
    \fun{\var}{\tPair{\Dsynrev[\vGamma,\var]{\trm}_1}
    {\lfun{\lvar}{\tSnd(\Dsynrev[\vGamma,\var]{\trm}_2)}}}
     \\&  
    \Dsynrev[\vGamma]{\trm\,\trm[2]}_1 \defeq
    \tFst(\Dsynrev[\vGamma]{\trm}_1\,\Dsynrev[\vGamma]{\trm[2]}_1) \\&  
    \Dsynrev[\vGamma]{\Cns\trm}_1\defeq\Cns(\Dsynrev[\vGamma]{\trm}_1) \\&  
    \Dsynrev[\vGamma]{\vMatch{\trm}{\Cns_1\var_1\To\trm[2]_1\mid\cdots\mid \Cns_n \var_n\To\trm[2]_n}}_1\defeq\\ &\qquad \vMatch{\Dsynrev[\vGamma]{\trm}_1}{\Cns_1\var_1\To\Dsynrev[\vGamma,\var_1]{\trm[2]_1}_1\mid\cdots\mid \Cns_n \var_n\To\Dsynrev[\vGamma,\var_n]{\trm[2]_n}_1} \\& 
    \Dsynrev[\vGamma]{\tRoll\trm}_1 \defeq \tRoll\Dsynrev[\vGamma]{\trm}_1\\& 
    \Dsynrev[\vGamma]{\tFold\trm\var{\trm[2]}}_1 \defeq \tFold{\Dsynrev[\vGamma]{\trm}_1}{\var}{\Dsynrev[\var]{\trm[2]}_1}\\&
    \Dsynrev[\vGamma]{\tGen\trm\var{\trm[2]}}_1 \defeq \tGen{\Dsynrev[\vGamma]{\trm}_1}{\var}{\Dsynrev[\var]{\trm[2]}_1}\\&
    \Dsynrev[\vGamma]{\tUnroll\trm}_1 \defeq \tUnroll\Dsynrev[\vGamma]{\trm}_1\\&
     \\&
     \Dsynrev[\vGamma]{\op(\trm_1,\ldots,\trm_k)}_2\defeq 
     \letin{\var_1}{\Dsynrev[\vGamma]{\trm_1}}{\cdots\letin{\var_k}{\Dsynrev[\vGamma]{\trm_k}}{\letin{\lvar}{\transpose{D\op}(\var_1,\ldots,\var_k;\lvar)}{}}}\\
     &\phantom{\Dsynrev[\vGamma]{\op(\trm_1,\ldots,\trm_k)}_2\defeq } (\letin{\lvar}{\tProj{1}{\lvar}}{\Dsynrev[\vGamma]{\trm_1}_2}) +\cdots+ (\letin{\lvar}{\tProj{1}{\lvar}}{\Dsynrev[\vGamma]{\trm_k}_2})\\&
    \Dsynrev[\vGamma]{\var}_2\defeq\tCoProj{\idx{\var}{\vGamma}}(\lvar) \\&  
    \Dsynrev[\vGamma]{\letin{\var}{\trm}{\trm[2]}}_2\defeq
    \letin{\var}{\Dsynrev[\vGamma]{\trm}_1}{
    \letin{\lvar}{\Dsynrev[\vGamma,\var]{\trm[2]}_2}{\tFst(\lvar)+
        \letin{\lvar}{\tSnd(\lvar)}{\Dsynrev[\vGamma]{\trm}_2}
    }   }  \\&  
    \Dsynrev[\vGamma]{\tUnit}_2 \defeq \zero \\&  
    \Dsynrev[\vGamma]{\tPair{\trm}{\trm[2]}}_2 \defeq (\letin{\lvar}{\tFst(\lvar)}{\Dsynrev[\vGamma]{\trm}_2})+(\letin{\lvar}{\tSnd(\lvar)}{\Dsynrev[\vGamma]{\trm[2]}_2}) \\&  
    \Dsynrev[\vGamma]{\tFst(\trm)}_2\defeq\letin{\lvar}{\tPair{\lvar}{\zero}}{\Dsynrev[\vGamma]{\trm}_2} \\&  
    \Dsynrev[\vGamma]{\tSnd(\trm)}_2\defeq\letin{\lvar}{\tPair{\zero}{\lvar}}{\Dsynrev[\vGamma]{\trm}_2} \\&  
    \Dsynrev[\vGamma]{\fun{\var}{\trm}}_2 \defeq
    \tensMatch{\lvar}{\var}{\lvar}{
    \tFst(\Dsynrev[\vGamma,\var]{\trm}_2)} \\&  
    \Dsynrev[\vGamma]{\trm\,\trm[2]}_2 \defeq 
    \letin{\var}{\Dsynrev[\vGamma]{\trm[2]}_1}{
    (\letin{\lvar}{!\var\otimes \lvar}{\Dsynrev[\vGamma]{\trm}_2})+ \\&  
    \qquad\qquad\qquad\quad\;\;
    (\letin{\lvar}{\lapp{(\tSnd(\Dsynrev[\vGamma]{\trm}_1\,\var))}{\lvar}}{\Dsynrev[\vGamma]{\trm[2]}_2})
    }
     \\&  
    \Dsynrev[\vGamma]{\Cns\trm}_2\defeq\Dsynrev[\vGamma]{\trm}_2 \\&  
    \Dsynrev[\vGamma]{\vMatch{\trm}{\Cns_1\var_1\To\trm[2]_1\mid\cdots\mid \Cns_n \var_n\To\trm[2]_n}}_2   \defeq\\& \qquad 
    \letin{\lvar}{
    \vMatch{\Dsynrev[\vGamma]{\trm}_1}{\Cns_1\var_1\To\Dsynrev[\vGamma,\var_1]{\trm[2]_1}_2\mid \cdots\mid
    \Cns_n\var_n\To\Dsynrev[\vGamma,\var_n]{\trm[2]_n}_2 }}{}\\
    &\qquad
    {\tFst\lvar+\letin{\lvar}{\tSnd\lvar}{\Dsynrev[\vGamma]{\trm}_2}}\\&
    \Dsynrev[\vGamma]{\tRoll\trm}_2 \defeq 
    \letin{\lvar}{\tUnroll\lvar}{\Dsynrev[\vGamma]{\trm}_2}\\ &
    \Dsynrev[\vGamma]{\tFold\trm\var{\trm[2]}}_2 \defeq
    \letin{\lvar}{\big(
    \tGen{\lvar}{\lvar}{
        \\&\qquad\letin{\var}{\tFold{\Dsynrev[\vGamma]{\trm}_1}{\var}{\subst{\Dsynrev{\ty}_1}{\sfor{\tvar}{\var\vdash \Dsynrev[\var]{\trm[2]}_1}}}}{\Dsynrev[\var]{\trm[2]}_2}
        }\big)}{\Dsynrev[\vGamma]{\trm}_2}\\ & 
        \Dsynrev[\vGamma]{\tGen{\trm}{\var}{\trm[2]}}_2 \defeq \letin{\lvar}{(\tFold{\lvar}{\lvar}
        {\letin{\var}{\Dsynrev[\vGamma]{\trm}_1}{\Dsynrev[\var]{\trm[2]}_2} })}{\Dsynrev[\vGamma]{\trm}_2}\\ &
    \Dsynrev[\vGamma]{\tUnroll\trm}_2 \defeq \letin{\lvar}{\tRoll\lvar}{\Dsynrev[\vGamma]{\trm}_2} 
    \end{align*}

\section{A Manual Proof of AD Correctness for Simply Typed Coproducts}\label{sec:simply-typed-coprod}
In many implementations of CHAD, we will not have access to dependent types.
Therefore, we need to give up a bit of type safety for AD on coproducts.
Here, we extend the applied, manual correctness proof of the applied CHAD implementation of \citep[Appendix A]{vakar2021chad}.

For coproducts, we have the following constructs in the source language:
\begin{align*}
\tInl \in \Syn(\ty,\ty\sqcup \ty[2])\\
\tInr \in \Syn(\ty[2],\ty\sqcup \ty[2])\\
[,]: \Syn(\ty,\ty[3])\times \Syn(\ty[2],\ty[3])\to \Syn(\ty\sqcup \ty[2], \ty[3]).
\end{align*}

\subsection*{Forward AD}
We can define 
\begin{align*}
    \Dsyn{\ty\sqcup \ty[2]}_1&\defeq \Dsyn{\ty}_1\sqcup \Dsyn{\ty[2]}_1\\
    \Dsyn{\ty\sqcup \ty[2]}_2&\defeq \Dsyn{\ty}_1\t* \Dsyn{\ty[2]}_1\\
    \Dsyn{\tInl}_1&\defeq \tInl\\
    \Dsyn{\tInl}_2 &\defeq \lfun\lvar\tPair{\lvar}{\zero}\\
    \Dsyn{\tInr}_1&\defeq \tInr\\
    \Dsyn{\tInr}_2 &\defeq \lfun\lvar\tPair{\zero}{\lvar}\\
    \Dsyn{[\trm,\trm[2]]}_1 & \defeq x\vdash \vMatch{x}{\tInl\,x\To \Dsyn{\trm}_1 | x\To \Dsyn{\trm[2]}_1}\\
    \Dsyn{[\trm,\trm[2]]}_2 & \defeq x\vdash \vMatch{x}{\tInr\,x\To \lfun\lvar\lapp{\Dsyn{\trm}_2}{(\tFst\lvar)} | x\To \lfun\lvar\lapp{\Dsyn{\trm[2]}_2}{(\tSnd\lvar)} }\\.
    \end{align*}
Then, we have that  
\begin{align*}
\Dsyn{\tInl}_1&\in \CSyn(\Dsyn{\ty}_1, \Dsyn{\ty}_1\sqcup  \Dsyn{\ty}_2)\\
\Dsyn{\tInl}_2& \in \CSyn(\Dsyn{\ty}_1,{\Dsyn{\ty}_2}\multimap{\Dsyn{\ty}_2\t*\Dsyn{\ty[2]}_2})\\
\Dsyn{\tInr}_1&\in \CSyn(\Dsyn{\ty[2]}_1, \Dsyn{\ty}_1\sqcup  \Dsyn{\ty}_2)\\
\Dsyn{\tInr}_2& \in \CSyn(\Dsyn{\ty[2]}_1,{\Dsyn{\ty[2]}_2}\multimap{\Dsyn{\ty}_2\t*\Dsyn{\ty[2]}_2})\\
\Dsyn{[\trm,\trm[2]]}_1 &\in \CSyn(\Dsyn{\ty}_1\sqcup \Dsyn{\ty[2]}_1, \Dsyn{\ty[3]}_1)\\
\Dsyn{[\trm,\trm[2]]}_2 &\in \CSyn(\Dsyn{\ty}_1\sqcup \Dsyn{\ty[2]}_1,{\Dsyn{\ty}_2\t*\Dsyn{\ty[2]}_2}\multimap{\Dsyn{\ty[3]}_2}).
\end{align*}

Then, we define the following semantics:
\begin{align*}
    \sem{\Dsyn{\ty\sqcup \ty[2]}_1}&\defeq \sem{\Dsyn{\ty}_1}\sqcup \sem{\Dsyn{\ty}_1}\\
    \sem{\Dsyn{\ty\sqcup \ty[2]}_2}&\defeq \sem{\Dsyn{\ty}_2}\times\sem{\Dsyn{\ty}_2}\\
    \sem{\Dsyn{\tInl}_1}&\defeq \inj{1}\\
    \sem{\Dsyn{\tInl}_2} &\defeq \_ \mapsto x\mapsto (x,0)\\
    \sem{\Dsyn{\tInr}_1}&\defeq \inj{2}\\
    \sem{\Dsyn{\tInr}_2} &\defeq \_\mapsto y\mapsto (0,y)\\
    \sem{\Dsyn{[\trm,\trm[2]]}_1}&\defeq [\sem{\Dsyn{\trm}_1}, \sem{\Dsyn{\trm[2]}_1}] \\ 
    \sem{\Dsyn{[\trm,\trm[2]]}_2}&\defeq [x\mapsto (x',\_)\mapsto \sem{\Dsyn{\trm}_2}(x)(x'), 
    y\mapsto (y',\_)\mapsto \sem{\Dsyn{\trm}_2}(y)(y')]\\     .
    \end{align*}
We define the forward AD logical relation $P_{\ty\sqcup \ty[2]}$ for coproducts on 
$$
(\RR\To(\sem{\ty}\sqcup \sem{\ty[2]}))\times( (\RR\To (\sem{\Dsyn{\ty}_1}\sqcup \sem{\Dsyn{\ty[2]}_1}))\times (\RR\To \RR\multimap (\sem{\Dsyn{\ty}_2}\times \sem{\Dsyn{\ty[2]}_2})))
$$
as 
\begin{align*}
&\set{(\inj{1}\circ f',(\inj{1}\circ g', x\mapsto x'\mapsto (h(x)(x'), 0)))\mid  (f',(g',h'))\in P_{\ty}}\cup \\
&\set{(\inj{2}\circ f',(\inj{2} \circ g', x\mapsto x'\mapsto (0,h(x)(x'))))\mid  (f',(g',h'))\in P_{\ty[2]}}.
\end{align*}
Then, clearly, $ \tInl$ and $\tInr$ respect this relation (almost by definition).
We verify that $[\trm,\trm[2]]$ also respects the relation provided that $\trm$ and $\trm[2]$ do.
Suppose that $(f, (g,h))\in P_{\ty\sqcup \ty[2]}$ and $(\sem{\trm}, (\sem{\Dsyn{\trm}_1}, \sem{\Dsyn{\trm}_2}))\in P_{\ty}$ and 
$(\sem{\trm[2]}, (\sem{\Dsyn{\trm[2]}_1}, \sem{\Dsyn{\trm[2]}_2}))\in P_{\ty[2]}$.
We have to show that 
\begin{align*}
(&[\sem{\trm},\sem{\trm[2]}]\circ f,\\
&\qquad([\sem{\Dsyn{\trm}_1}, \sem{\Dsyn{\trm[2]}_1}] \circ g, \\
&\qquad\;z\mapsto z'\mapsto [x\mapsto (x',\_)\mapsto \sem{\Dsyn{\trm}_2}(x)(x'), \\&\hspace{68pt}
    y\mapsto (y',\_)\mapsto \sem{\Dsyn{\trm}_2}(y)(y')](g(z))(h(z)(z'))
))\in P_{\sem{\ty[3]}}.
\end{align*}
Now, we have two cases:
\begin{itemize}
\item $(f,(g,h))=(\inj{1}\circ f',(\inj{1}\circ g', x\mapsto x'\mapsto (h'(x)(x'), 0)))$, for $(f',(g',h'))\in P_{\ty}$.
Then, 
\begin{align*}
    &([\sem{\trm},\sem{\trm[2]}]\circ f,\\
    &\;\qquad([\sem{\Dsyn{\trm}_1}, \sem{\Dsyn{\trm[2]}_1}] \circ g, \\
    &\;\qquad\;z\mapsto z'\mapsto [x\mapsto (x',\_)\mapsto \sem{\Dsyn{\trm}_2}(x)(x'), \\&\hspace{68pt}
        y\mapsto (y',\_)\mapsto \sem{\Dsyn{\trm}_2}(y)(y')](g(z))(h(z)(z'))
    ))=\\
    &(\sem{\trm}\circ f',(\sem{\Dsyn{\trm}_1}\circ g', z\mapsto z'\mapsto \sem{\Dsyn{\trm}_2}(g(z))(h(z)(z'))
    )),
    \end{align*}
    which is a member of $P_{\ty[3]}$ because $\trm$ respects the logical relation by assumption.
\item $(f,(g,h))=(\inj{2}\circ f',(\inj{2} \circ g', x\mapsto x'\mapsto (0,h'(x)(x'))))$ for $(f',(g',h'))\in P_{\ty[2]}$.
Then, 
\begin{align*}
    &([\sem{\trm},\sem{\trm[2]}]\circ f,\\
    &\;\qquad([\sem{\Dsyn{\trm}_1}, \sem{\Dsyn{\trm[2]}_1}] \circ g,\\
    &\;\qquad\;z\mapsto z'\mapsto [x\mapsto (x',\_)\mapsto \sem{\Dsyn{\trm}_2}(x)(x'), \\&\hspace{68pt}
        y\mapsto (y',\_)\mapsto \sem{\Dsyn{\trm}_2}(y)(y')](g'(z))(h'(z)(z'))
        ))=\\
    &(\sem{\trm[2]}\circ f',(\sem{\Dsyn{\trm[2]}_1}\circ g',z\mapsto z'\mapsto \sem{\Dsyn{\trm}_2}(g'(z))(h'(z)(z'))
    )),
    \end{align*}
    which is a member of $P_{\ty[3]}$ because $\trm[2]$ respects the logical relation by assumption.
\end{itemize}
It follows that our implementation of forward AD for coproducts is correct.

\subsection*{Reverse AD}
We can define 
\begin{align*}
    \Dsynrev{\ty\sqcup \ty[2]}_1&\defeq \Dsynrev{\ty}_1\sqcup \Dsynrev{\ty[2]}_1\\
    \Dsynrev{\ty\sqcup \ty[2]}_2&\defeq \Dsynrev{\ty}_1\t* \Dsynrev{\ty[2]}_1\\
    \Dsynrev{\tInl}_1&\defeq \tInl\\
    \Dsynrev{\tInl}_2 &\defeq \lfun\lvar\tFst\lvar\\
    \Dsynrev{\tInr}_1&\defeq \tInr\\
    \Dsynrev{\tInr}_2 &\defeq \lfun\lvar\tSnd\lvar\\
    \Dsynrev{[\trm,\trm[2]]}_1 & \defeq x\vdash \vMatch{x}{\tInl\,x\To \Dsynrev{\trm}_1 | x\To \Dsynrev{\trm[2]}_1}\\
    \Dsynrev{[\trm,\trm[2]]}_2 & \defeq x\vdash \vMatch{x}{\tInr\,x\,\To \lfun\lvar \tPair{\lapp {\Dsynrev{\trm}_2}{\lvar}}{\zero} | x\To \lfun\lvar\tPair{\zero}{\lapp{\Dsynrev{\trm[2]}_2}{\lvar}}}\\.
    \end{align*}
    Then, we have that  
    \begin{align*}
    \Dsynrev{\tInl}_1&\in \CSyn(\Dsynrev{\ty}_1, \Dsynrev{\ty}_1\sqcup  \Dsynrev{\ty}_2)\\
    \Dsynrev{\tInl}_2& \in \CSyn(\Dsynrev{\ty}_1,{\Dsynrev{\ty}_2\t*\Dsynrev{\ty[2]}_2}\multimap{\Dsynrev{\ty}_2})\\
    \Dsynrev{\tInr}_1&\in \CSyn(\Dsynrev{\ty[2]}_1, \Dsynrev{\ty}_1\sqcup  \Dsynrev{\ty}_2)\\
    \Dsynrev{\tInr}_2& \in \CSyn(\Dsynrev{\ty[2]}_1,{\Dsynrev{\ty}_2\t*\Dsynrev{\ty[2]}_2}\multimap{\Dsynrev{\ty[2]}_2})\\
    \Dsynrev{[\trm,\trm[2]]}_1 &\in \CSyn(\Dsynrev{\ty}_1\sqcup \Dsynrev{\ty[2]}_1, \Dsynrev{\ty[3]}_1)\\
    \Dsynrev{[\trm,\trm[2]]}_2 &\in \CSyn(\Dsynrev{\ty}_1\sqcup \Dsynrev{\ty[2]}_1,{\Dsynrev{\ty[3]}_2}\multimap{\Dsynrev{\ty}_2\t*\Dsynrev{\ty[2]}_2}).
    \end{align*}
Then, 
\begin{align*}
    \sem{\Dsynrev{\ty\sqcup \ty[2]}_1}&\defeq \sem{\Dsynrev{\ty}_1}\sqcup \sem{\Dsynrev{\ty}_1}\\
    \sem{\Dsynrev{\ty\sqcup \ty[2]}_2}&\defeq \sem{\Dsynrev{\ty}_2}\times\sem{\Dsynrev{\ty}_2}\\
    \sem{\Dsynrev{\tInl}_1}&\defeq \inj{1}\\
    \sem{\Dsynrev{\tInl}_2} &\defeq \_ \mapsto (x,\_)\mapsto x\\
    \sem{\Dsynrev{\tInr}_1}&\defeq \inj{2}\\
    \sem{\Dsynrev{\tInr}_2} &\defeq \_\mapsto (\_,y)\mapsto y\\
    \sem{\Dsynrev{[\trm,\trm[2]]}_1}&\defeq [\sem{\Dsynrev{\trm}_1}, \sem{\Dsynrev{\trm[2]}_1}] \\ 
    \sem{\Dsynrev{[\trm,\trm[2]]}_2}&\defeq [x\mapsto z'\mapsto (\sem{\Dsynrev{\trm}_2}(x)(z'),0), 
    y\mapsto z'\mapsto (0,\sem{\Dsynrev{\trm}_2}(y)(z'))]\\     .
    \end{align*}
We define the reverse AD logical relation $P_{\ty\\ty[2]}$ for coproducts on 
$$
(\RR\To(\sem{\ty}\sqcup \sem{\ty[2]}))\times( (\RR\To (\sem{\Dsynrev{\ty}_1}\sqcup \sem{\Dsynrev{\ty[2]}_1}))\times (\RR\To (\sem{\Dsynrev{\ty}_2}\times \sem{\Dsynrev{\ty[2]}_2})\multimap \RR))
$$
as 
\begin{align*}
&\set{(\inj{1}\circ f',(\inj{1}\circ g', z\mapsto (x',\_)\mapsto h'(z)(x')))\mid  (f',(g',h'))\in P_{\ty}}\cup \\
&\set{ (\inj{2}\circ f',(\inj{2} \circ g', z\mapsto (\_,y')\mapsto h'(z)(y')))\mid  (f',(g',h'))\in P_{\ty[2]}}.
\end{align*}
Then, clearly, $ \tInl$ and $\tInr$ respect this relation (almost by definition).
We verify that $[\trm,\trm[2]]$ also respects the relation provided that $\trm$ and $\trm[2]$ do.
Suppose that $(f, (g,h))\in P_{\ty\sqcup \ty[2]}$ and $(\sem{\trm}, (\sem{\Dsynrev{\trm}_1}, \sem{\Dsynrev{\trm}_2}))\in P_{\ty}$ and 
$(\sem{\trm[2]}, (\sem{\Dsynrev{\trm[2]}_1}, \sem{\Dsynrev{\trm[2]}_2}))\in P_{\ty[2]}$.
We have to show that 
\begin{align*}
(&[\sem{\trm},\sem{\trm[2]}]\circ f,\\
&\qquad([\sem{\Dsynrev{\trm}_1}, \sem{\Dsynrev{\trm[2]}_1}]\circ g, \\
&\qquad\;z\mapsto x'\mapsto h(z)([x\mapsto z'\mapsto (\sem{\Dsynrev{\trm}_2}(x)(z'),0), \\&\hspace{90pt}
y\mapsto z'\mapsto (0,\sem{\Dsynrev{\trm[2]}_2}(y)(z'))](g(x))(x'))
))\in P_{\sem{\ty[3]}}.
\end{align*}
Now, we have two cases:
\begin{itemize}
\item $(f,(g,h))=(\inj{1}\circ f',(\inj{1}\circ g', z\mapsto (x',\_)\mapsto h'(z)(x')))$, for $(f',(g',h'))\in P_{\ty}$.
Then, 
\begin{align*}
    &([\sem{\trm},\sem{\trm[2]}]\circ f,\\
&\;\qquad([\sem{\Dsynrev{\trm}_1}, \sem{\Dsynrev{\trm[2]}_1}]\circ g, \\
&\;\qquad\;z\mapsto x'\mapsto h(z)([x\mapsto z'\mapsto (\sem{\Dsynrev{\trm}_2}(x)(z'),0), \\&\hspace{92pt}
y\mapsto z'\mapsto (0,\sem{\Dsynrev{\trm[2]}_2}(y)(z'))](g(x))(x'))
))=\\
    &(\sem{\trm}\circ f',(\sem{\Dsynrev{\trm}_1}\circ g', z\mapsto x'\mapsto h'(z)(\sem{\Dsynrev{\trm}_2}
    (g'(x))(x')))),
    \end{align*}
    which is a member of $P_{\ty[3]}$ because $\trm$ respects the logical relation by assumption.
\item $(f,(g,h))=(\inj{2}\circ f',(\inj{2} \circ g', z\mapsto (\_,y')\mapsto h'(z)(y')))$ for $(f',(g',h'))\in P_{\ty[2]}$.
Then, 
\begin{align*}
    &([\sem{\trm},\sem{\trm[2]}]\circ f,\\
&\;\qquad([\sem{\Dsynrev{\trm}_1}, \sem{\Dsynrev{\trm[2]}_1}]\circ g, \\
&\;\qquad\;z\mapsto x'\mapsto h(z)([x\mapsto z'\mapsto (\sem{\Dsynrev{\trm}_2}(x)(z'),0), \\&\hspace{92pt}
y\mapsto z'\mapsto (0,\sem{\Dsynrev{\trm[2]}_2}(y)(z'))](g(x))(x'))
))=\\
    &(\sem{\trm[2]}\circ f',(\sem{\Dsynrev{\trm[2]}_1}\circ g', z\mapsto x'\mapsto h'(z)(\sem{\Dsynrev{\trm[2]}_2}
    (g'(x))(x')))),
    \end{align*}
    which is a member of $P_{\ty[3]}$ because $\trm[2]$ respects the logical relation by assumption.
\end{itemize}
It follows that our implementation of reverse AD for coproducts is correct.

A categorical way to understand this proof is that 
$(A_1,A_2)\sqcup (B_1,B_2)\defeq (A_1\sqcup B_1, A_2\times B_2)$ lifts the coproduct in $\catC$ to a 
\emph{weak} (fibred) coproduct in $\Sigma_\catC\catL$ and $\Sigma_\catC\catL^{op}$.
This weak coproduct lifts to the subscone, in the manner outlined above.
One consequence is that the AD transformations no longer respect the $\eta$-rule for coproducts
(unlike in the dependently typed setting).

\end{document}